\documentclass[sn-chicago]{sn-jnl}

\usepackage{graphicx}%
\usepackage{array,multirow}%
\usepackage{amsmath,amssymb,amsfonts}%
\usepackage{amsthm}%
\usepackage{mathrsfs}%
\usepackage{xcolor}%
\usepackage{textcomp}%
\usepackage{manyfoot}%
\usepackage{booktabs}%
\usepackage{natbib}
\usepackage{hyperref}

\newcommand{\bI}{\textbf{I}}
\newcommand{\bq}{\textbf{q}}
\newcommand{\bs}{\textbf{s}}
\newcommand{\bu}{\textbf{u}}

\newcommand{\bX}{\textbf{X}}
\newcommand{\by}{\textbf{y}}
\newcommand{\bY}{\textbf{Y}}

\newcommand{\bOmega}{\mbox{\boldmath $\Omega$}}

\newtheorem{theorem}{Theorem}

\begin{document}

\title[Bayesian joint quantile autoregression]{Bayesian joint quantile autoregression}

\author*[1]{\fnm{Jorge} \sur{Castillo-Mateo}}\email{jorgecm@unizar.es}
\author[2]{\fnm{Alan E.} \sur{Gelfand}}\email{alan@stat.duke.edu}
\author[1]{\fnm{Jes\'us} \sur{As\'in}}\email{jasin@unizar.es}
\author[1]{\fnm{Ana C.} \sur{Cebri\'an}}\email{acebrian@unizar.es}
\author[1]{\fnm{Jes\'us} \sur{Abaurrea}}\email{abaurrea@unizar.es}

\affil[1]{\orgdiv{Department of Statistical Methods}, \orgname{University of Zaragoza}, \orgaddress{\street{Pedro Cerbuna 12}, \city{Zaragoza}, \postcode{50009}, \country{Spain}}}
\affil[2]{\orgdiv{Department of Statistical Science}, \orgname{Duke University}, \orgaddress{\street{415 Chapel Dr}, \city{Durham}, \postcode{27705}, \state{NC}, \country{USA}}}

\abstract{Quantile regression continues to increase in usage, providing a useful alternative to customary mean regression.  Primary implementation takes the form of so-called \emph{multiple} quantile regression, creating a separate regression for each quantile of interest.  However, recently, advances have been made in \emph{joint} quantile regression, supplying  a quantile function which avoids crossing of the regression across quantiles.  Here, we turn to quantile autoregression (QAR), offering a fully Bayesian version.  We extend the initial quantile regression work of \cite{koenker2006} in the spirit of \cite{tokdar2012}.  We offer a directly interpretable parametric model specification for QAR.  Further, we offer a p-th order QAR(p) version, a multivariate QAR(1) version, and a spatial QAR(1) version.  We illustrate with simulation as well as a temperature dataset collected in Arag\'on, Spain.}



\maketitle

\section{Introduction}

For time series data, autoregressive (AR) modeling is perhaps the most common approach.  A lag one, AR(1), takes the form $Y_{t} = \mu + \rho(Y_{t-1} - \mu) +\epsilon_{t}$, with $\epsilon_{t}$ following a suitable zero-mean error distribution; a conditional mean is provided. By analogy, quantile autoregression (QAR) considers  conditional quantiles. 

An issue with quantile regression (QR) is the so-called quantile crossing problem.  Modeling quantiles individually enables rich modeling for a given quantile but allows for crossing of quantiles across quantile level $\tau$.  For arbitrary values of the regressors, $\bX$, we can not ensure that the resulting modeled quantiles will increase in $\tau$.  Such modeling is referred to as \emph{multiple} QR.  Inference typically proceeds by minimizing a check loss function or, more formally, assuming an asymmetric Laplace (AL) error term. Examples of multiple QR with AL errors appear in \cite{yu2001}, and \cite{kozumi2011} present a Gibbs sampler model fitting implementation.  \cite{lum2012}  work in the context of spatially referenced data and extend the AL model to a spatial process. \cite{castillo2022b} propose a very flexible spatial AL mixed effects QAR model.

Recent effort has focused on a \emph{joint} QR modeling to avoid quantile crossing.  Adopting restricted support for the regressors, $\bX$, the $\tau$-quantile will increase monotonically over $\tau \in (0,1)$. Foundational work appears in \cite{tokdar2012} using a Gaussian process (GP) with follow on work in \cite{das2017a} using splines. \cite{reich2011} developed a spatial joint QR model through spatially varying regression coefficients using Bernstein polynomials. \cite{yang2017} propose a novel parameterization that characterizes any collection of noncrossing quantile planes over arbitrarily shaped convex predictor domains. This parameterization was extended to spatial data by \cite{chen2021} through a copula process but a non-spatially varying quantile function results.  Joint modeling imposes strong restrictions on the class of permissible specifications; models outside of this class may be preferred.

\cite{koenker2006} offered an initial version of a joint QAR($p$) model.  Illustrating with $p = 1$, they consider
\begin{equation} \label{eq:process1}
	Y_t = \theta_0(U_t) + \theta_1(U_t) Y_{t-1},
\end{equation}
where $U_t$ is a sequence of IID standard uniform random variables. The $\theta$ functions, from $[0,1] \rightarrow \mathbb{R}$, need to be estimated. Provided that the right side of expression \eqref{eq:process1} is monotone increasing in $U_t$, the $\tau$ conditional quantile function of $Y_t$ given $y_{t-1}$ increases in $\tau$ and is
\begin{equation}\label{eq:model1}
    Q_{Y_t}(\tau \mid y_{t-1}) = \theta_0(\tau) + \theta_1(\tau) y_{t-1}.
\end{equation}
\cite{koenker2006} required both $\theta_0$ and $\theta_1$ to be strictly increasing  functions (referred to as comonotonicity).  Their suggested choices were $\theta_{0}(\tau) = \sigma\Phi^{-1}(\tau)$ with $\Phi$ the cdf of a standard normal distribution and $\theta_{1}(\tau) = \min\{\gamma_{0} + \gamma_{1}\tau, 1\}$ for $\gamma_0 \in (0,1)$ and $\gamma_1 > 0$.  If $y_{t-1} \geq 0$, comonotonicity ensures that $Q_{Y_{t}}(\tau \mid y_{t-1})$ will not cross as $\tau$ increases but under the restrictive assumption that the autoregression coefficient strictly increases in $\tau$.

Our contribution is to reconsider the work of \cite{koenker2006} in the context of \cite{tokdar2012}, providing flexible joint QAR modeling in a Bayesian framework. We characterize noncrossing QAR(1) also using two monotone curves, through a convenient class of cdf's. We note extension to the QAR($p$) model. We consider bivariate QAR, capturing dependence through a copula process. Then, for spatially referenced time series, we introduce  spatial dependence in the realizations and obtain spatially varying QAR's through spatially varying coefficients.

QAR models arise when time series are observed to display asymmetric dynamics; such data often appears in economic applications. \cite{koenker2006} show empirical applications of the QAR model to the USA unemployment rate and gasoline prices. Further examples in the literature consider dynamic additive quantile models, QR with cointegrated time series, and conditional quantiles with GARCH models. Applications include stock returns, house price returns, and gold prices.  See, e.g., \cite{li2015} or \cite{yang21} and references therein.

QR use is popular for climate data \citep[see][for an extensive review]{gao2017}. \cite{yang2018} propose a semiparametric QAR model including lagged data to develop quantile-based temperature extreme indices.  \cite{castillo2022b} use a rich QAR model to compare the effects of climate change in daily maximum temperature.  \cite{zhang2022evolutionary} use QR models conditional on the state of the previous observation time to predict short-term wind speed or velocity.

The outline of the paper is as follows. Section~\ref{sec:QAR1} provides a model characterization for the QAR(1) case.  Further, it offers explicit parametric model specifications, the resulting likelihood for Bayesian model fitting, some criteria for model assessment, and a simulation study. Section~\ref{sec:QARp} looks at the QAR($p$) case. Section~\ref{sec:MQAR1} considers the bivariate QAR(1) setting. Section~\ref{sec:SQAR1} develops a fully spatial version through the use of a Gaussian copula. Section~\ref{sec:temperature} employs time series of daily temperature data from 18 spatial locations to illustrate the previous four sections. Finally, Section~\ref{sec:summary} presents a brief summary and possibilities for future work.

\section{The QAR(1) case} \label{sec:QAR1}

\subsection{The support of the data}

For a noncrossing QAR specification we need to restrict the support of the time series data, $\{y_t^* : t=1,\ldots,T\}$, to a bounded interval on the real line.\footnote{This is the analogue of the restriction over the predictor domain in \cite{yang2017}.}  We take this interval to be $[0,1]$ and implement this by making a transformation of the data,
\begin{equation} \label{eq:trans}
    y_{t} = \frac{y_t^* - m}{M - m},
\end{equation}
where $m < \min y_t^*$ and $M > \max y_t^*$.  In fact, $m$ and $M$ are chosen such that $\min y_{t}$ is close to but above $0$ and $\max y_{t}$ is close to but below $1$.  This enables the most flexibility for the quantile function under our proposed QAR modeling and we offer an automatic selection approach below.  

Two points are important to note.  First, we can not take $m = \min y_{t}^*$ and $M = \max y_{t}^*$.  The data must be in the interior of the unit interval in order to enable distinct quantiles as $\tau$ varies across $(0,1)$.  Second, choosing $m$ and $M$ is merely a device for working on the unit interval.  There is no connection between these values and the potential practical support of the $y_t^*$'s.  Imposing bounding on the support is unavoidable for a valid linear specification of $Q_{Y_{t}}(\tau \mid y_{t-1})$ of the form $ \theta_0(\tau) + \theta_1(\tau) y_{t-1}$ because the only nonintersecting lines under unbounded support are parallel lines. 

A convenient ``automatic'' strategy for selecting $m$ and $M$ is as follows.  The idea recalls basic results from the theory of order statistics.  If we have $T$ independent observations from a uniform distribution on $(m,M)$, $\{y_{t}^* : t=1,\ldots,T\}$, then $[E(Y_{(1)}^*) - m]/(M-m) = 1/(T+1)$ and $[E(Y_{(T)}^*) - m]/(M-m) = T/(T+1)$. So we can say $y_{(1)}^* \approx (m T + M)/(T+1)$ and $y_{(T)}^* \approx (m + T M)/(T+1)$. This gives two equations in two unknowns to solve for $m$ and $M$.  We obtain  
\begin{equation} \label{eq:m}
    m = \frac{T y_{(1)}^* - y_{(T)}^*}{T-1}
    \qquad \text{and} \qquad
    M = \frac{T y_{(T)}^* - y_{(1)}^*}{T-1}.
\end{equation}
Of course the $Y_{t}^*$'s are not independent, they do not come from a distribution on a bounded interval, and marginally, we would not expect them to follow a uniform distribution on $(m,M)$.  We only implement a simple automatic bounding strategy.

\subsection{The model} \label{SecModel}

A straightforward characterization of the required monotonicity of the QAR lines is offered by the following result, inspired from \cite{tokdar2012}.

\begin{theorem}
An autoregressive specification of the form of \eqref{eq:model1} with $\theta_1(\tau) \in [-1,1]$ for $\tau \in [0,1]$, is monotonically increasing in $\tau$ for $Y_{t}$ taking values in $[0,1]$ and $y_{t-1} \in [0,1]$ if and only if
\begin{equation} \label{eq:model2}
    Q_{Y_{t}}(\tau \mid y_{t-1}) = 
    y_{t-1} \eta_1(\tau) + (1-y_{t-1}) \eta_2(\tau)
\end{equation}
where $\eta_1,\eta_2: [0,1] \longrightarrow [0,1]$ are monotonically increasing.
\end{theorem}

\begin{proof}
Any monotonicity obeying $Q_{Y_{t}}(\tau \mid y_{t-1})$ given by \eqref{eq:model1} can be expressed as \eqref{eq:model2} by taking $\eta_1(\tau) =  \theta_0(\tau) + \theta_1(\tau) = Q_{Y_{t}}(\tau \mid 1)$ and $\eta_2(\tau) = \theta_0(\tau) = Q_{Y_{t}}(\tau \mid 0)$.
For the converse, if $Q_{Y_{t}}(\tau \mid y_{t-1})$ is given by \eqref{eq:model2} then it must be monotonically increasing in $\tau$ for every $y_{t-1} \in [0,1]$ for which both $y_{t-1}$ and $1 - y_{t-1}$ are nonnegative. One can express such a $Q_{Y_{t}}(\tau \mid y_{t-1})$ by defining $\theta_0(\tau) = \eta_2(\tau)$ and $\theta_1(\tau) = \eta_1(\tau) - \eta_2(\tau) \in [-1,1]$. 
\end{proof}

If we focus on \eqref{eq:model2}, a model for functions $\eta_1$ and $\eta_2$, each from $[0,1] \rightarrow [0,1]$, induces a QAR(1) model over all valid QAR(1) specifications of $Q_{Y_{t}}(\tau \mid y_{t-1})$, provided the boundary conditions $Q_{Y_{t}}(0 \mid y_{t-1}) = 0$ and $Q_{Y_{t}}(1 \mid y_{t-1}) = 1$ for all $y_{t-1} \in [0,1]$ are satisfied. The above condition can be rewritten as $\eta_j(0) = 0$ and $\eta_j(1) = 1$ for $j=1,2$.  Next we show how to specify these two monotone functions.

\subsubsection{Specification for the two monotone curves}

Specifically, both $\eta_1(\cdot)$ and $\eta_2(\cdot)$ again must be strictly monotone from $[0,1] \rightarrow [0,1]$. A convenient class to work with are cdf's for continuous random variables with support $[0,1]$. In fact, a rich class would arise as probabilistic mixtures of such cdf's, leading to the general form $ \eta(\tau) = \sum_{k=1}^{K} \lambda_k F(\tau \mid \bOmega_{k})$ such that $\lambda_k \geq 0$, $\sum_k \lambda_k = 1$ and $F: [0,1] \rightarrow [0,1]$ is strictly increasing for any parameters $\bOmega_{k}$. 

A convenient class of $F$'s to work with are the cdf's of the two parameter \cite{kumaraswamy1980} distribution \citep[also known as the minimax distribution,][]{jones2009}. Specifically, the probability density function (pdf) and cdf are
\begin{equation} \label{eq:kum}
    f(x \mid a, b) = a b x^{a-1} (1 - x^a)^{b-1} 
    \quad \text{and} \quad
    F(x \mid a, b) = 1 - (1 - x^a)^b,
\end{equation}
where $x \in [0,1]$ and $a,b > 0$. The Kumaraswamy distributions are a family with behavior similar to the beta distribution.  However, for our purposes, they are much simpler to use especially in the context of simulation since the cdf and quantile function can be expressed in closed form, i.e.,  $Q(\tau \mid a,b) = (1-(1-\tau)^{1/b})^{1/a}$ where $\tau \in [0,1]$. The flexibility of the Kumaraswamy distributions is shown in Section~S1 of the Supplementary Information (SI) employing different combinations of parameters $(a,b)$.

To work with the mixture form for $\eta(\tau)$, we investigated two mixture strategies.  The first lets $K$ be small but assumes the $a$'s and $b$'s are unknowns.  The second lets $K$ be larger but adopts a fixed set of $a$'s and $b$'s, in the spirit of basis function forms.  Specifically, we consider $K$ Kumaraswamy distributions with medians $k/(K+1)$, respectively.  In the former, with $K=2$ we have a total of five parameters (two $a$'s, two $b$'s, and a $\lambda$) while in the latter, with $K=6$ again we have five parameters (five $\lambda$'s).  Increasing the number of ``basis'' components in the specification of the $\eta$'s need not provide better model performance.  From considerable simulation experience, model performance is very sensitive to the choice of parameters in the mixture components.  So, in the sequel, we work with $K=1$ or $2$ (QAR1K1 and QAR1K2, hereafter) and fit the $a$'s and $b$'s. As for priors, with $K=1$, we consider $\log a_1,\log b_1 \sim N(0, \sigma_{ab}^2)$ with $\sigma_{ab} = 3$, which gives a weak prior on the log-scale. With $K=2$, we consider $\lambda_1 \sim U(0, 1/2)$ and $\log a_1,\log a_2,\log b_1,\log b_2 \sim N(0, \sigma_{ab}^2)$ with $\sigma_{ab} = 1.5$.  Restricting $\lambda_1$ to $(0,1/2)$ avoids identification issues, while $\sigma_{ab}$ is taken smaller than in the $K=1$ case to penalize values of $a$'s and $b$'s too small or large. Values of $a$'s and $b$'s that are close to zero or very large can cause negligible numerical errors in the rootfinder to generate a numerical overflow in the likelihood (see Equations~\ref{eq:density} and \ref{eq:likelihood} below) and thus degeneracy.

\subsection{Likelihood evaluation and model fitting} \label{sec:likelihood}

An important feature of a valid joint specification of $Q_{Y_{t}}(\tau \mid y_{t-1})$ for all $\tau \in (0,1)$, following \cite{tokdar2012}, is that it uniquely defines the conditional response density given $y_{t-1} \in [0,1]$.  Specifically, this density is given by
\begin{equation} \label{eq:density}
    f_{Y_{t}}(y_{t} \mid y_{t-1}) = 
    \left. \frac{1}{\frac{d}{d\tau} Q_{Y_{t}}(\tau \mid y_{t-1})} \right\rvert_{\tau=\tau_{y_{t-1}}(y_{t})},
\end{equation}
where $\tau_{y_{t-1}}(y_{t})$ solves $y_{t} = y_{t-1}\eta_1(\tau) +  (1-y_{t-1}) \eta_2(\tau)$ in $\tau$ and is numerically approximated to arbitrary precision via a one-dimensional rootfinder. We implement the hybrid rootfinding algorithm combining the bisection method, the secant method, and inverse quadratic interpolation, called Brent's method \citep{brent1973}. Consequently, given the data at $t=1$, $y_1$, we can write a valid log-likelihood score
\begin{equation} \label{eq:likelihood}
    \ell(\bOmega \mid \by) =
     \sum_{t=2}^{T} \log f_{Y_{t}}(y_{t} \mid y_{t-1}) =
      - \sum_{t=2}^{T} \log\big\{y_{t-1}\dot{\eta}_1(u_t) +  (1-y_{t-1}) \dot{\eta}_2(u_t)\big\},
\end{equation}
where $u_t = \tau_{y_{t-1}}(y_{t})$, $\by^\top = (y_1,\ldots,y_T)$ are all of the observed data, $\bOmega$ are the model parameters, and the $\dot{\eta}$'s are the derivatives of the $\eta$'s. 

We implement an adaptive block-Metropolis sampler algorithm \citep{haario2001} to obtain Markov chain Monte Carlo (MCMC) samples from the posterior distribution of the parameters and to summarize the posterior distribution of the conditional quantile function.  Furthermore, with a posterior realization of the model parameters and a given value of $y_{t-1}$, we can use \eqref{eq:density} with discretization, to obtain a posterior realization of the density function that is driving the joint quantiles.  Averaging over these realizations provides the posterior mean of the density.

\subsection{Model comparison and simulation study}

Working within our parametric Bayesian framework, for any $\tau$, posterior samples of the model parameters, $\{\bOmega_{b}^{*} : b=1,\ldots,B\}$, produce posterior samples of the conditional quantile function for $Y_{t}$, $Q_{Y_{t}}(\tau \mid y_{t-1}; \bOmega_{b}^{*})$.  Essentially, for each $Y_{t}$ (with associated $y_{t-1}$) and any $\tau$, we obtain the posterior distribution of $Q_{Y_{t}}(\tau \mid y_{t-1}; \bOmega)$.  We use these posterior distributions along with the dataset, $\by$, to offer model assessment.  Specifically, we propose two metrics.  One, denoted by $\tilde{p}_{v}$, is created by comparing averaged estimates of $p_{t}(\tau) \equiv E[\textbf{1}(y_{t} < Q_{Y_{t}}(\tau \mid y_{t-1}; \bOmega)) \mid \by]$ with $\tau$.  The second, denoted by $\bar{R^1}$ is developed through estimation of $\Delta_{t}(\tau) \equiv \delta_{\tau}(y_{t} - E[Q_{Y_{t}}(\tau \mid y_{t-1}; \bOmega) \mid \by])$ where $\delta_{\tau}(u) = u(\tau - \textbf{1}(u<0))$ is the check loss function associated with the AL distribution.  Full details are provided in Section~S2 of the SI.

In Section~S3 of the SI we also present the results of a brief simulation study where the goals were (i) to illustrate parameter recovery under fitting for several models, (ii) to investigate model flexibility, i.e., performance when the sampling model is not the same as the fitting model, and (iii) to consider the effect of sample size with regard to (i) and (ii).

\section{The QAR(\texorpdfstring{$p$}{TEXT}) case} \label{sec:QARp}

We provide a straightforward extension of our joint QAR(1) model to the lag $p$ case. It is not a characterization of the QAR($p$) function of $Y_t$ but offers a flexible specification. In this regard, we obtain a form with some restrictions on the autoregressive coefficients but no constraints on the $y_t$’s beyond the bounded interval support. By interpreting $\eta_1(\tau)$ and $\eta_2(\tau)$ in \eqref{eq:model2} as the conditional quantiles of $Y_t$ at $y_{t-1} \in \{0,1\}$, we build a similar construction for an autoregressive process of order $p$ as follows. Define
\begin{equation} \small
\begin{aligned}
    Q_{Y_t}(\tau &\mid y_{t-1}, \ldots, y_{t-p}) \\
    &= (\eta_1(\tau),\ldots,\eta_{p+1}(\tau))
    \begin{pmatrix}
    0 & \pi_1 & 0 & \cdots & 0\\
    0 & 0 & \pi_2 & \cdots & 0\\
    \vdots & \vdots & \vdots & \ddots & \vdots &\\
    0 & 0 & 0 & \cdots & \pi_p\\
    1 & - \pi_1 & - \pi_2 & \cdots & -\pi_p\\
    \end{pmatrix}
    \begin{pmatrix}
    1\\
    y_{t-1}\\
    y_{t-2}\\
    \vdots\\
    y_{t-p}\\
    \end{pmatrix},
\end{aligned}
\end{equation}
where the functions $\eta_1,\ldots,\eta_{p+1}: [0,1] \rightarrow [0,1]$ are monotonically increasing and the weights $\pi_1,\ldots,\pi_p$ are such that $\pi_j \geq 0$ and $\sum_j \pi_j = 1$. It is easy to see that such $Q_{Y_t}(\tau \mid y_{t-1}, \ldots, y_{t-p})$ is monotonically increasing in $\tau \in [0,1]$ for every $y_{t-1}, \ldots, y_{t-p} \in [0,1]$.

In particular, for QAR(2), let $\tau,\pi \in [0,1]$.  Then, define 
\begin{equation} \small
    Q_{Y_t}(\tau \mid y_{t-1}, y_{t-2}) = 
    \pi y_{t-1} \eta_1(\tau) + 
    (1 - \pi) y_{t-2} \eta_2(\tau) +
    (1 - \pi y_{t-1} - (1 - \pi) y_{t-2}) \eta_3(\tau)
\end{equation}
where the three $\eta$ functions are all strictly increasing, using forms as above. Rewriting the expression as
\begin{equation} \small
    Q_{Y_t}(\tau \mid y_{t-1}, y_{t-2}) = 
    \eta_3(\tau) + 
    \pi (\eta_1(\tau) - \eta_3(\tau)) y_{t-1} + 
    (1 - \pi) (\eta_2(\tau) - \eta_3(\tau)) y_{t-2},
\end{equation}
both autoregressive coefficients belong to $[-1,1]$ and need not be increasing in $\tau$.  We fit this QAR(2) model to our real data in Section~\ref{sec:temp_qar2}.  In fact, we only attempt this with $K=1$ mixture components (seven parameters) to keep the model simple.  Further, the second autoregressive term results are not influential for our data. Also, we choose $\log a$'s and $\log b$'s to follow a $N(0, 1.5^2)$ prior and $\pi \sim U(0,1)$ as a noninformative prior for $\pi$.

\section{Multivariate QAR(1)} \label{sec:MQAR1}

Often a collection of dependent times series is gathered over a common time window.  For instance, our illustration below considers the dependent pairs $\{(y_{t}^{\text{max}},y_{t}^{\text{min}}):t=1,\ldots,T\}$, the daily maximum and minimum temperature for day $t$ at a site. In fact, the collection of time series might be spatially referenced (leading to a spatial copula model construction, as developed in the next section).  What we have is the quantile analogue of usual multivariate AR for time series.  Implementation using the class of joint QAR(1) models we have proposed has not appeared in the literature. Our interest is in the quantile function for each time series. We are asking about the amount of dependence between quantile levels regarding the marginal quantile functions.

Here, we illustrate with the bivariate case where we have two models each defined as in \eqref{eq:process1}, introducing dependence in the two time series by making the associated $U_t$'s dependent through $T-1$ IID $2$-dimensional Gaussian copulas.  This specification captures the acknowledged dependence between the pair of time series.  We postpone to Section~\ref{sec:SQAR1} the details of modeling using copulas; in particular, that section develops the form of the general $n$-dimensional joint density.  The only detail that we advance here is that the correlation matrix associated with the copulas contains $1$'s on the diagonal and $\rho$ on the off-diagonal, where $\rho \sim U(-1,1)$ measures the correlation between series.

Apart from introducing dependence through $U_{t}^{\text{max}}$ and $U_{t}^{\text{min}}$, we could introduce dependence in the $\eta$'s.  For instance, using Kumaraswamy cdf's, under the $K=1$ case, we consider the pairs $\log a_{j}^{\text{max}}$ and $\log a_{j}^{\text{min}}$ and the pairs $\log b_{j}^{\text{max}}$ and $\log b_{j}^{\text{min}}$ ($j=1,2$) to be bivariate normal.  In our data we found little or no correlation between the parameters of the two time series, so in subsequent analyzes we will consider them independent.  We do not pursue this case further here except to note the analogy with dependent responses in linear regression models. Introducing dependence through the $U_t$'s is analogous to introducing dependence through the errors in the linear regression while introducing dependence through the $\eta$'s is analogous to introducing dependence in the mean structure through shared parameters.

An example is presented in Section~\ref{sec:temp_mqar1}. Again, with $K=1$, this yields four $\eta$'s, i.e., four independent $\log a$'s and four independent $\log b$'s, each following a weak, say $N(0,3^2)$ prior, as well as the copula parameter.  As a by-product, we show the induced bivariate conditional pdf (arising from the bivariate analogue of Equation~\ref{eq:density}) for $(Y_{t}^{\text{max}}, Y_{t}^{\text{min}})$ with some choices for the $y_{t-1}$'s.

\section{Spatial QAR(1)} \label{sec:SQAR1}

In the spatial setting, we consider spatial point-referenced time series data. Here, $Y_t(\bs)$ denotes the observation for time $t=1,\ldots,T$ at location $\bs \in \mathcal{D}$, where $\mathcal{D} \subset \mathbb{R}^2$ is the study region.  We have a time series at each of the locations, $\{\bs_1,\ldots,\bs_n\}$, say, the locations of the monitoring stations. The joint spatial QAR model is given by
\begin{equation} \label{eq:SQAR}
    Y_t(\bs) = \theta_0(U_t(\bs);\bs) + \theta_1(U_t(\bs);\bs) Y_{t-1}(\bs),
\end{equation}
where the $\theta$ functions are quantile and spatially varying. \cite{chen2021} propose to model the spatial dependence of the realizations in a QR model using a spatial copula process. Generalizing it to our model, the vectors $(U_t(\bs_1),\ldots,U_t(\bs_n))^\top$ follow an independent copula distribution for every $t$.  

Supplementing \cite{chen2021}, in \eqref{eq:SQAR} we introduce spatially varying coefficients rather than global coefficients. As a consequence, we have dependence in the time series realizations as well as spatially varying quantile functions.

\subsection{Modeling spatial dependence}

\subsubsection{Spatially varying quantiles}

For the spatially varying coefficients, we consider only one Kumaraswamy cdf for each $\eta(\tau;\bs)$. In fact, at location $\bs$, let assume $\eta_j(\tau; \bs) = 1 - (1 - \tau^{a_j(\bs)})^{b_j(\bs)}$ with parameters $a_{j}(\bs)$ and $b_{j}(\bs)$ for $j=1,2$. Then, we introduce four independent GP's for the $a$'s and $b$'s on the log-scale. In particular, we model $\log a_j(\bs) \sim GP(a_j, \sigma_{a_j}^2 \rho(\bs,\bs';\phi_{a_j}))$ and $\log b_j(\bs) \sim GP(b_j, \sigma_{b_j}^2 \rho(\bs,\bs';\phi_{b_j}))$ where the $\rho(\bs, \bs';\phi)$'s are exponential correlation functions with $\phi$'s as corresponding decay parameters.

We take the $\phi$'s to be fixed values, according to the spatial scale, because it is usually difficult to estimate them from the data and typically interest focuses on the $\sigma^2$'s, the spatial uncertainties \citep{BCG}.  Specifically, we fix $\phi = 3 / d_{\text{max}}$, with $d_{\text{max}}$ the maximum distance between any pair of spatial locations.  Thus, the spatial GP's are only indexed by a mean and a variance parameter.  We choose the priors $a_j,b_j,\log \sigma_{a_j}^2,\log \sigma_{b_j}^2 \sim N(0,3^2)$ ($j=1,2$).

\subsubsection{The spatial copula process}

A copula is a multivariate cdf for which the marginal distribution of each variable is $U(0,1)$.  Copulas are used to model the dependence between random variables.  Particularly, Sklar's theorem \citep{sklar1959} states that any multivariate joint pdf can be written in terms of univariate marginal pdf's and a copula which describes the dependence structure between the variables.  

Gaussian spatial copulas enable computational advantages, e.g., ease of parameter estimation and scalability with sample size.  For a given correlation matrix $R$, the $n$-dimensional Gaussian copula function with parameter matrix $R$ becomes
\begin{equation}
    C_{\Phi}(\bu \mid R) = \Phi_{R}(\Phi^{-1}(u_1), \ldots, \Phi^{-1}(u_n)),
\end{equation}
where $\bu^\top = (u_1,\ldots,u_n) \in [0,1]^n$, $\Phi_{R}$ is the joint cumulative distribution function of a multivariate normal distribution with zero-mean vector and covariance matrix $R$.  According to \cite{song2000}, the associated copula density is
\begin{equation} \label{eq:copula}
    c_{\Phi}(\bu \mid R) =
    |R|^{-1/2} 
    \exp\left\{\frac{1}{2} \bq^\top (\bI_n - R^{-1}) \bq \right\},
\end{equation}
with $\bq^\top = (\Phi^{-1}(u_1), \ldots, \Phi^{-1}(u_n))$.

With regard to the copula model for \eqref{eq:SQAR}, we take the processes $U_t(\bs)$'s to follow a Gaussian copula for each $t$, induced by a stationary spatial GP.  In the spirit of \cite{chen2021}, we define 
\begin{equation} \label{eq:U}
\begin{gathered}
    U_t(\bs) = \Phi(Z_t(\bs)),\quad
    Z_t(\bs) = W_t(\bs) + \epsilon_t(\bs),\\
    W_t(\bs) \sim GP(0, \gamma \rho(\bs,\bs'; \phi)),\quad \epsilon_t(\bs) \sim \text{IID } N(0, 1 - \gamma).
\end{gathered}
\end{equation}
The process $W_t(\bs)$ captures spatial dependence while $\epsilon_t(\bs)$ is independent pure error. The parameter $\gamma \in [0,1]$ determines the proportion of spatial and independent variation. When $\gamma = 1$, the specification for $Z_{t}(\bs)$ is purely spatial. When $\gamma = 0$, we have an independent noise model.  With this approach, the Gaussian copula density has correlation matrix $R \equiv \gamma R(\phi) + (1 - \gamma) \bI_{n}$ where $R(\phi)$ is the $n \times n$ correlation matrix induced by $\rho(\bs,\bs';\phi)$.  To address the final copula piece of our model, we fix $\phi$ as above, and adopt $\gamma \sim U(0,1)$ as a noninformative prior for $\gamma$.

\subsection{Likelihood evaluation}

We are interested in the likelihood under model \eqref{eq:SQAR} using \eqref{eq:copula} and \eqref{eq:U}.  It is convenient to first obtain the joint distribution for $\bY^\top = (\bY_1^\top,\ldots,\bY_T^\top)$ where $\bY_t^\top = (Y_t(\bs_1), \ldots, Y_t(\bs_n))$, $t=1,\ldots,T$.  That is, each $\bY_t$ is $n \times 1$ and $\bY$ is $Tn \times 1$.  By Sklar’s theorem, the joint conditional density of responses, $\bY$, given initial time's data, $\by_1$, can be partitioned into a marginal part and a copula part,
\begin{equation} \small
\begin{aligned}
f_{\bY}(\by \mid \by_1) = \prod_{t=2}^{T} \Bigg[& \prod_{i=1}^{n} f_{Y_t(\bs_i)}\big(y_{t}(\bs_i) \mid y_{t-1}(\bs_i)\big) \Bigg. \\ &\times \Bigg. c_{\Phi}\big(F_{Y_t(\bs_1)}(y_{t}(\bs_1) \mid y_{t-1}(\bs_1)), \ldots, F_{Y_t(\bs_n)}(y_{t}(\bs_n) \mid y_{t-1}(\bs_n))\big)\Bigg],
\end{aligned}    
\end{equation}
where the cdf $F_{Y_t(\bs_i)}$ corresponds to the pdf $f_{Y_t(\bs_i)}$ and $c_{\Phi}$ is the Gaussian copula density in \eqref{eq:copula}. As in Section~\ref{sec:likelihood}, we evaluate $f_{Y_t(\bs_i)}$ and $F_{Y_t(\bs_i)}$ using:
\begin{equation}
\begin{aligned}
    f_{Y_t(\bs_i)}(y_{t}(\bs_i) \mid y_{t-1}(\bs_i)) &= 
    \left. \frac{1}{\frac{d}{d\tau} Q_{Y_{t}(\bs_i)}(\tau \mid y_{t-1}(\bs_i))} \right\rvert_{\tau=\tau_{y_{t-1}(\bs_i)}(y_{t}(\bs_i))}, \\
    F_{Y_t(\bs_i)}(y_{t}(\bs_i) \mid y_{t-1}(\bs_i)) &= \tau_{y_{t-1}(\bs_i)}(y_{t}(\bs_i)),
\end{aligned}
\end{equation}
where $\tau_{y_{t-1}(\bs_i)}(y_{t}(\bs_i))$ solves $y_t(\bs_i) = y_{t-1}(\bs_i)\eta_1(\tau;\bs_i) +  (1-y_{t-1}(\bs_i)) \eta_2(\tau;\bs_i)$ in $\tau$. Then, the log-likelihood score of the model parameters $\bOmega$ can be expressed by
\begin{equation}
\begin{aligned}
     \ell(\bOmega \mid \by) =  
     \sum_{t=2}^{T} \Bigg[ &- \sum_{i=1}^{n} \log\big\{y_{t-1}(\bs_i)\dot{\eta}_1(u_t(\bs_i);\bs_i) + (1-y_{t-1}(\bs_i)) \dot{\eta}_2(u_t(\bs_i);\bs_i)\big\} \Bigg. \\ & \Bigg. + \log
     c_{\Phi}\left(u_t(\bs_1),\ldots,u_t(\bs_n) \mid R\right) \Bigg],
\end{aligned}
\end{equation}
with $u_t(\bs_i) = \tau_{y_{t-1}(\bs_i)}(y_{t}(\bs_i))$.  
Finally, note that, for the calculation of the log-likelihood, the value of the $u_t(\bs_i)$'s must be solved for, so the number of rootfinders needed at each iteration of the MCMC is $n(T-1)$.  As a result, likelihood evaluation is expensive, leading to long MCMC run times.

\subsection{Spatial interpolation}

The quantile $Q_{Y_t(\bs)}(\tau \mid y_{t-1}(\bs))$ is a function of process realizations. Posterior samples for the hyperparameters are available from the model fitting. Posterior samples for the GP's are available, using posterior samples of the hyperparameters, through usual Bayesian kriging \citep{BCG}. This yields prediction of $a_j(\bs_0)$ and $b_j(\bs_0)$ ($j=1,2$) at a new $\bs_0 \in \mathcal{D}$, enabling spatially varying quantile functions.  Therefore, we can interpolate conditional quantiles to any desired location in the study region given any proposed or reference value for the previous day's temperature at that location. If we do this over a sufficiently spatially resolved grid, we can obtain the posterior mean at each point and show the posterior $\tau$ conditional quantile surface for the given day.

\section{Application to temperature data} \label{sec:temperature}

\subsection{The data} \label{sec:temp_data}

We illustrate the proposed modeling methods with analyses of persistence in point-referenced daily maximum temperatures ($^\circ\text{C}$) from $n=18$ locations in Arag\'on, Spain. We bring in daily minimum temperatures for the bivariate QAR(1) case. The data is provided by the State Meteorological Agency (AEMET, in its Spanish acronym).  \cite{castillo2022a} provide exploratory analysis and spatial hierarchical modeling for this dataset. We analyze responses for 2015, an interesting year because the summer was especially hot in Europe \citep{dong201612}.  There were numerous locations with record-breaking temperatures in July 2015 and the heat was maintained over time.  The monthly average value of temperatures was a record in July 2015 for 6 of the 18 locations and in the entire region it was among the 10 hottest monthly averages. We restrict analysis to observations from May, June, July, August, and September (denoted as MJJAS), i.e., the hottest months of the year, resulting in $T=153$ days.  The location of the  18 observatories is shown in Figure~\ref{fig:map} and their time series in Figure~S4 of the SI.    

\begin{figure}[!t]
\centering
 \includegraphics[width=12cm]{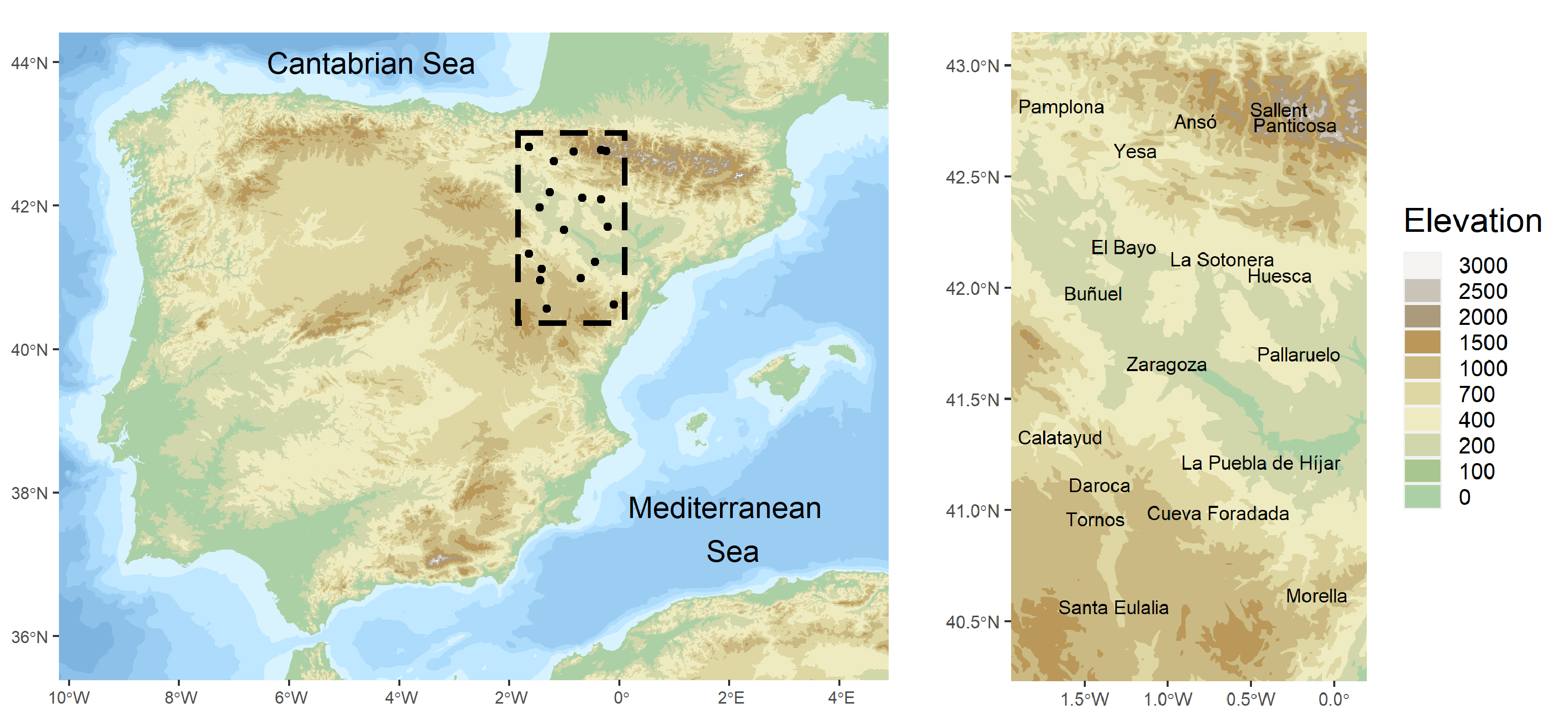}
 \caption{Location of the 18 sites around Arag\'on, northeastern Spain.}
 \label{fig:map}
\end{figure}

We begin with a model comparison using QAR(1) and QAR(2) models for all locations. Then, we analyze two illustrative locations within the region, Pamplona and Zaragoza, separately. Subsequently, we implement the bivariate QAR(1) model to the daily maximum and minimum temperature series in Zaragoza. Finally, we implement the general model for spatial QAR(1) with all the locations.  Before model fitting, we scale each of the temperature time series to $(0,1)$ using the transformation in \eqref{eq:trans} with $m$ and $M$ in \eqref{eq:m}. We adopt site-level values for $m$ and $M$.

\subsection{The QAR(1) case} \label{sec:temp_qar1}

Table~\ref{tab:measurements} shows, averaged across locations, the metrics of model adequacy $\tilde{p}_2$ and model comparison $\bar{R}^1$ defined in Section~S2 from the SI for the QAR1K1 and QAR1K2 models, and the model from \cite{koenker2006} fitted under our Bayesian framework using the density in \eqref{eq:density}. Table~S6 in the SI shows the metrics for each location. For this latter model, denoted as KX2006, we also consider a location parameter $\mu$ in the intercept, i.e., $\theta_{0}(\tau) = \mu + \sigma\Phi^{-1}(\tau)$ and $\theta_{1}(\tau) = \min\{\gamma_{0} + \gamma_{1}\tau, 1\}$ for $\gamma_0 \in (0,1)$ and $\gamma_1 > 0$. With KX2006 we work on the original scale of the data since they are all positive. We choose the priors $\mu \sim N(0,10^2)$, $\log \sigma, \log \gamma_1 \sim N(0, 3^2)$, and $\gamma_0 \sim U(0,1)$.  The $\bar{R}^1$ does not discriminate much between the proposed models, i.e., the autoregressive term explains much more variability than the difference in specification between the models.  However, our proposed models have a slightly higher performance, $0.365$, than the KX2006 model, around $0.34$.  Also, the $\tilde{p}_2$ directly measures how well the quantiles are captured, and its discriminative capacity is much higher. While QAR1K1 obtains a value of $0.633$, adding a second component to the mixing improves this measure to $0.402$. For its part, the KX2006 model obtains the worst value, $0.683$, indicating an overall poorer fitting of the quantiles.

\begin{table}[!t]
\caption{Adequacy and comparison metrics for QAR1K1, QAR1K2, QAR2K1, and KX2006 models averaged across locations.} \label{tab:measurements}
\begin{tabular}{@{}cc|cc|cc|cc@{}}
\toprule
 \multicolumn{2}{c|}{QAR1K1} & \multicolumn{2}{c|}{QAR1K2} & \multicolumn{2}{c|}{QAR2K1} & \multicolumn{2}{c}{KX2006} \\
 $\tilde{p}_2$ & $\bar{R}^1$ & $\tilde{p}_2$ & $\bar{R}^1$ & $\tilde{p}_2$ & $\bar{R}^1$ & $\tilde{p}_2$ & $\bar{R}^1$ \\ 
 \midrule
 0.633&0.365&0.402&0.365&0.542&0.365&0.683&0.339 \\
 \botrule
\end{tabular}
\end{table}

Figure~\ref{fig:theta} shows the posterior mean of the functions $\theta_0$ and $\theta_1$ in Pamplona and Zaragoza for the models QAR1K1 (dashed) and QAR1K2 (solid).  Note that we could recover the intercepts on the original scale as $\theta_0^*(\tau) = m (1-\eta_1(\tau)) + M \eta_2(\tau)$ and the autoregressive coefficients remain invariant.  Further, $\theta_1$ is not monotonic; this aspect of temperature dependence with respect to the previous day's temperature was also observed by \cite{castillo2022b}.  It cannot be reproduced by KX2006.   In Pamplona, the QAR1K2 model (the best) estimates a lower autoregressive coefficient than the QAR1K1 for $\tau \in (0.1, 0.7)$.  In Zaragoza, similar curves appear for the two values of $K$, as shown by $\tilde{p}_2$ and $\bar{R}^1$.

\begin{figure}[!t] 
\centering
\includegraphics[width=5cm]{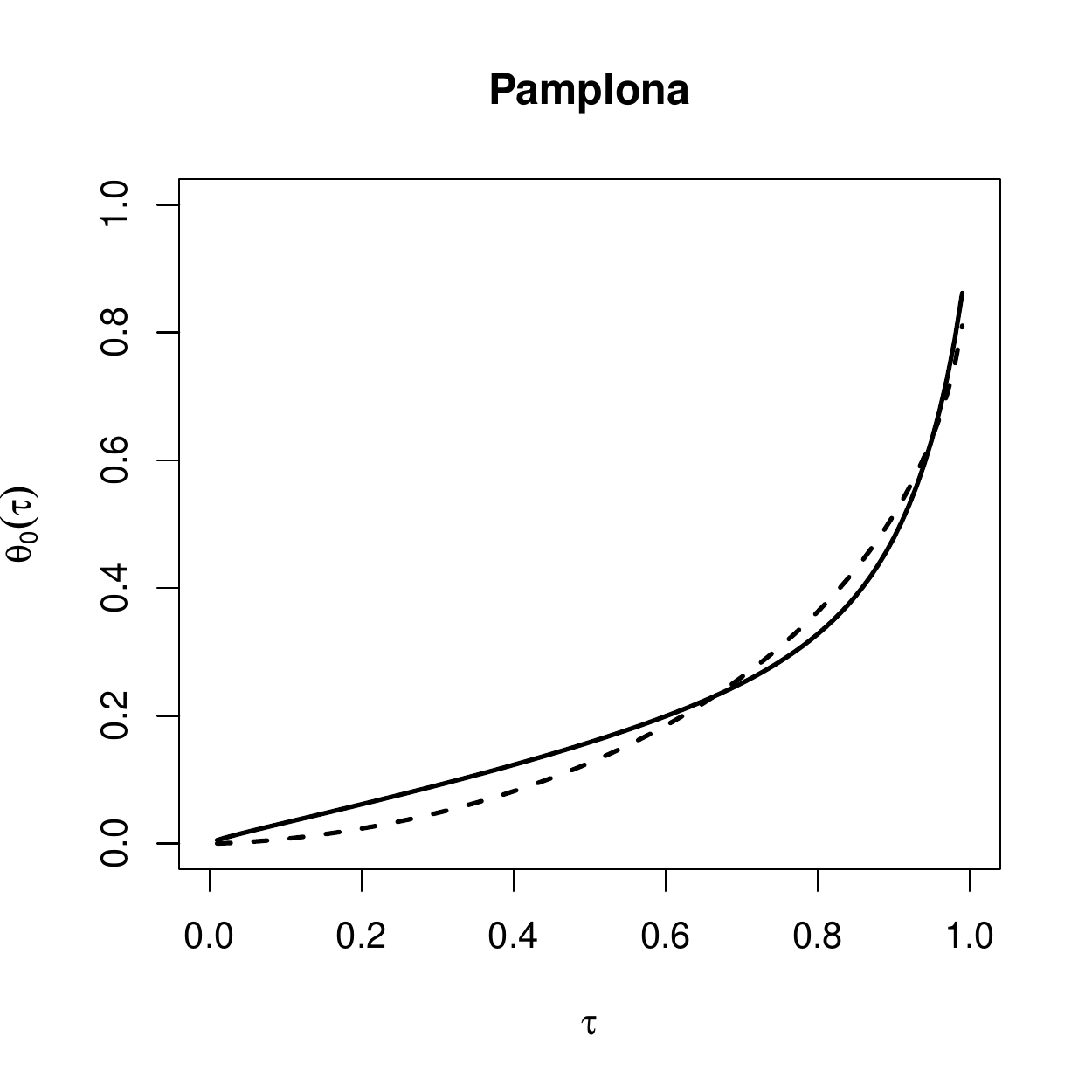}
\includegraphics[width=5cm]{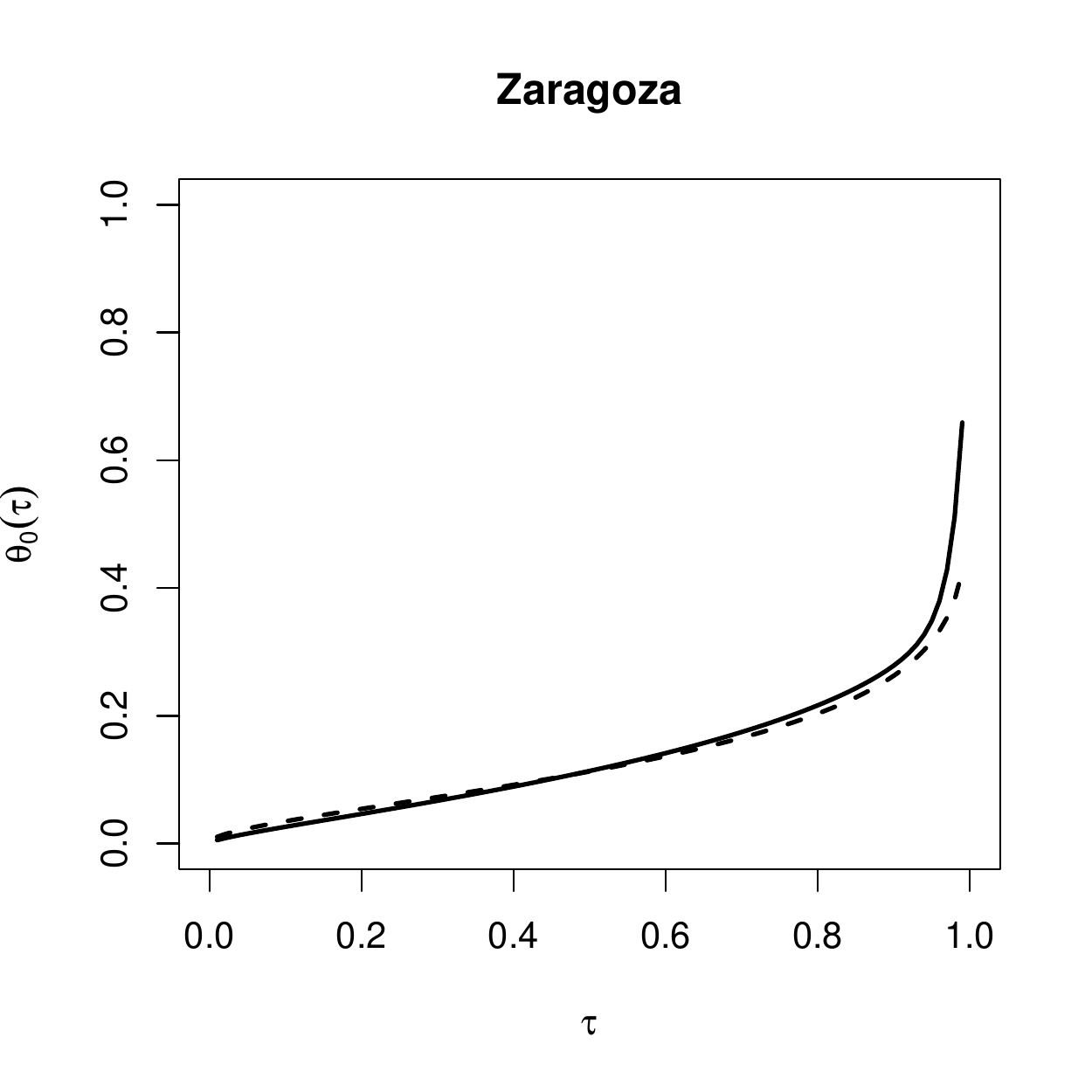} \\
\includegraphics[width=5cm]{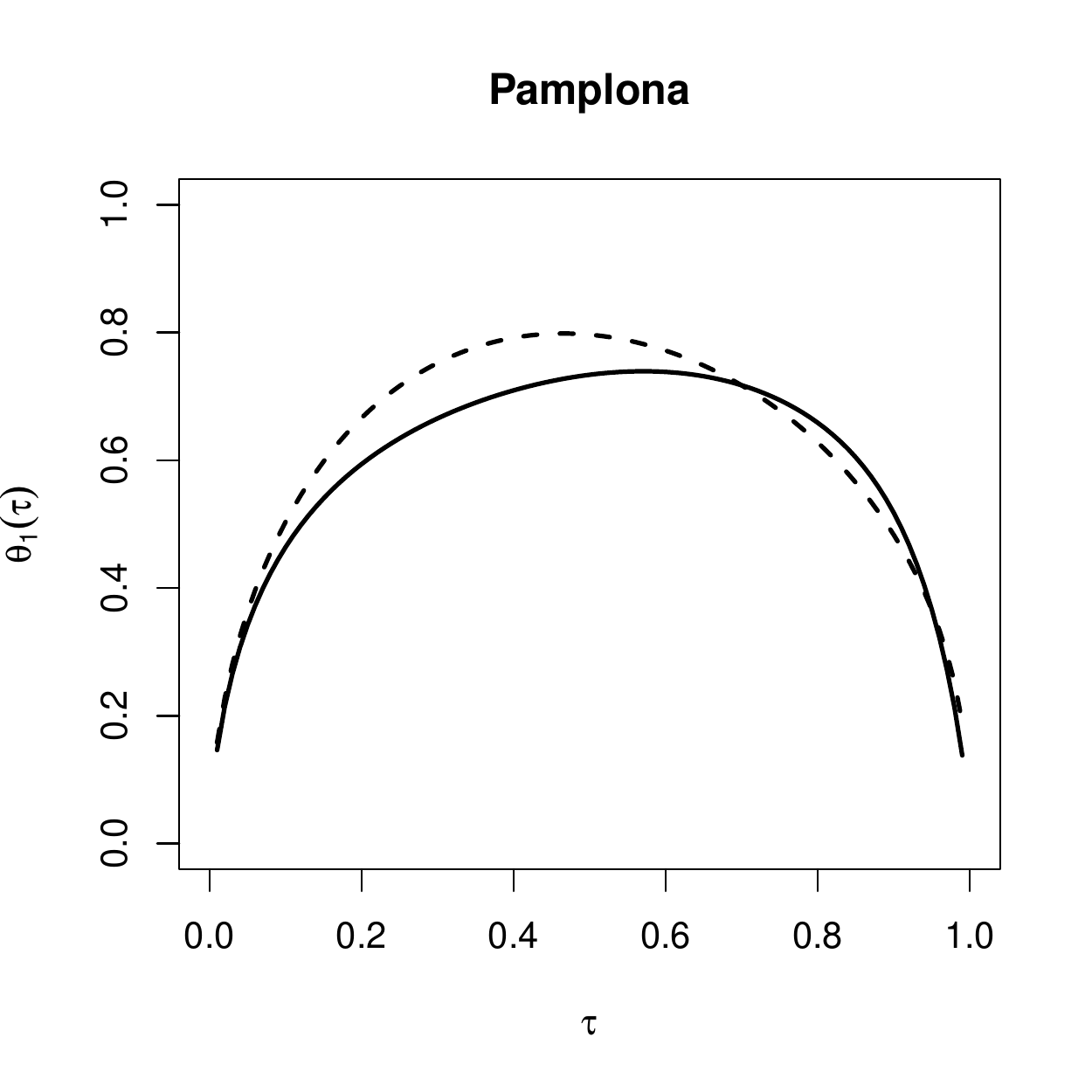}
\includegraphics[width=5cm]{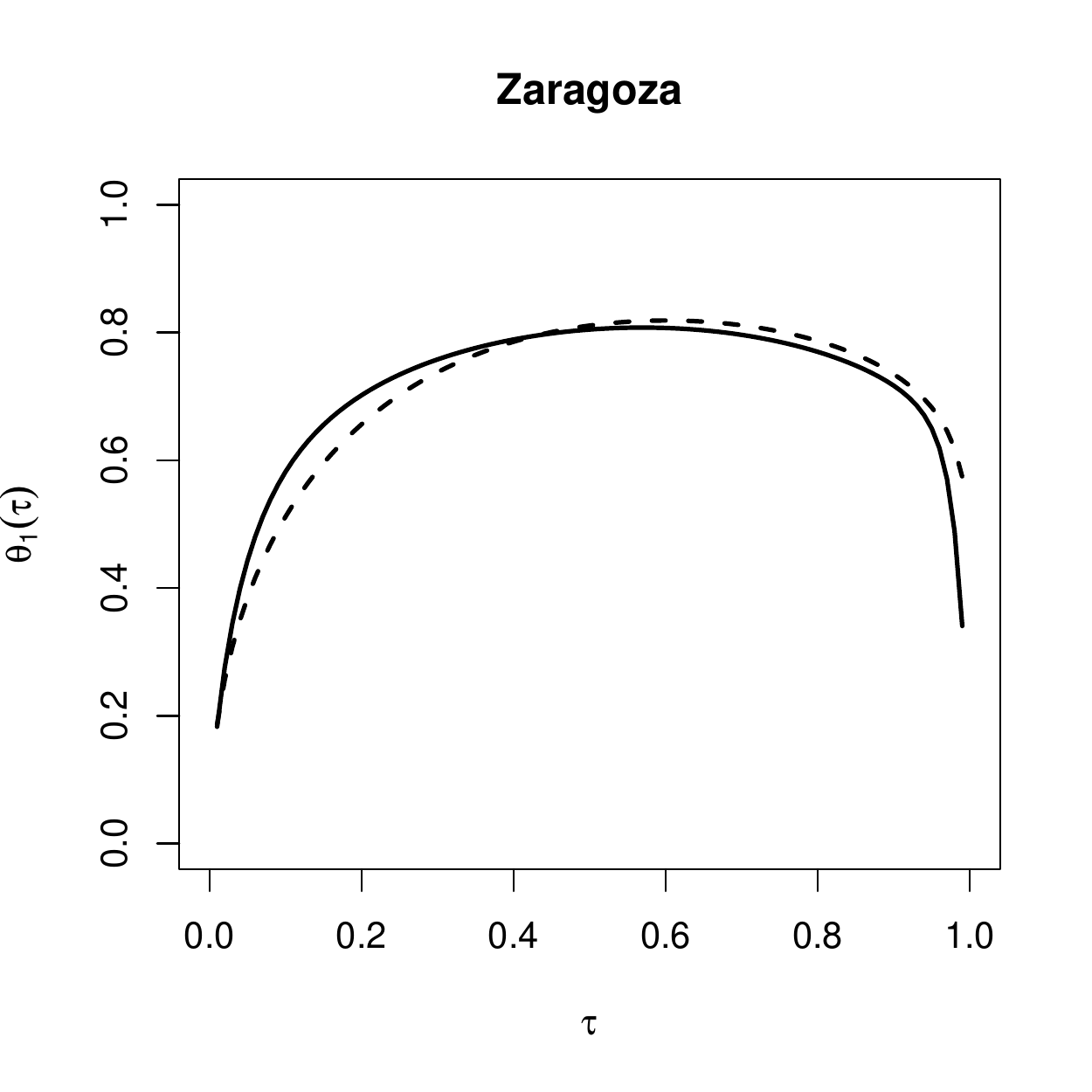}
\caption{Posterior mean of $\theta_0(\tau)$ (above) and $\theta_1(\tau)$ (below) vs. $\tau$ for QAR1K1 (dashed) and QAR1K2 (solid). Pamplona (left) and Zaragoza (right), MJJAS, 2015. } \label{fig:theta}
\end{figure}

Figure~\ref{fig:quantile} shows the posterior mean of the conditional quantile functions $Q_{Y_{t}}(\tau \mid y)$ for three situations where $y$ is the empirical $\tau$ marginal quantile for $\tau=0.1,0.5,0.9$; the legend shows the values that are conditioned on both the original scale and the $(0,1)$ scale.  The smallest values of $\theta_1$ are in extreme $\tau$'s, this means that the previous day's temperature is less influential for the extreme quantiles. In fact, the conditional quantiles in Figure~\ref{fig:quantile} overlap for $\tau$'s near 0 or near 1.

\begin{figure}[!t] 
\centering
\includegraphics[width=5cm]{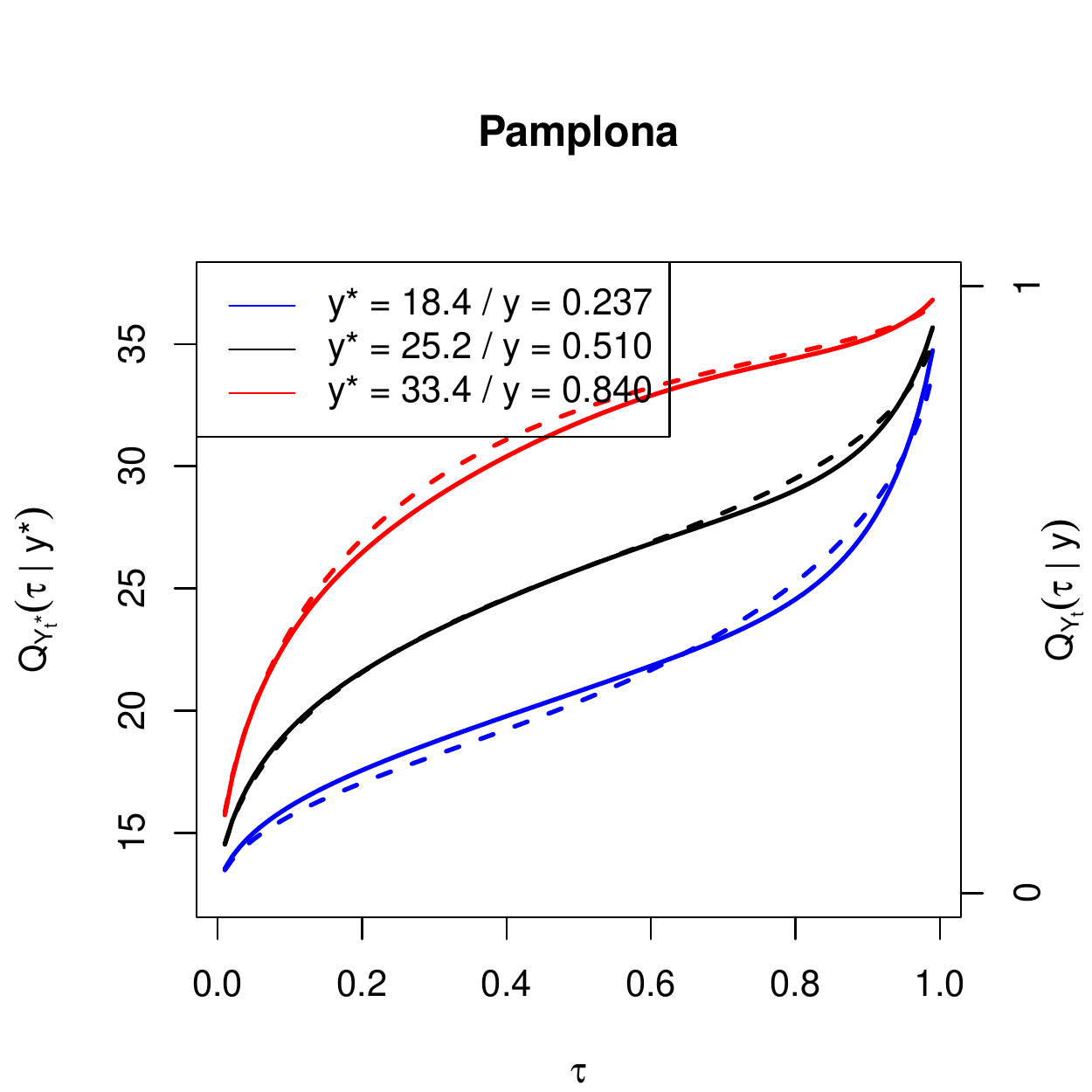}
\includegraphics[width=5cm]{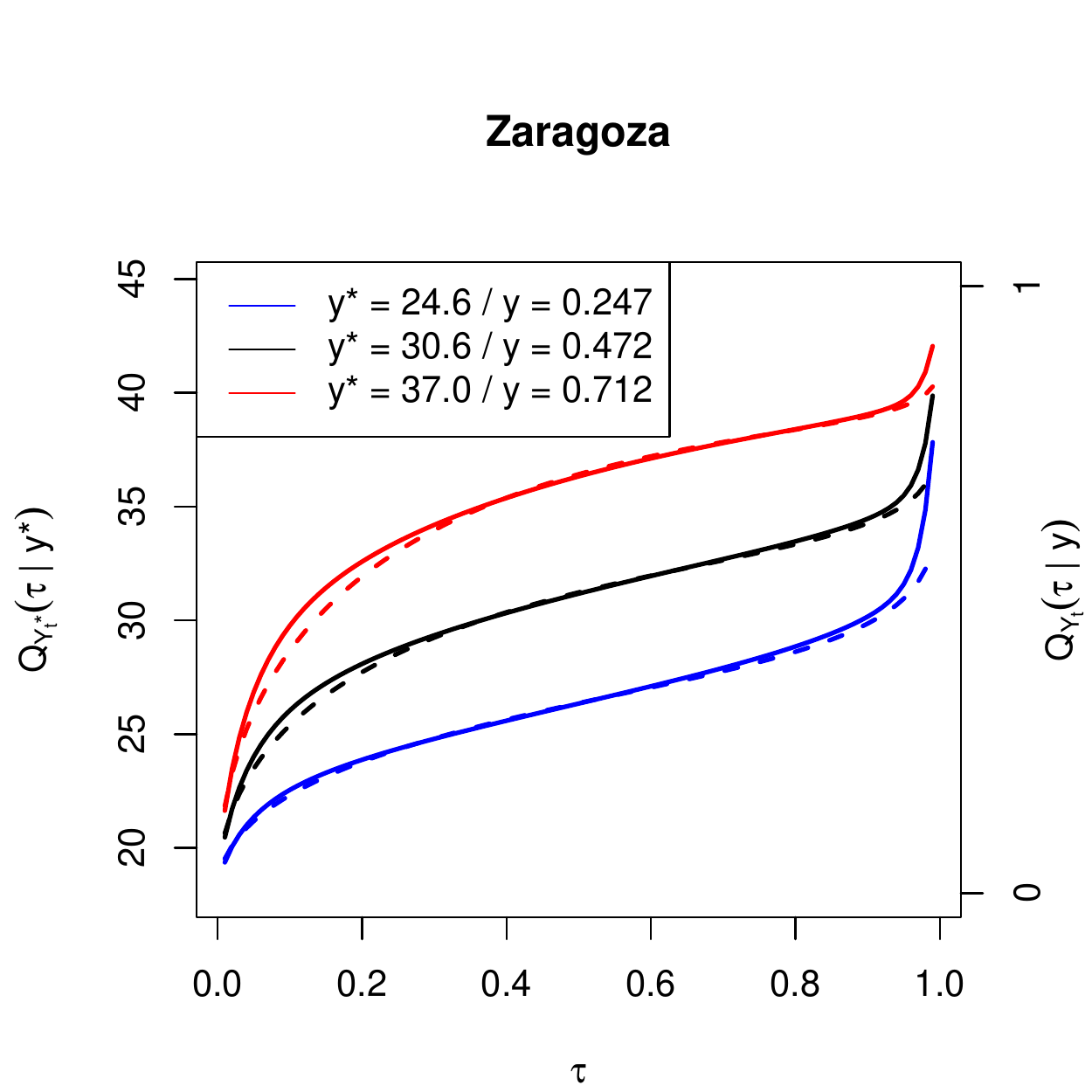}
\caption{Posterior mean of the quantile function $Q_{Y_{t}}(\tau \mid y)$ vs. $\tau$ for QAR1K1 (dashed) and QAR1K2 (solid).  Here, $y$ is the empirical $\tau$ marginal quantile for $\tau=0.1$ (blue), $0.5$ (black), $0.9$ (red). Pamplona (left) and Zaragoza (right), MJJAS, 2015. }
\label{fig:quantile}
\end{figure}

Figure~\ref{fig:density} shows the posterior mean of the conditional density function in \eqref{eq:density} under the same conditions as Figure~\ref{fig:quantile}.  Pamplona presents different shapes in $f_{Y_{t}}(y_t \mid y)$ for different values of $y$.  The distribution is asymmetrical with positive skewness if we condition on a small value for the previous day's temperature, and negative skewness if we condition on a big value. A general pattern is common in the region, the conditional distribution conditional on the $0.9$ marginal quantile is more concentrated than those conditional on the $0.1$ quantile.  Figures~S5, S6, S7, and S8 in the SI present the plots  for the 18 locations.

\begin{figure}[!t] 
\centering
\includegraphics[width=5cm]{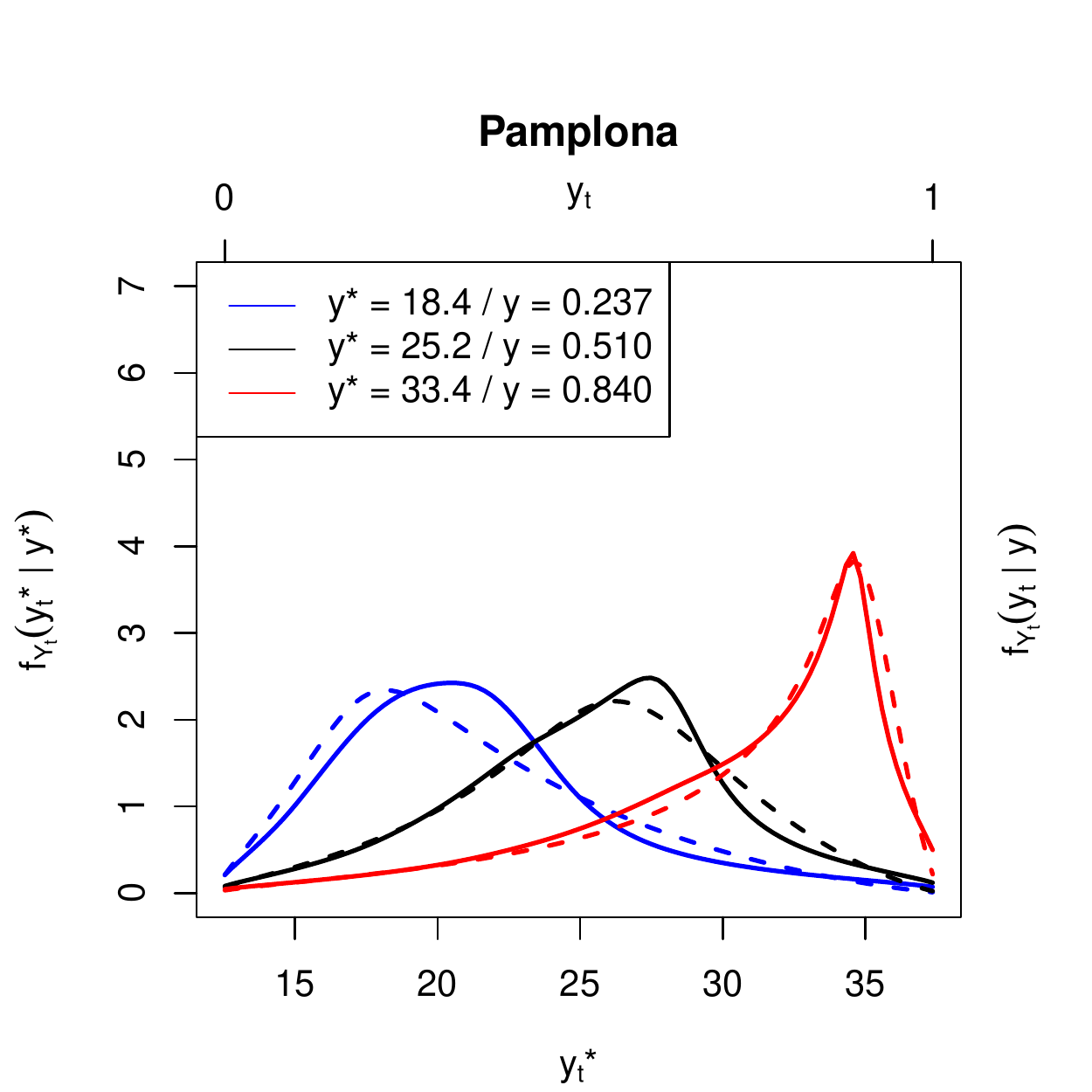}
\includegraphics[width=5cm]{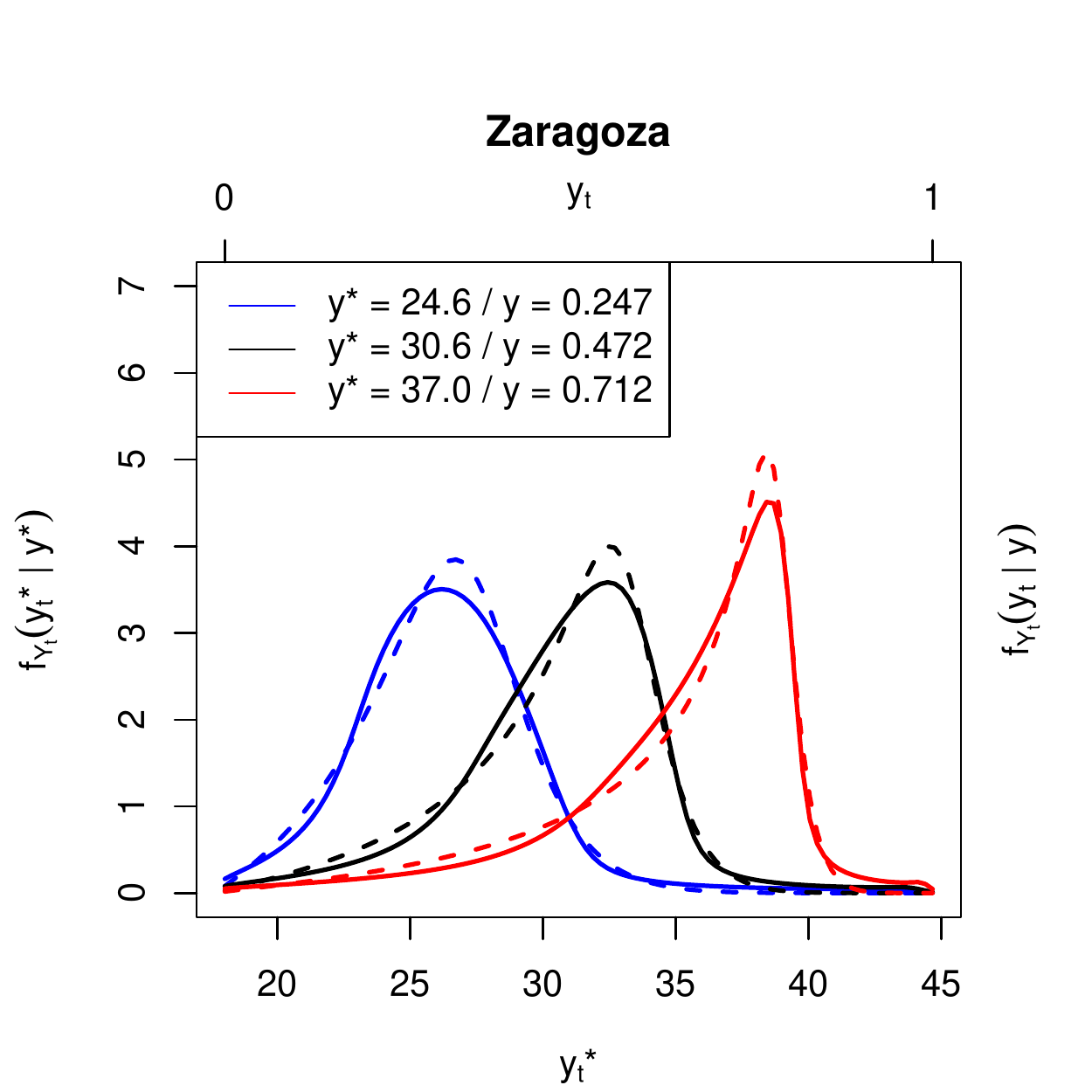}
\caption{Posterior mean of the density function $f_{Y_{t}}(x \mid y)$ for QAR1K1 (dashed) and QAR1K2 (solid).  Here, $y$ is the empirical $\tau$ marginal quantile for $\tau=0.1$ (blue), $0.5$ (black), $0.9$ (red). Pamplona (left) and Zaragoza (right), MJJAS, 2015. }
\label{fig:density}
\end{figure}

\subsection{The QAR(2) case} \label{sec:temp_qar2}

Table~\ref{tab:measurements} uses the criteria $\tilde{p}_2$ and $\bar{R}^1$ for the QAR(2) model with $K=1$ (QAR2K1). The previous subsection showed that including a first lag improved the performance of the model with respect to an empirical null model.  However, including a second lag does not increase the value of $\bar{R}^1$ with respect to a QAR(1) model. On the other hand, the measure of $\tilde{p}_2$ is somewhat better for QAR2K1 than for QAR1K1 but it is still inferior to the QAR1K2 case.  Since QAR2K1 does not improve performance, and, as we will see below, there is no evidence that the term $\theta_2(\tau)$ is different from zero for any $\tau$ across most locations, there seems to be no value in exploring a QAR(2) model with $K=2$.

Figure~\ref{fig:thetaQAR2K1} shows the $\theta$ functions in the QAR2K1 model for Pamplona and Zaragoza (see Figure~S9 in the SI for all locations). The $\theta_0$ and $\theta_1$ functions have a shape very similar to the QAR1K1 case. The $\theta_2$ functions have values that are essentially centered at zero in most locations, giving more evidence that it is not necessary to introduce a lag of order 2 in the model. However, there are four locations with a coefficient slightly away from zero; Bu\~nuel and La Sotonera have a value of $\theta_2(\tau)$ close to $0.2$ for non-extreme $\tau$'s while Huesca and La Puebla de H\'ijar have similar behavior with values around $0.1$.

\begin{figure}[!t]
\centering
\includegraphics[width=5cm]{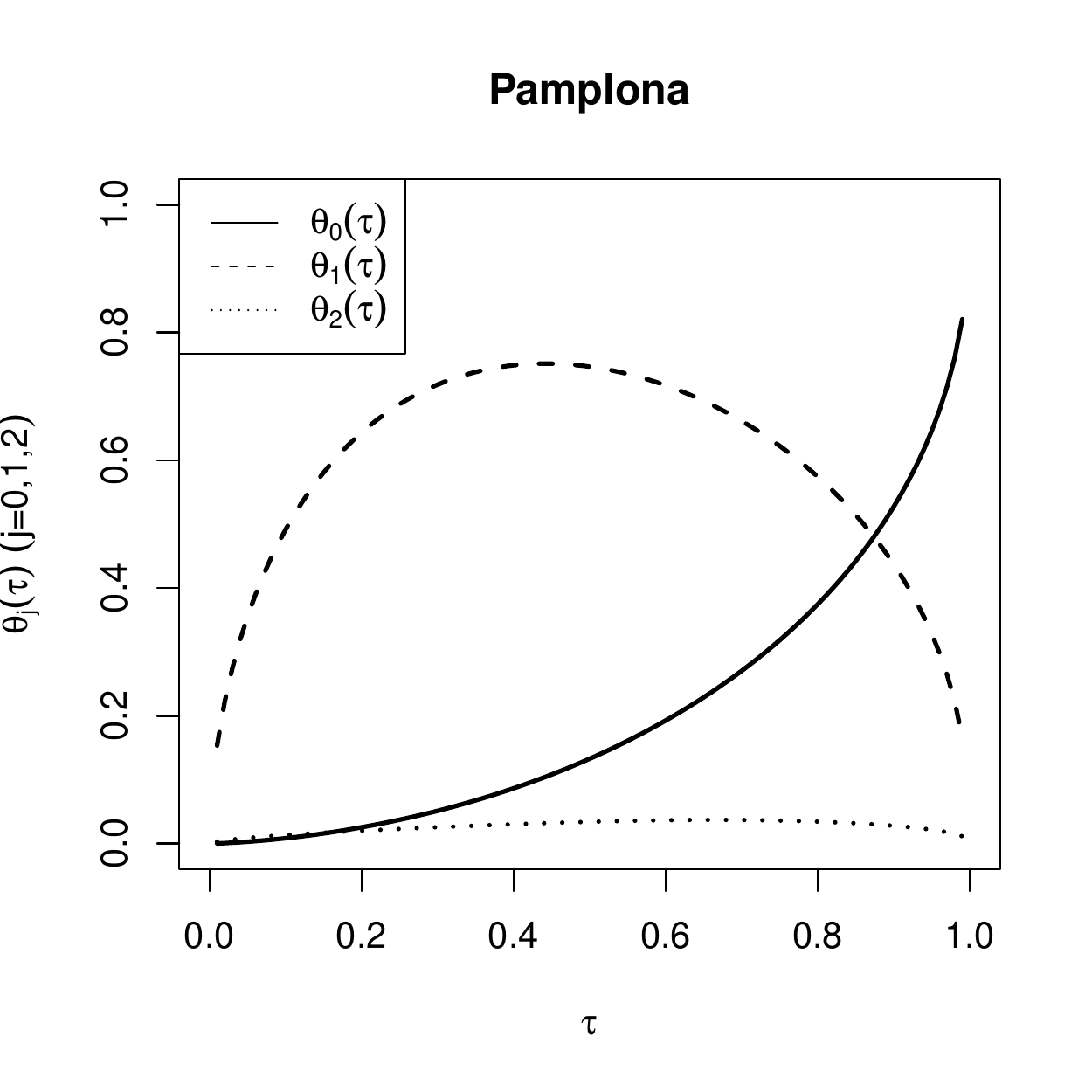}
\includegraphics[width=5cm]{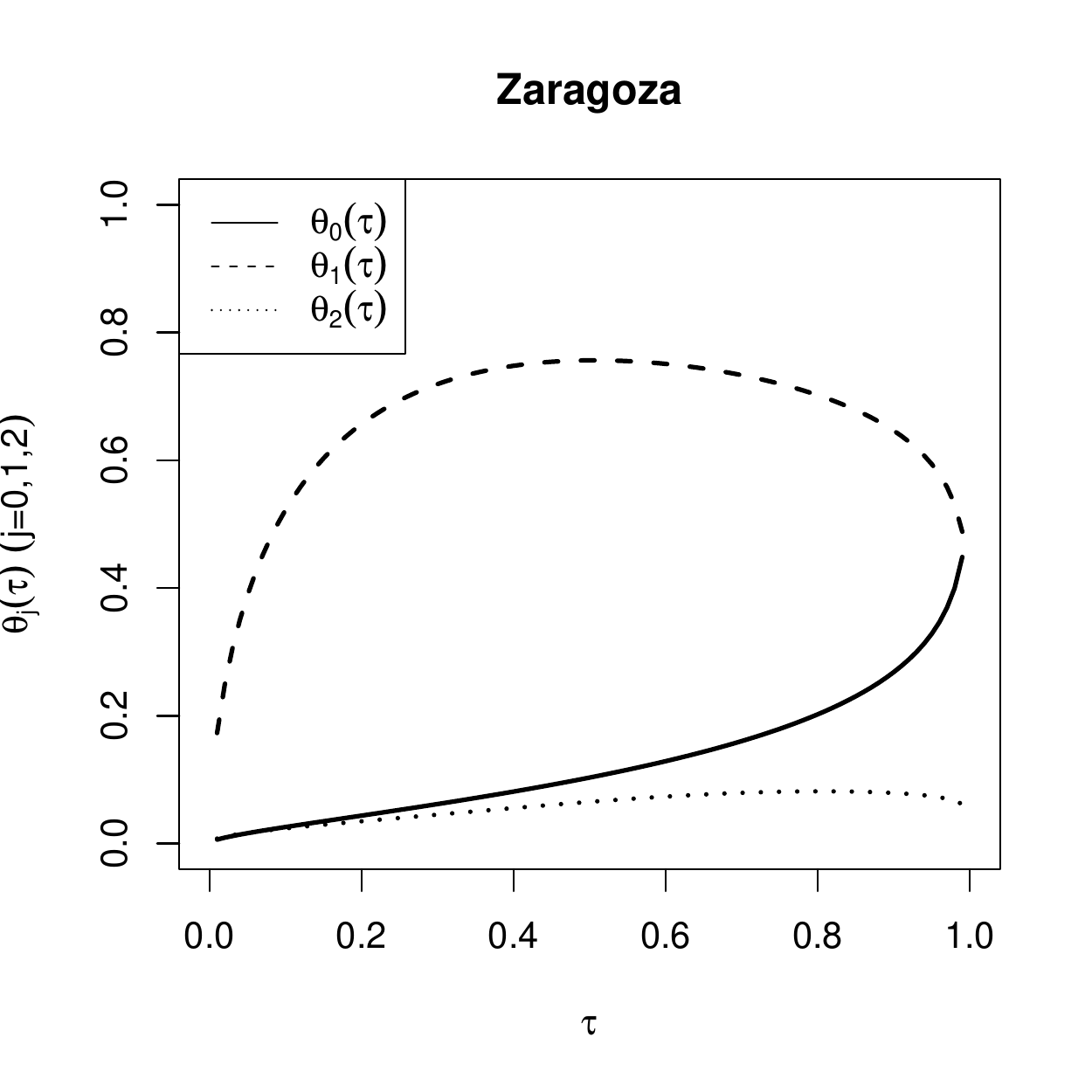}
\caption{Posterior mean of $\theta_0(\tau)$ (solid), $\theta_1(\tau)$ (dashed) and $\theta_2(\tau)$ (dotted) vs. $\tau$ for QAR2K1. Pamplona (left) and Zaragoza (right), MJJAS, 2015. } \label{fig:thetaQAR2K1}
\end{figure}

\subsection{Multivariate QAR(1)} \label{sec:temp_mqar1}

Here, we fit the multivariate QAR(1) model (MQAR1K1) to the daily maximum and minimum temperature series at Zaragoza, $\{(y_{t}^{\text{max}},y_{t}^{\text{min}}):t=1,\ldots,T\}$. The same analyses were developed for Pamplona and Daroca, but with different conclusions. Figure~\ref{fig:fitted:theta:Tx.Tn:Zaragoza:MJJAS:2015} shows the $\theta$ functions for the $y_{t}^{\text{max}}$ (red) and $y_{t}^{\text{min}}$ (blue) series.  
We see different patterns for $\theta_{1}^{\text{max}}$ and $\theta_{1}^{\text{min}}$; $y_{t}^{\text{max}}$ shows high autocorrelation for high quantiles while $y_{t}^{\text{min}}$ has less persistence for those quantiles.

\begin{figure}[!t]
\centering
\includegraphics[width=5cm]{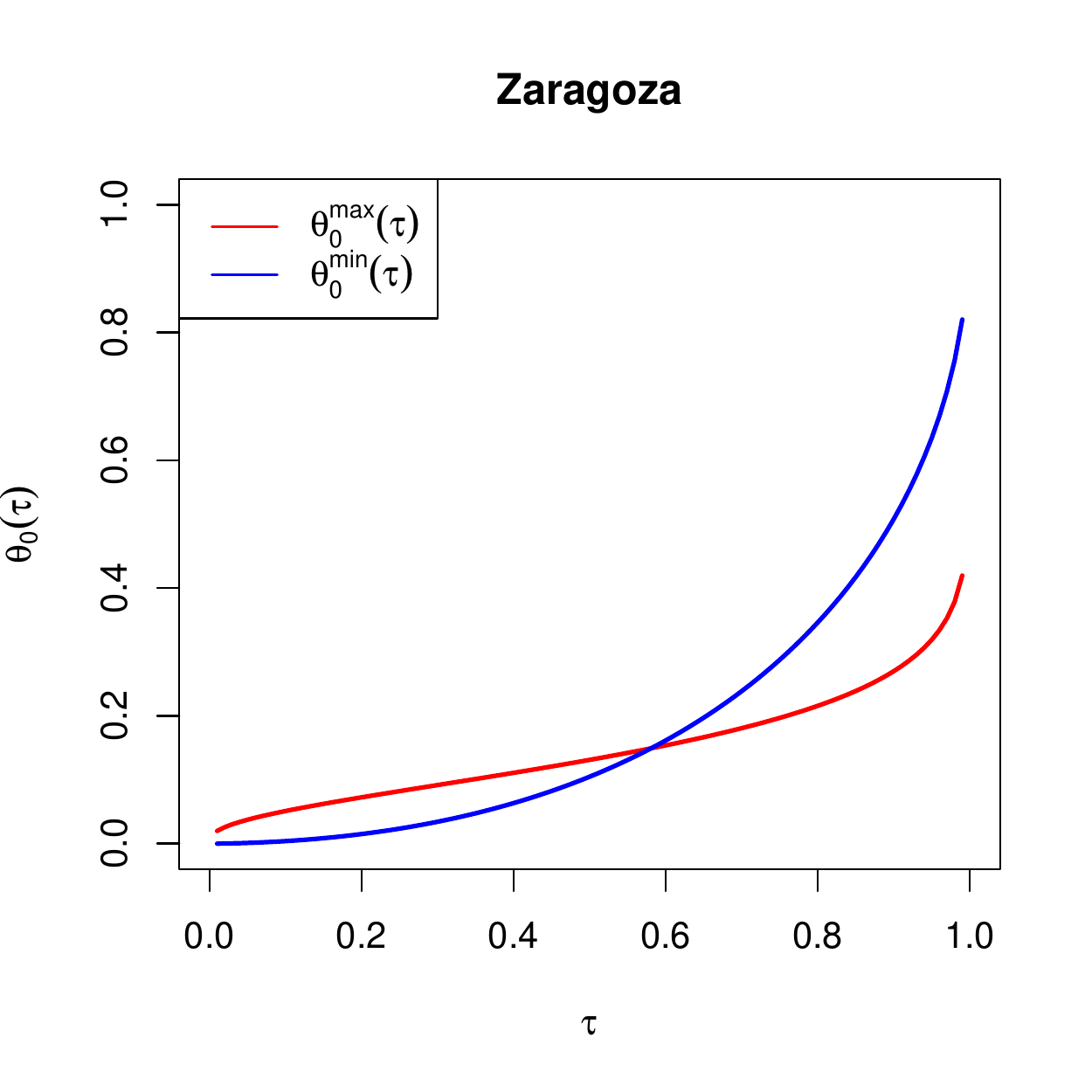}
\includegraphics[width=5cm]{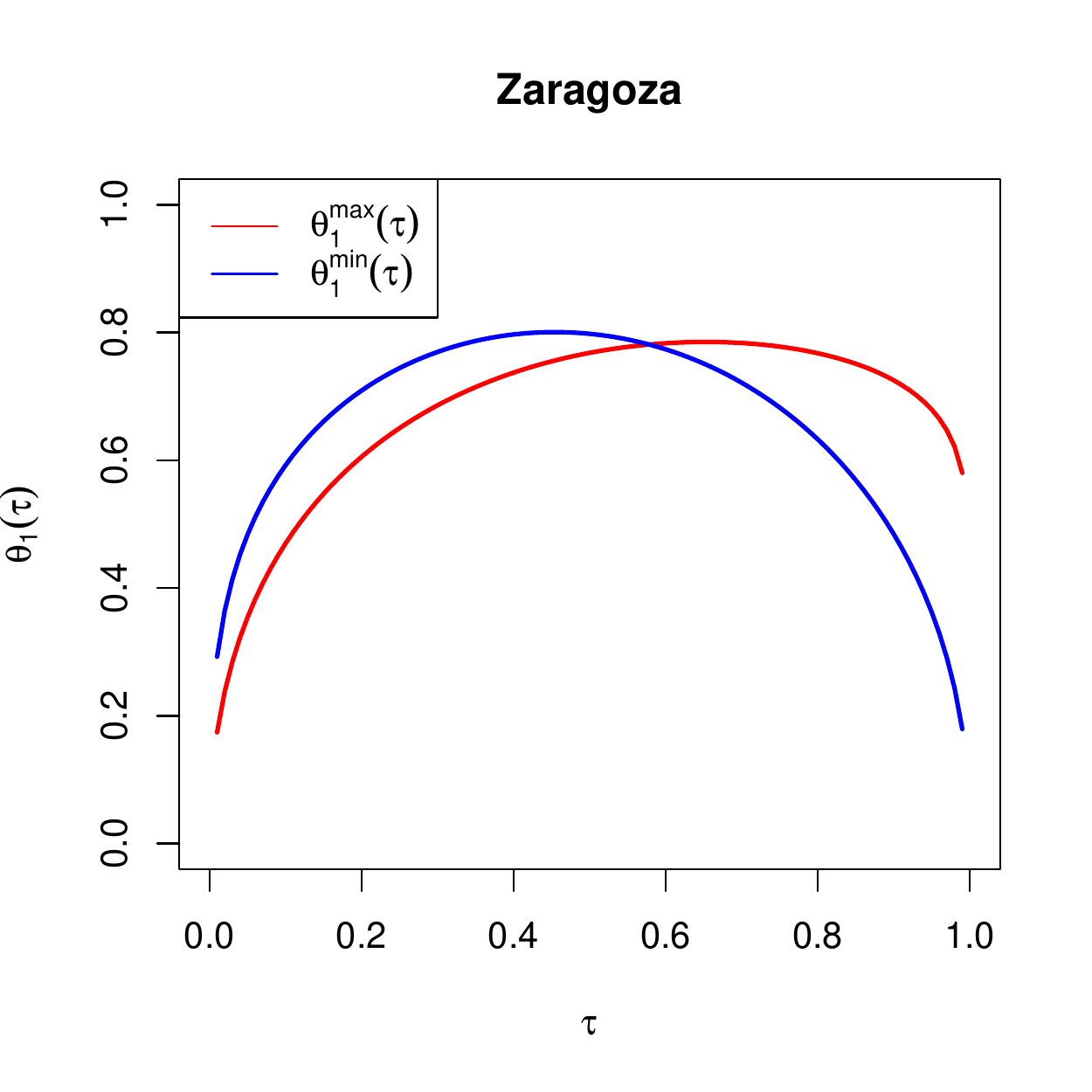}
\caption{Posterior mean of $\theta_{0}^{\text{max}}(\tau)$ and $\theta_{0}^{\text{min}}(\tau)$ (left), and $\theta_{1}^{\text{max}}(\tau)$ and $\theta_{1}^{\text{min}}(\tau)$ (right), vs. $\tau$ for MQAR1K1. Zaragoza, MJJAS, 2015. }
\label{fig:fitted:theta:Tx.Tn:Zaragoza:MJJAS:2015}
\end{figure}

For Zaragoza, the posterior mean of $\rho$ is $0.32$ with  $95\%$ credible interval $(0.17,0.45)$, indicating the need to include dependence in the quantile levels of both series. For Pamplona, the posterior mean of $\rho$ is $0.06$ with $95\%$ credible interval $(-0.11, 0.23)$. Here, independent models for the conditional quantiles could be adopted. An explanation is the frequent appearance of fresh wind from the northwest during the night in Pamplona, resulting from proximity to the Cantabrian Sea.

\begin{figure}[!t]
\centering
\includegraphics[width=5cm]{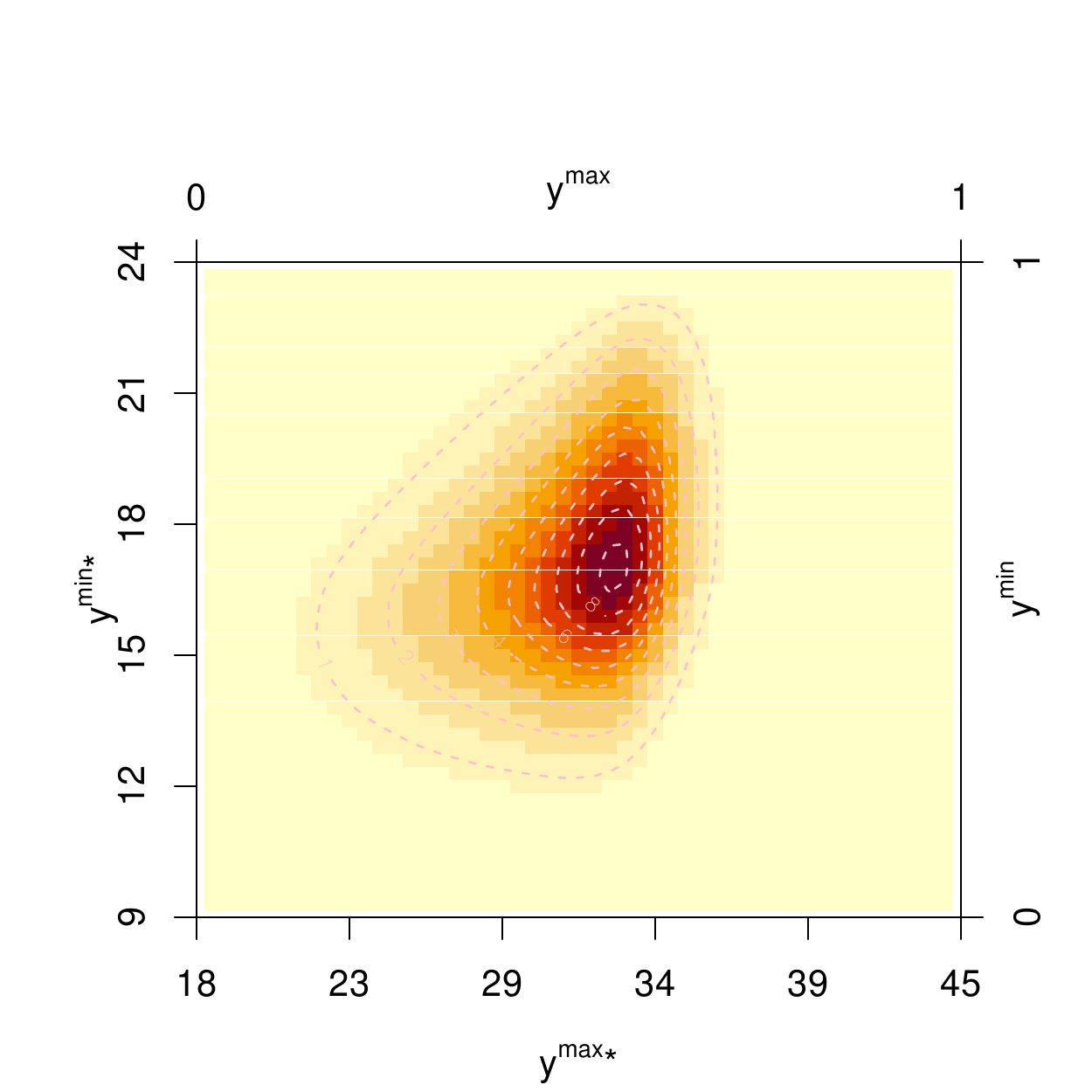}
\includegraphics[width=5cm]{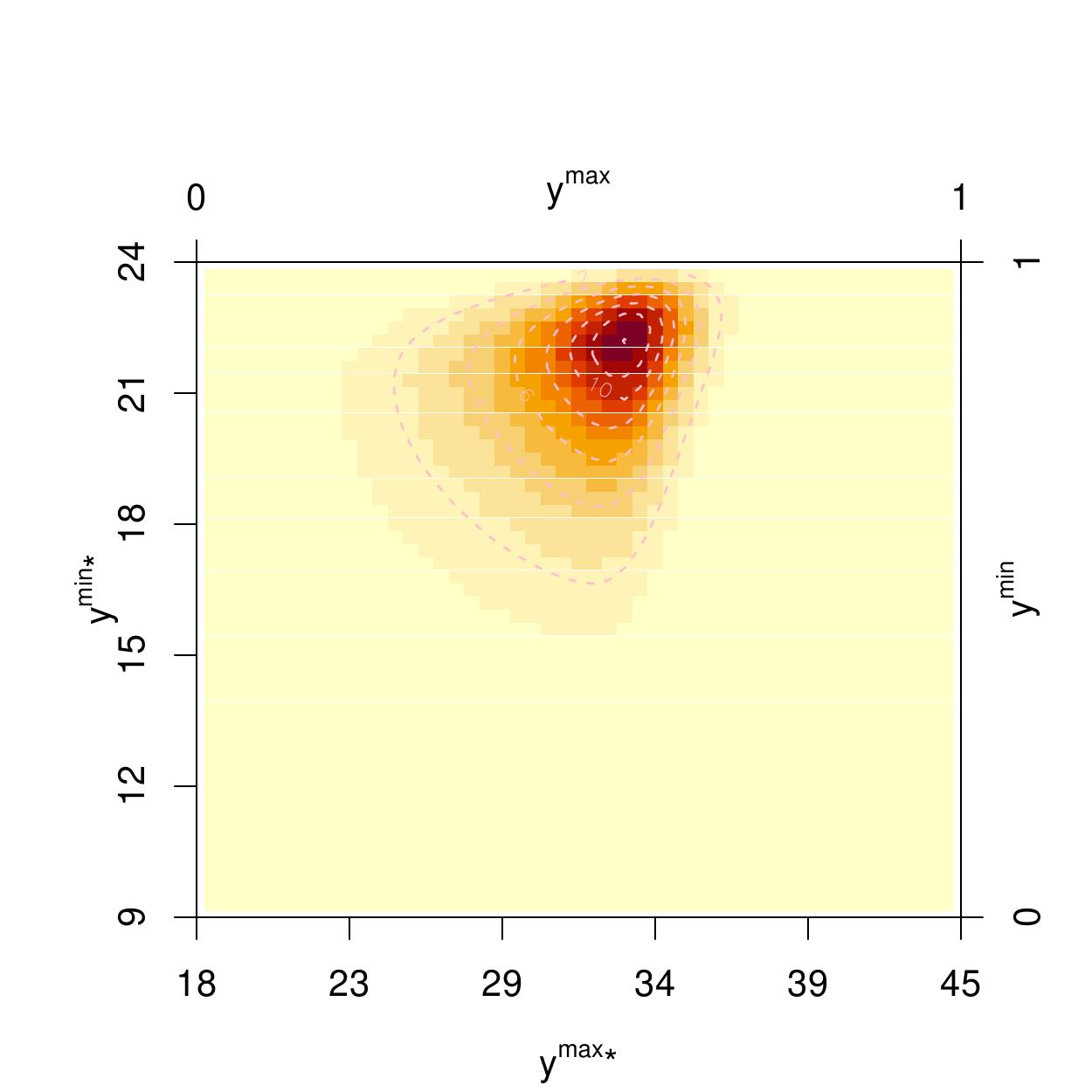}\\
\includegraphics[width=5cm]{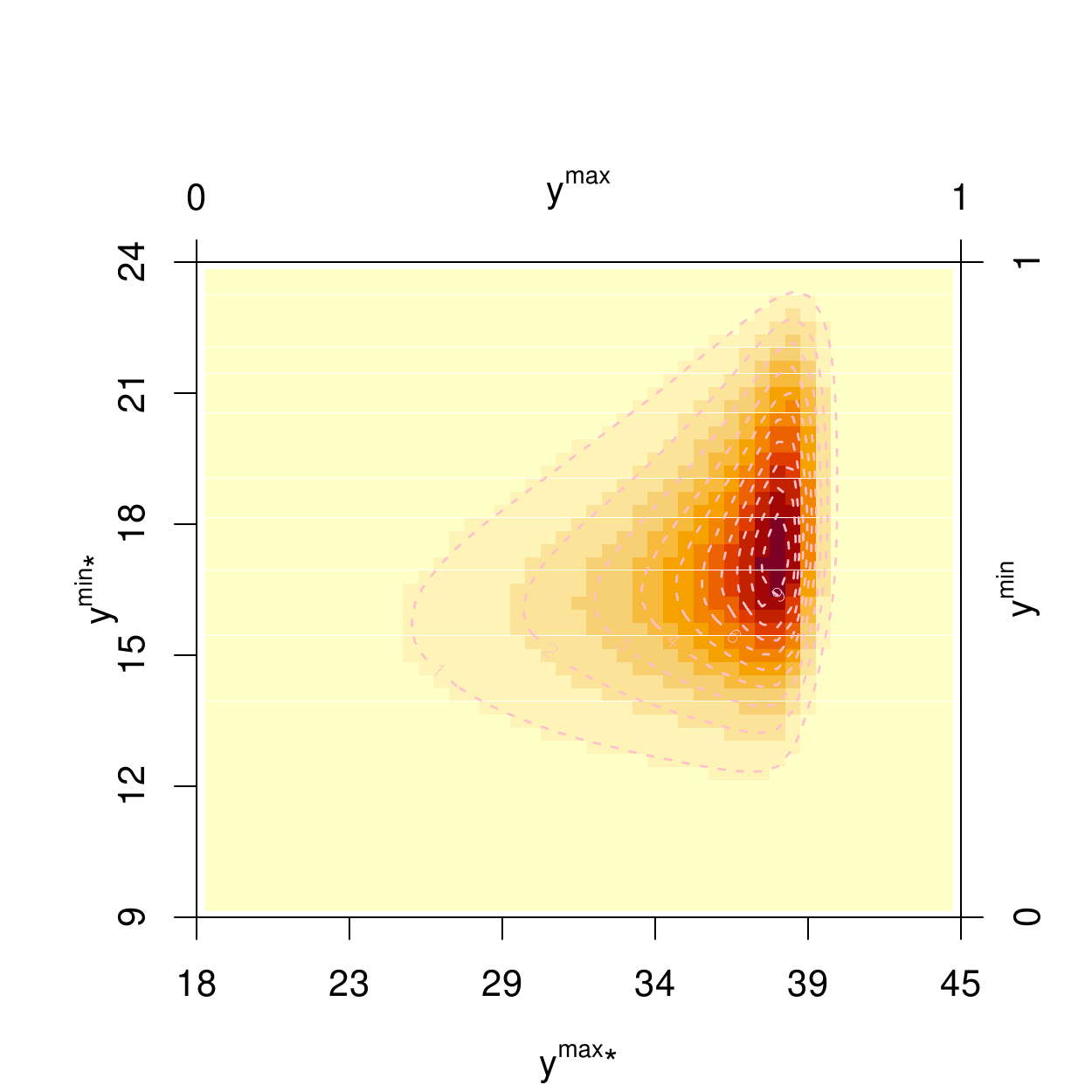}
\includegraphics[width=5cm]{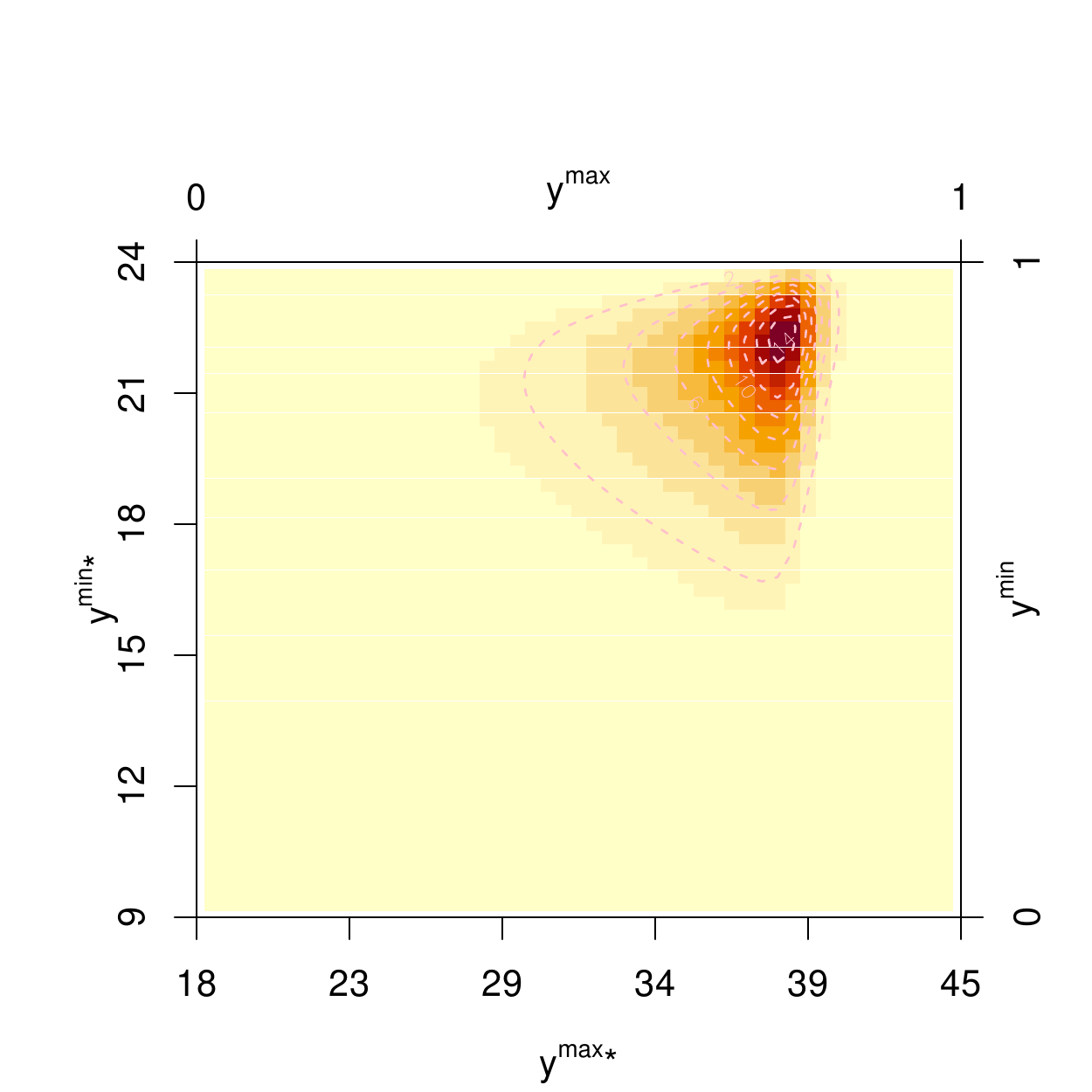}
\caption{Posterior mean of the density function of $(Y_{t}^{\text{max}}, Y_{t}^{\text{min}})$ conditioned on $(y^{\text{max}}, y^{\text{min}})$ for MQAR1K1. Here, $(y^{\text{max}}, y^{\text{min}})$ is equal to the respective empirical marginal quantiles for $\tau = 0.5$ (above), $0.9$ (below) for the maximum; and $\tau = 0.5$ (left), $0.9$ (right) for the minimum. Zaragoza, MJJAS, 2015. }
\label{fig:fitted:density:Tx.Tn:Zaragoza:MJJAS:2015}
\end{figure}

Figure~\ref{fig:fitted:density:Tx.Tn:Zaragoza:MJJAS:2015} shows level curves of the posterior conditional joint density of the vector $(Y_{t}^{\text{max}}, Y_{t}^{\text{min}})$ given the previous day's max and min temperatures, in Zaragoza (see Figures~S10 and S11 in the SI for Pamplona and Daroca). The conditioning values are empirical marginal quantiles of $Y_{t}^{\text{max}}$ and $Y_{t}^{\text{min}}$. The first row conditions on the quantile $\tau = 0.5$ ($30.6^\circ$C) and the second row on the quantile $\tau = 0.9$ ($37.0^\circ$C) of $Y_{t}^{\text{max}}$, and the same quantiles of $Y_{t}^{\text{min}}$, for the first ($17.2^\circ$C) and second ($21.8^\circ$C) columns. The different patterns observed in the plots reveal a different relation between $Y_{t}^{\text{max}}$ and $Y_{t}^{\text{min}}$ depending on the previous day's temperatures.  The conditional posterior distribution is not symmetric, with a different mean vector depending on the conditioning temperatures; the  variability of the distribution  is smaller when it is conditioned on high quantiles.

\subsection{Spatial QAR(1)} \label{sec:temp_sqar1}

The spatial QAR  model is fitted to the series of MJJAS in 18 locations in Arag\'on for the year 2015.  The posterior mean of $\gamma$, the proportion of spatial dependence in \eqref{eq:U}, is $0.955$ with $95\%$ credible interval $(0.934,0.973)$  indicating very strong spatial dependence in the quantile levels of the temperature series.  Figure~S12 in the SI provides maps of the posterior mean surface of the model GP's. We notice that $b_1(\bs)$ and $b_2(\bs)$ show approximately opposite spatial behavior since $b_1(\bs)$ has the highest values where $b_2(\bs)$ has the lowest, in the central and southeastern areas. Figure~S13 of the SI shows boxplots of the posterior distribution of the GP's at each observed location; locations are sorted by elevation. The results suggest that the GP of $a_2(\bs)$  might be not necessary since the boxplots in the 18 locations have very similar ranges.  The spatial variability of  $a_1(\bs)$ is higher and, although it is not related to the elevation, it could be related to the distance to the sea.    

The posterior distribution of $\theta_1(\tau;\bs)$, which captures the autoregressive structure, is summarized using the same type of plots.  Figure~S14 of the SI shows boxplots presenting the posterior distribution of $\theta_1(\tau;\bs)$ at the observed locations while Figure~\ref{fig:meantheta1} the maps of the posterior mean surface of $\theta_1(\tau;\bs)$, both for $\tau = 0.05,0.50,0.95$.  The spatial GP's in the parameters of the Kumaraswamy distribution allow the model to fit different spatial patterns in each $\tau$. The results show that the posterior mean of $\theta_1(\tau;\bs)$ is higher in the central quantiles.  The spatial pattern of $\theta_1(\tau;\bs)$ is not symmetric around  $\tau=0.5$ and, e.g., values of $\theta_1(0.95;\bs)$ in the Pyrenees and northwestern areas are smaller than  $\theta_1(0.05;\bs)$ in the same areas.  Although $\theta_1(\tau;\bs)$ tends to be lower in locations with higher elevation, its spatial pattern cannot be explained by elevation alone.  Consequently, the spatial GP's cannot be replaced with an elevation fixed effect. 

\begin{figure}[!t]
\centering
\includegraphics[width=4cm]{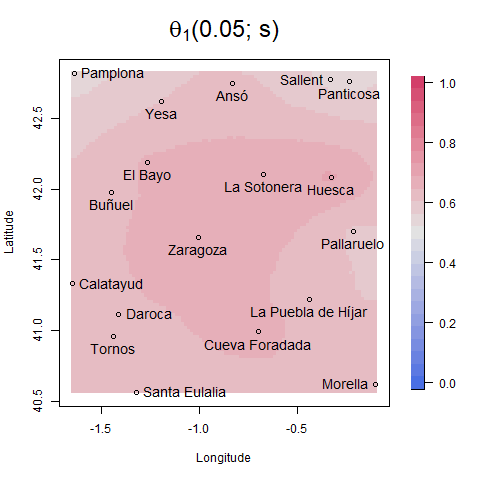}
\includegraphics[width=4cm]{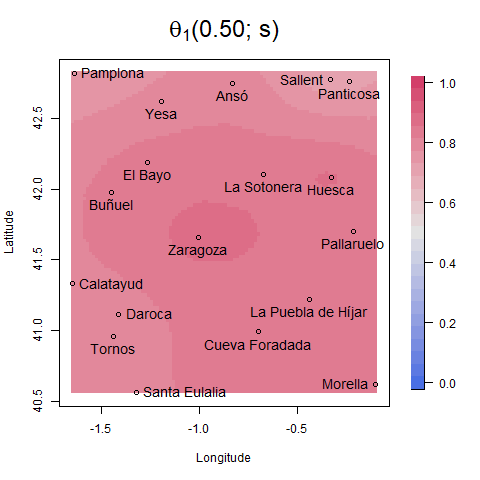}
\includegraphics[width=4cm]{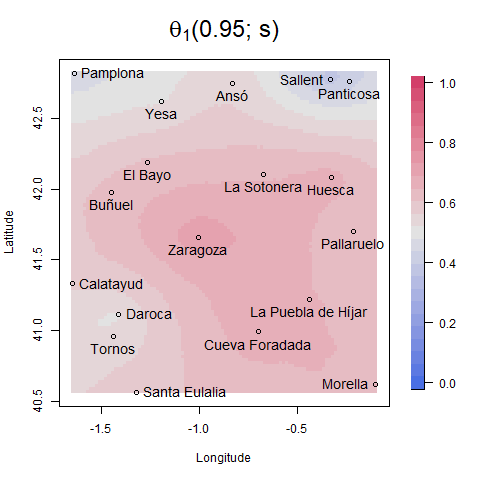}
\caption{Maps of the posterior mean of $\theta_1(\tau;\bs)$ for $\tau = 0.05,0.50,0.95$. } \label{fig:meantheta1}
\end{figure}

The spatial joint model  can also be used to estimate parameters related to the conditional distribution, e.g., conditional quantiles  at unobserved locations.  As a brief example, Figure~S15 of the SI shows this through maps of the posterior mean of $Q_{Y_t(\bs)}(\tau \mid y)$ for $\tau = 0.05,0.50,0.95$, and $y = 0.05,0.50,0.95$.  If it were desired to obtain the quantiles on the original scale of the data rather than the scale $(0,1)$, we could consider a kriging of $m(\bs)$ and $M(\bs)$. With the same kriging procedure we could condition on values $y(\bs)$'s relative to a certain empirical marginal quantile for each location.
  
The spatial modeling here is primarily illustrative.  For instance, the assumption of asymptotic tail independence, imposed by the Gaussian copula, may not be suitable. Examination of alternative copulas is beyond the scope of this work.

\section{Summary and future work} \label{sec:summary}

We have presented consequentially expanded modeling for joint (non-quantile crossing) QAR.  In particular, we have characterized the QAR(1) setting in a way that allows for a more flexible autocorrelation structure than the one in the seminal paper by \cite{koenker2006}.  We have extended this to the QAR($p$) case. We have offered a novel multiple time series version using a Gaussian copula.  We have elaborated a spatial version, using a GP copula based upon a GP in conjunction with four additional GP's.  This model enables spatially varying quantile functions. Our modeling is entirely parametric through the use of the  Kumaraswamy distributions. A software implementation of our methods is available as the R-package ``QAR" through GitHub: \url{https://github.com/JorgeCastilloMateo/QAR}.

We have illustrated the above contributions through time series of daily temperatures from sites in Arag\'on, Spain.  The joint QAR model, with greater flexibility in the modeling of the $\theta$ functions, allows us to capture autoregression structure in daily temperature data, which is not strictly increasing in $\tau$, but decreasing in both tails.

A critical challenge in employing this work is model fitting.  We can make specifications as rich as needed through the use of probabilistic mixtures of Kumaraswamy cdf's.  However, it is well-known that model fitting  employing MCMC with mixture specifications  is often poorly identified. This issue is compounded in our case by the fact that calculation of the likelihood requires constant use of a one-dimensional rootfinder.  Ongoing work is attempting to address these computational difficulties. 

We digress briefly to note a simple approximation strategy to incorporate regressors, e.g., seasonality, into our joint QAR approach.  Suppose we introduce a regression structure, $\mu_{t}$ into the QAR(1) and estimate by  $\hat{\mu}_{t}$, creating residuals $r_{t} = Y_{t} - \hat{\mu}_{t}$.  Then, we could apply the above methodology to obtain the QAR(1) for $r_{t}$.  Our strategy for selecting $m$ and $M$ can be applied to residuals.  More precisely, let $r_{t} = \theta_{(0)}(U_{t}) + \theta_{(1)}(U_{t})r_{t-1}$. This would yield the conditional quantile function, $Q_{r_{t}}(\tau \mid r_{t-1}) = \theta_{(0)}(\tau) + \theta_{(1)}(\tau)r_{t-1}$. Solving for the quantile function for $Y_{t}$ we obtain
\begin{equation}
    Q_{Y_{t}}(\tau \mid Y_{t-1}) = \hat{\mu}_{t} + \theta_{(0)}(\tau) + \theta_{(1)}(\tau)(Y_{t-1} - \hat{\mu}_{t-1}).
\end{equation}
This approximation can be criticized for two reasons:  (i) we are creating $\hat{\mu}_{t}$ as if we were fitting a usual AR(1) and (ii) the resulting quantiles are not coherent since $\hat{\mu}_{t}$ is a function of $\{Y_t : t=1,\ldots,T\}$.  The QAR(1) is not defined until the end of the observation window. Coherent implementation of covariates in the QAR setting, seeking to bridge our modeling with the work of \cite{yang2017} is in progress.  

Sections~\ref{sec:MQAR1} and \ref{sec:SQAR1} could be combined to build a bivariate spatial QAR model for daily max and daily min temperature.  Another challenge for the multivariate and spatial modeling would be to consider alternative copula choices, e.g., t-copulas in order to allow tail dependence for extreme quantiles.

\section*{Supplementary Information}
SI for ``Bayesian joint quantile autoregression" contains details on the Kumaraswamy distribution. Details on the model comparison metrics and simulation study. More results on the application with temperature series.

\section*{Acknowledgments}
This work has been supported in part by the Grants PID2020-116873GB-I00 and TED2021-130702B-I00 funded by MCIN/AEI/10.13039/501100011033 and Uni\'on Europea NextGenerationEU; and the Research Group E46\_20R: Modelos Estoc\'asticos funded by Gobierno de Arag\'on. J. C.-M. was supported by Gobierno de Arag{\'o}n under Doctoral Scholarship ORDEN CUS/581/2020. This work was done in part while J. C.-M. was a Visiting Scholar at the Department of Statistical Science from Duke University. The authors also thank Surya T. Tokdar from Duke University for fruitful discussions on joint QR modeling.

\section*{Declarations}


\textbf{Conflict of interest.} The authors declare that they have no conflict of interest.

\bibliography{arxiv_v1}

\end{document}


\maketitle

\section{Kumaraswamy distribution}

Figures~\ref{Fig.minimax.density.base} and \ref{Fig.minimax.CDF.base} show the pdf and the cdf of the Kumaraswamy distribution for different combinations of parameters $(a,b)$.

\begin{figure}[ht]
\begin{center}
\includegraphics[width=10cm]{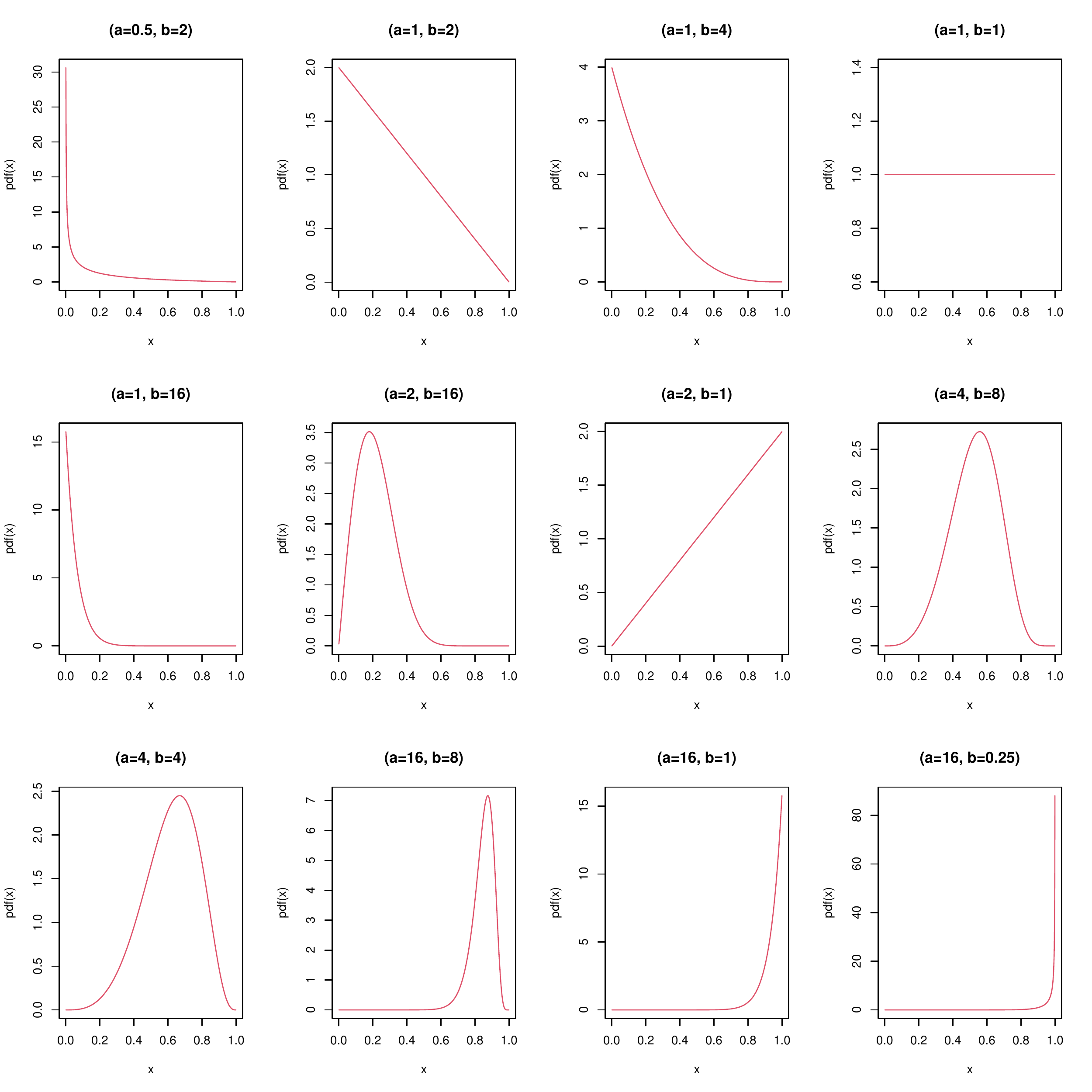}  
\caption{Probability density function of the Kumaraswamy distribution for different parameters.} \label{Fig.minimax.density.base}
\end{center}
\end{figure}

\begin{figure}[ht]
\begin{center}
\includegraphics[width=10cm]{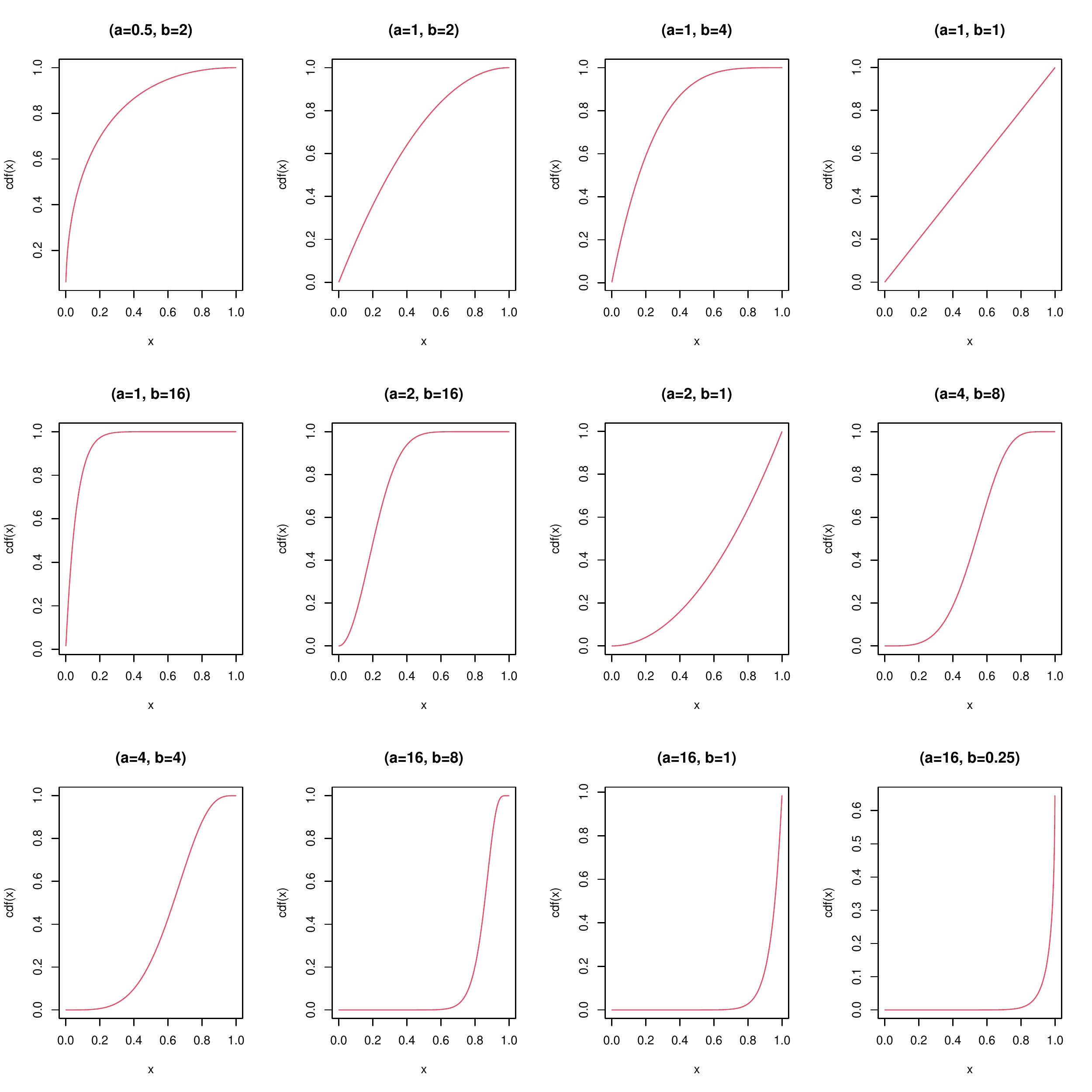}  
\caption{Cumulative distribution function of the Kumaraswamy distribution for different parameters.} \label{Fig.minimax.CDF.base}
\end{center}
\end{figure}

\clearpage

\section{Model adequacy and comparison} \label{sec:metrics}

Working within our parametric Bayesian framework, for any $\tau$, posterior samples of the model parameters, $\{\bOmega_{b}^{*} : b=1,\ldots,B\}$, produce posterior samples of the conditional quantile function for $Y_{t}$, $Q_{Y_{t}}(\tau \mid y_{t-1}; \bOmega_{b}^{*})$.  Essentially, for each $Y_{t}$ (with associated $y_{t-1}$) and any $\tau$, we obtain the posterior distribution of $Q_{Y_{t}}(\tau \mid y_{t-1}; \bOmega)$.  We use these posterior distributions along with the dataset, $\by$, to offer model assessment.  

We propose two novel approaches.  First, for any $y$, consider $\textbf{1}(y<Q_{Y_{t}}(\tau \mid y_{t-1}; \bOmega))$ where $\textbf{1}$ denotes the indicator function.  Then, let $p_{t}(\tau) \equiv E[\textbf{1}(y_{t} < Q_{Y_{t}}(\tau \mid y_{t-1}; \bOmega)) \mid \by]$, i.e., the posterior probability that $Q_{Y_{t}}(\tau \mid y_{t-1}; \bOmega)$ exceeds $y_{t}$.  Suppose we compute $p(\tau) \equiv \sum_{t=2}^{T} p_{t}(\tau) / (T-1)$.  We note that for any $\tau$ and $Y$ regardless of its distribution, $E[\bone(Y < Q_{Y}(\tau)] = \tau$ and $Var[\bone(Y < Q_{Y}(\tau)] = \tau(1 - \tau)$.  If we let $v \ge 1$ be a real number, then 
\begin{equation} \label{eq:p}
    \tilde{p}_v \equiv 
    \sqrt[v]{\int_{0}^{1} \left|\frac{p(\tau) - \tau}{\sqrt{\tau(1-\tau) / (T-1)}}\right|^v \,d\tau}
\end{equation}
provides a standardized deviation form as a dimensionless measure of how well the quantile function under the model is capturing conditional quantiles for the given time series. We propose this as a (\emph{global}) measure of model accuracy.  With a minimum value of zero, a smaller $\tilde{p}_v$ indicates better accuracy of the model.  We would approximate the integral by discretizing $\tau$, in particular, we consider $\tau \in \{0.01,0.02,\ldots,0.99\}$.

As a second measure, we turn to the check loss function, usually employed as an optimality function to obtain the $\tau$ empirical quantile \citep{koenker1978}.  Here, we adopt $\delta_{\tau}(u) = u(\tau - \textbf{1}(u<0))$, the check loss function associated with the AL distribution.  Again, from the posterior distribution of $Q_{Y_{t}}(\tau \mid y_{t-1}; \bOmega)$, for any $Y_{t}$ (with associated $y_{t-1}$) and $\tau$, we can obtain $\Delta_{t}(\tau) \equiv \delta_{\tau}(y_{t} - E[Q_{Y_{t}}(\tau \mid y_{t-1}; \bOmega) \mid \by])$.  As above suppose we compute $\Delta(\tau) \equiv \sum_{t=2}^{T} \Delta_{t}(\tau) / (T-1)$.  Then, for a given $\tau$, $\Delta(\tau)$ provides an average discrepancy for the $\tau$ quantile function.  The smaller the value, the better the model performance.  Then, we propose to weight $\Delta(\tau)$,
\begin{equation}
    \tilde{\Delta} \equiv \int_{0}^{1} \omega(\tau) \Delta(\tau) \,d\tau
    \label{eq:Delta}
\end{equation}
to provide a global measure of model performance.   We propose this as a \emph{relative} measure of model performance in making model comparison.  The weighting function, $\omega(\tau)$, compensates for the variation in mean of $\Delta(\tau)$ across $\tau$.  Again, we would approximate the integral by discretizing $\tau$.  

For the weight function, we consider 
\begin{equation} 
    \omega(\tau \mid \by) = \frac{1}{\sum_{t=2}^{T} \delta_{\tau}(y_{t} - Q^{emp}_{Y}(\tau)) / (T-1)}.
\end{equation}
This choice leads to a measure that is closely related to the $R^1(\tau)$ metric by \cite{koenker1999}. The $R^{1}(\tau)$ measure is essentially, $1- \omega(\tau \mid \by) \Delta(\tau)$. This measure is viewed as an analogue of $R^2$ for the classical residual sum of squares, i.e., the check loss function for quantiles replaces the least-squares loss function and the $\tau$ empirical marginal quantile $Q^{emp}_{Y}(\tau)$ replaces the sample mean.  With a maximum value of $1$, the best model performance is reached at this maximum. Then,
\begin{equation} \label{eq:R1}
    \bar{R^1} \equiv \int_{0}^{1} R^1(\tau) \,d\tau = 1 - \int_{0}^{1} \omega(\tau \mid \by) \Delta(\tau) \,d\tau = 1- \tilde{\Delta},
\end{equation}
provides a dimensionless global measure of model performance which can be used for model comparison.

\section{Simulation study} \label{sec:sims}

To illustrate parameter recovery under the modeling in Section~2 of the Main Manuscript, we perform a simulation study where we consider QAR(1) models with $K=1$ and $K=2$ (QAR1K1 and QAR1K2) Kumaraswamy cdf's in the probabilistic mixture from Section~2.2.1, and scenarios with different combinations of parameter values. For each scenario we simulate $B=100$ datasets of length $T=150$ and $T=500$ time points, using the model in (5) from the Main Manuscript with the $\eta$'s again specified as probabilistic mixtures.  For each simulated dataset, a QAR(1) model with $K$ equal to the value used in the data generating process is fitted. First, $90\%$ credible intervals are computed for $\theta_0(\tau)$ and $\theta_1(\tau)$ in a grid of values of $\tau$. Then, the coverage (CVG) of each value of each function is computed as the proportion of the $B$ computed credible intervals that contain the corresponding true values.

Table~\ref{tab:compK1} shows $\overline{CVG}(\theta_0)$ and $\overline{CVG}(\theta_1)$, the average of the CVG's for the values of the functions in $\tau \in \{0.01,0.02,\ldots,0.99\}$ based on the QAR1K1 model fitted to data generated with the same model, under four scenarios. The individual CVG's of $\theta_0(\tau)$ and $\theta_1(\tau)$  for $\tau=0.01, 0.50, 0.99$ are shown in Table~\ref{tab:app:compK1}. Scenario~SC1 considers independent series of data with $a_{1}=0.5$, $b_{1}=2$, $a_{2}=0.5$ and $b_{2}=2$ (here, the subscript indicates the function $\eta_j$ to which the parameter corresponds).  The other scenarios correspond to series with different correlation structures.  Parameters in Scenario~SC2 ($a_{1}=4$, $b_{1}=4$, $a_{2}=1$, $b_{2}=2$) give values of $\theta_1(\tau)$ moving from $-0.5$ to $0$, with the maximum negative correlation for the central values of $\tau$. Parameters in Scenario~SC3 ($a_{1}=0.5$, $b_{1}=2$, $a_{2}=2$, $b_{2}=1$) give $\theta_1(\tau)$ values from $0.2$ to $0.7$  for $\tau<0.75$.  Parameters in Scenario~SC4 ($a_{1}=0.3$, $b_{1}=6$, $a_{2}=12$, $b_{2}=8$) give values of $\theta_1(\tau)$ from $0.8$ to $1$ for $\tau<0.75$. The $\theta$ functions corresponding to the generated data are plotted in Figure~\ref{fig:app:scenariosK1}.  The same study is repeated  for series of data generated with a QAR1K2 model under three scenarios: SC5 gives, again, uncorrelated series; SC6 positive $\theta_1(\tau)$ between $0.5$ and $0.91$ for $\tau<0.8$; and SC7 gives $\theta_1(\tau)$ moving from $-0.5$ to $0$, with the maximum negative values for the central values of $\tau$ and decreasing slowly towards $0$ in the extremes.  Table~\ref{tab:app:param} contains the parameters used in each scenario.  The average of the CVG's, $\overline{CVG}(\theta_0)$ and $\overline{CVG}(\theta_1)$, when fitting the QAR1K2 model under these scenarios are also shown in Table~\ref{tab:compK1}, and plots of $\theta_0(\tau)$ and $\theta_1(\tau)$ are shown in Figure~\ref{fig:app:scenariosK1}.

\begin{table}[b]
\centering
\begin{tabular}{cc|ccccccc}
Function & $T$ & SC1 & SC2   & SC3  & SC4  & SC5  & SC6  & SC7 \\ 
  \hline
 \multirow{2}{*}{$\theta_0$}  & $150$ & 0.87 & 0.86 & 0.91 & 0.87 & 0.92 & 0.88 & 0.88  \\
 & $500$ & 0.89 & 0.89 & 0.88 & 0.89 & 0.91 & 0.89 & 0.91  \\
 \multirow{2}{*}{$\theta_1$}  & $150$ & 0.88 & 0.90 & 0.88 & 0.86 & 0.90 & 0.87 & 0.87  \\
 & $500$ & 0.89 & 0.89 & 0.90 & 0.87 & 0.91 & 0.90 & 0.90  \\ 
\end{tabular}
\caption{The $90\%$ $\overline{CVG}$ of $\theta_0(\tau)$ and $\theta_1(\tau)$ across $\tau \in \{0.01,0.02,\ldots,0.99\}$ obtained from fitting QAR1K1 models in Scenarios~SC1--SC4 and QAR1K2 models in Scenarios~SC5--SC7.}
\label{tab:compK1}
\end{table}

\begin{table}[tb]
\centering
\begin{tabular}{cc|cccccc}
Scenario & $T$ & $\theta_0(0.01) $ & $\theta_0(0.99) $ & $\theta_0(0.50) $  & $\theta_1(0.01) $ & $\theta_1(0.99) $ & $\theta_1(0.50) $ \\ 
  \hline
\multirow{2}{*}{SC1} & $150$ & 0.88 & 0.86 & 0.85  & 0.87 & 0.91 & 0.86  \\ 
 & $500$  & 0.91 & 0.89 & 0.88 & 0.89 & 0.90 & 0.90   \\
\multirow{2}{*}{SC2} & $150$ & 0.87 & 0.90 & 0.92  & 0.87 & 0.89 & 0.90  \\ 
 & $500$ & 0.87 & 0.89 & 0.88  & 0.87 & 0.88 & 0.91  \\ 
\multirow{2}{*}{SC3} & $150$ & 0.91 & 0.83 & 0.89  & 0.90 & 0.83 & 0.89  \\ 
 & $500$ & 0.90 & 0.87 & 0.88 & 0.83 & 0.89 & 0.95  \\ 
\multirow{2}{*}{SC4} & $150$ & 0.88 & 0.88 & 0.86 & 0.88 & 0.88 & 0.86  \\ 
 & $500$ & 0.86 & 0.86 & 0.80  & 0.93 & 0.86 & 0.81  \\ 
\end{tabular}
\caption{The $90\%$ CVG of $\theta_0(\tau)$ and $\theta_1(\tau)$ ($\tau = 0.01,0.50,0.99$) in QAR1K1 models fitted to Scenarios~SC1--SC4. }
\label{tab:app:compK1}
\end{table}

\begin{figure}[tb]
\centering
\includegraphics[width=5cm]{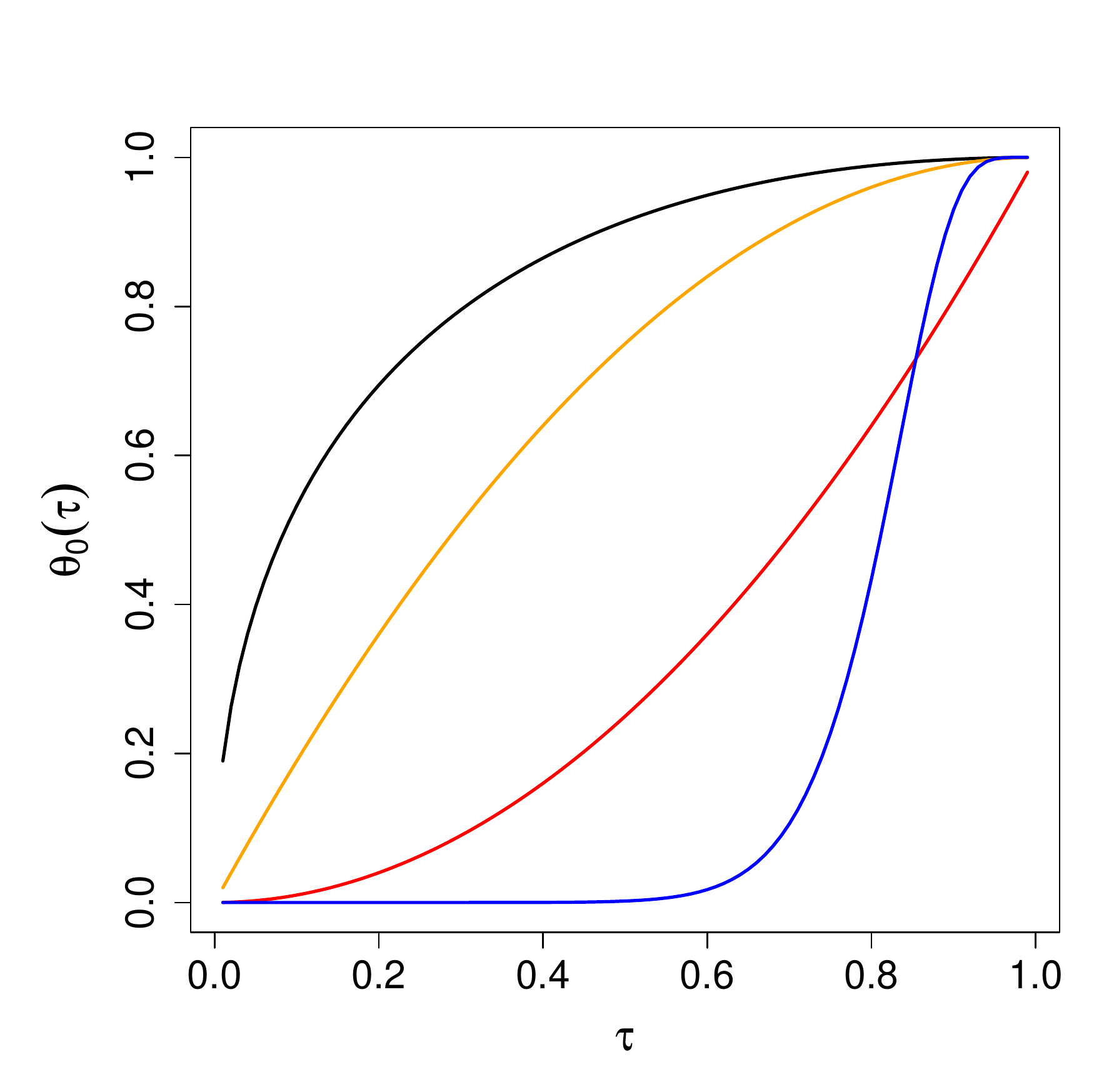}
\includegraphics[width=5cm]{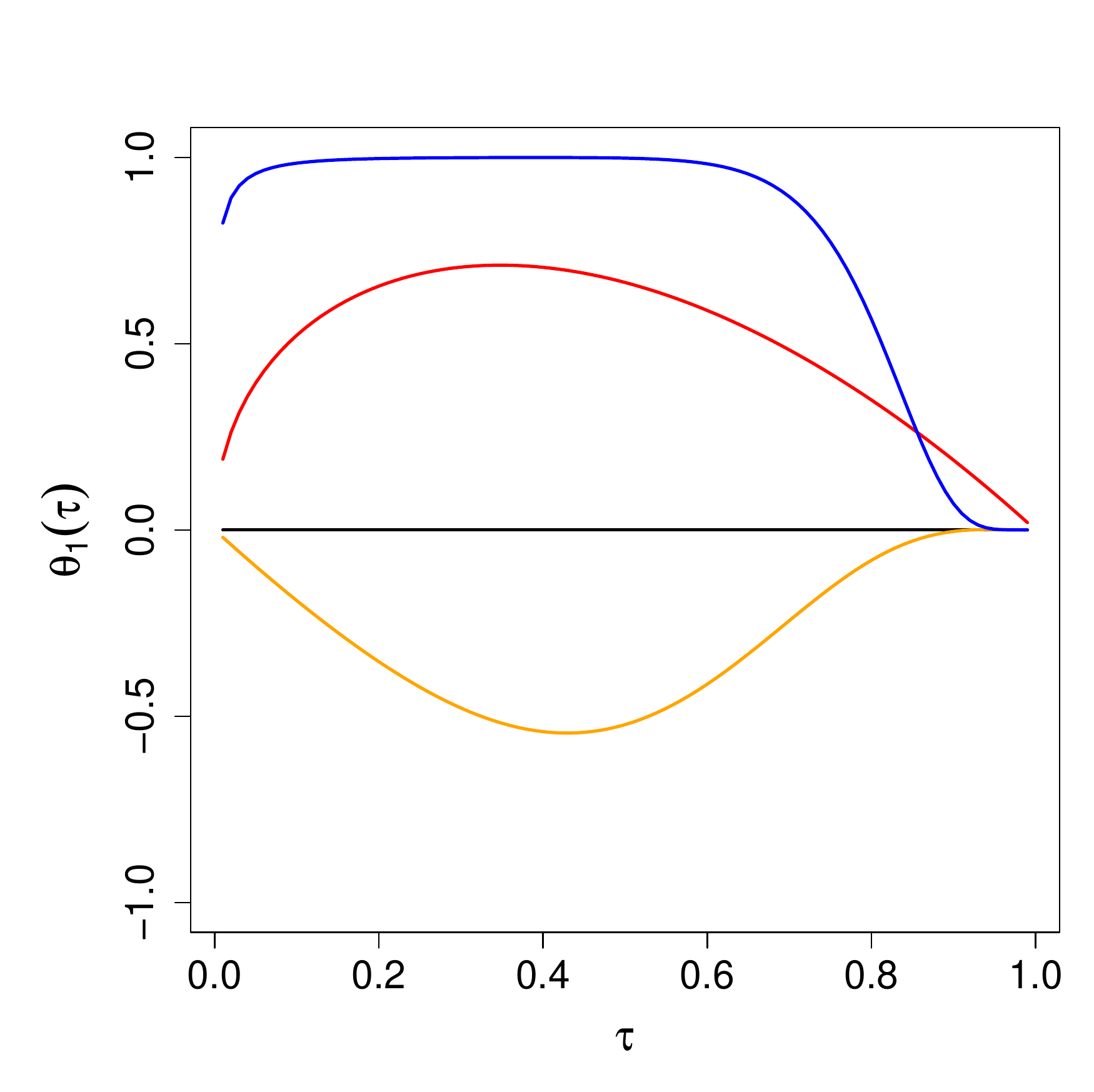}\\
\includegraphics[width=5cm]{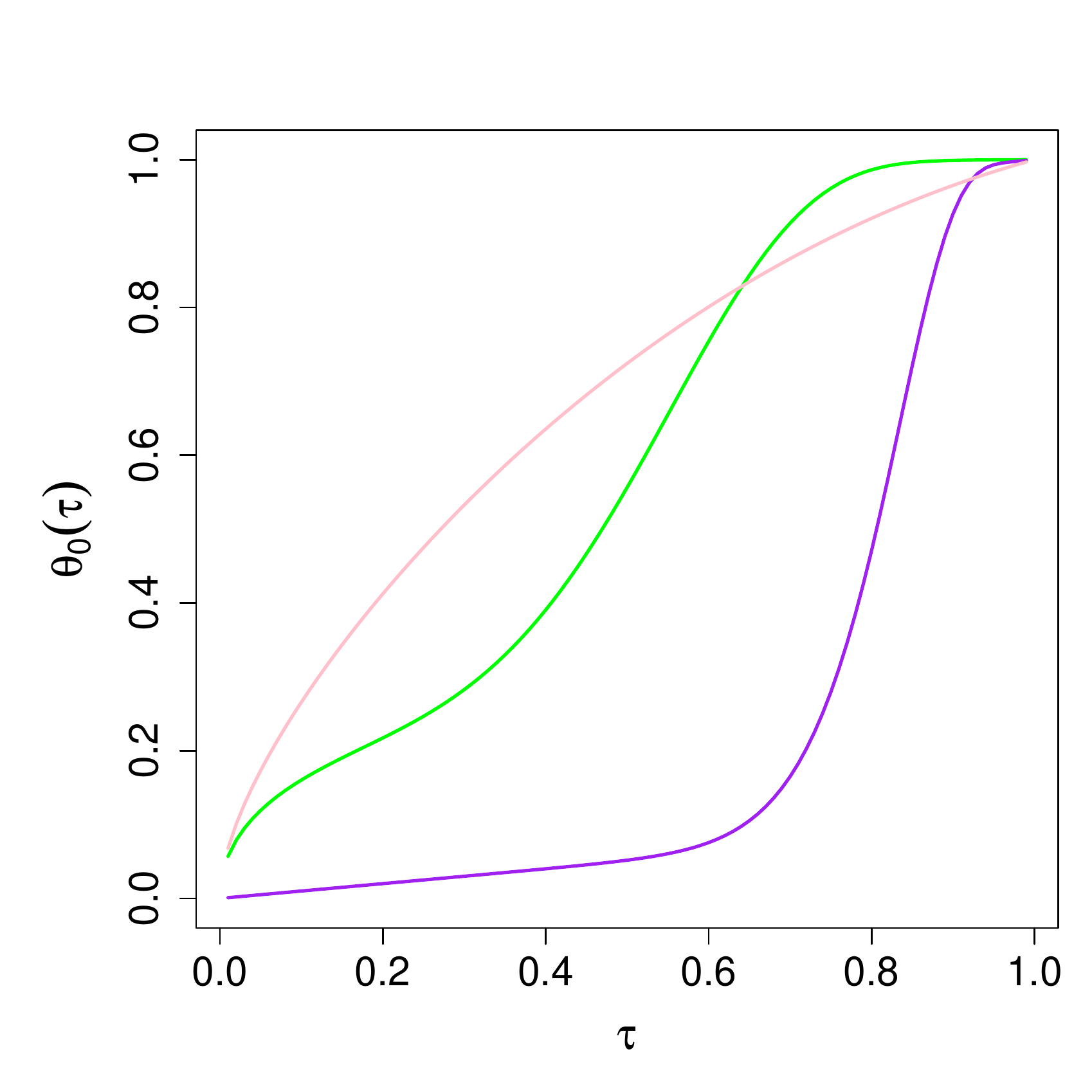}
\includegraphics[width=5cm]{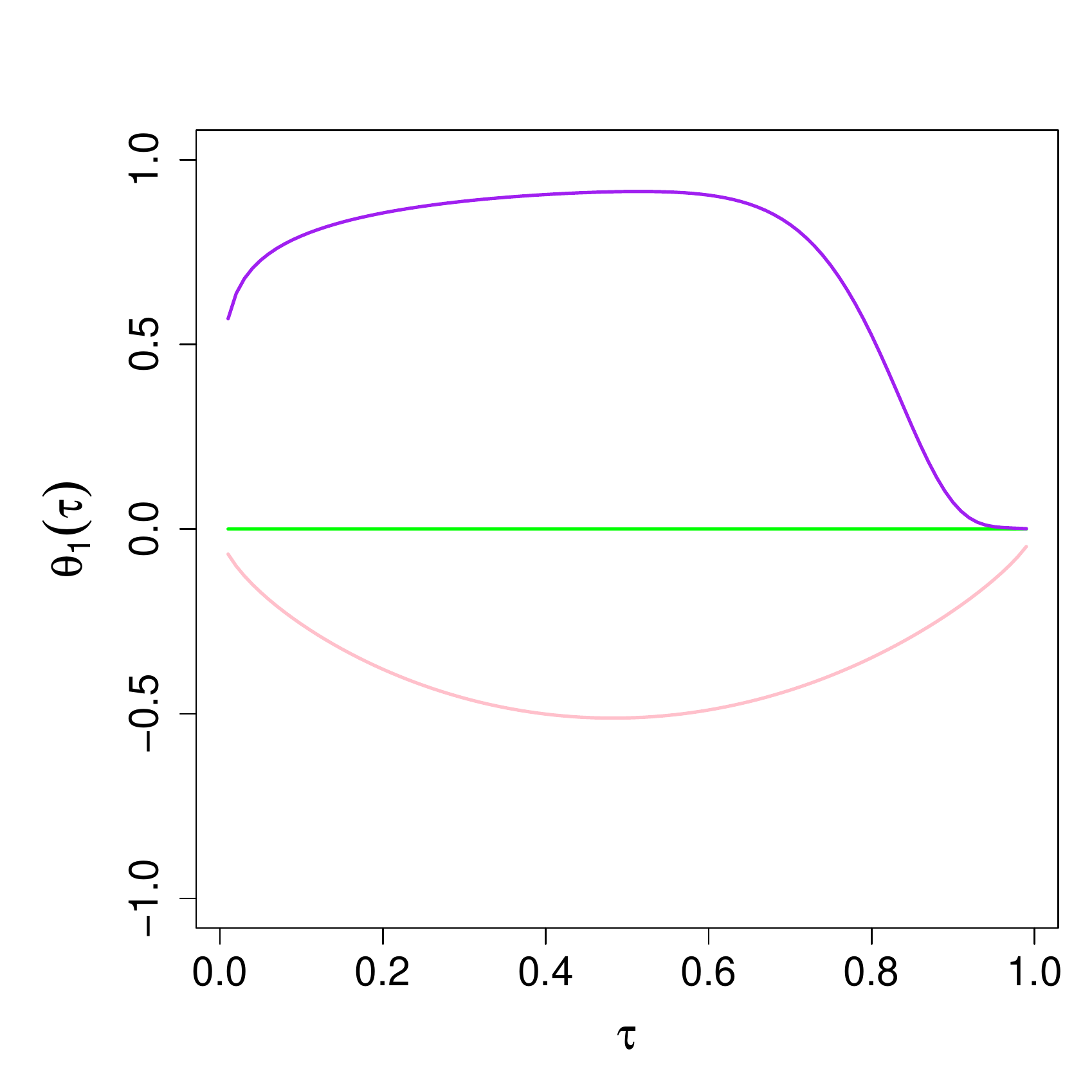}
\caption{Functions $\theta_0(\tau)$ (left) and $\theta_1(\tau)$ (right) vs. $\tau$. Under Scenarios~SC1, (black), SC2 (orange), SC3 (red), and SC4 (blue) used to generate QAR1K1 series (above); and under  Scenarios~SC5 (green), SC6 (purple), and SC7 (pink) used to generate QAR1K2 series (below). }
\label{fig:app:scenariosK1}
\end{figure}

\begin{table}[tb]
\centering
\begin{tabular}{c|ccccccccccc}
Scenario & $a_{1,1}$ &  $b_{1,1}$ & $a_{1,2}$ & $b_{1,2}$ & $a_{2,1}$ &  $b_{2,1}$ & $a_{2,2}$ &  $b_{2,2}$  & $\lambda_1$ & $\lambda_2$  \\ 
  \hline
SC5  & 0.5 & 2.0 & 4.0 & 8.0 & 0.5 & 2.0 & \phantom{1}4.0 & 8.0 & 0.3 & 0.3 \\ 
SC6  & 0.5 & 2.0 & 0.3 & 6.0 & 1.0 & 1.0 & 12.0 & 8.0 & 0.4 & 0.1 \\ 
SC7  & 3.0 & 0.5 & 2.0 & 1.0 & 1.0 & 2.0 & \phantom{1}0.5 & 1.0 & 0.2 & 0.4 \\ 
\end{tabular}
\caption{Values of the parameters of the QAR1K2 models used to generate data under Scenarios~SC5--SC7. (First subscript indicates the mixture function $\eta_j$ and second subscript indicates the $k$-th component of that mixture to which the parameter corresponds.)}
\label{tab:app:param}
\end{table}

The results are satisfactory, with the average of the CVG's close to the nominal level of $0.90$, varying between $0.86$ and $0.92$ in all the scenarios for both $T=150$ and $500$.  Not only the averages but also the individual CVG's for each $\tau$, even for extreme quantiles ($\tau=0.01, 0.99$), are close to $0.90$ (see Tables~\ref{tab:app:compK1} and \ref{tab:app:scenarioK2thetas}).  It is noteworthy that Scenarios~SC3, SC4 and SC6 could correspond to a usual correlation structure in climate and environmental data, with a strong positive dependence in the central values of $\tau$ that weakens at the extremes. 

\begin{table}[tb]
\centering
\begin{tabular}{cc|cccccc}
Scenario & $T$ & $\theta_0(0.01) $ & $\theta_0(0.99) $ & $\theta_0(0.50) $ & $\theta_1(0.01) $ & $\theta_1(0.99) $ & $\theta_1(0.50) $ \\ 
  \hline
 \multirow{2}{*}{SC5} & $150$ & 0.97 & 0.91 & 0.97  & 0.86 & 0.91 & 0.97  \\ 
  & $500$ & 0.90 & 0.80 & 0.93 & 0.86 & 0.80 & 0.93  \\ 
 \multirow{2}{*}{SC6} & $150$  & 0.95 & 0.85 & 0.90  & 0.87 & 0.83 & 0.91  \\ 
  & $500$  & 0.90 & 0.93 & 0.89 & 0.80 & 0.93 & 0.89  \\ 
 \multirow{2}{*}{SC7} & $150$  & 0.95 & 0.85 & 0.90  & 0.87 & 0.83 & 0.91 \\ 
  & $500$  & 0.94 & 0.88 & 0.89  & 0.94 & 0.86 & 0.89  \\
\end{tabular}
\caption{The $90\%$ CVG of $\theta_0(\tau)$ and $\theta_1(\tau)$ ($\tau = 0.01,0.50,0.99$) in QAR1K2 models fitted to Scenarios~SC5--SC7. }
\label{tab:app:scenarioK2thetas}
\end{table}

To study the flexibility of the QAR1K1 model, Table~\ref{tab:compK1K2} summarizes the metrics described in Section~\ref{sec:metrics}, obtained from  fitting the QAR1K1 and QAR1K2 models to the previous data series generated with QAR1K2 models.  The values of $\bar R^1$ and $\tilde p_2$ are the metrics averaged across the $B=100$ simulations.  QAR1K1 models are quite flexible and their metrics are only slightly poorer than QAR1K2 models. Note that series generated under Scenario~SC5 are independent and consequently $\bar R^1$ is expected to be zero. 

\begin{table}[tb]
\centering
\begin{tabular}{cc|cccc}
 &  &\multicolumn{2}{c}{QAR1K1} & \multicolumn{2}{c}{QAR1K2} \\
Scenario & $T$ & $\bar R^1$ & $\tilde p_2$ & $\bar R^1$ & $\tilde p_2$ \\
  \hline
\multirow{2}{*}{SC5} & 150 & $-0.0387$ & 0.0019 & 0.0021 & 0.0002 \\
 & 500 & $-0.0439$ & 0.0019 & 0.0002 & 0.0001 \\
\multirow{2}{*}{SC6} & 150 & \phantom{$-$}0.3082 & 0.0003 & 0.3091 & 0.0001 \\
 & 500 & \phantom{$-$}0.3155 & 0.0002 & 0.3171 & 0.0000 \\ 
\multirow{2}{*}{SC7} & 150 & \phantom{$-$}0.0781 & 0.0003 & 0.0776 & 0.0001 \\
 & 500 & \phantom{$-$}0.0737 & 0.0001 & 0.0737 & 0.0001 \\
\end{tabular}
\caption{Adequacy and comparison metrics for QAR1K1 and QAR1K2 models in Scenarios~SC5--SC7. (Values are averaged across the $B=100$ simulations.)}
\label{tab:compK1K2}
\end{table}

\clearpage

\section{Application to temperature data}

\subsection{Time series data}

\begin{figure}[!ht]
\centering
\includegraphics[width=3cm]{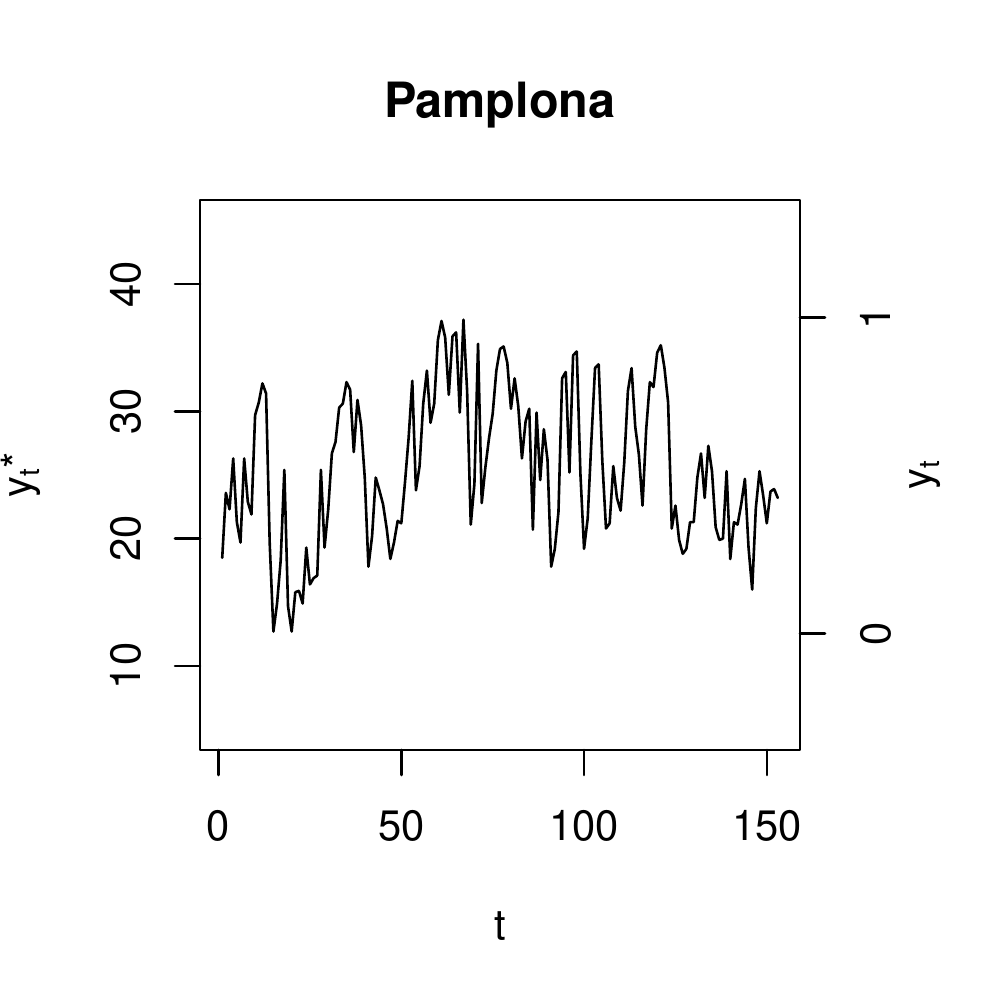}
\includegraphics[width=3cm]{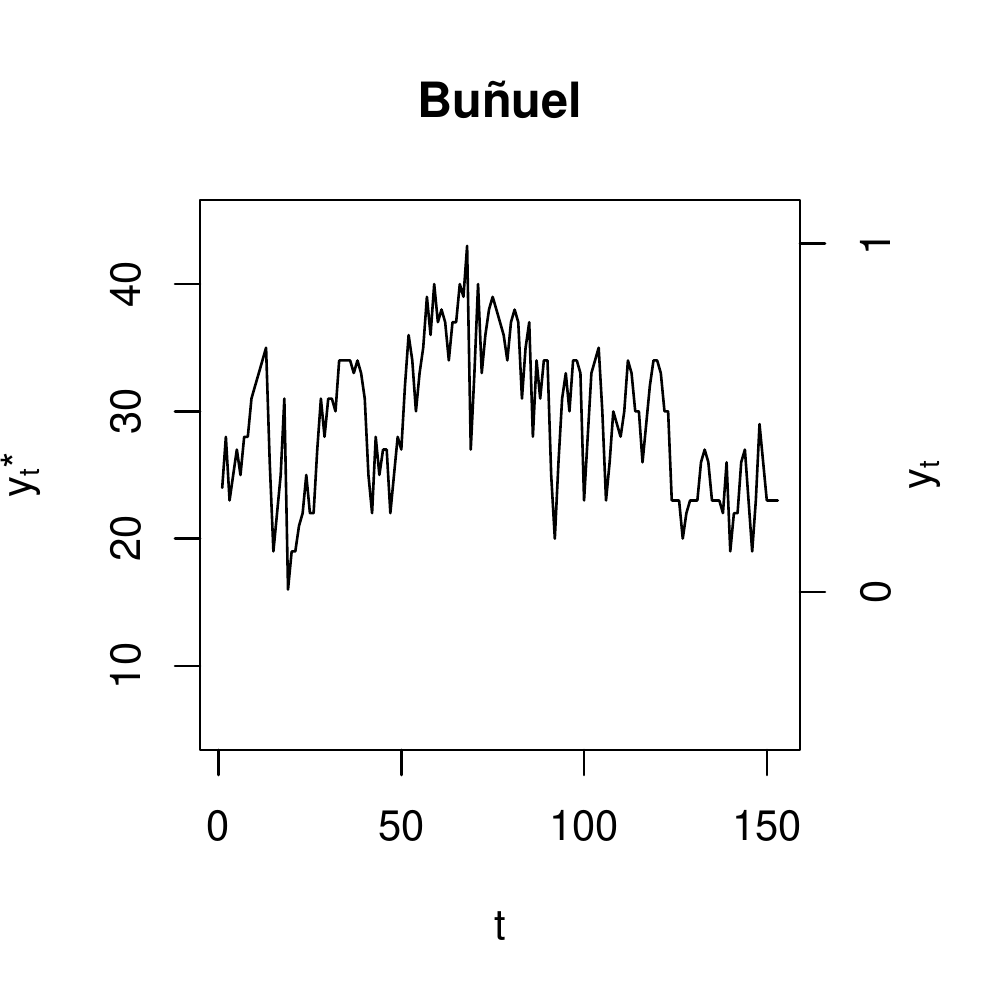}
\includegraphics[width=3cm]{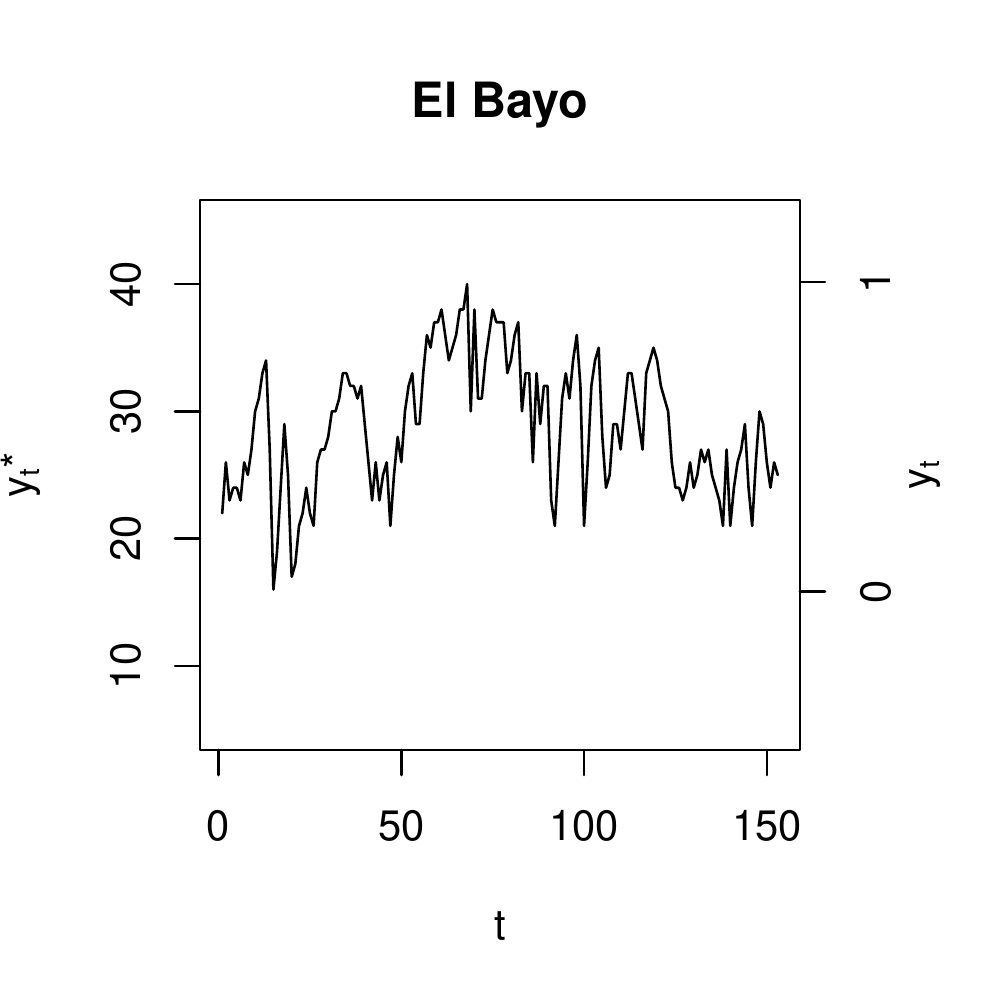}
\includegraphics[width=3cm]{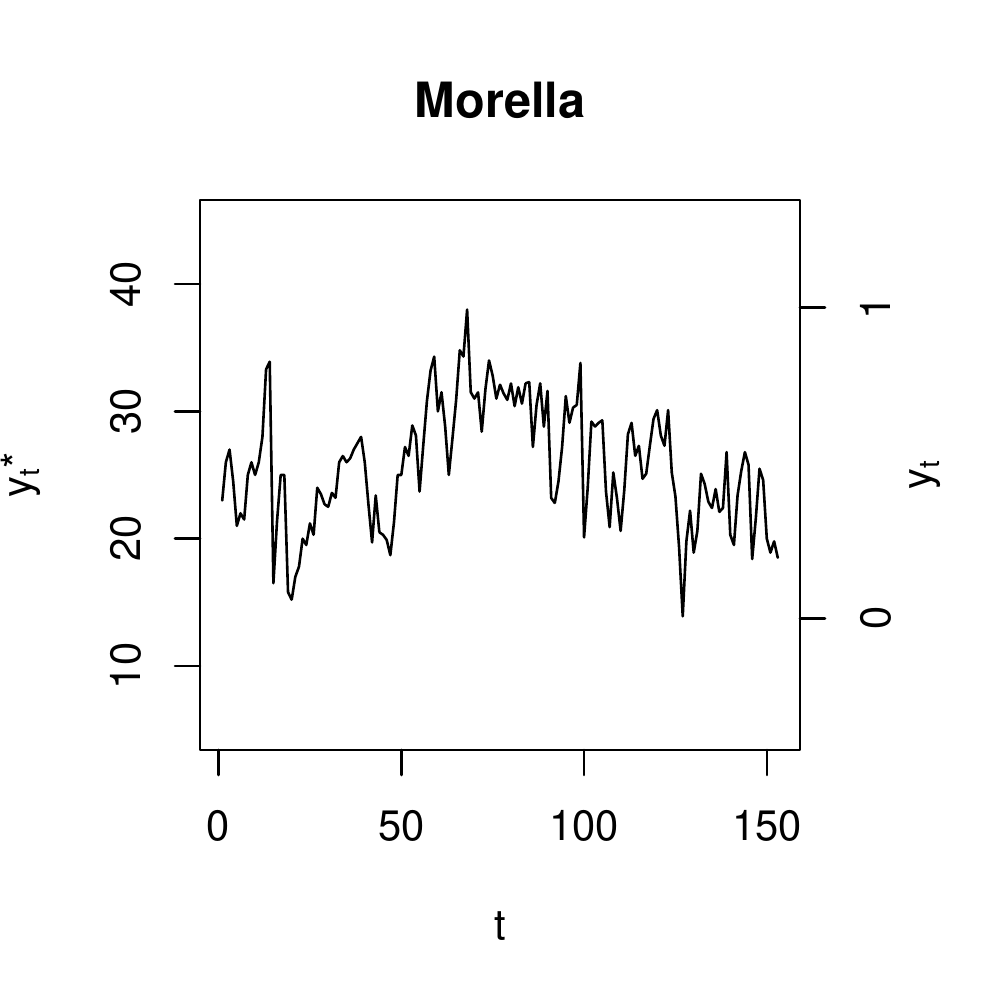} \\
\includegraphics[width=3cm]{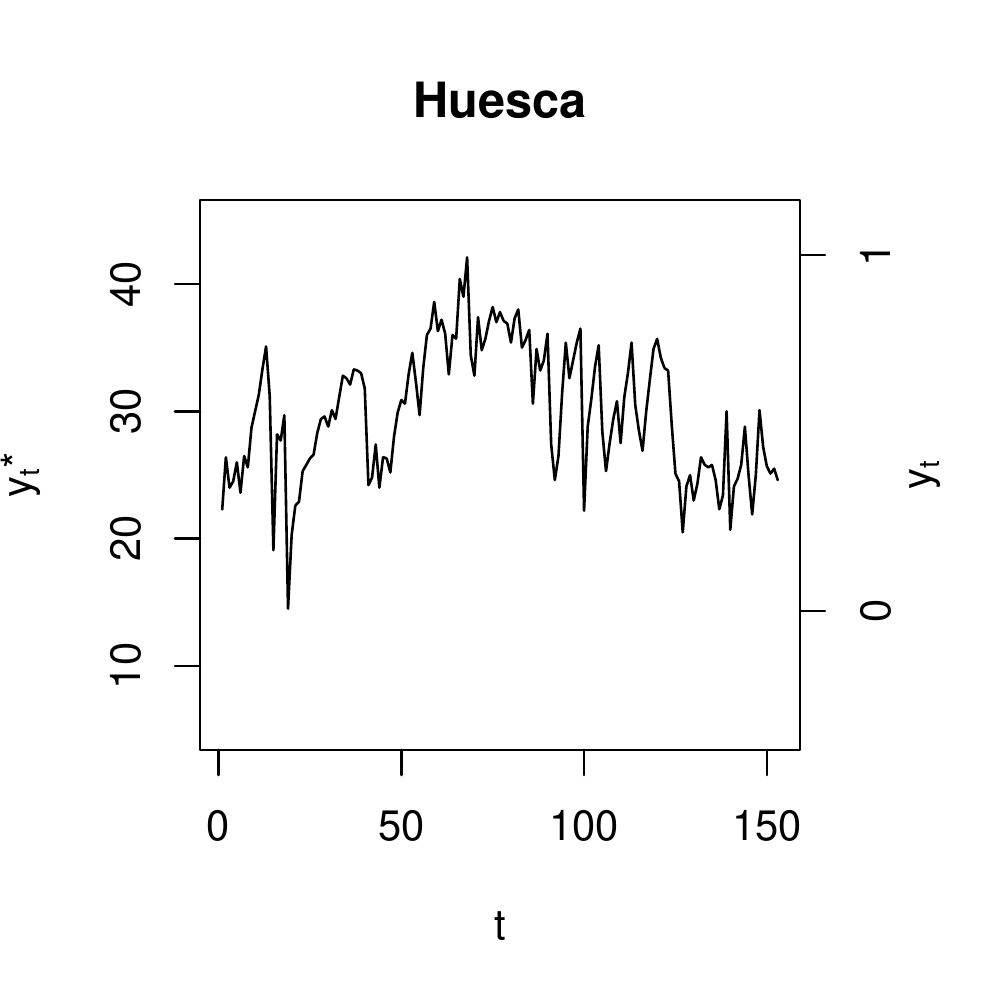}
\includegraphics[width=3cm]{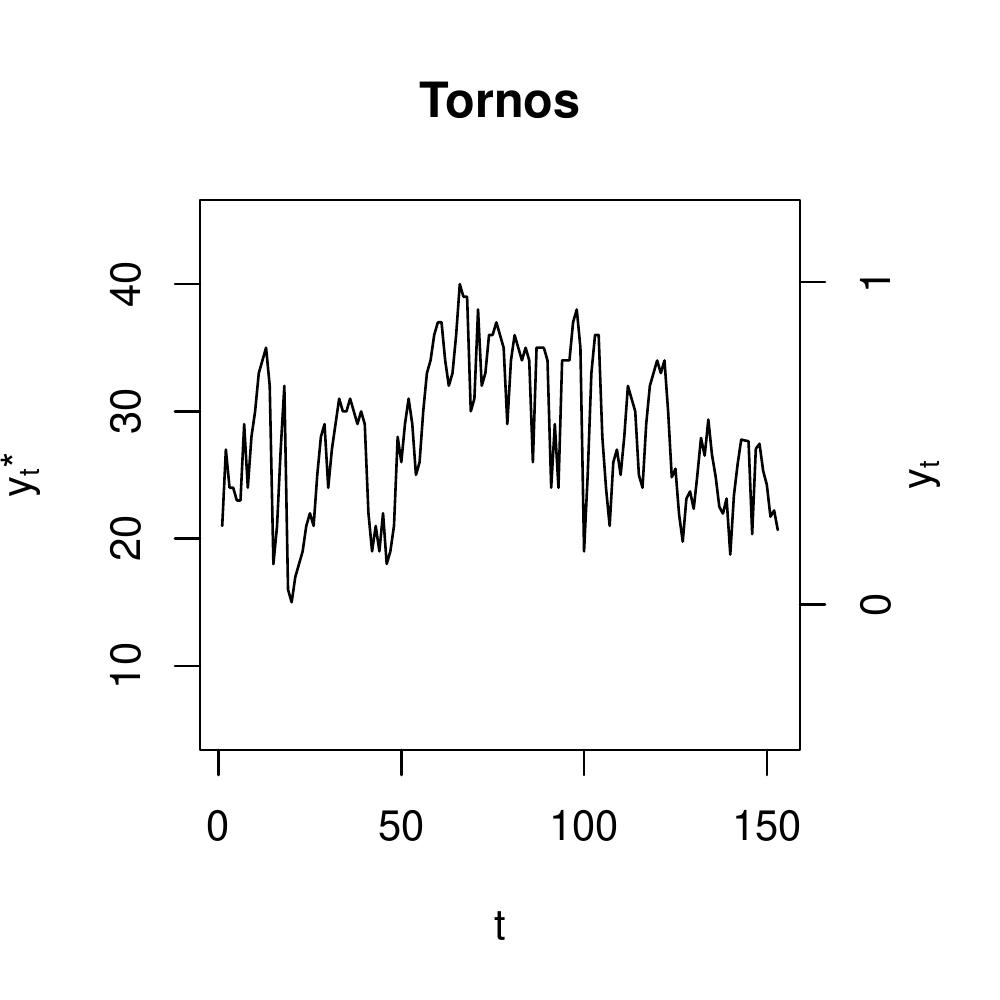}
\includegraphics[width=3cm]{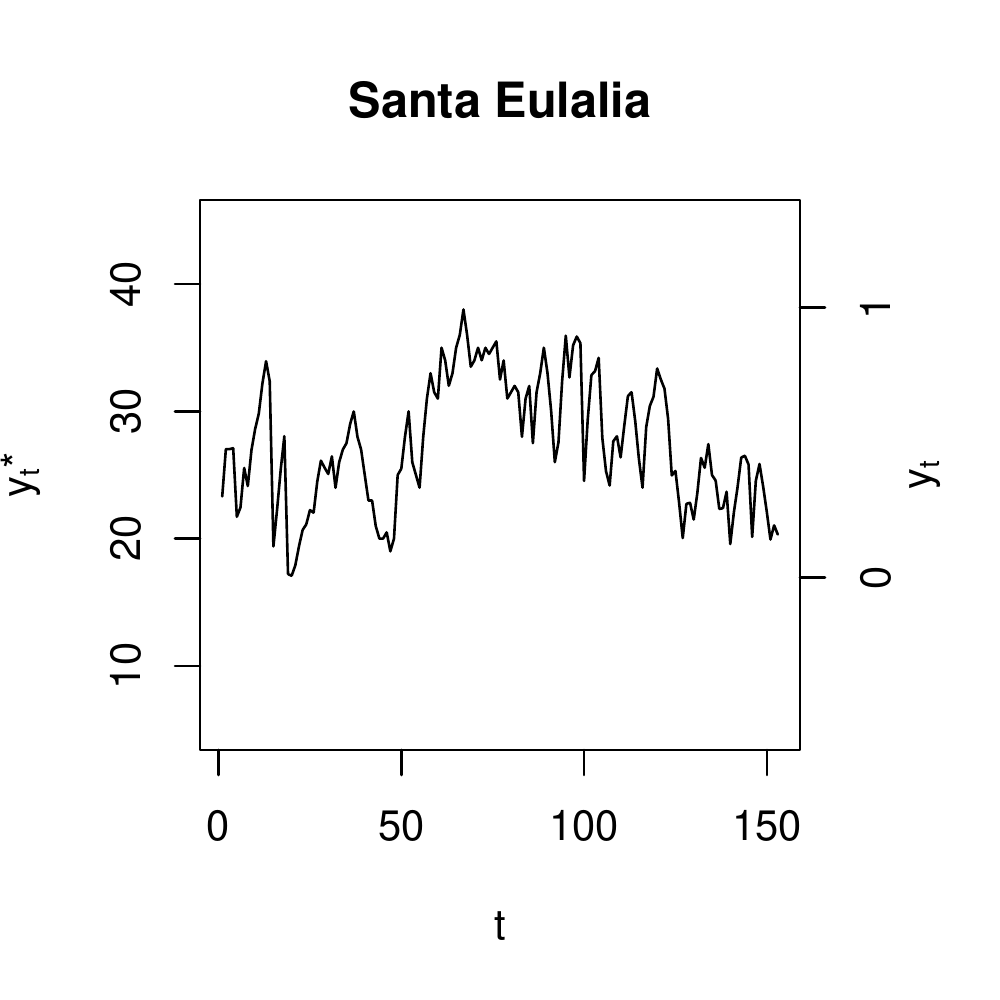}
\includegraphics[width=3cm]{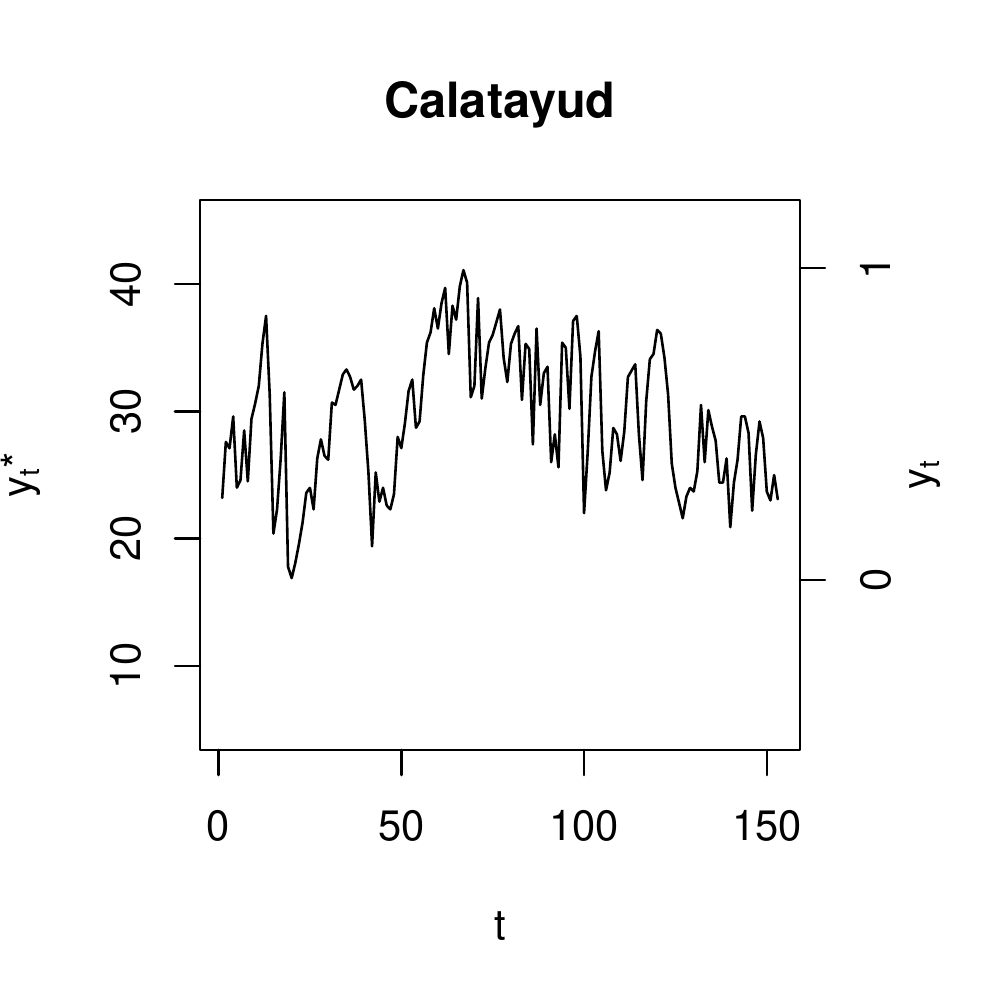} \\
\includegraphics[width=3cm]{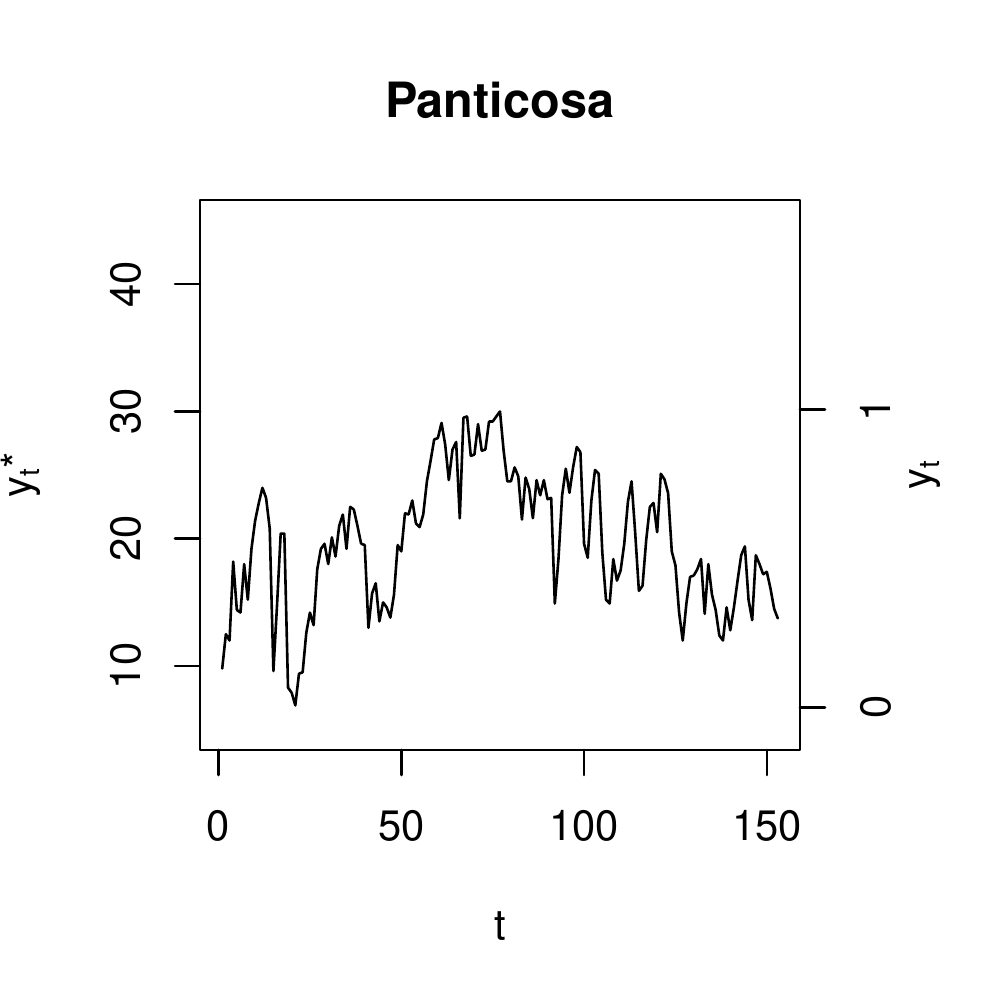}
\includegraphics[width=3cm]{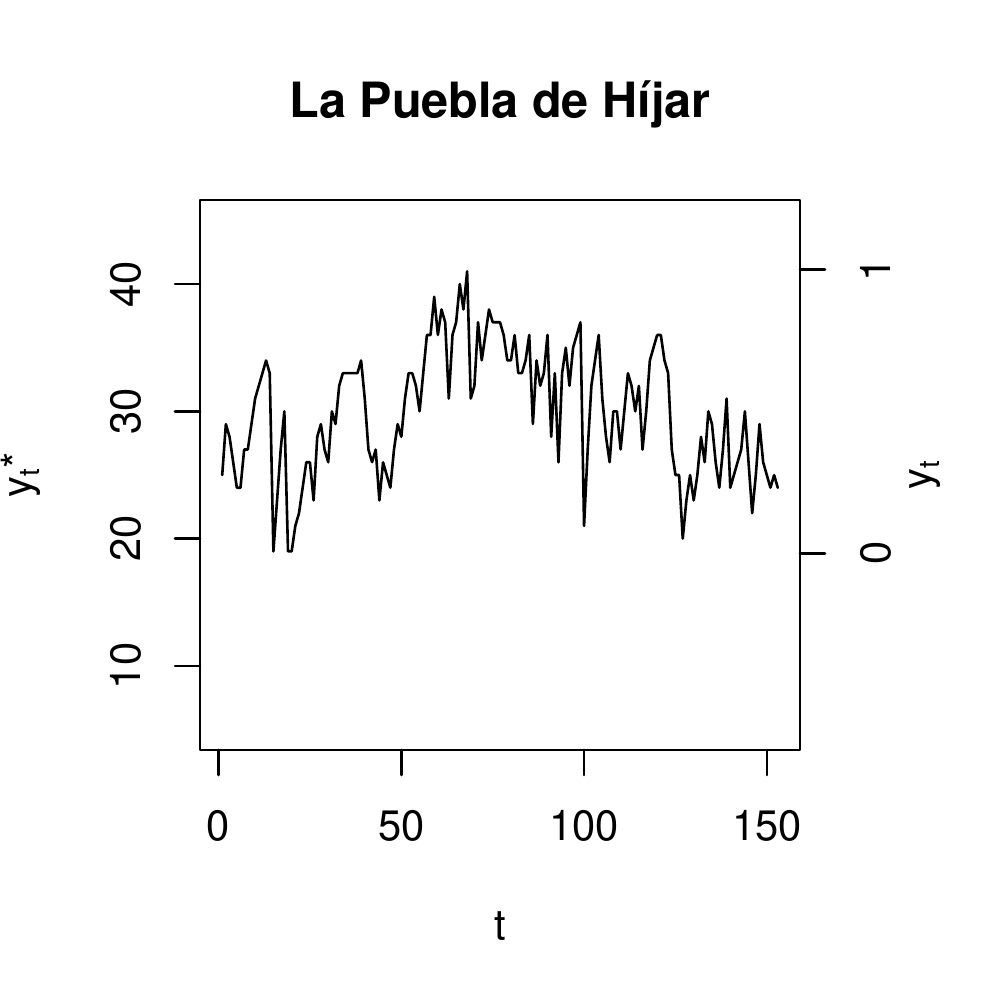}
\includegraphics[width=3cm]{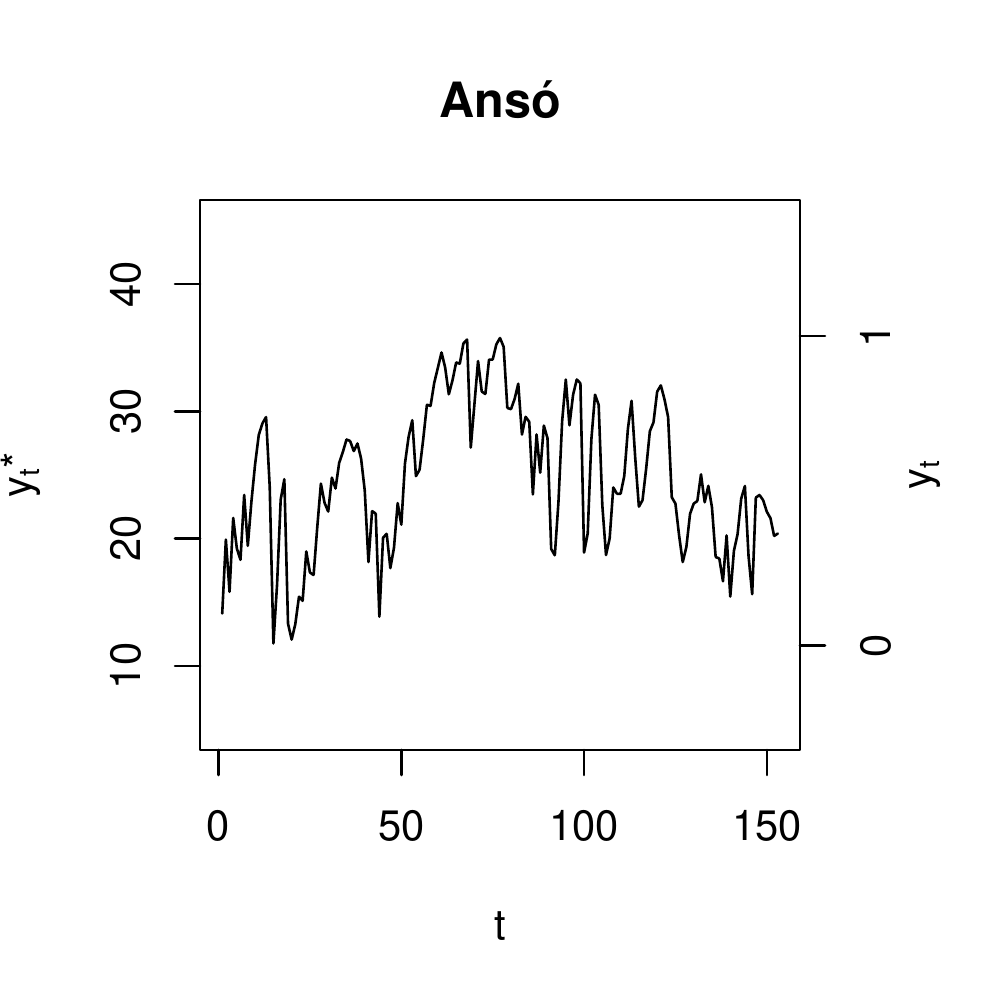}
\includegraphics[width=3cm]{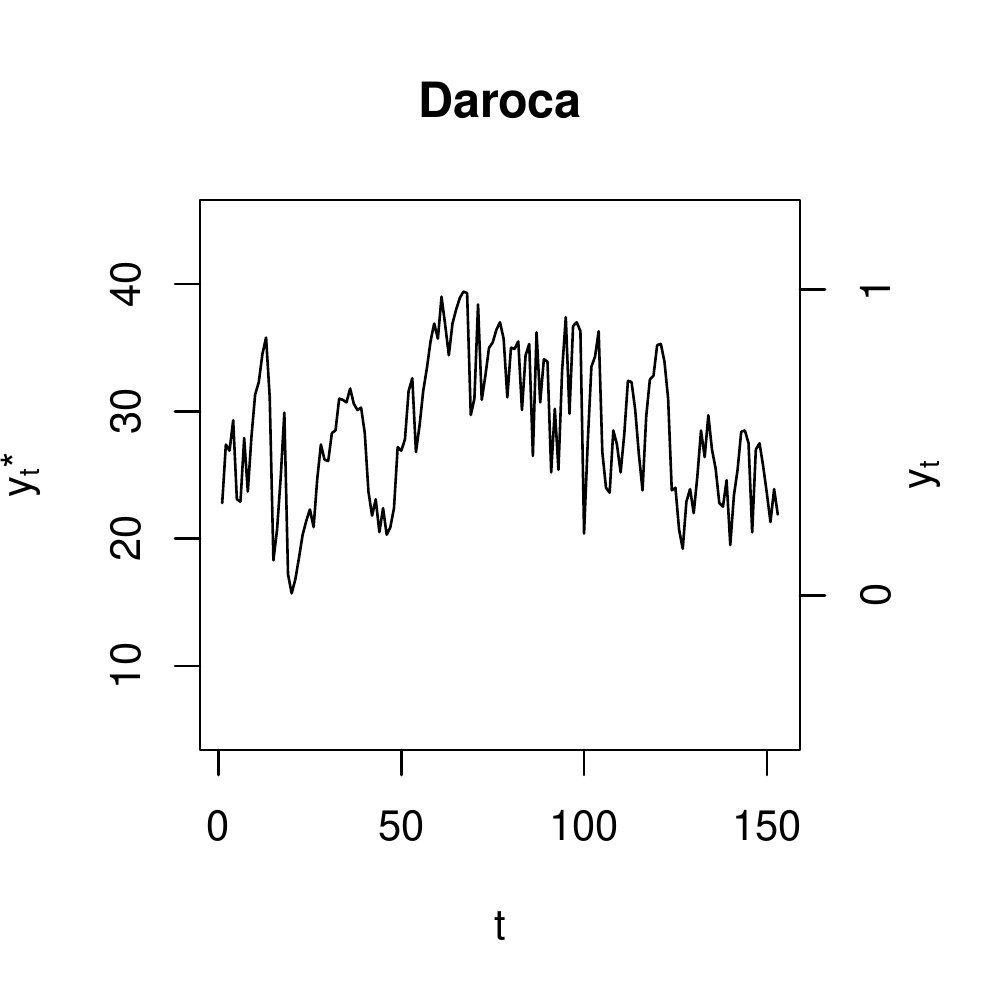} \\
\includegraphics[width=3cm]{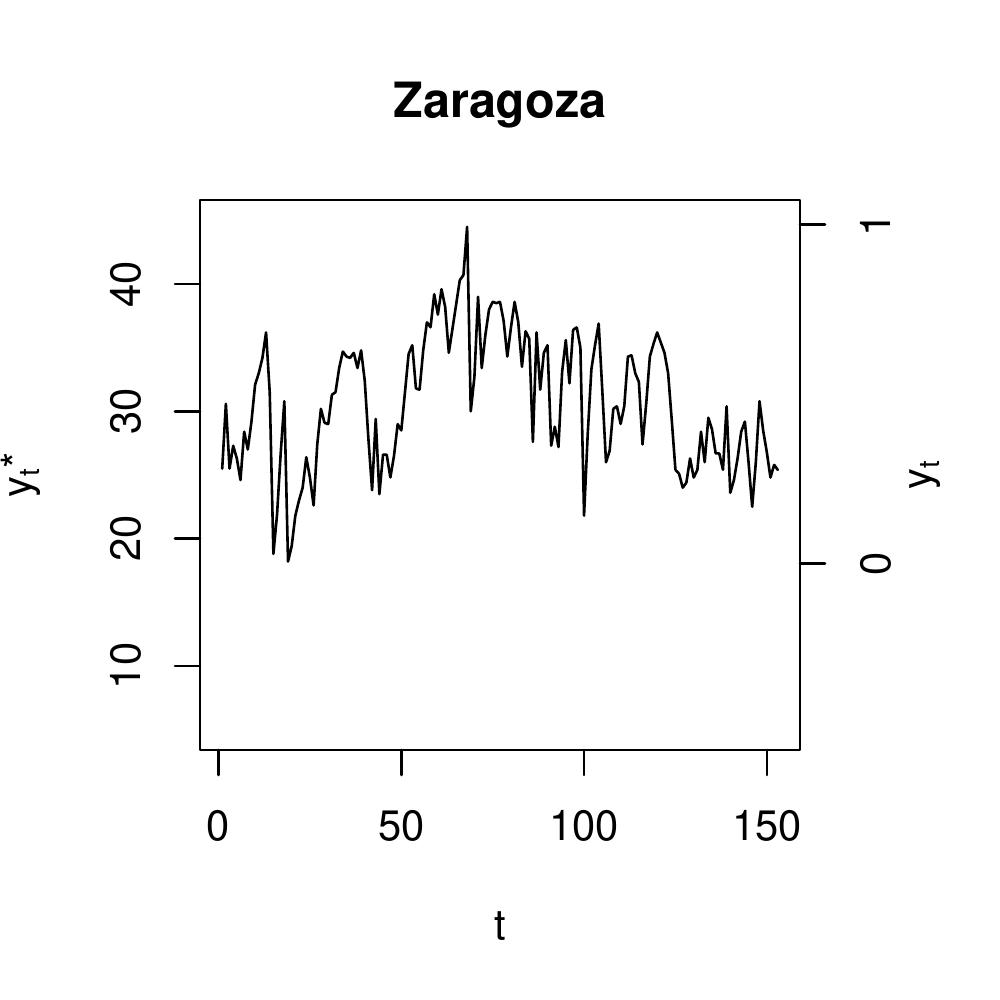}
\includegraphics[width=3cm]{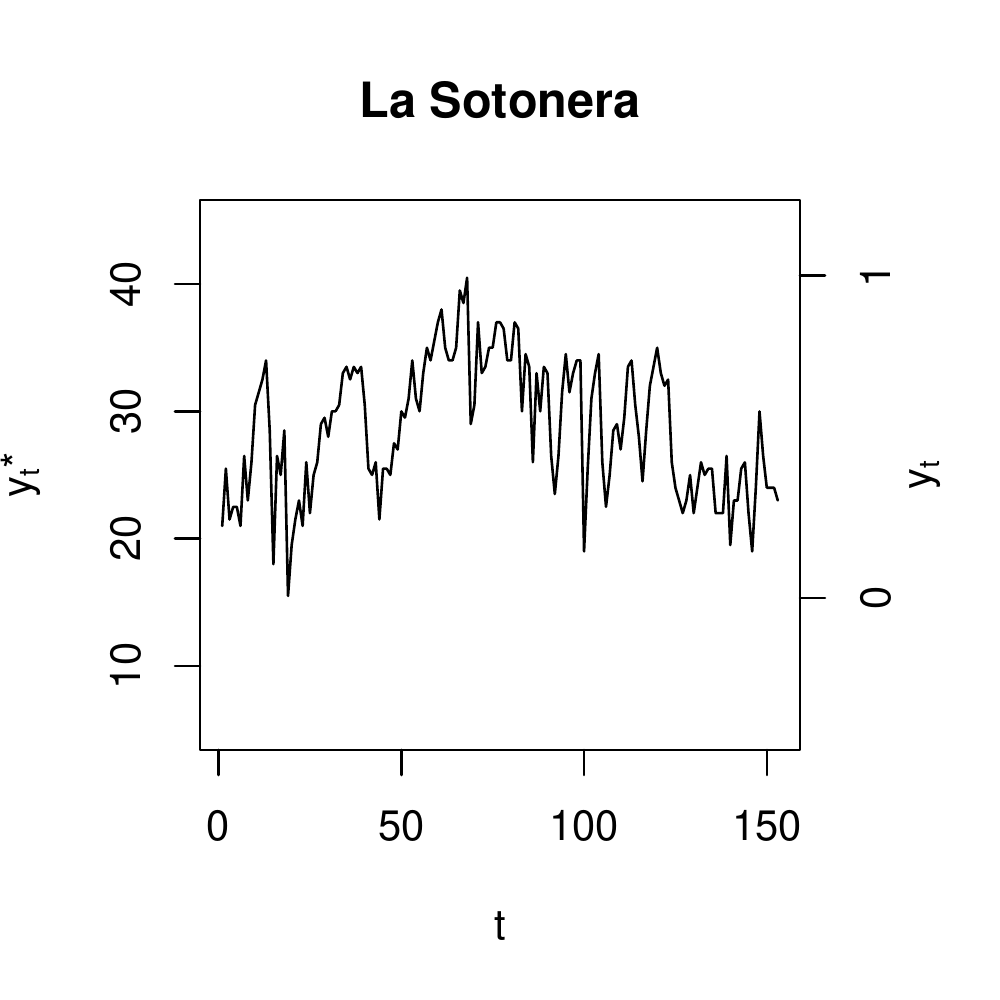}
\includegraphics[width=3cm]{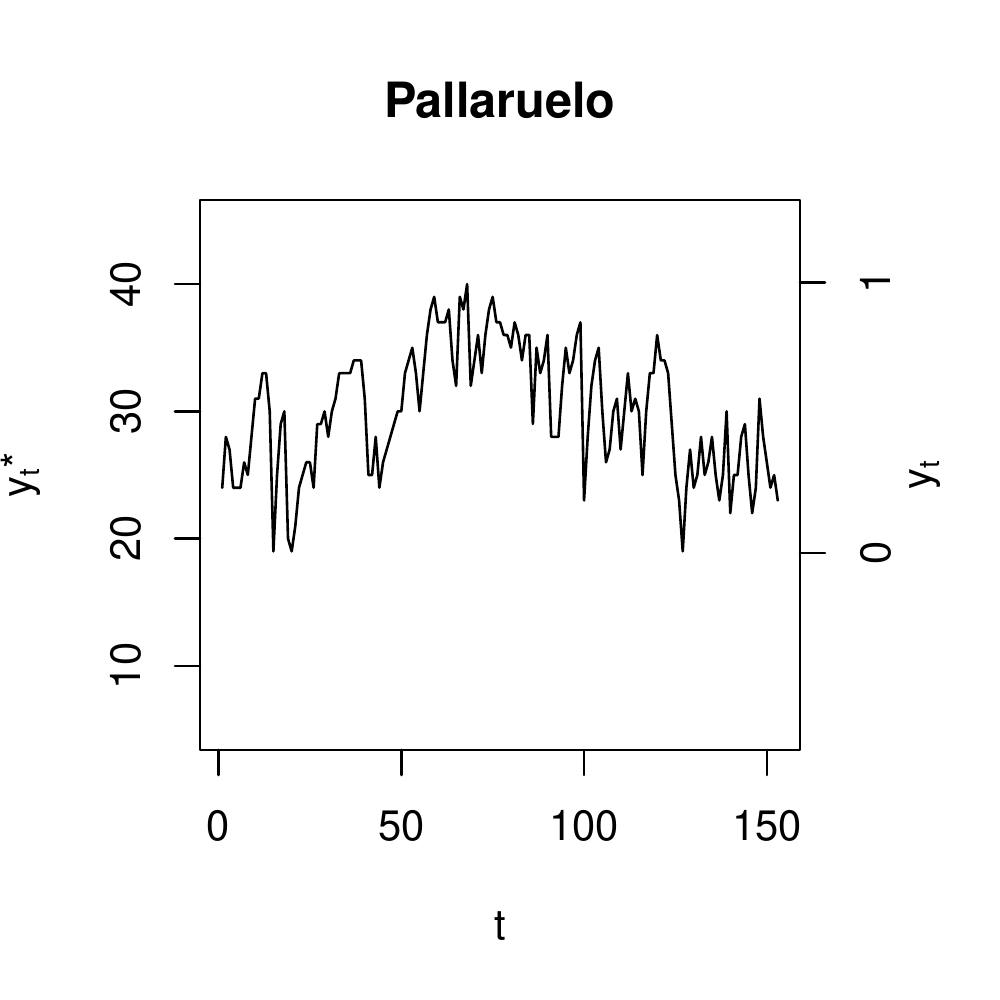}
\includegraphics[width=3cm]{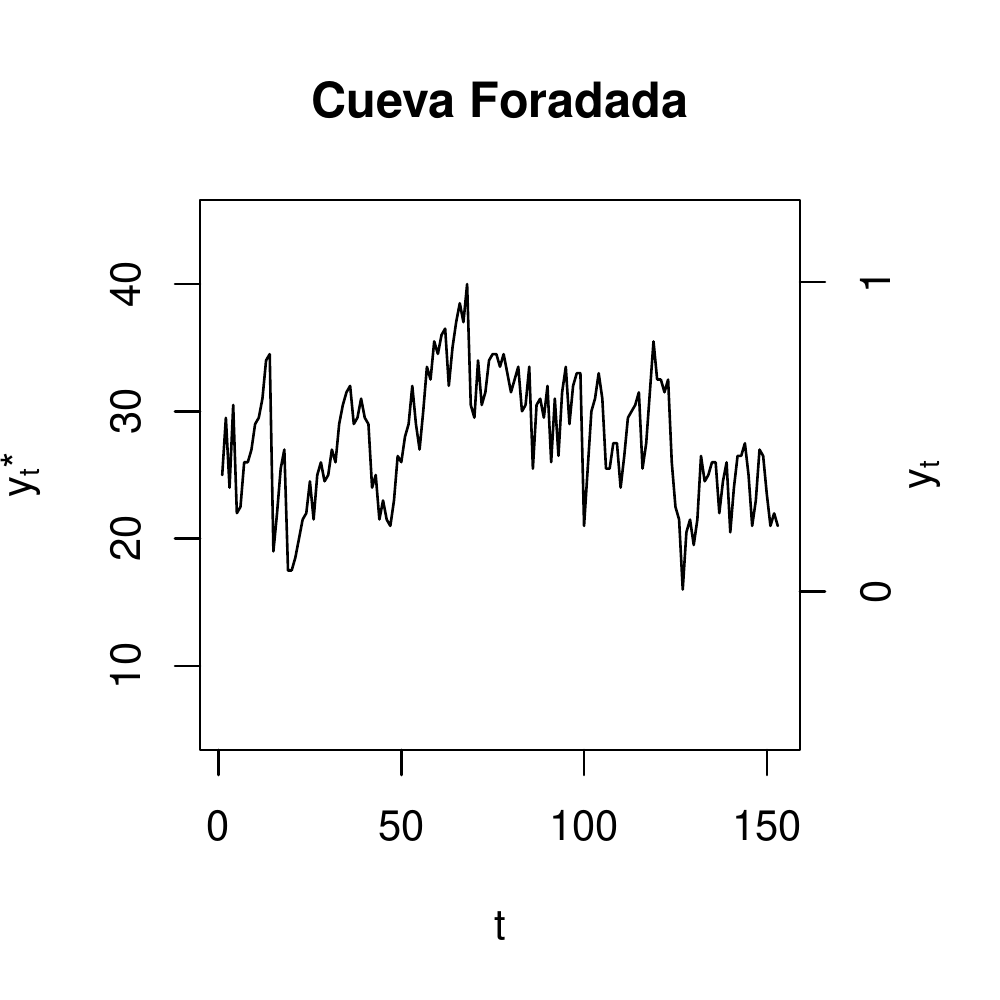} \\
\includegraphics[width=3cm]{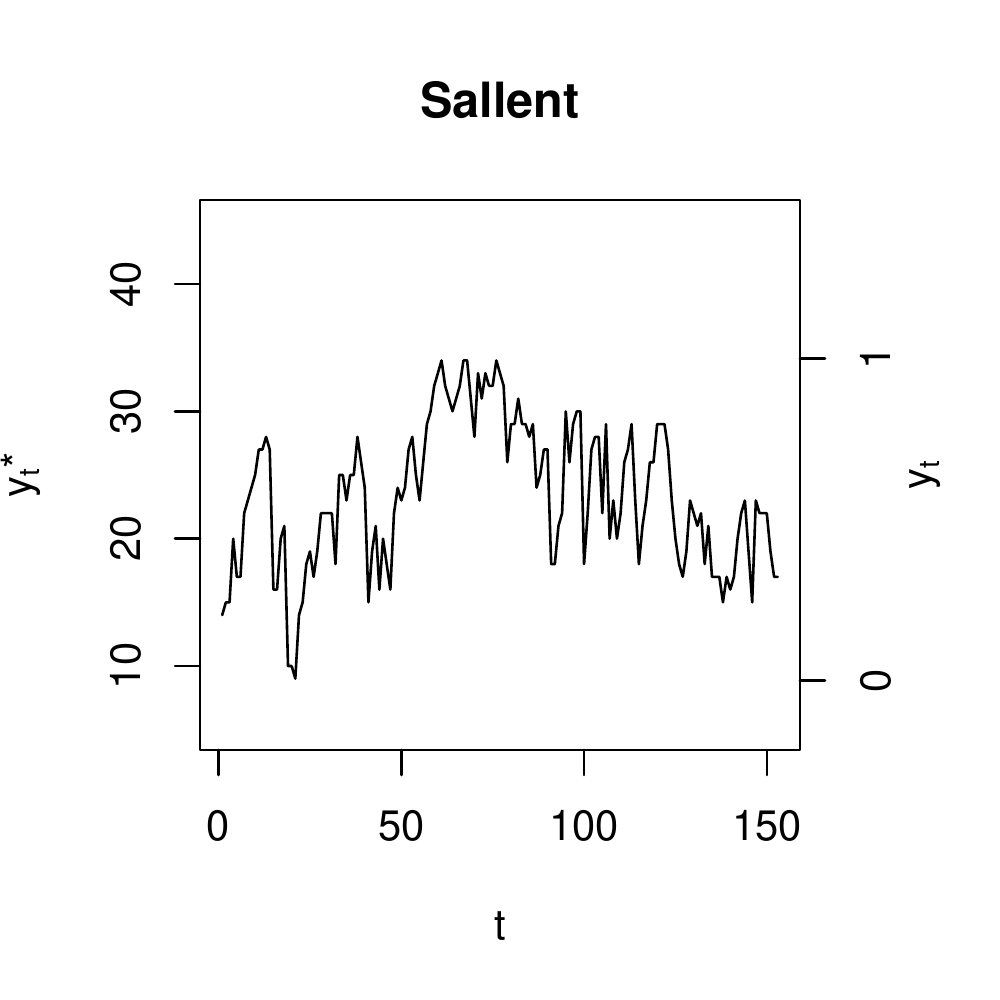}
\includegraphics[width=3cm]{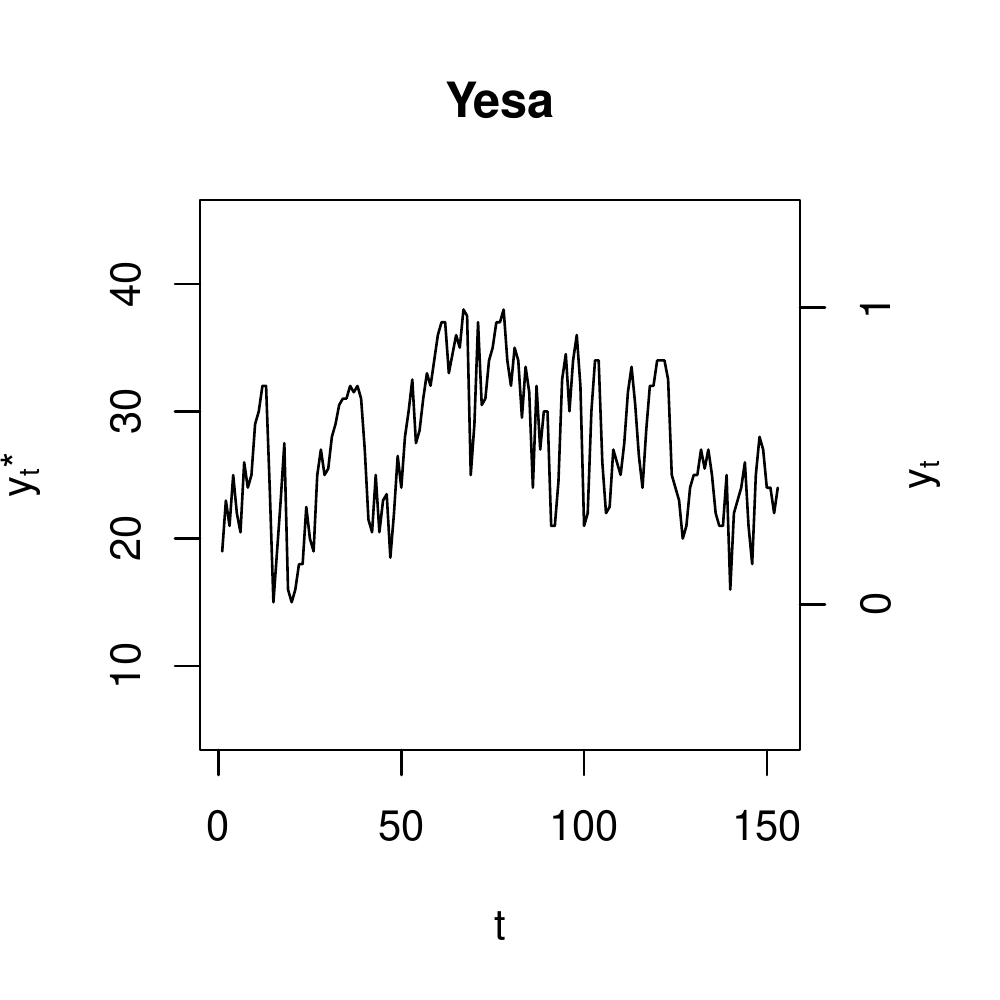}
\caption{Daily maximum temperature time series of the 18 locations in MJJAS, 2015.}
\label{fig:ts}
\end{figure}

\clearpage

\subsection{Model comparison}

\begin{table}[!ht] \footnotesize
\centering
\begin{tabular}{l|cc|cc|cc|cc}
 & \multicolumn{2}{c|}{QAR1K1} & \multicolumn{2}{c|}{QAR1K2} & \multicolumn{2}{c|}{QAR2K1} & \multicolumn{2}{c}{KX2006} \\
Location & $\tilde{p}_2$ & $\bar{R}^1$ & $\tilde{p}_2$ & $\bar{R}^1$ & $\tilde{p}_2$ & $\bar{R}^1$ & $\tilde{p}_2$ & $\bar{R}^1$ \\ 
  \hline
Pamplona & 0.75 & 0.267 & 0.46 & 0.270 & 0.65 & 0.265 & 0.65 & 0.274 \\ 
  Buñuel & 0.56 & 0.358 & 0.35 & 0.356 & 0.52 & 0.357 & 0.82 & 0.337 \\ 
  El Bayo & 0.67 & 0.362 & 0.39 & 0.364 & 0.61 & 0.361 & 0.44 & 0.331 \\ 
  Morella & 0.59 & 0.352 & 0.35 & 0.351 & 0.45 & 0.351 & 0.73 & 0.318 \\ 
  Huesca & 0.69 & 0.381 & 0.50 & 0.382 & 0.39 & 0.390 & 0.83 & 0.342 \\ 
  Tornos & 0.61 & 0.352 & 0.39 & 0.352 & 0.55 & 0.352 & 0.54 & 0.306 \\ 
  Santa Eulalia & 0.60 & 0.454 & 0.31 & 0.455 & 0.52 & 0.453 & 0.66 & 0.445 \\ 
  Calatayud & 0.57 & 0.335 & 0.44 & 0.335 & 0.50 & 0.336 & 1.01 & 0.328 \\ 
  Panticosa & 0.60 & 0.447 & 0.36 & 0.448 & 0.51 & 0.443 & 0.40 & 0.431 \\ 
  La Puebla de Híjar & 0.80 & 0.333 & 0.37 & 0.336 & 0.64 & 0.339 & 0.69 & 0.301 \\ 
  Ansó & 0.42 & 0.411 & 0.34 & 0.410 & 0.37 & 0.410 & 0.51 & 0.384 \\ 
  Daroca & 0.58 & 0.334 & 0.44 & 0.335 & 0.56 & 0.337 & 0.37 & 0.297 \\ 
  Zaragoza & 0.62 & 0.351 & 0.45 & 0.351 & 0.50 & 0.356 & 0.74 & 0.304 \\ 
  La Sotonera & 0.67 & 0.341 & 0.40 & 0.344 & 0.58 & 0.345 & 0.51 & 0.300 \\ 
  Pallaruelo & 0.80 & 0.362 & 0.50 & 0.362 & 0.63 & 0.361 & 0.80 & 0.335 \\ 
  Cueva Foradada & 0.46 & 0.359 & 0.47 & 0.353 & 0.40 & 0.365 & 0.76 & 0.334 \\ 
  Sallent & 0.71 & 0.418 & 0.30 & 0.419 & 0.73 & 0.413 & 0.53 & 0.401 \\ 
  Yesa & 0.72 & 0.343 & 0.42 & 0.344 & 0.64 & 0.340 & 1.31 & 0.344 \\  \hline
  $\boldsymbol{\sum / 18}$ & \textbf{0.633}&\textbf{0.365}&\textbf{0.402}&\textbf{0.365}&\textbf{0.542}&\textbf{0.365}&\textbf{0.683}&\textbf{0.339}
\end{tabular}
\caption{Adequacy and comparison metrics for QAR1K1, QAR1K2, QAR2K1, and KX2006 models for the 18 locations and averaged across locations.} \label{tab:measurements}
\end{table}

\clearpage

\subsection{The QAR(1) case}

\begin{figure}[!ht]
\centering
\includegraphics[width=3cm]{theta0QAR1s1.pdf}
\includegraphics[width=3cm]{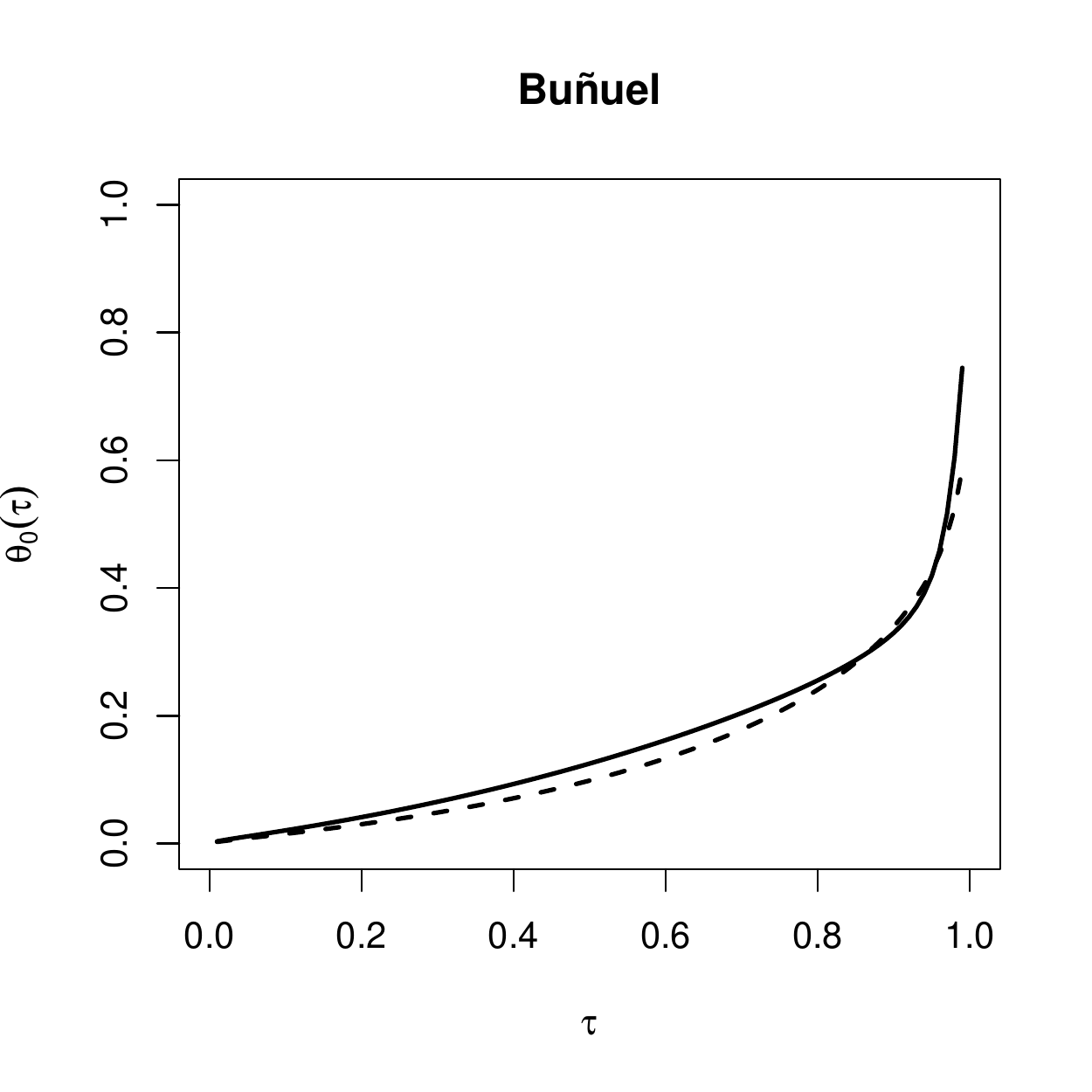}
\includegraphics[width=3cm]{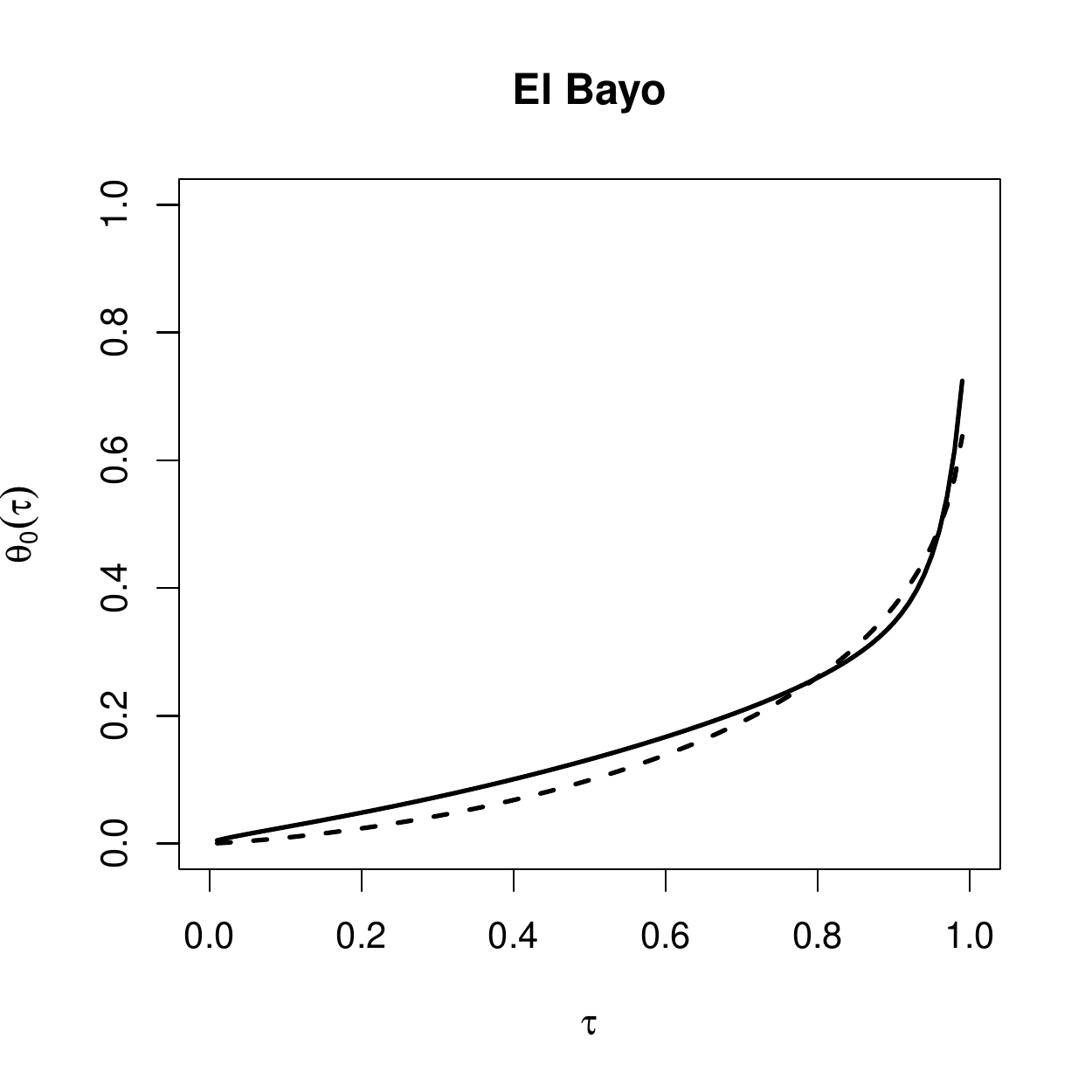}
\includegraphics[width=3cm]{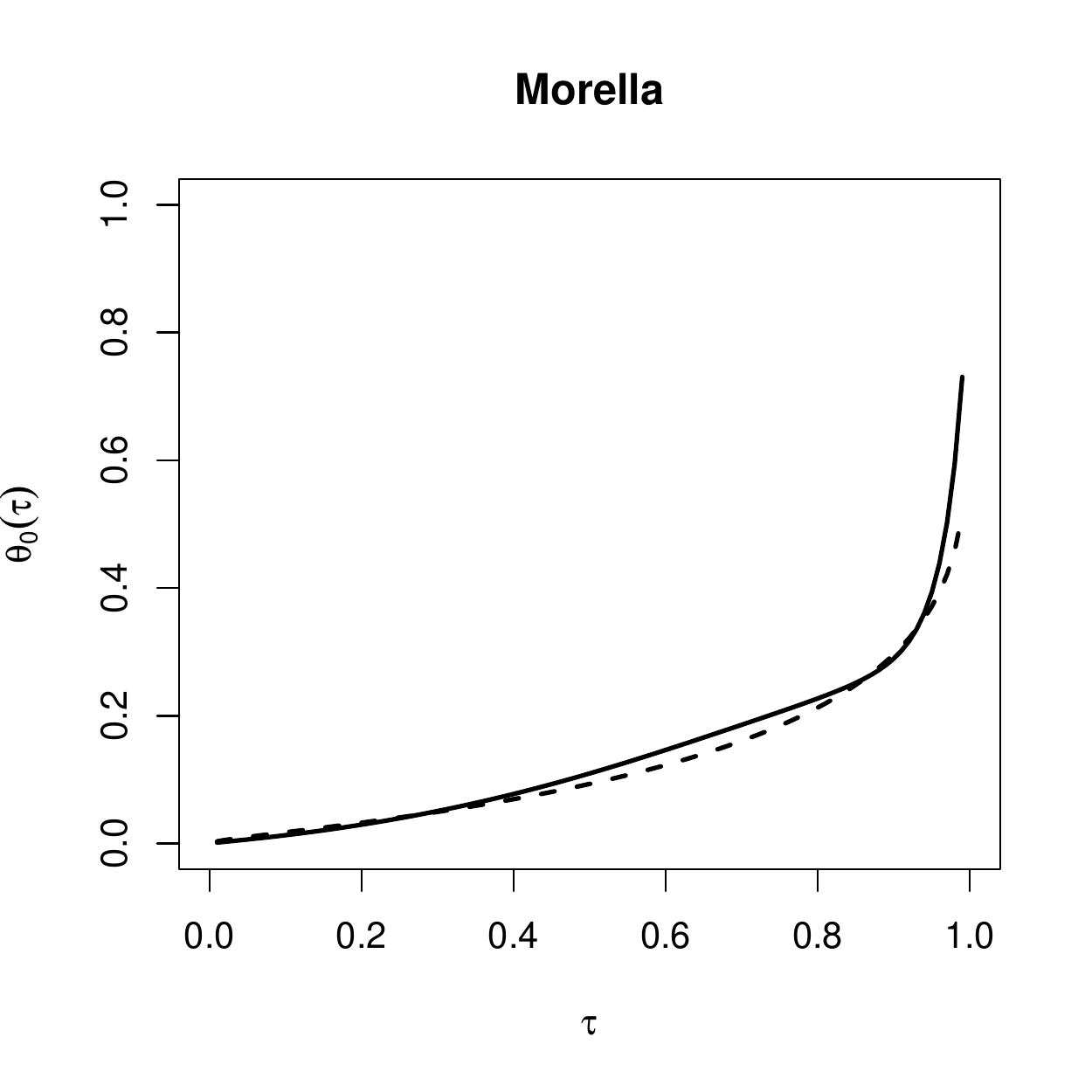} \\
\includegraphics[width=3cm]{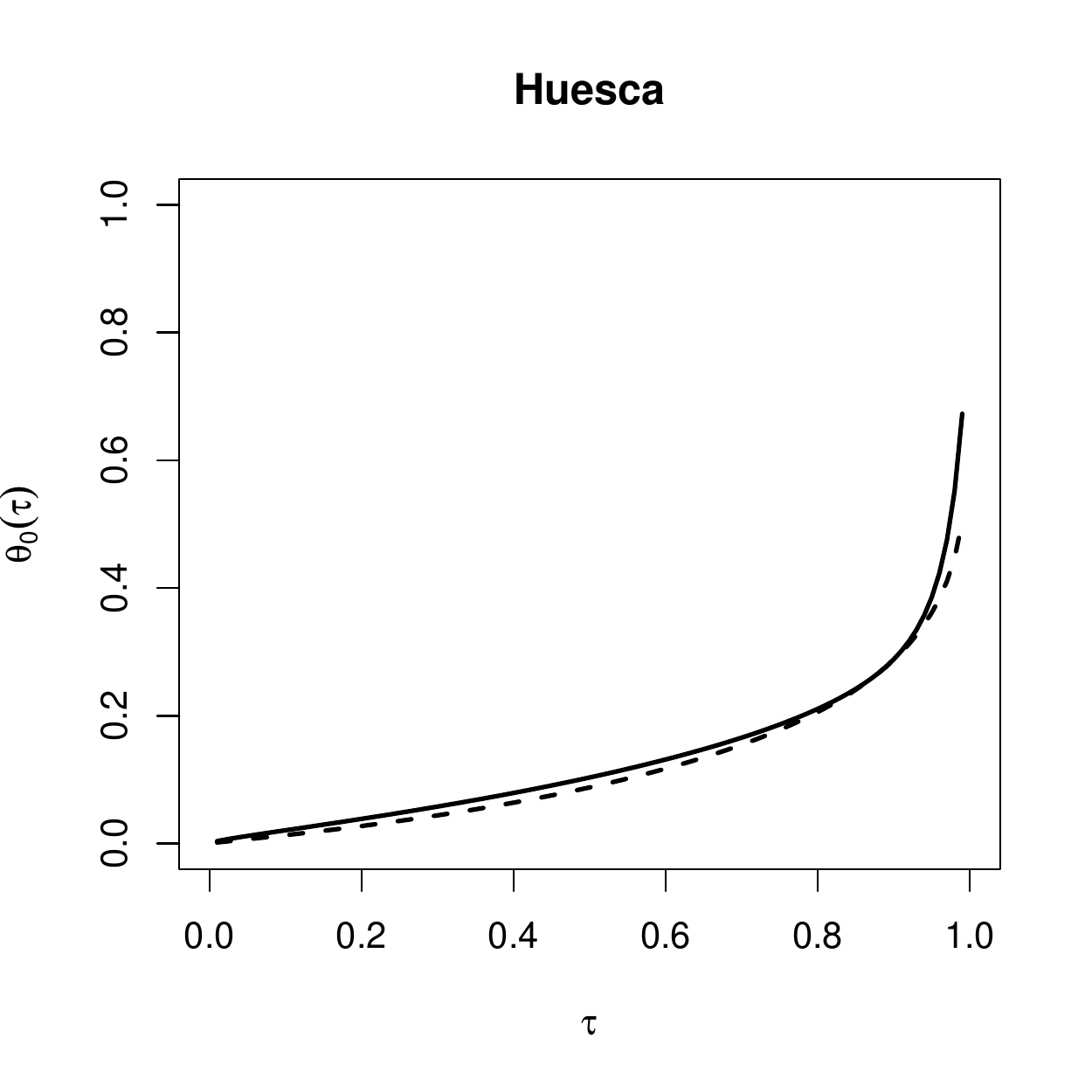}
\includegraphics[width=3cm]{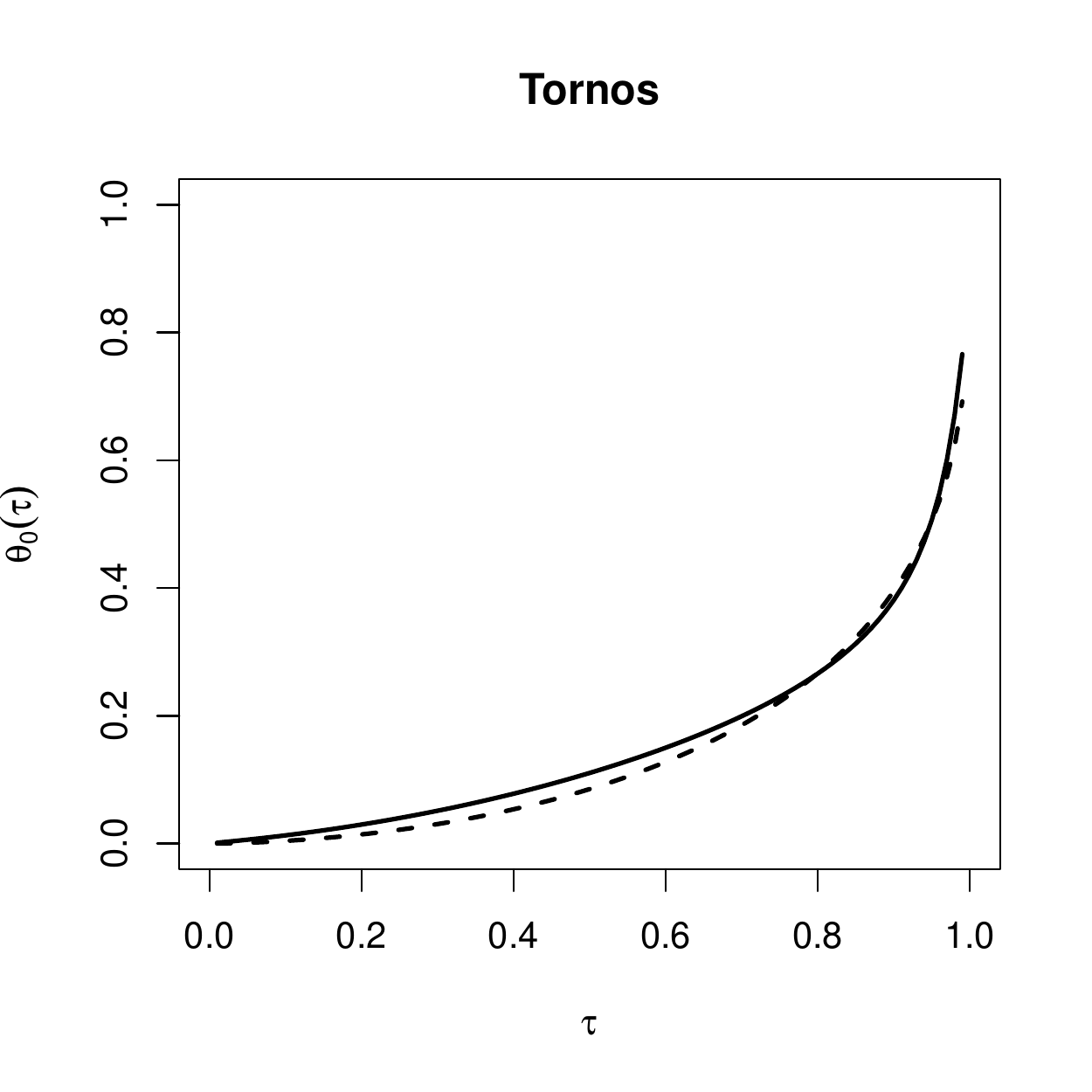}
\includegraphics[width=3cm]{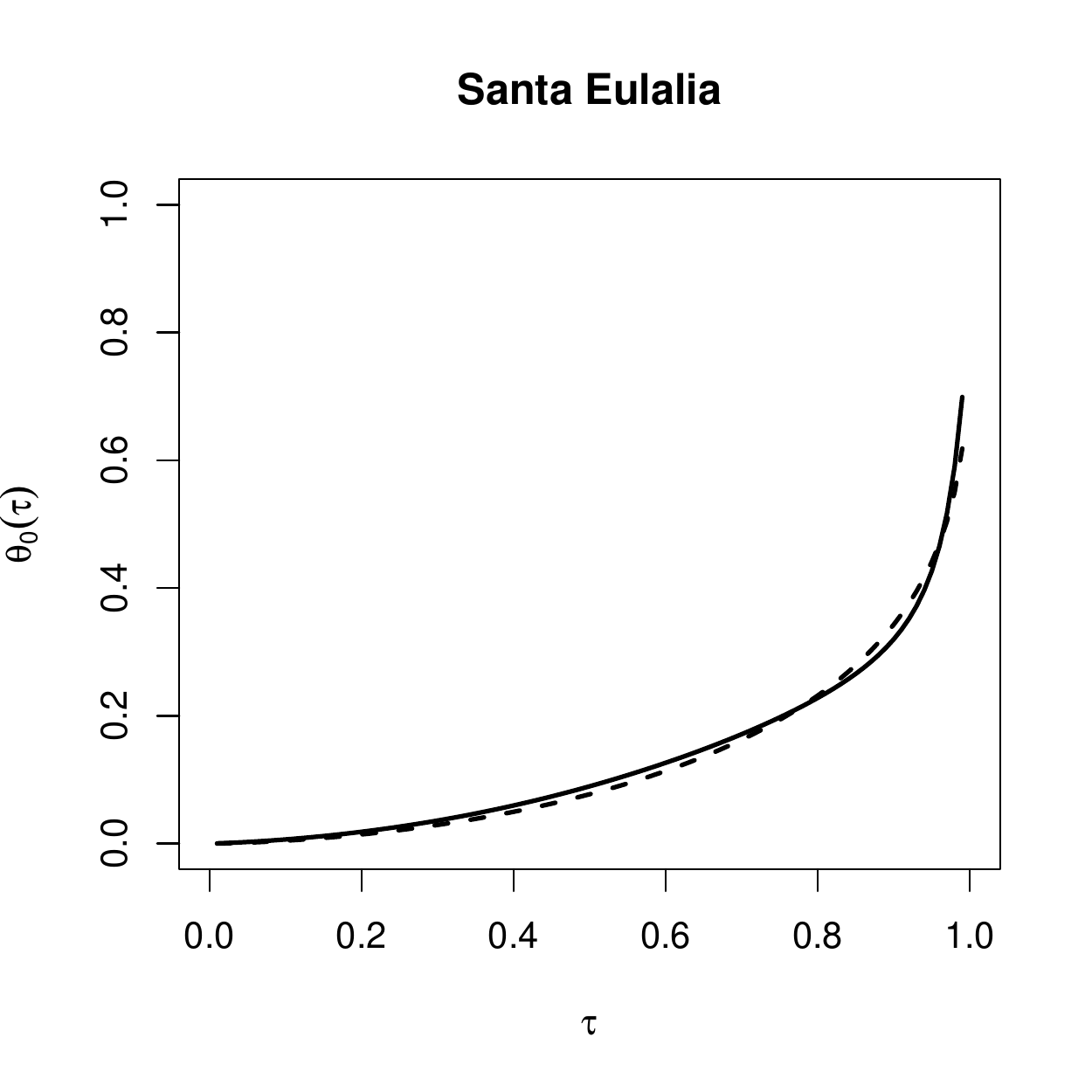}
\includegraphics[width=3cm]{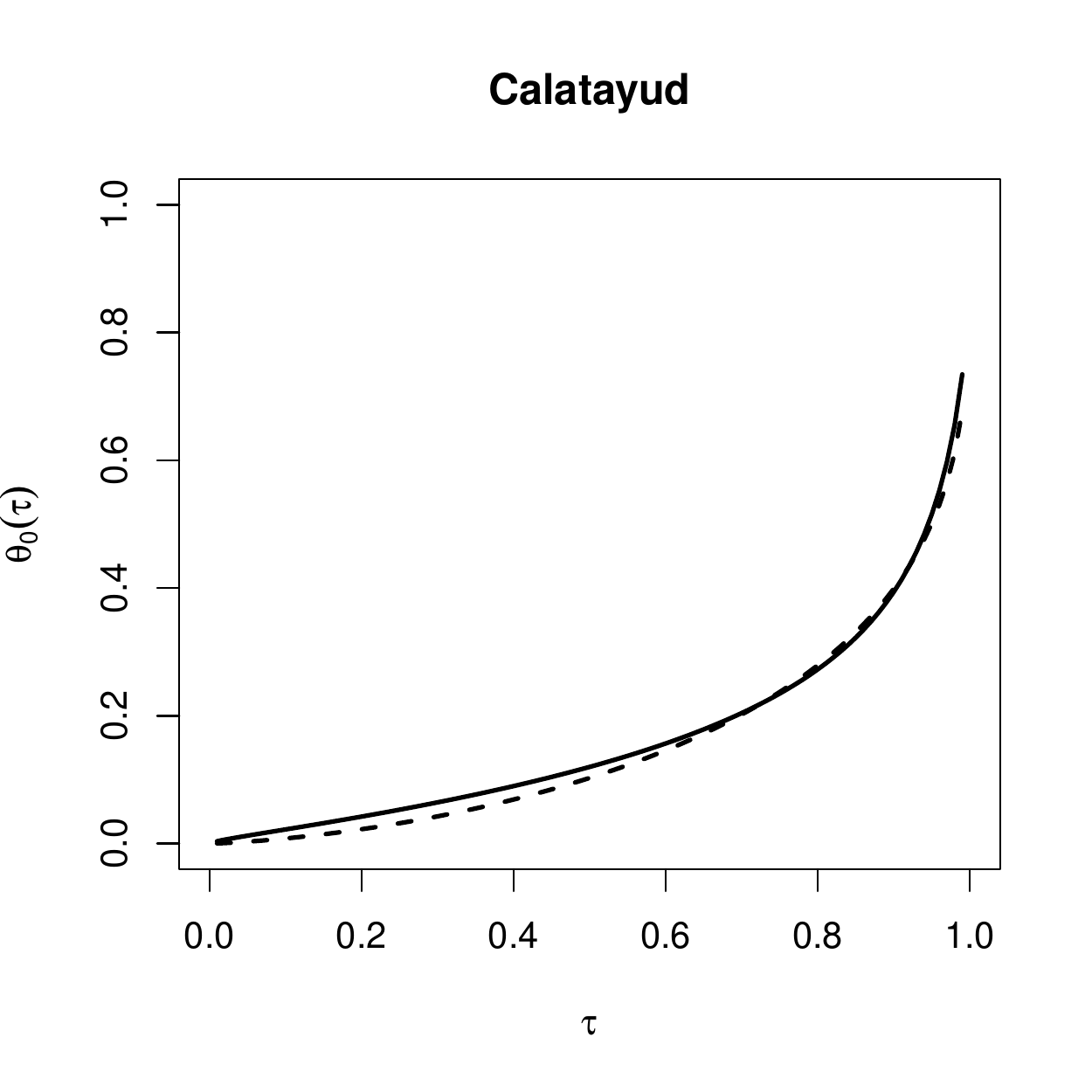} \\
\includegraphics[width=3cm]{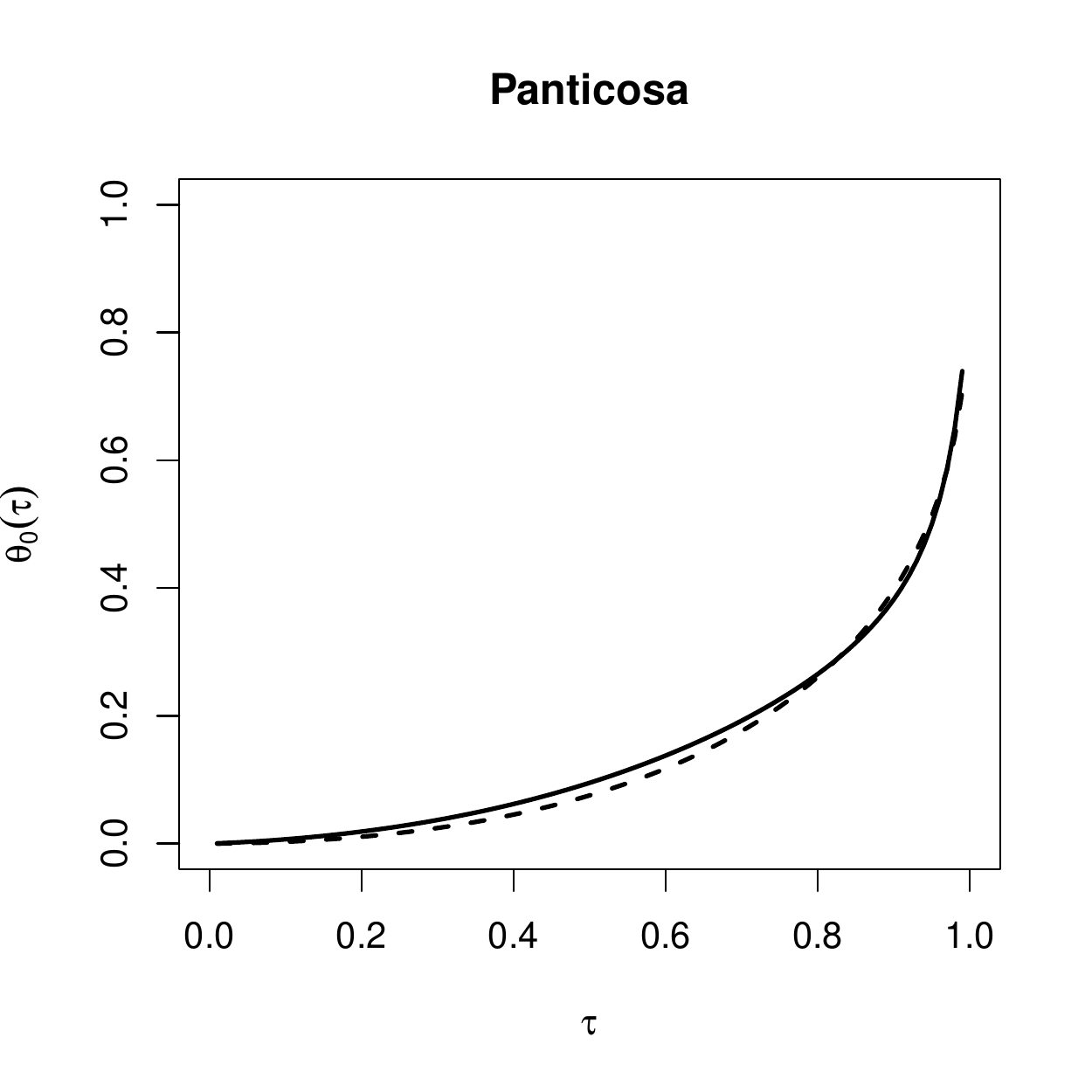}
\includegraphics[width=3cm]{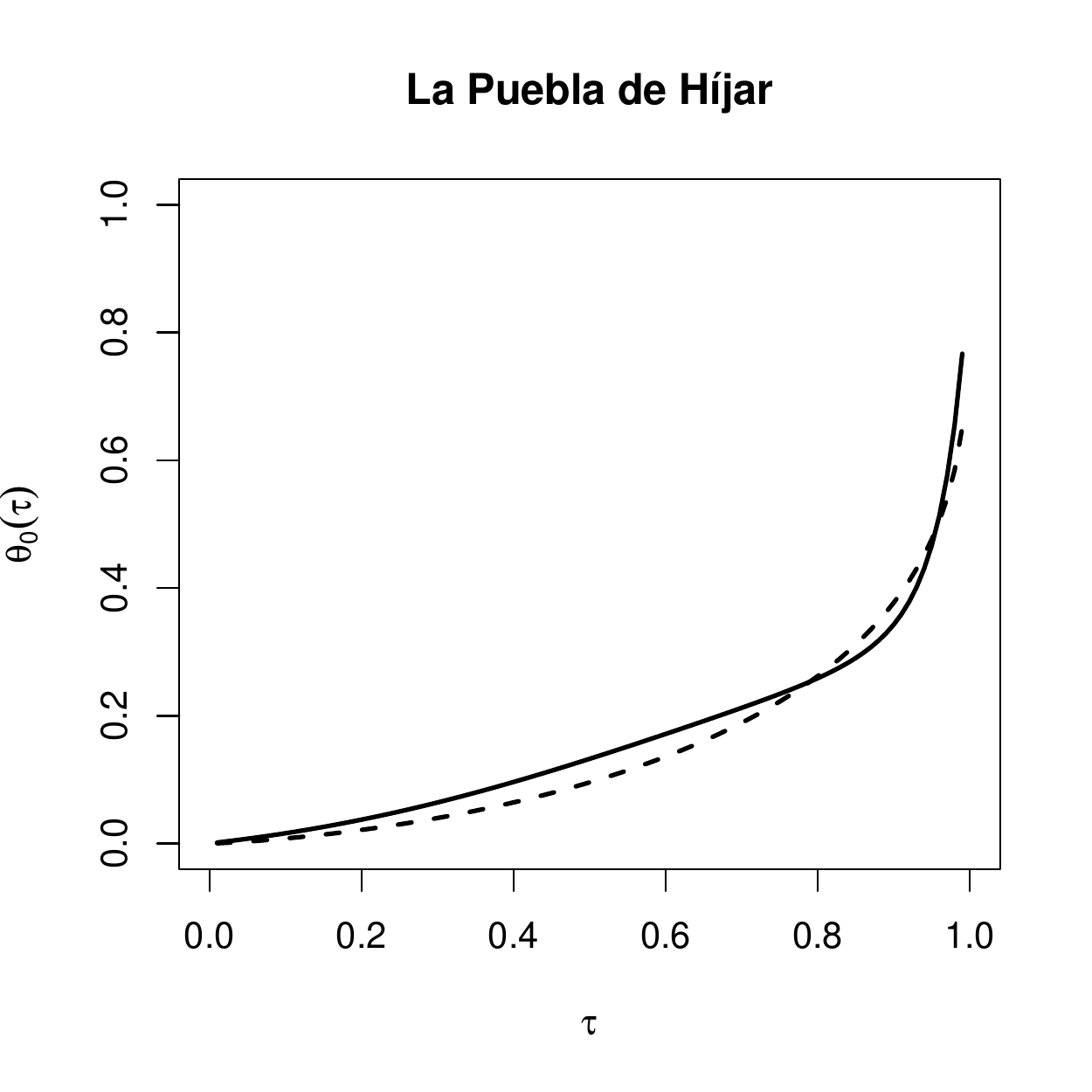}
\includegraphics[width=3cm]{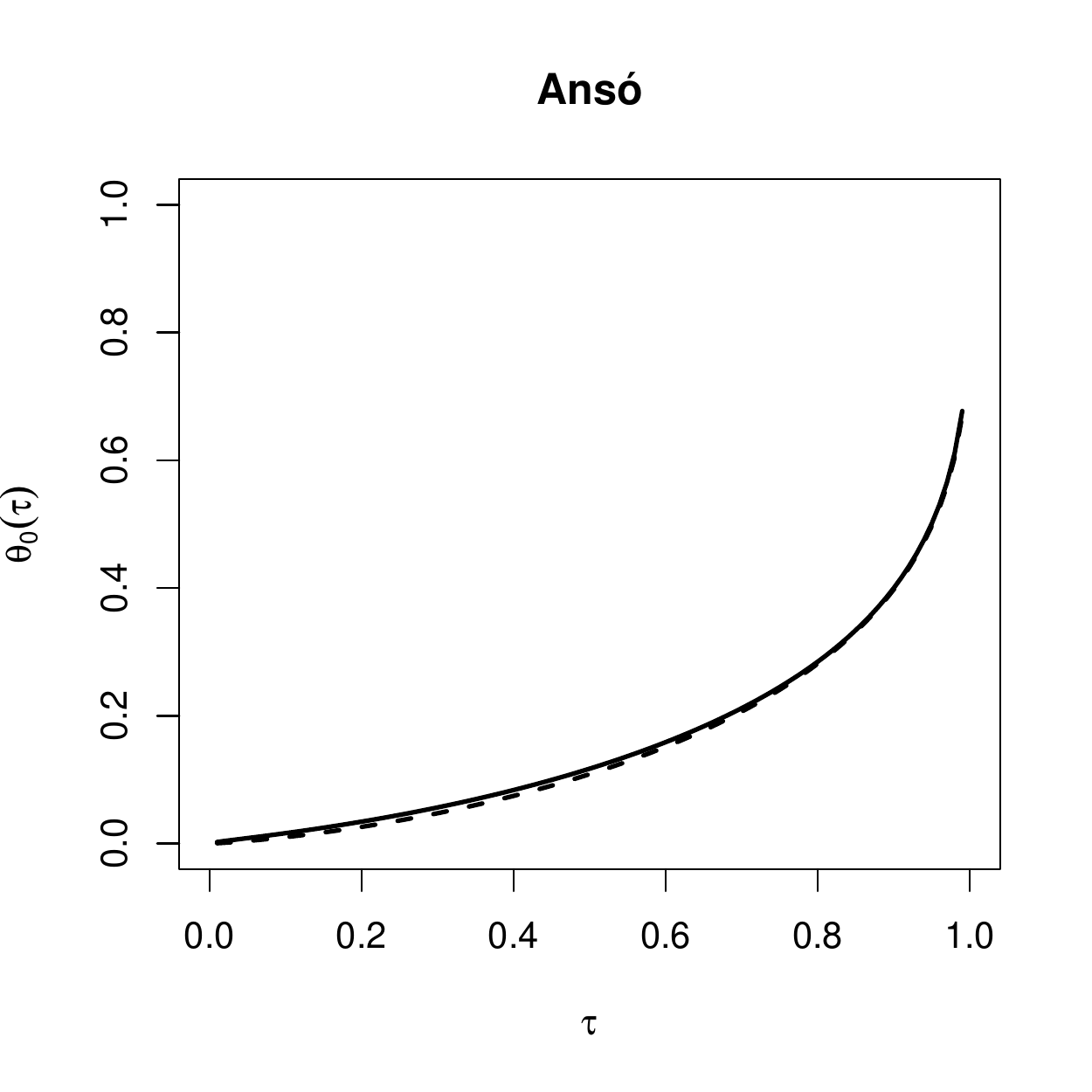}
\includegraphics[width=3cm]{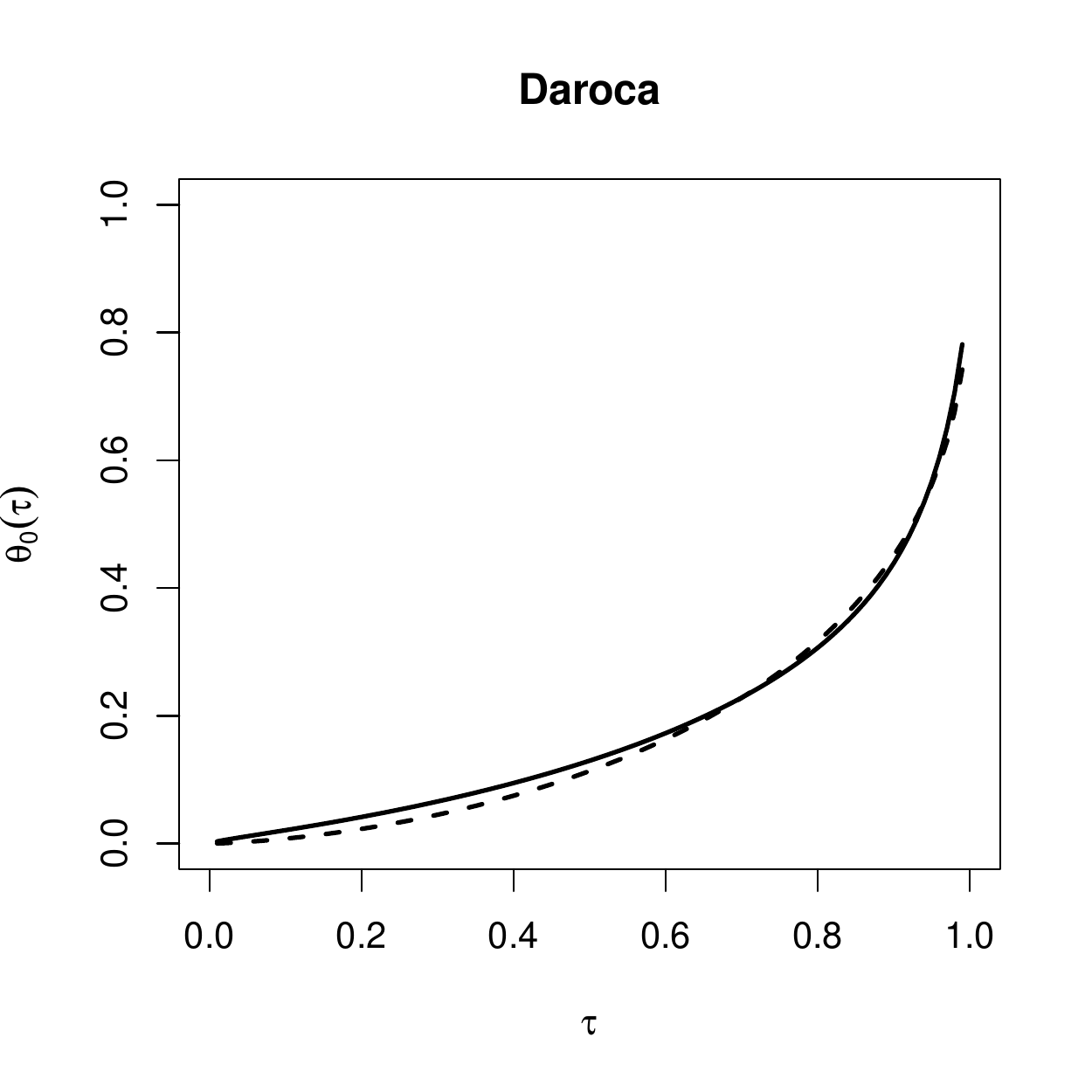} \\
\includegraphics[width=3cm]{theta0QAR1s13.pdf}
\includegraphics[width=3cm]{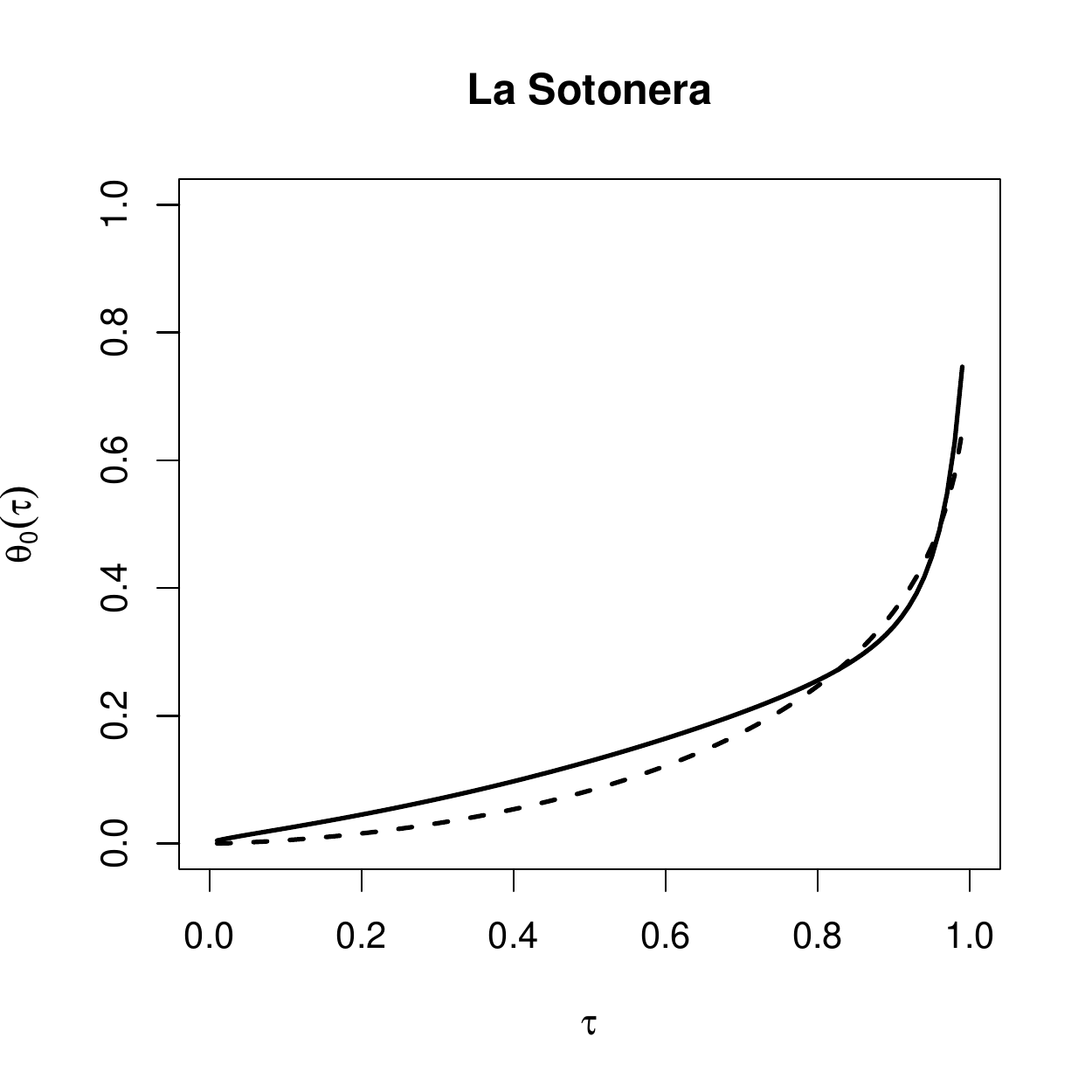}
\includegraphics[width=3cm]{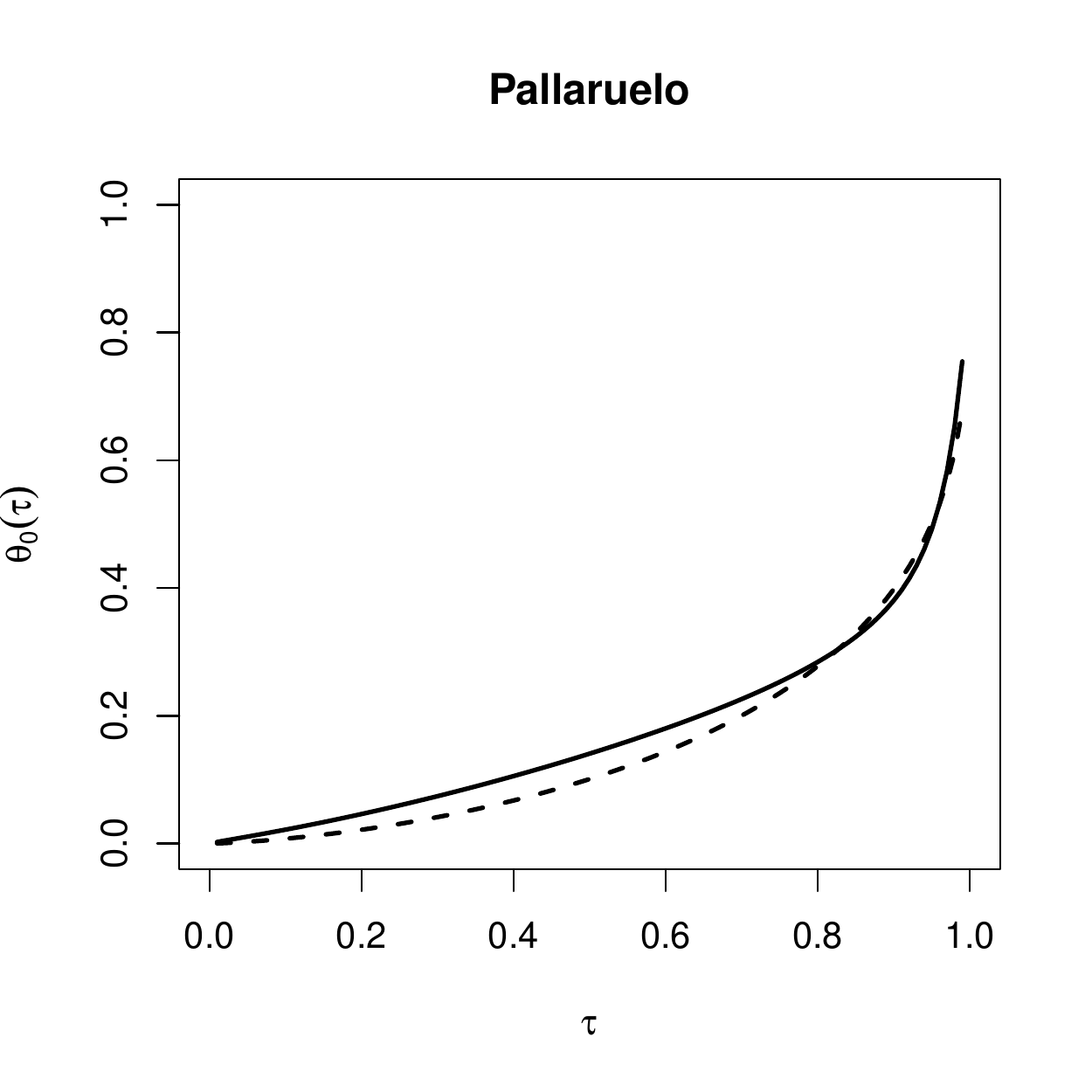}
\includegraphics[width=3cm]{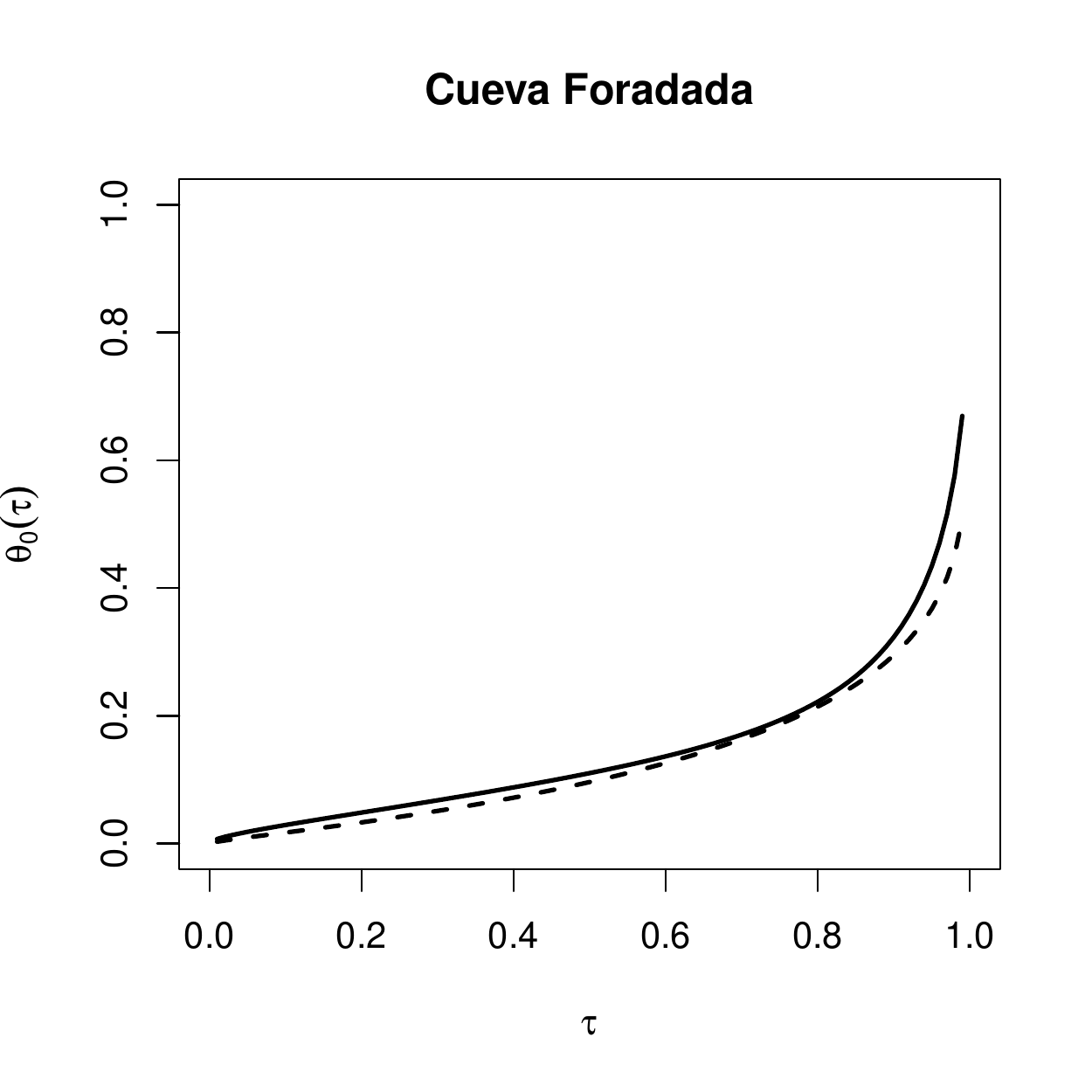} \\
\includegraphics[width=3cm]{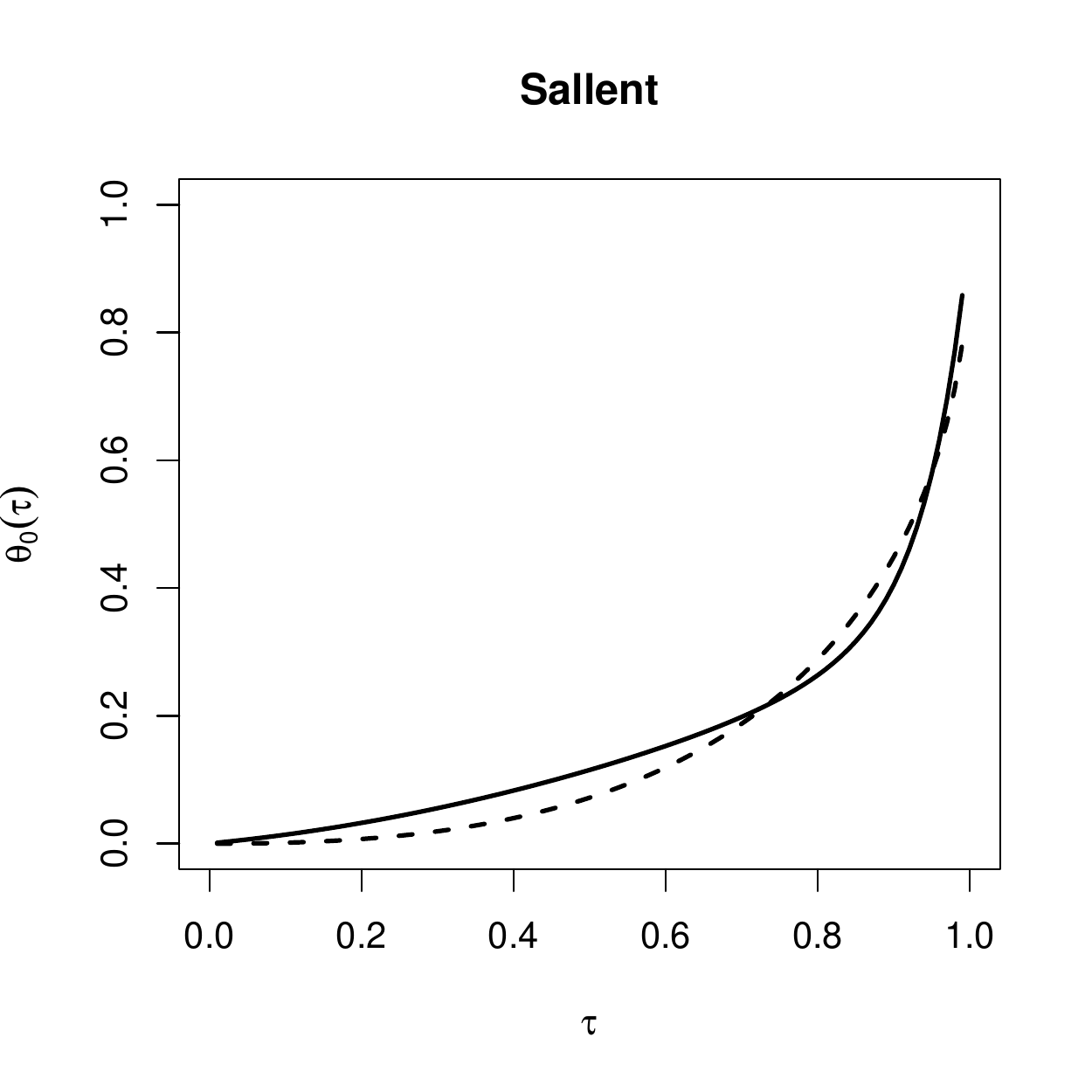}
\includegraphics[width=3cm]{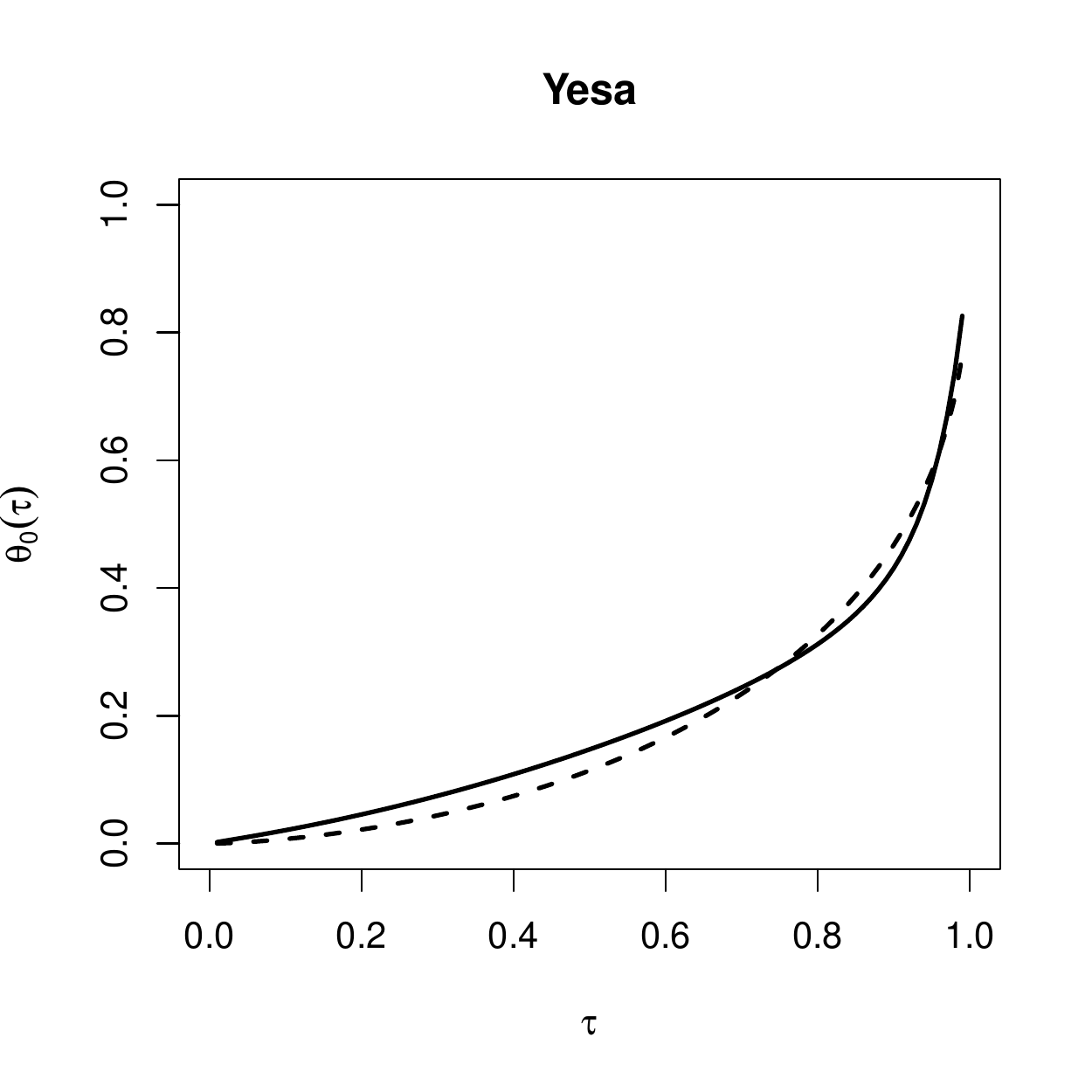}
\caption{Posterior mean of $\theta_0(\tau)$ vs. $\tau$ for QAR1K1 (dashed) and QAR1K2 (solid). All locations, MJJAS, 2015. }
\end{figure}

\begin{figure}[!ht]
\centering
\includegraphics[width=3cm]{theta1QAR1s1.pdf}
\includegraphics[width=3cm]{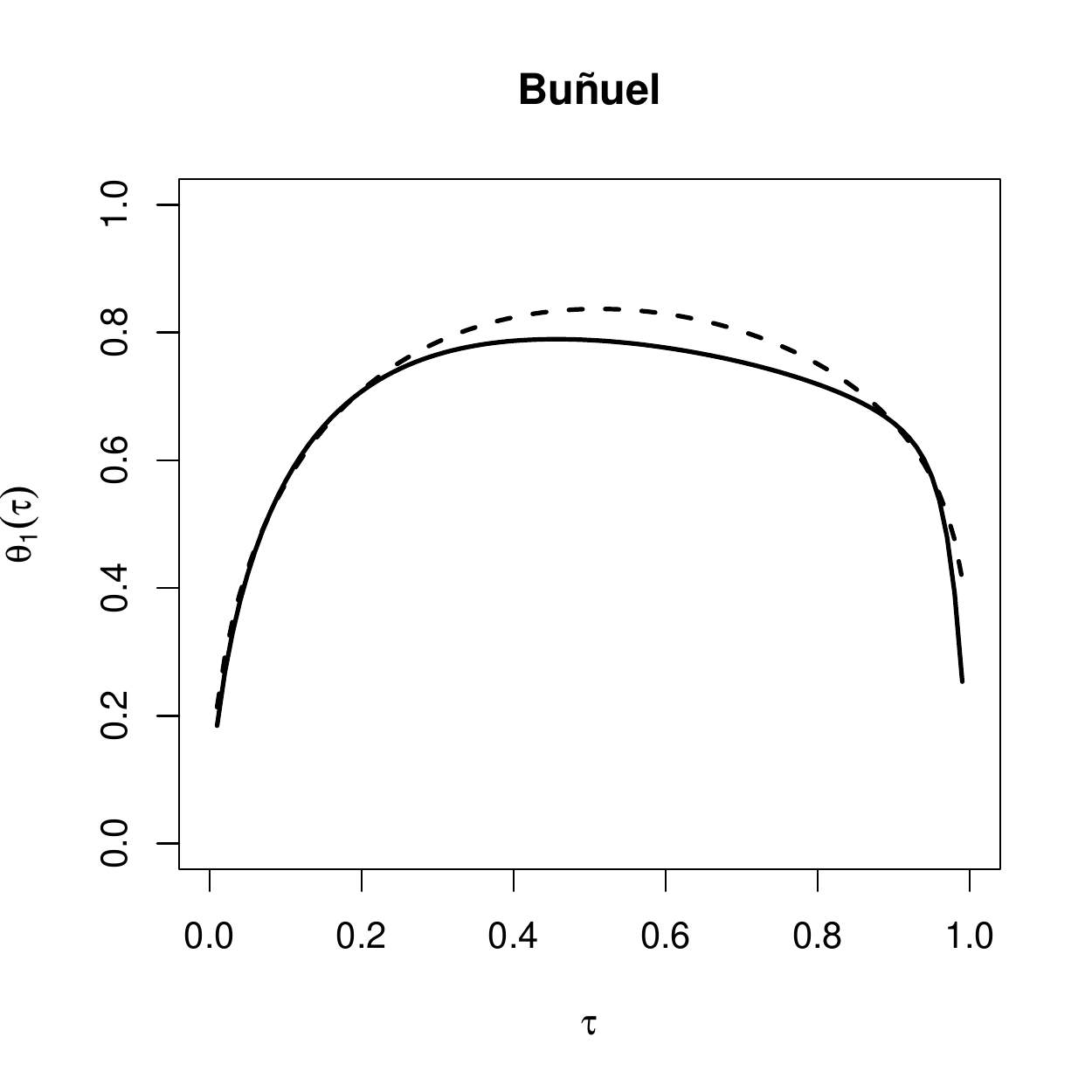}
\includegraphics[width=3cm]{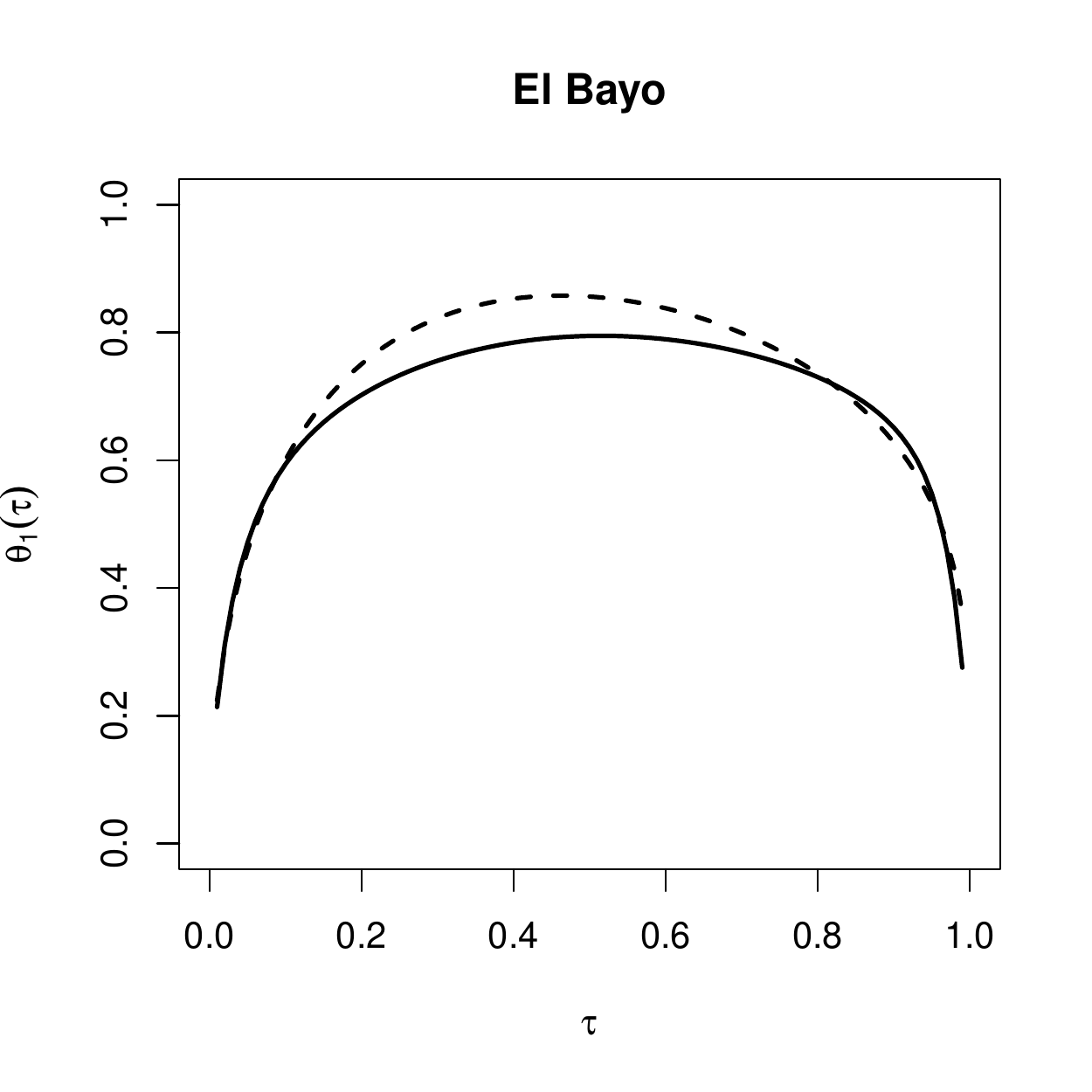}
\includegraphics[width=3cm]{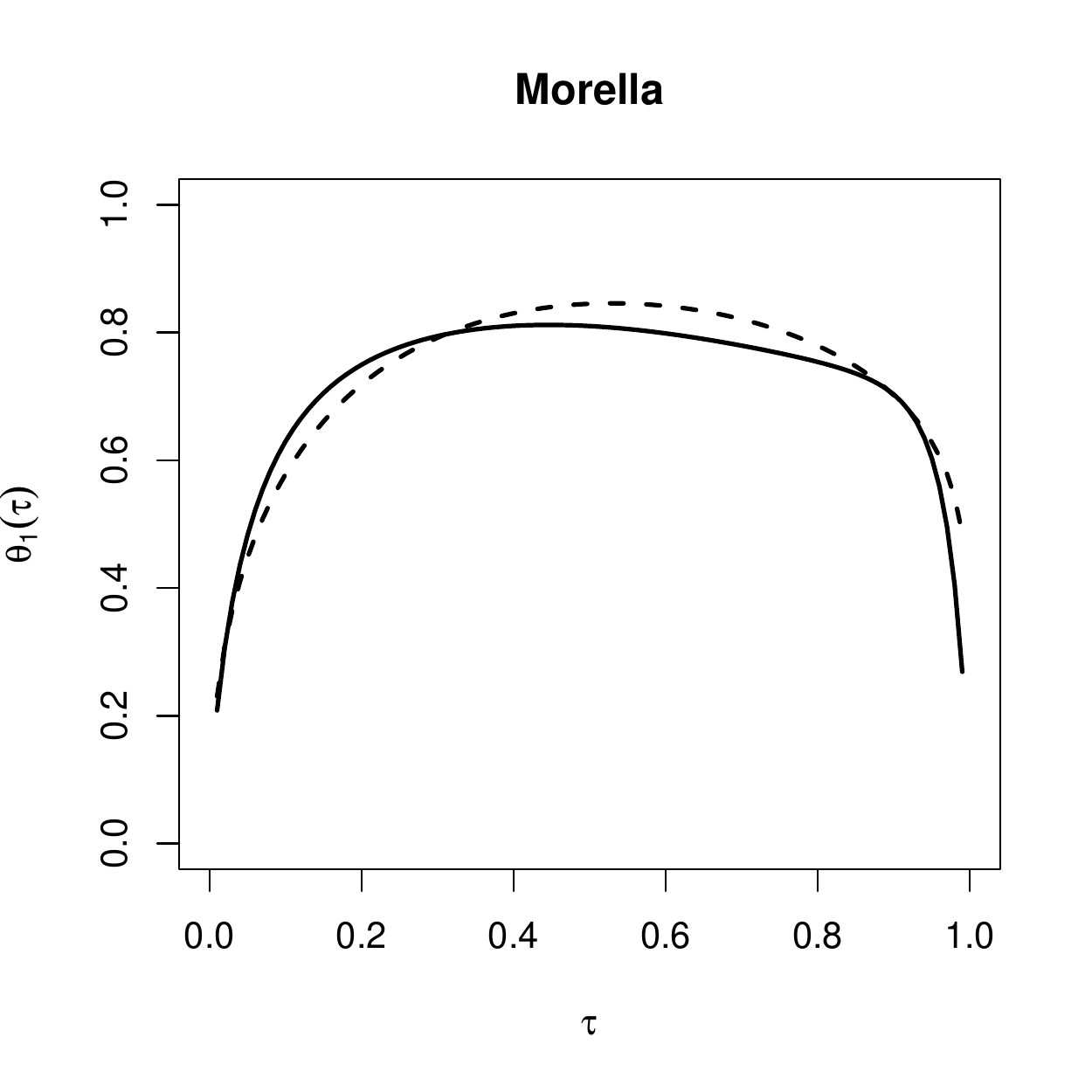} \\
\includegraphics[width=3cm]{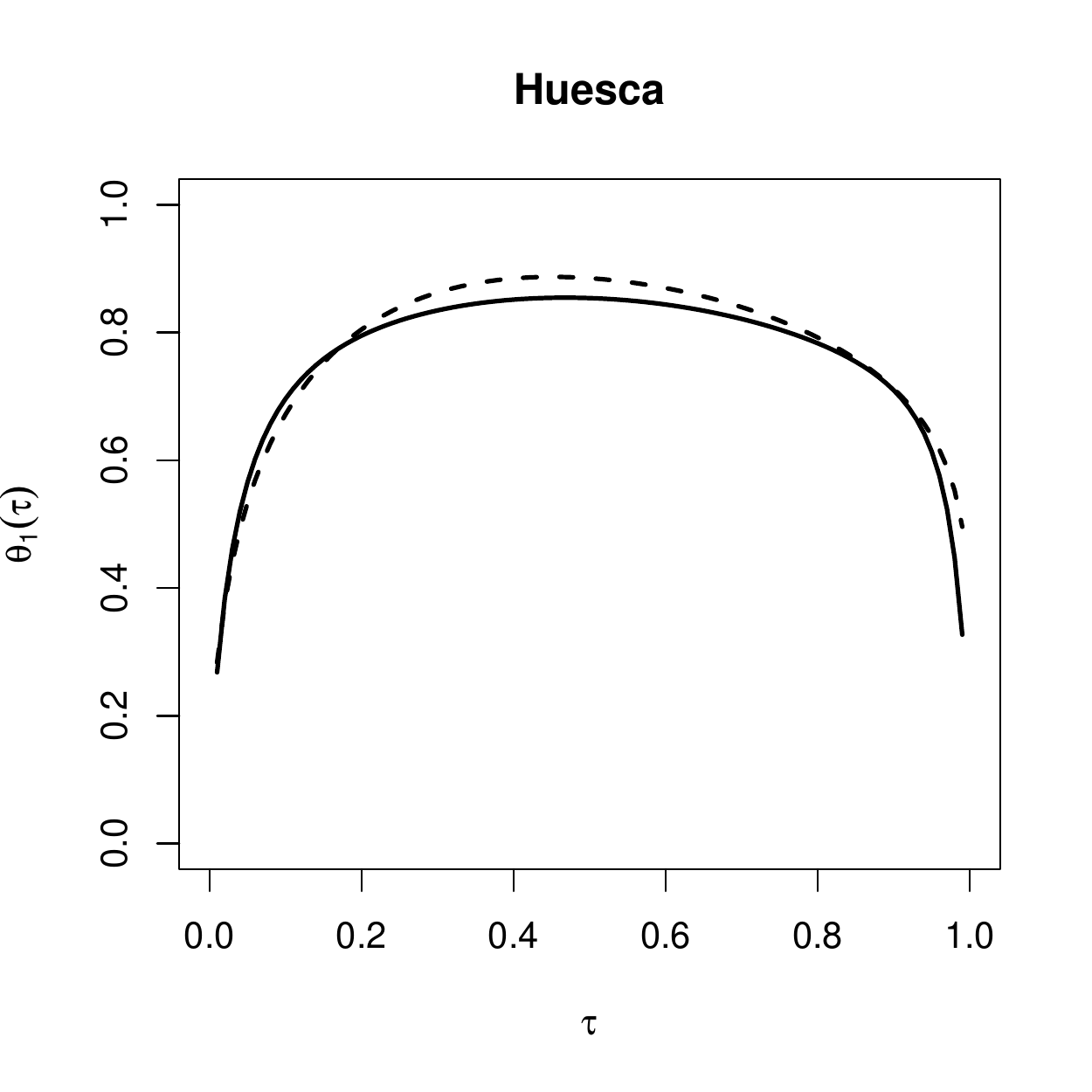}
\includegraphics[width=3cm]{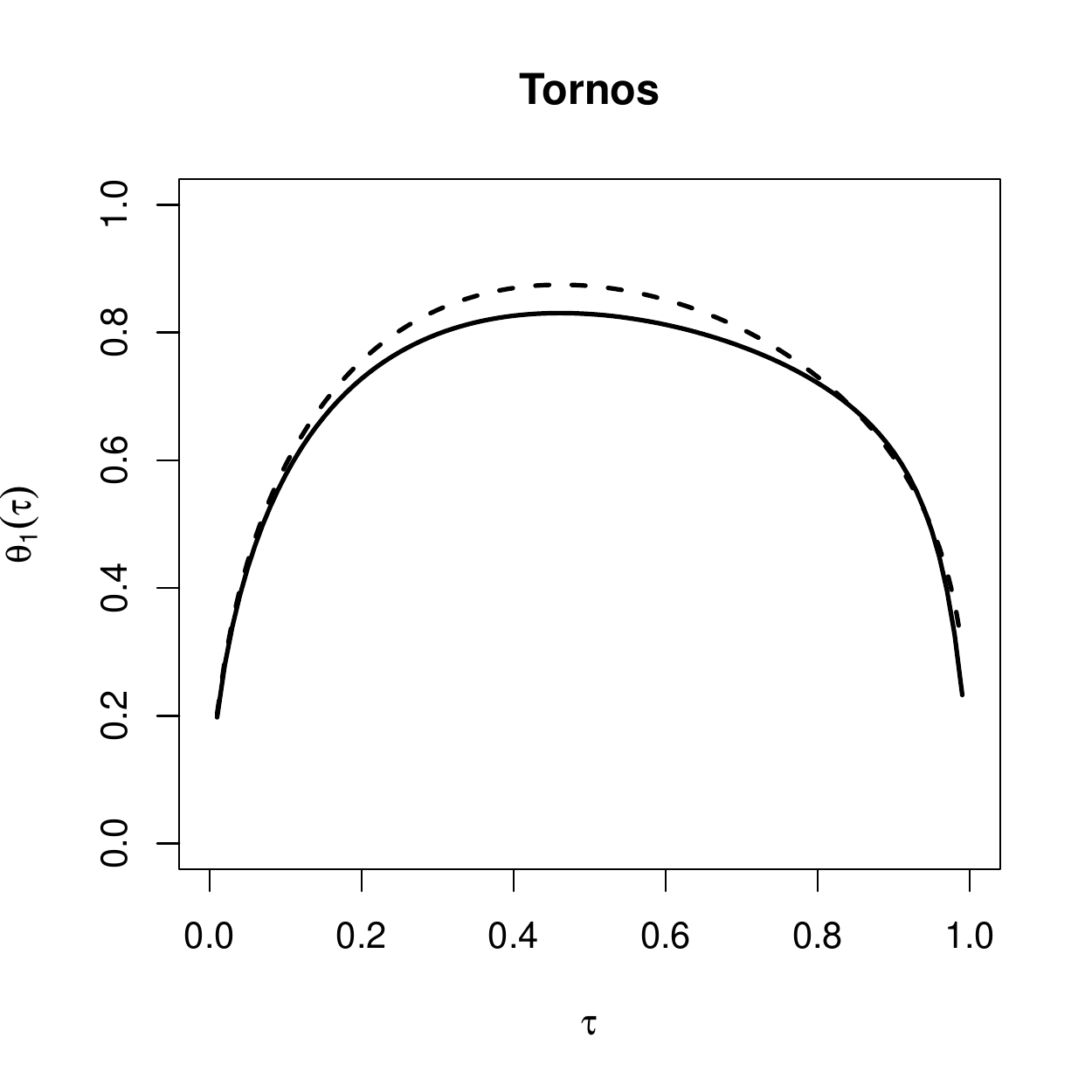}
\includegraphics[width=3cm]{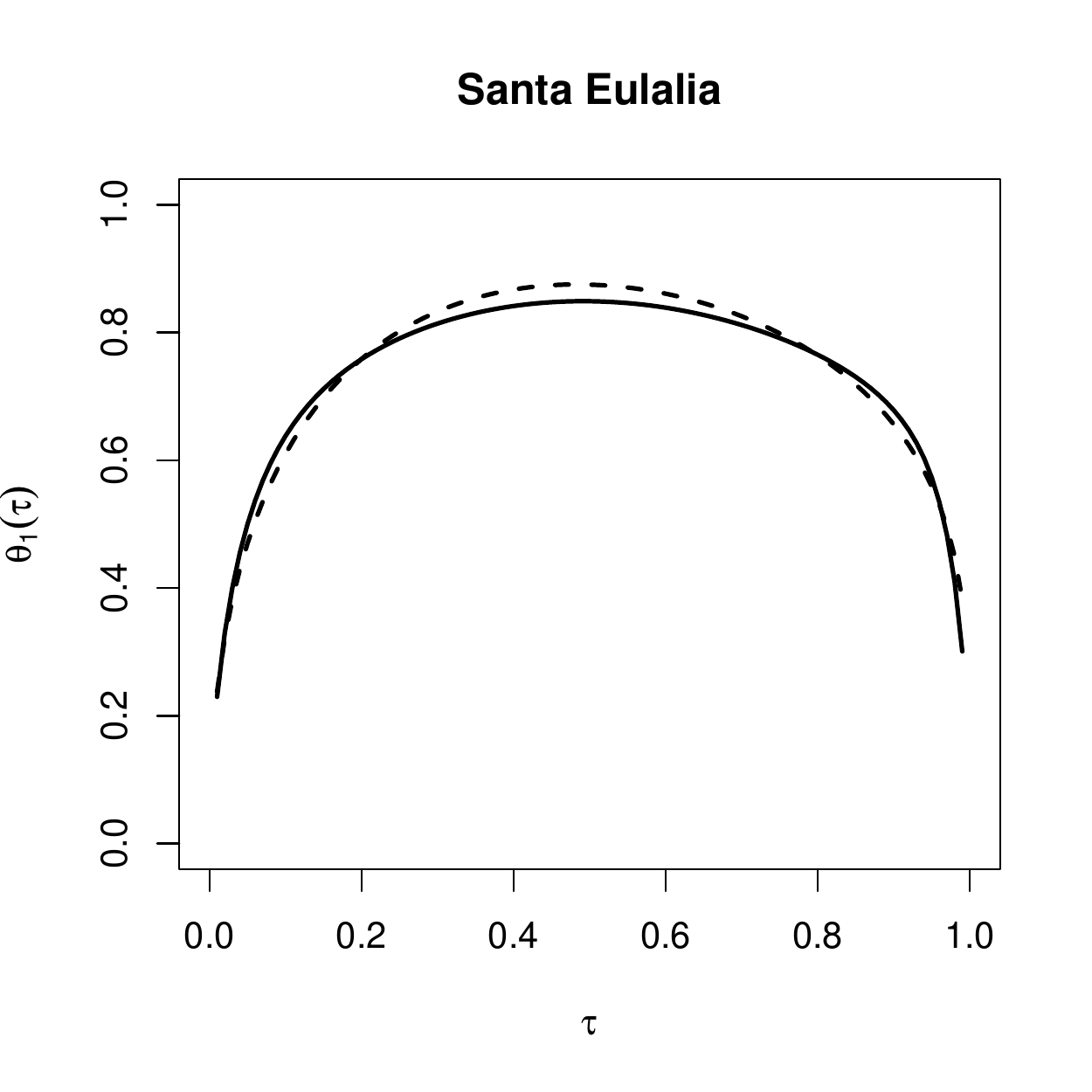}
\includegraphics[width=3cm]{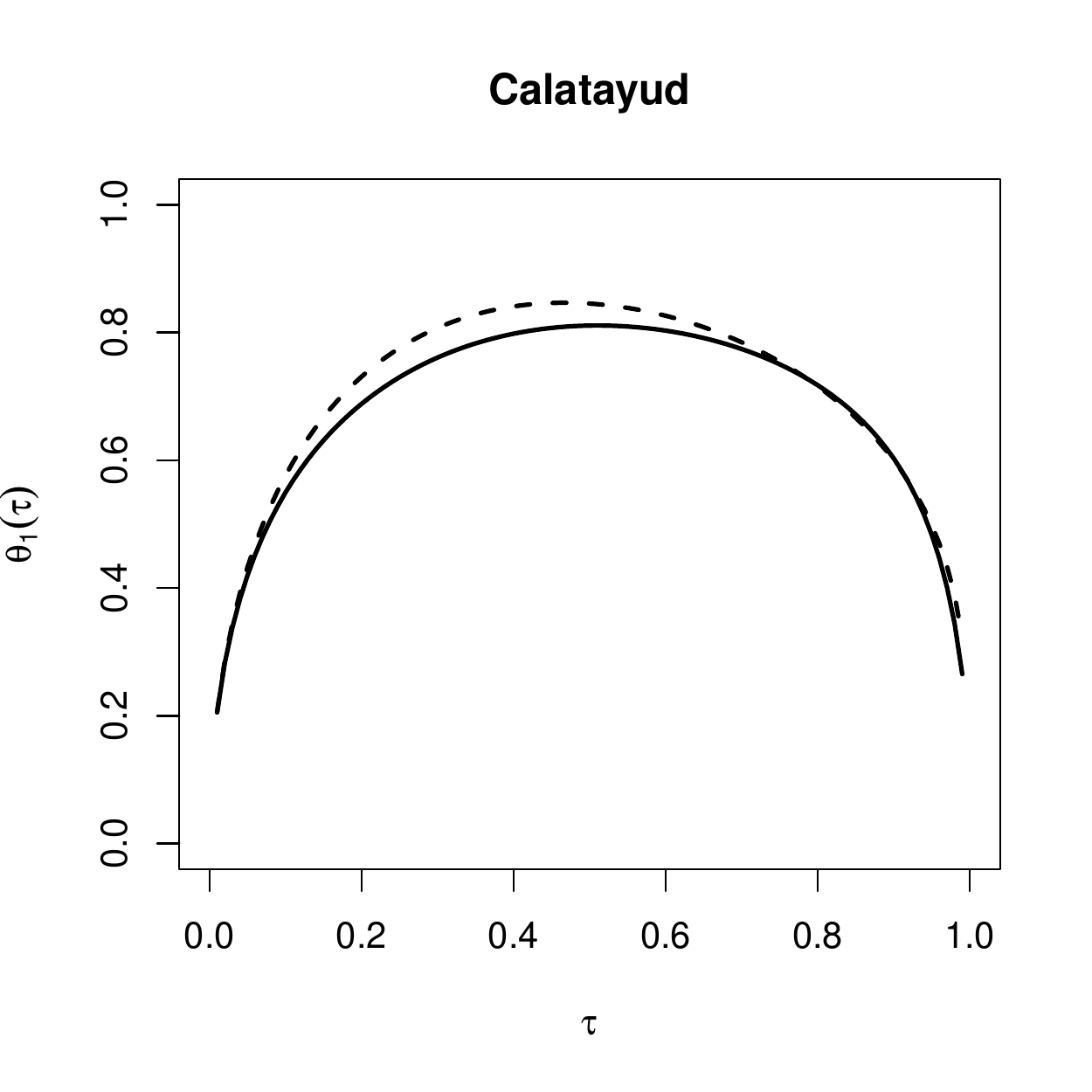} \\
\includegraphics[width=3cm]{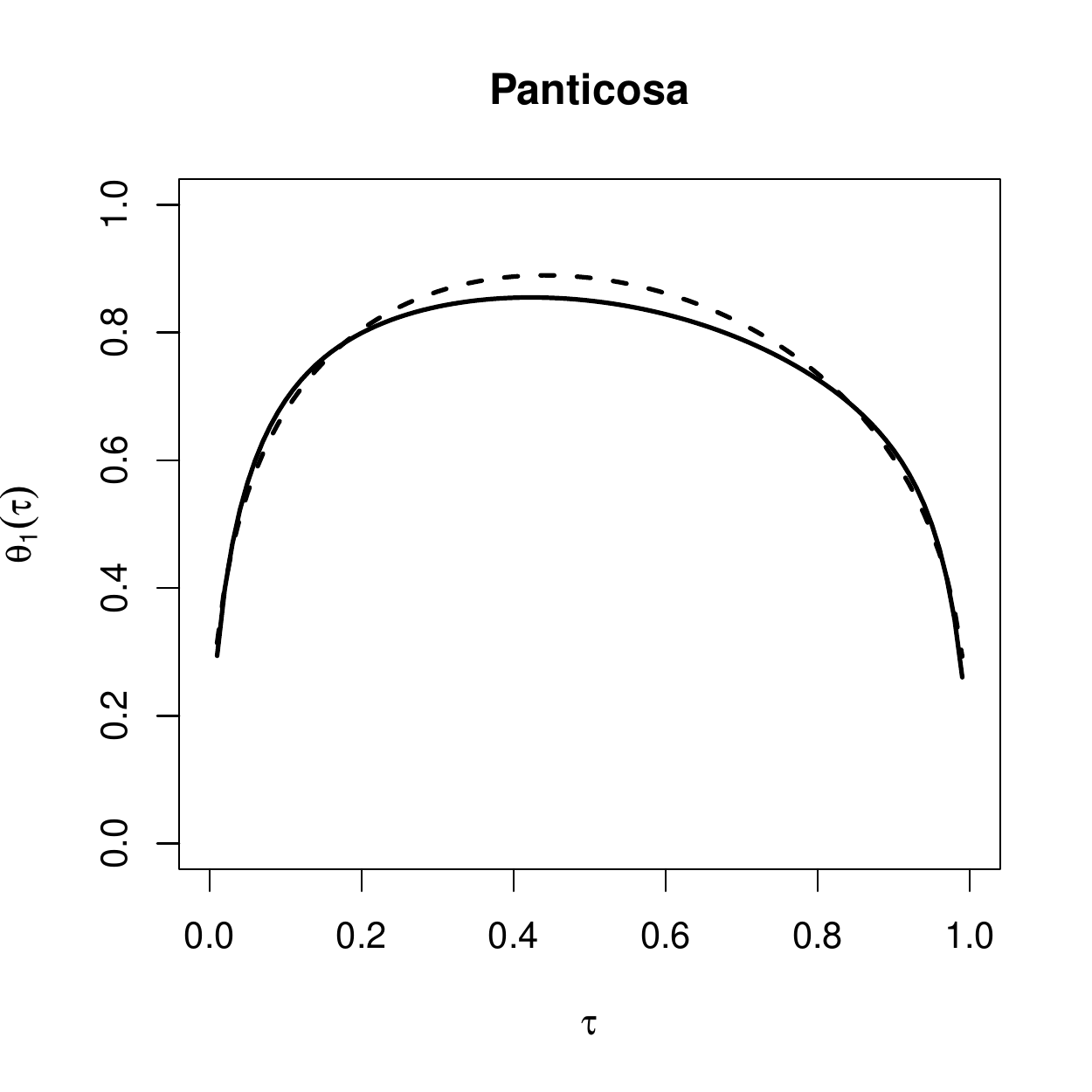}
\includegraphics[width=3cm]{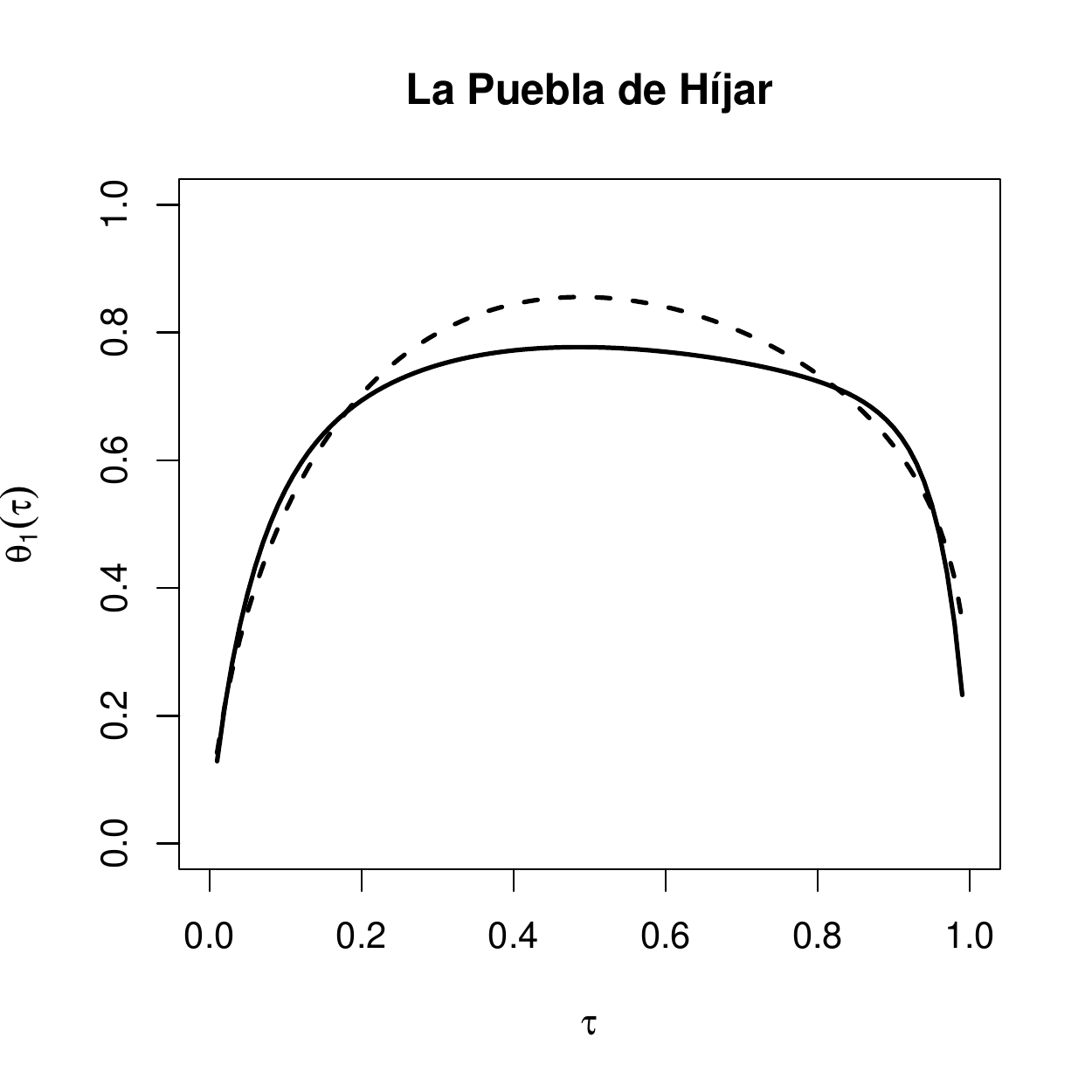}
\includegraphics[width=3cm]{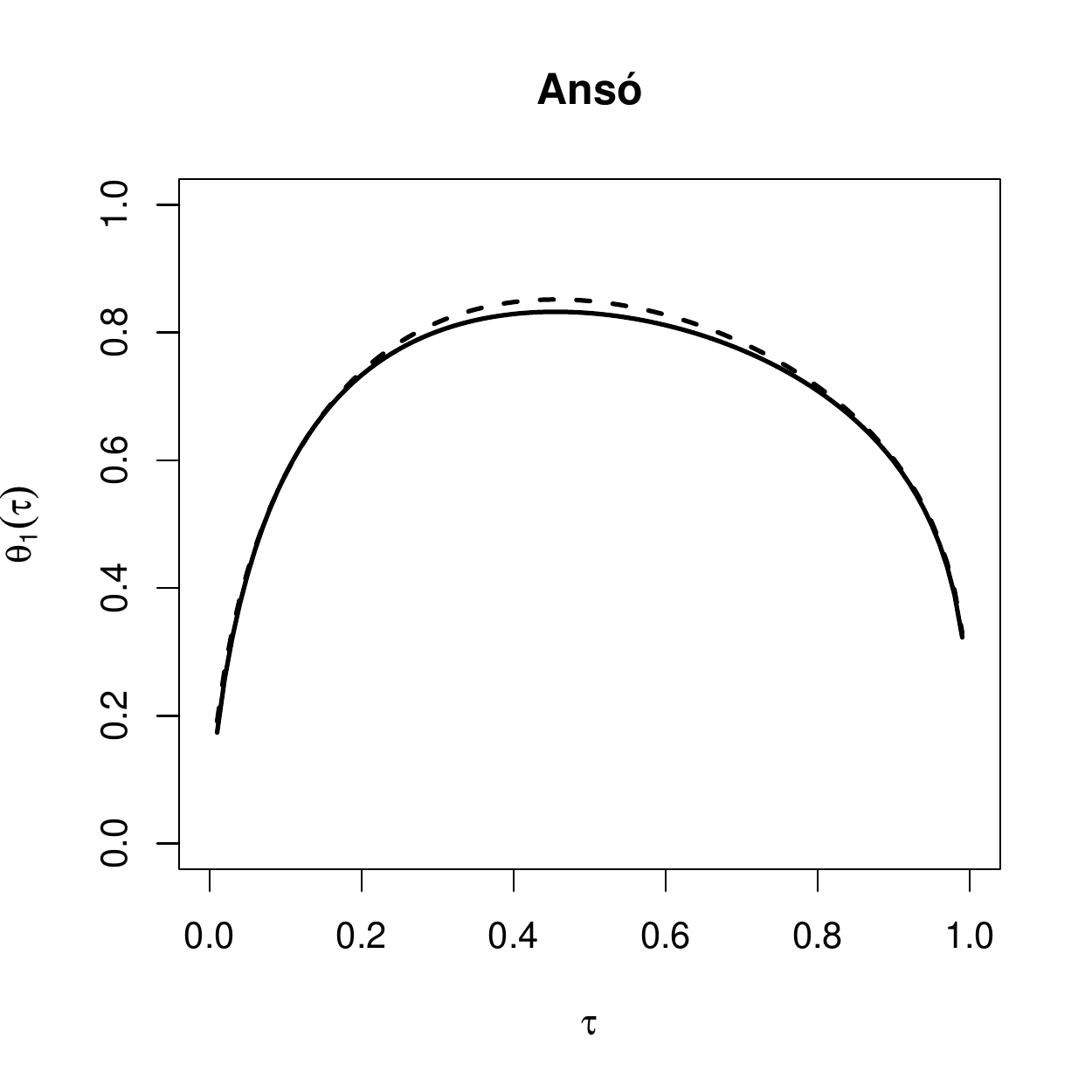}
\includegraphics[width=3cm]{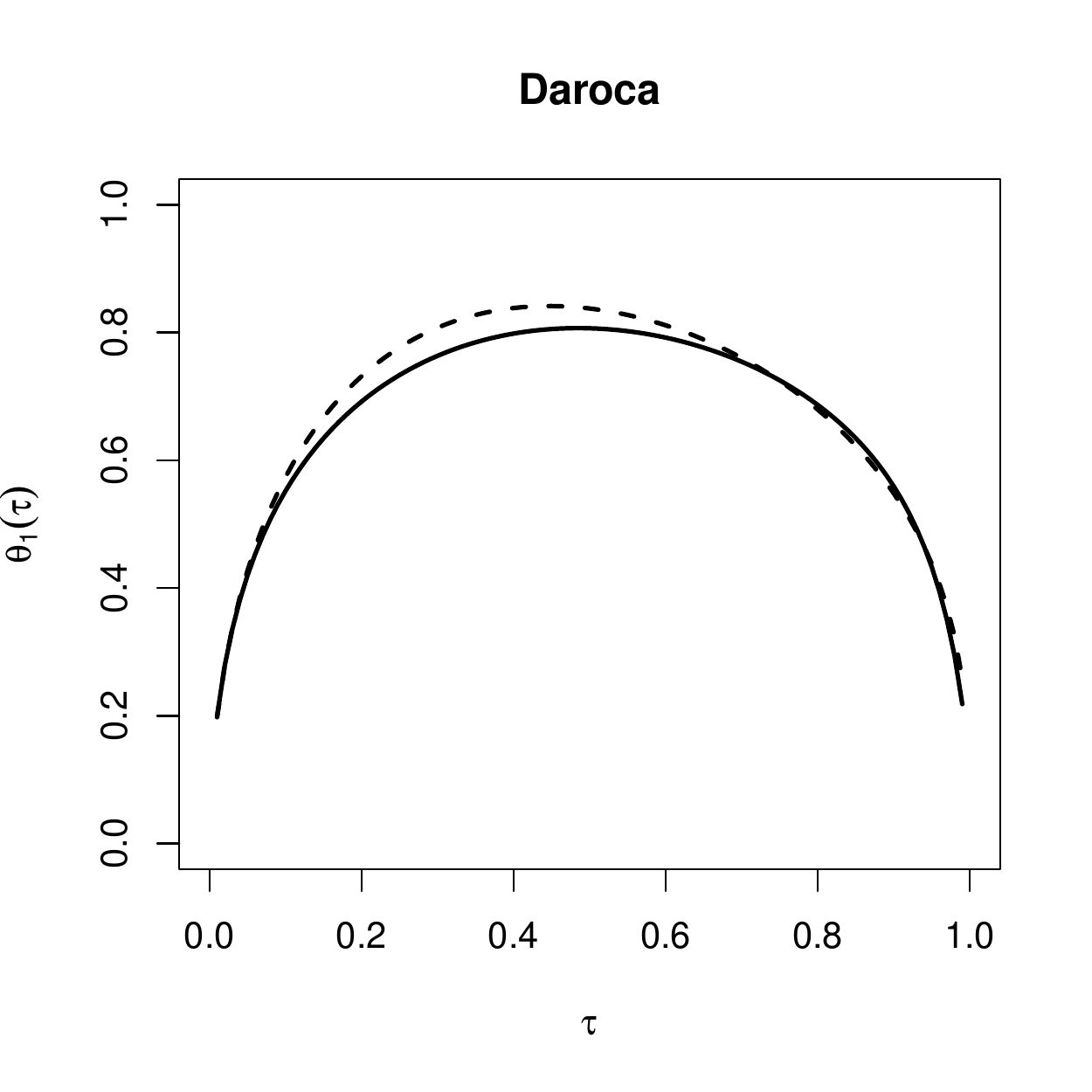} \\
\includegraphics[width=3cm]{theta1QAR1s13.pdf}
\includegraphics[width=3cm]{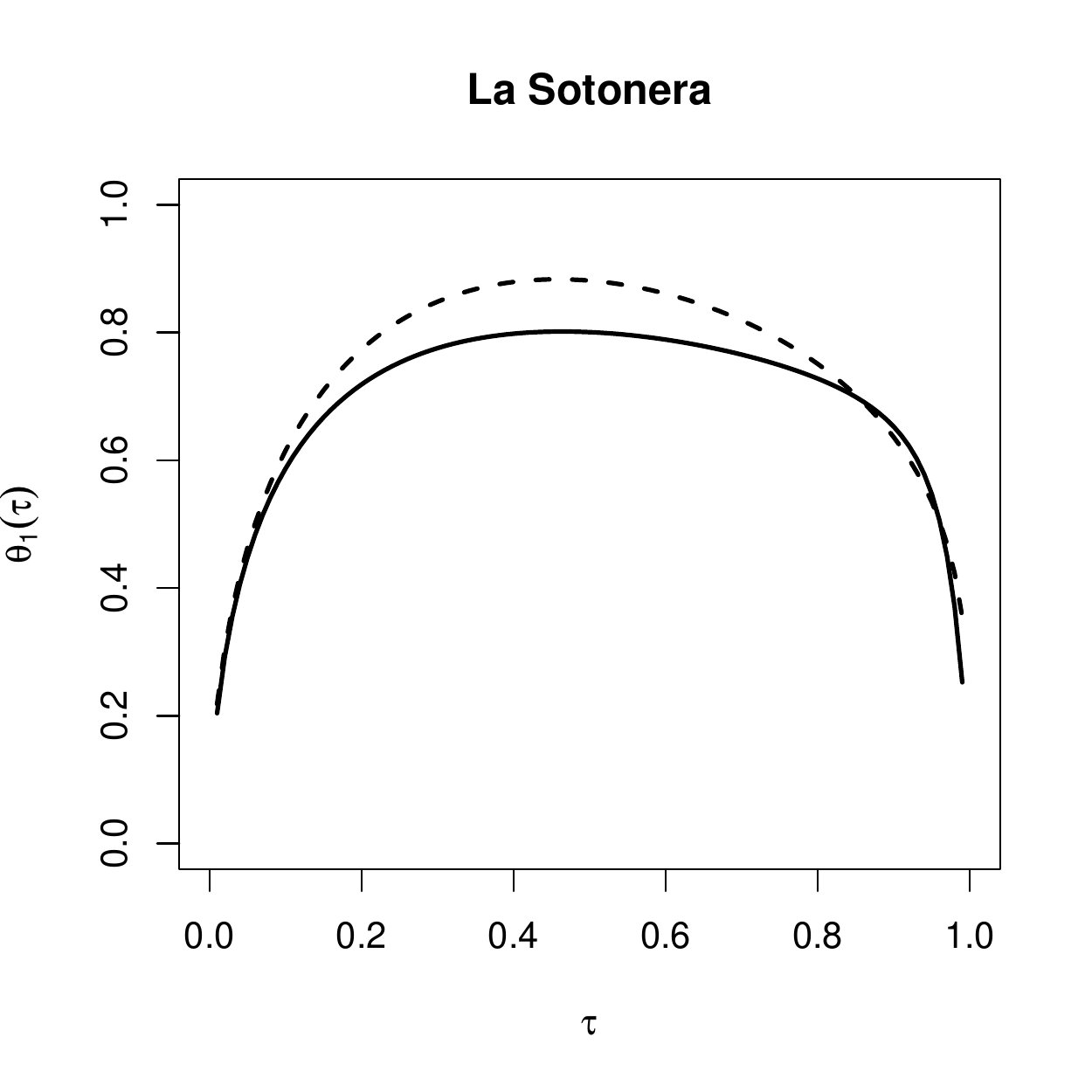}
\includegraphics[width=3cm]{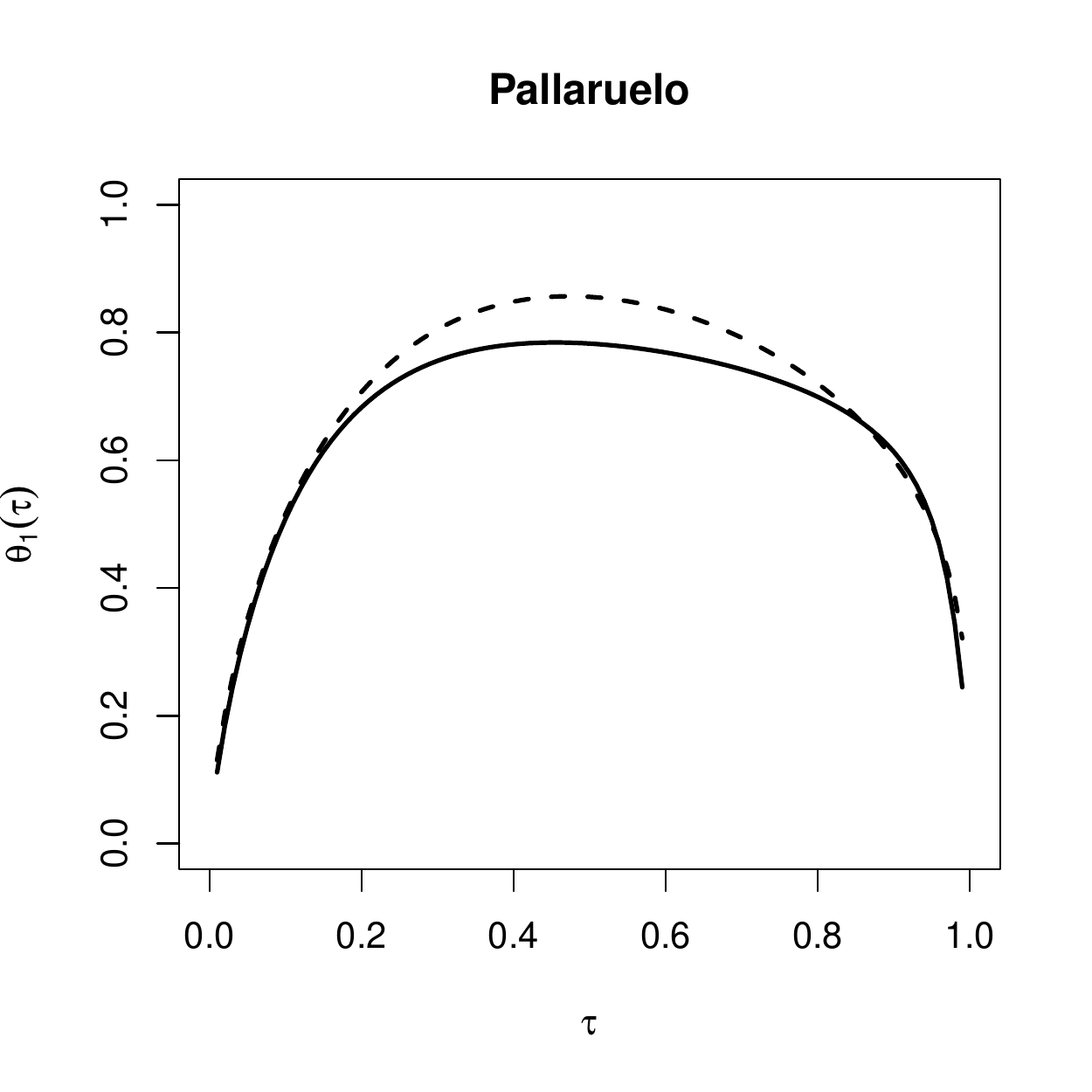}
\includegraphics[width=3cm]{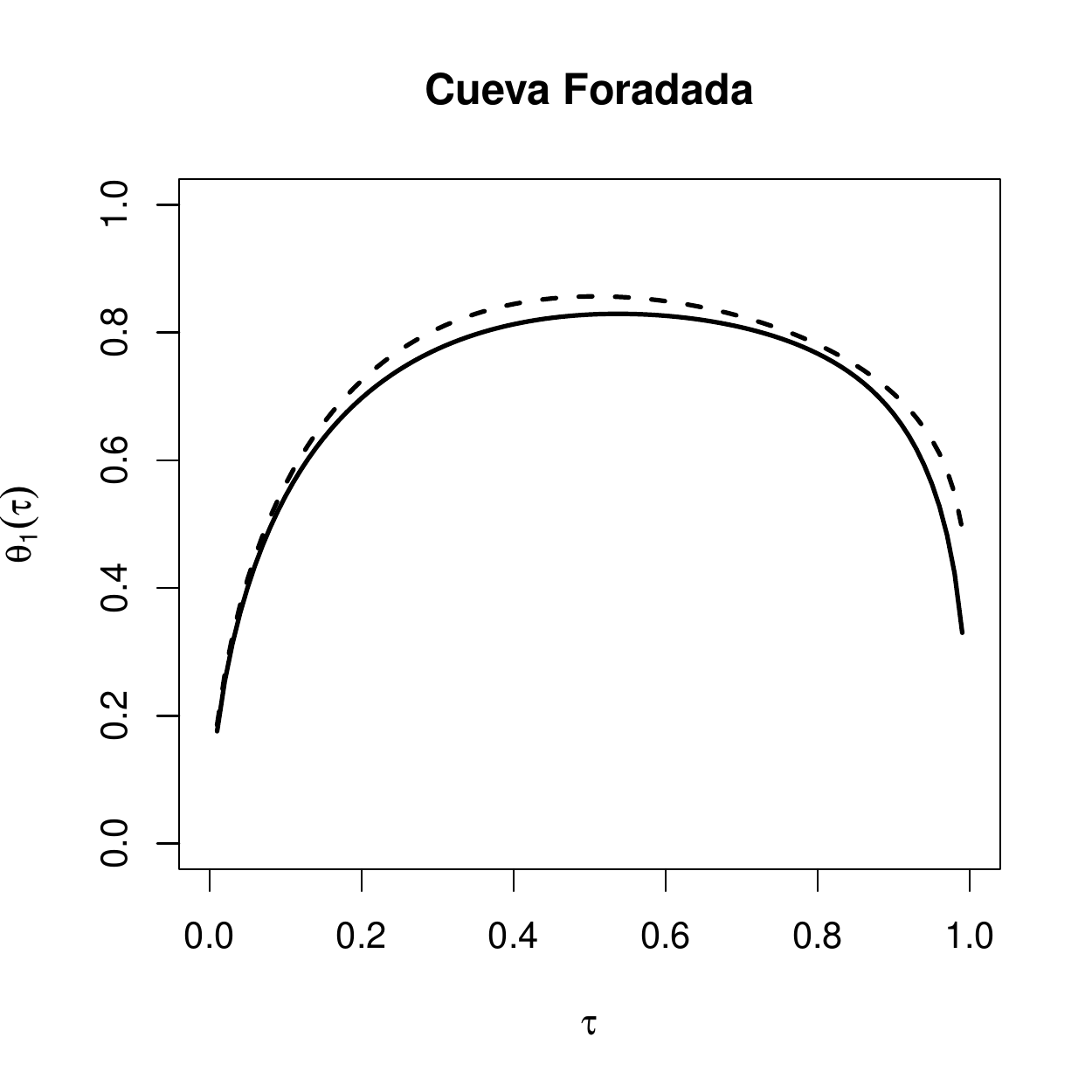} \\
\includegraphics[width=3cm]{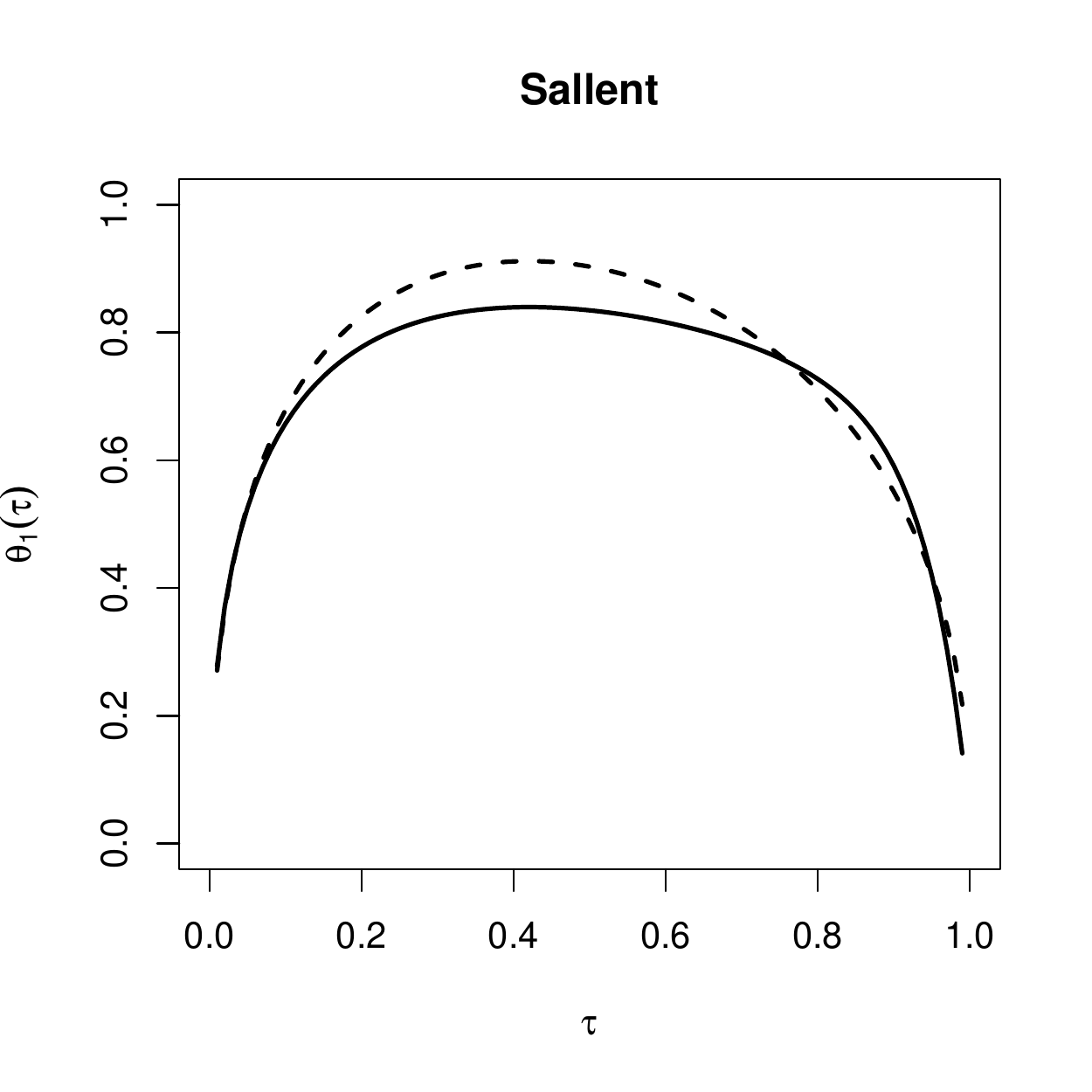}
\includegraphics[width=3cm]{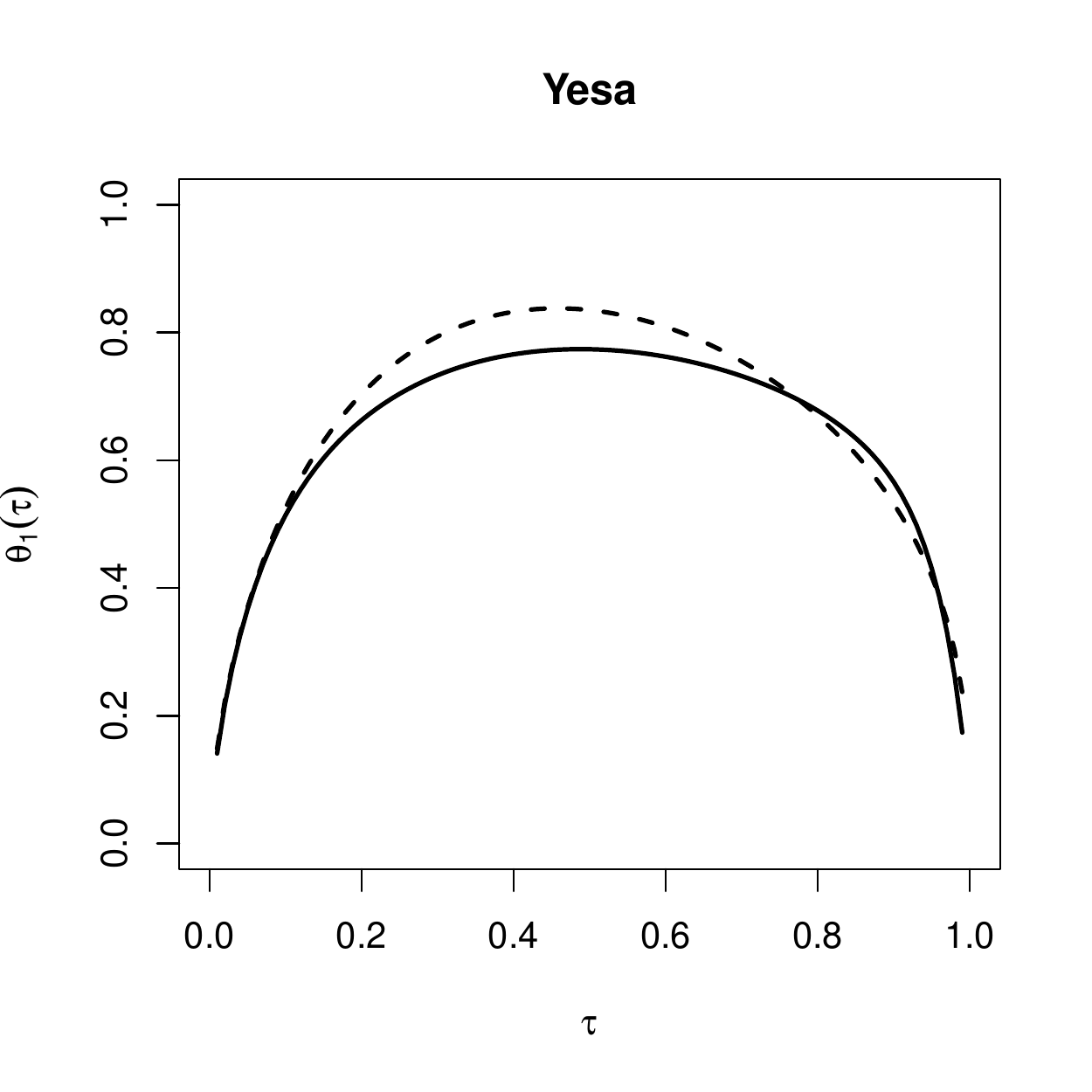}
\caption{Posterior mean of $\theta_1(\tau)$ vs. $\tau$ for QAR1K1 (dashed) and QAR1K2 (solid). All locations, MJJAS, 2015. }
\end{figure}

\begin{figure}[!ht]
\centering
\includegraphics[width=3cm]{quantileQAR1s1.pdf}
\includegraphics[width=3cm]{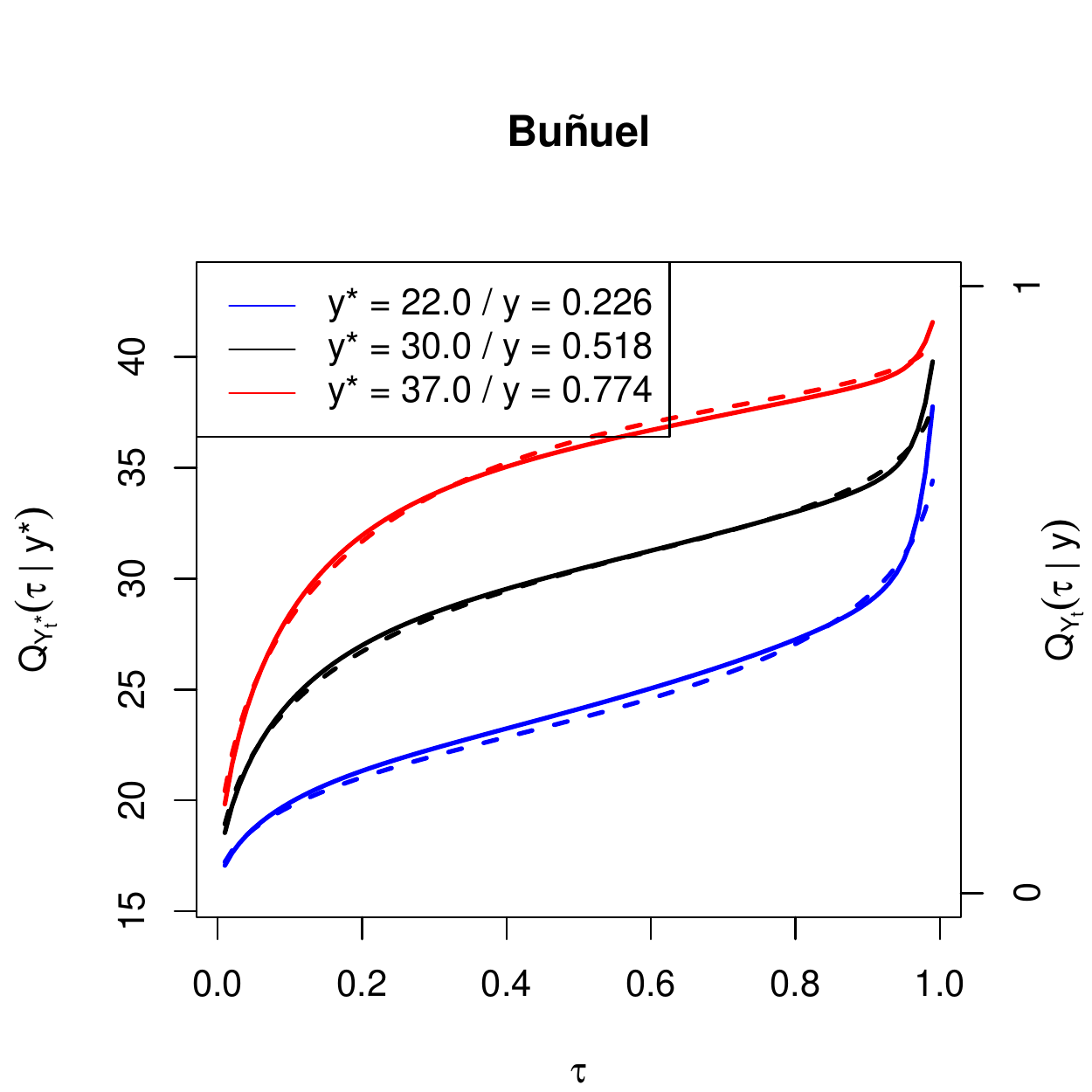}
\includegraphics[width=3cm]{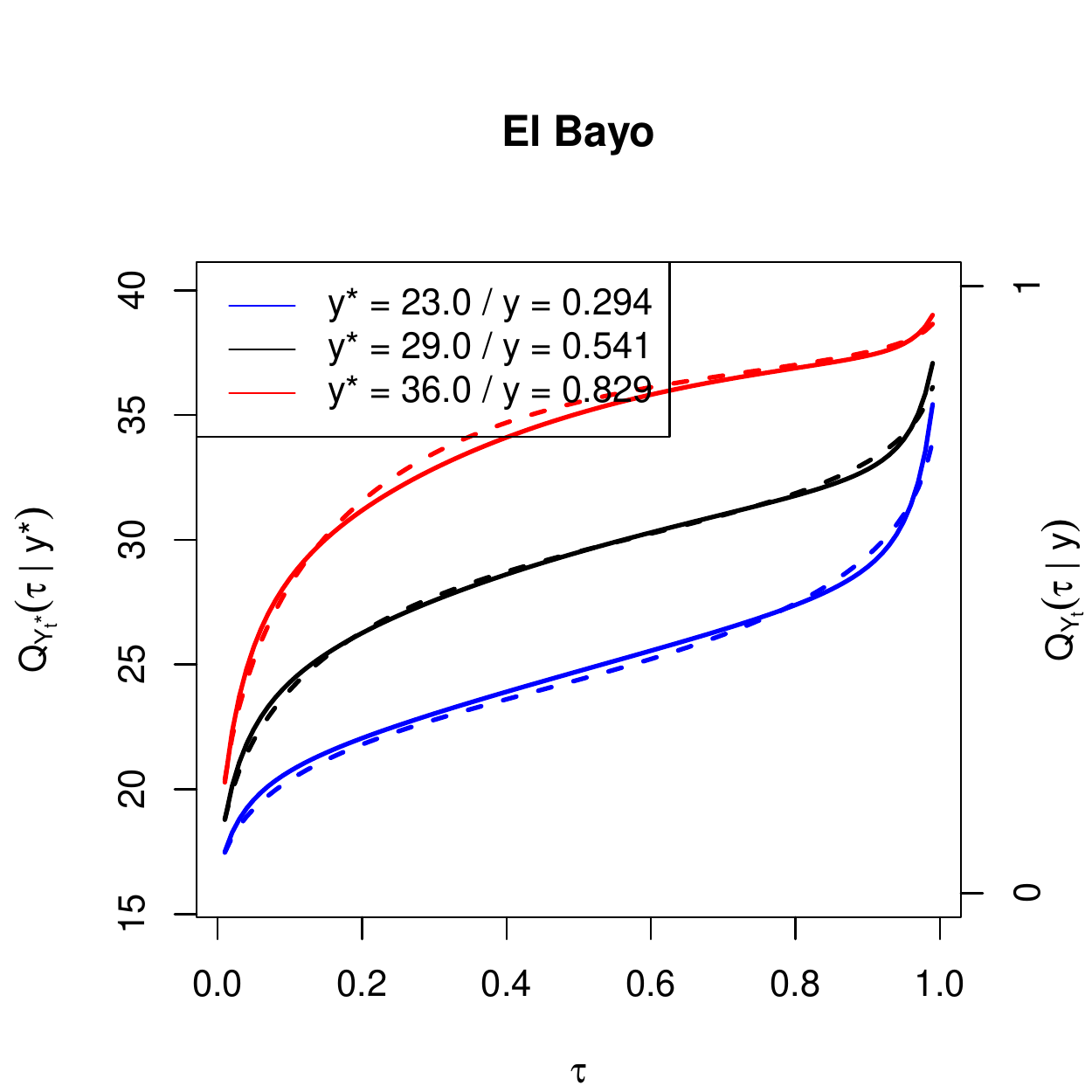}
\includegraphics[width=3cm]{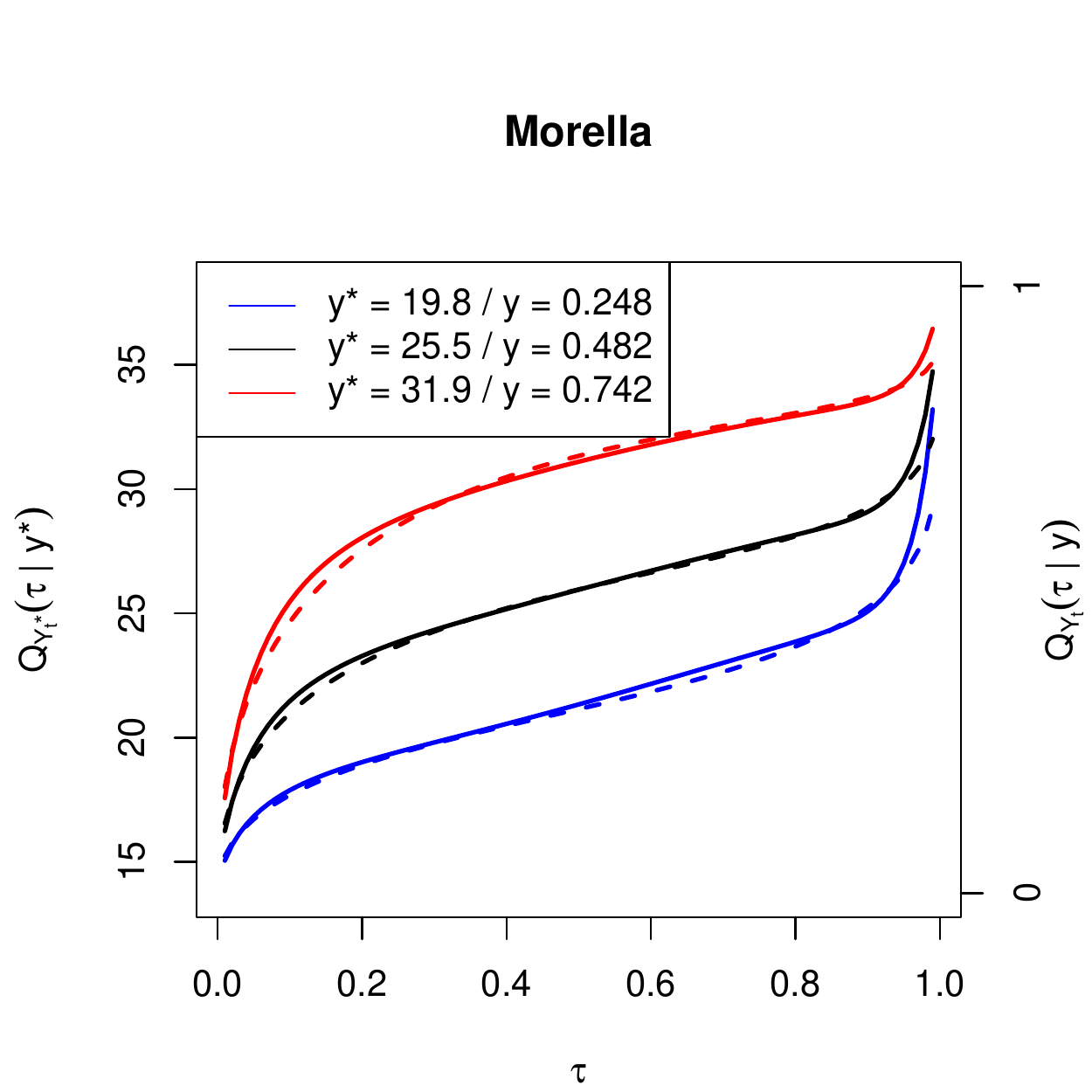} \\
\includegraphics[width=3cm]{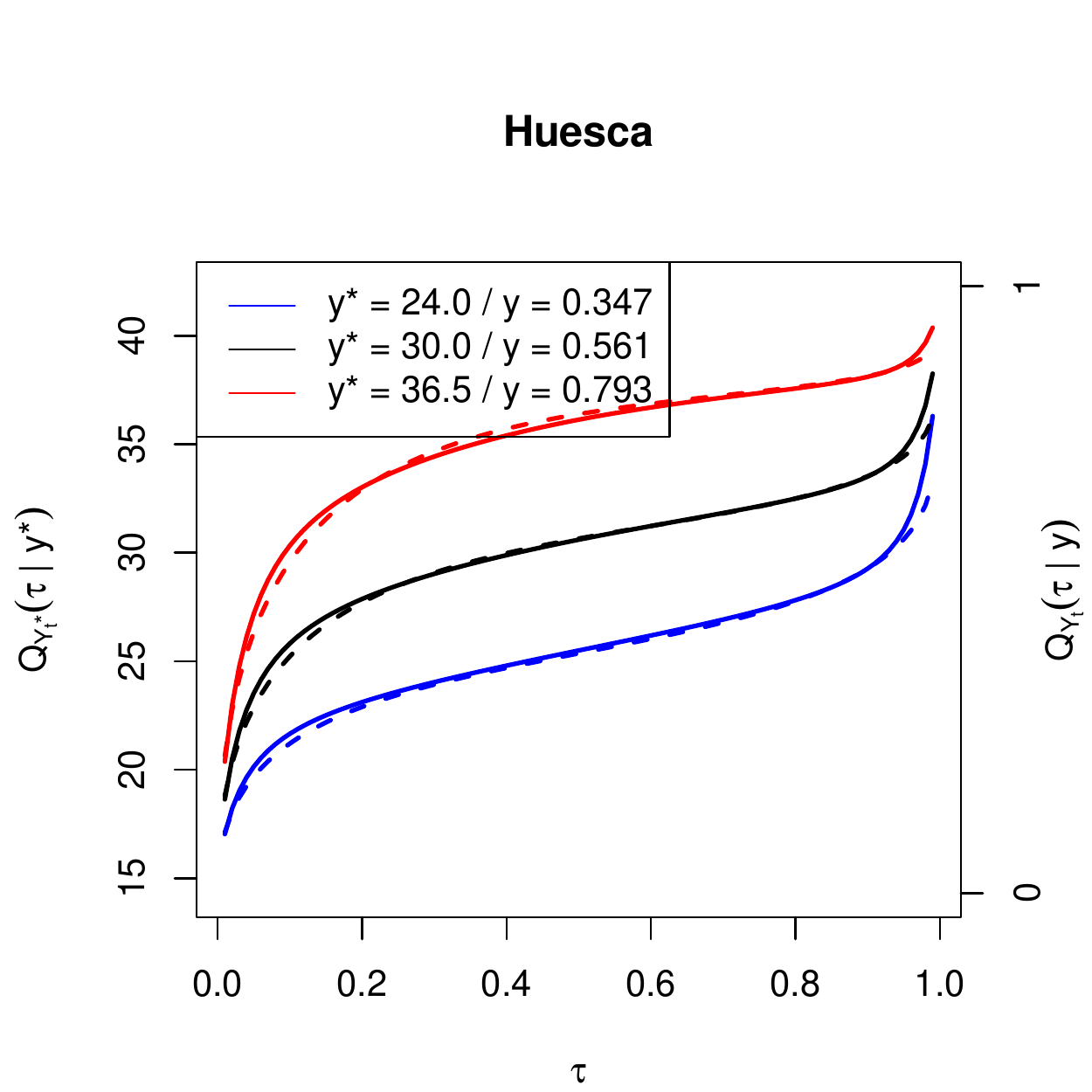}
\includegraphics[width=3cm]{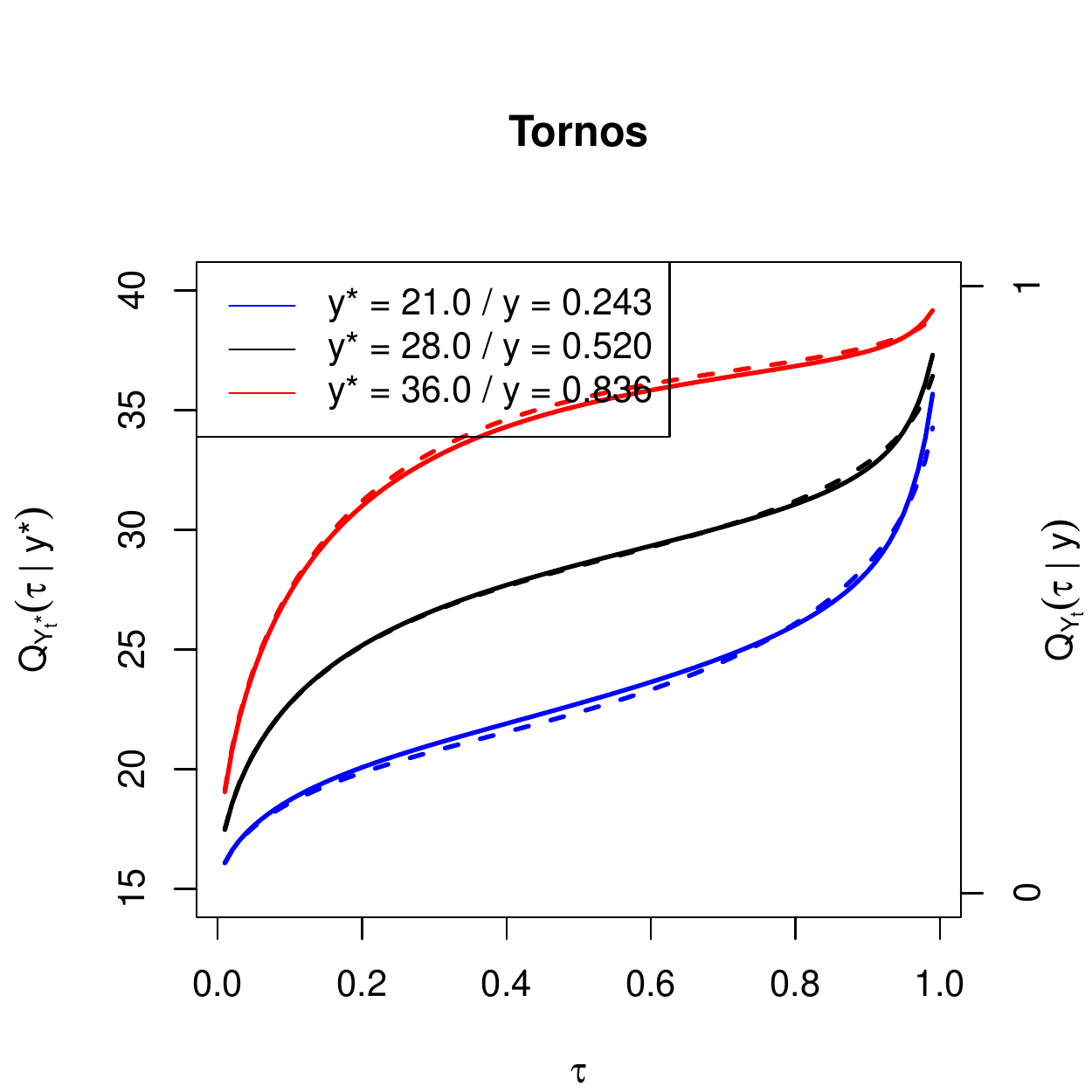}
\includegraphics[width=3cm]{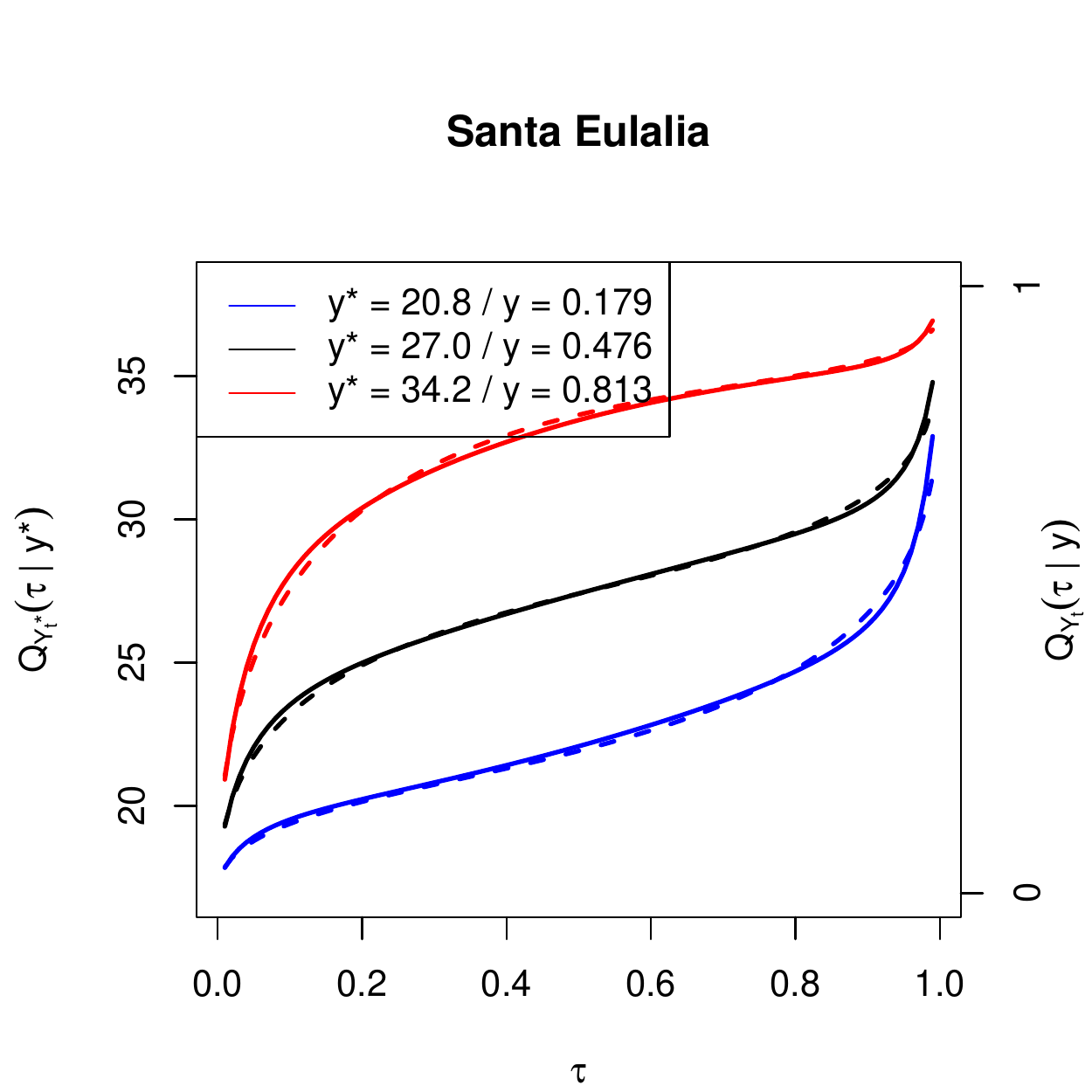}
\includegraphics[width=3cm]{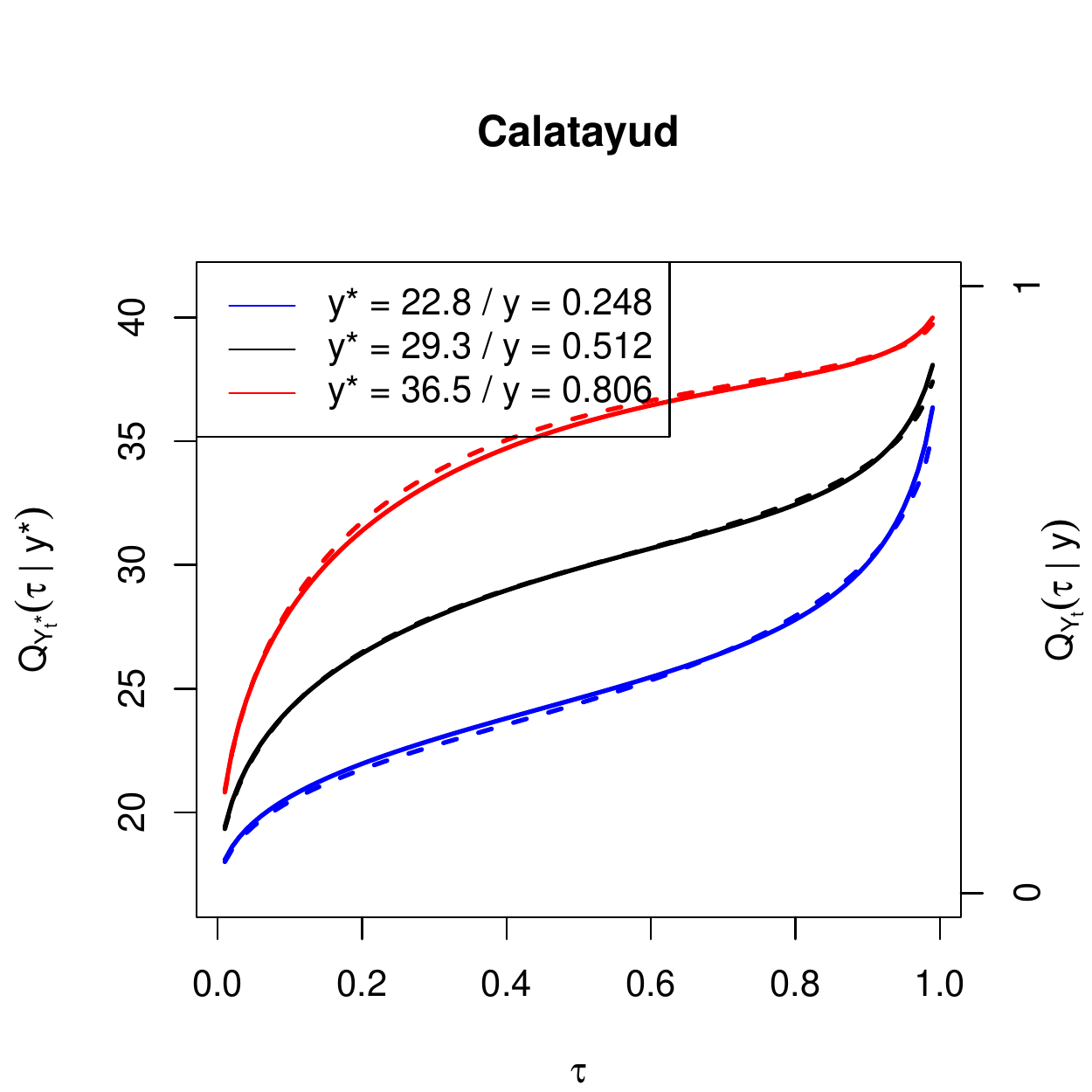} \\
\includegraphics[width=3cm]{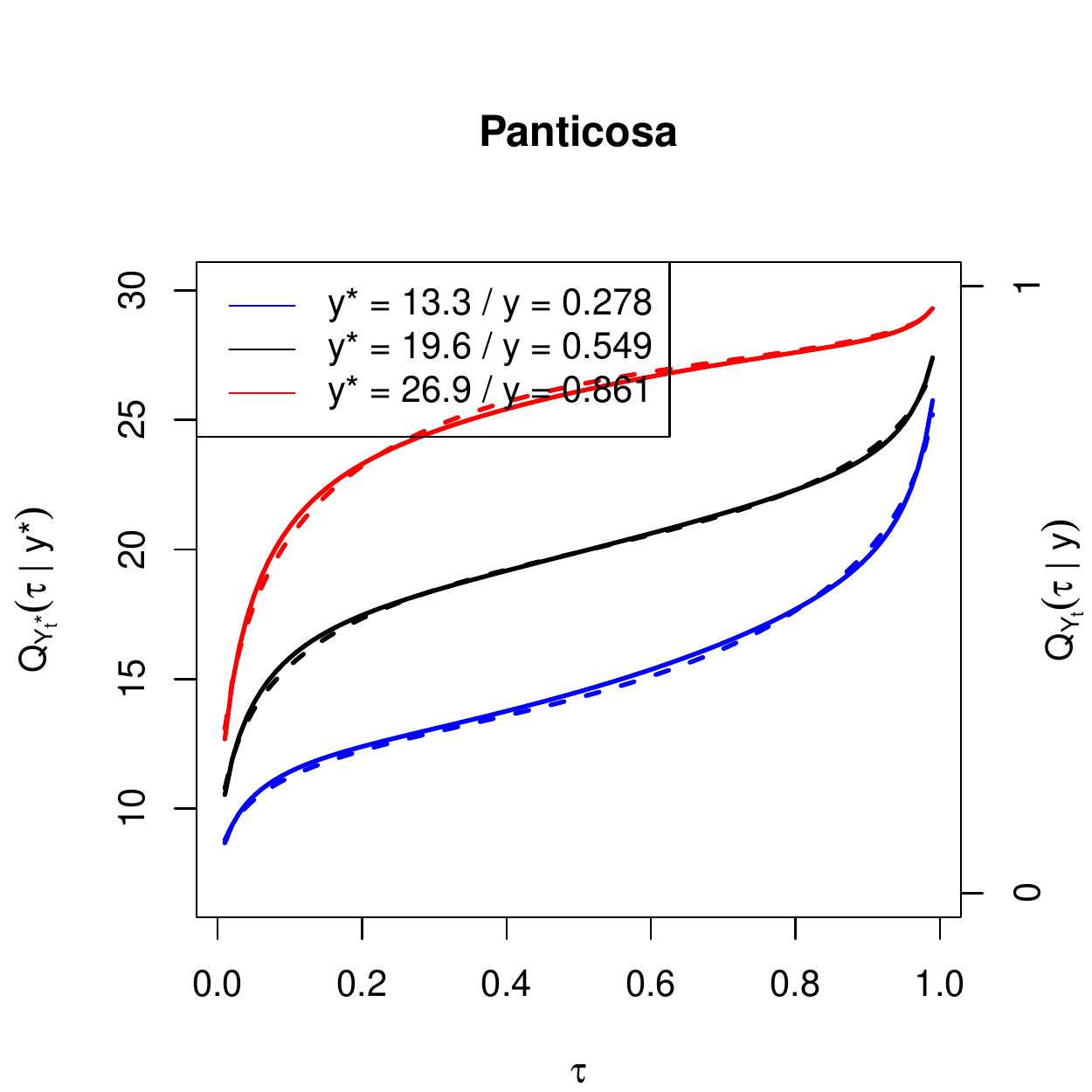}
\includegraphics[width=3cm]{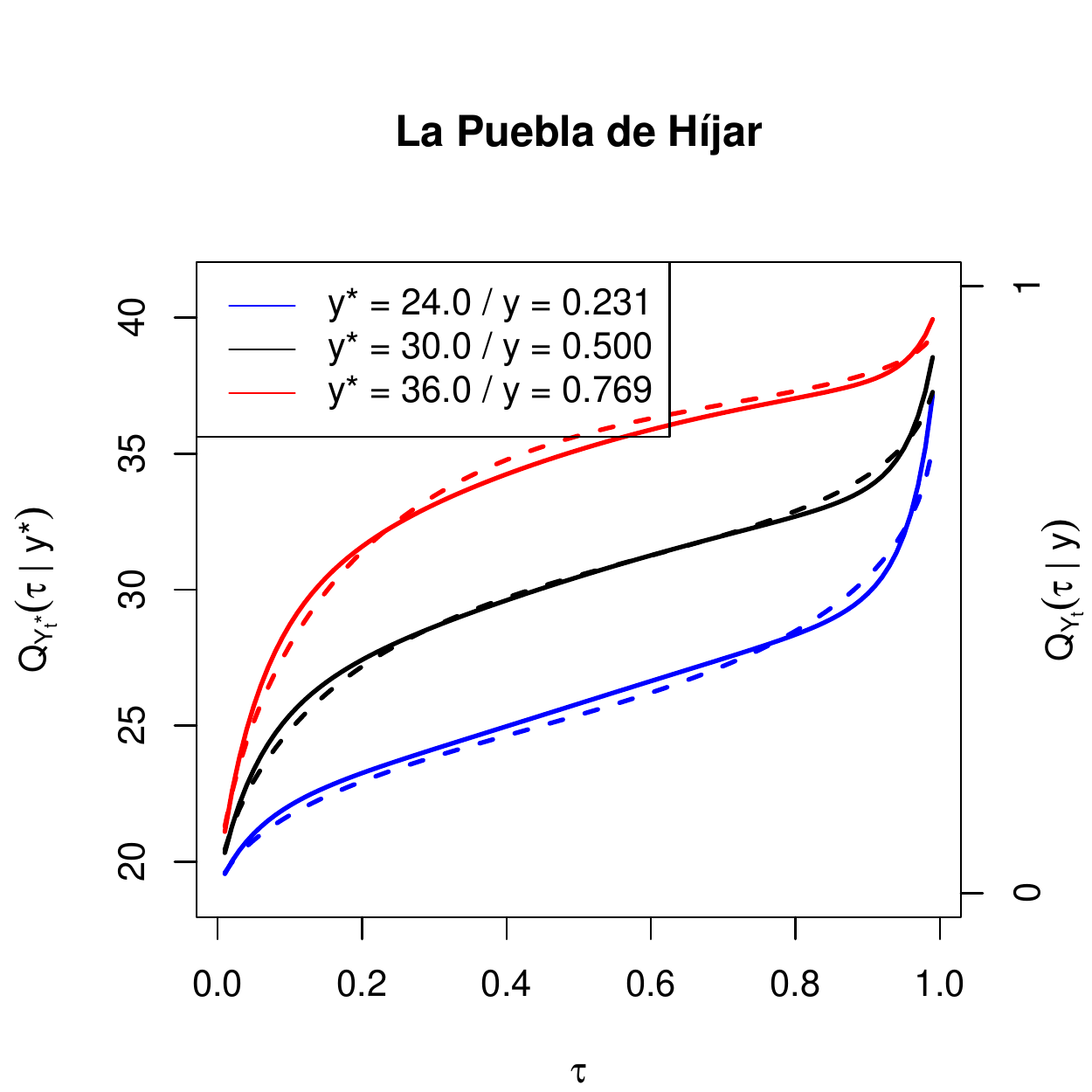}
\includegraphics[width=3cm]{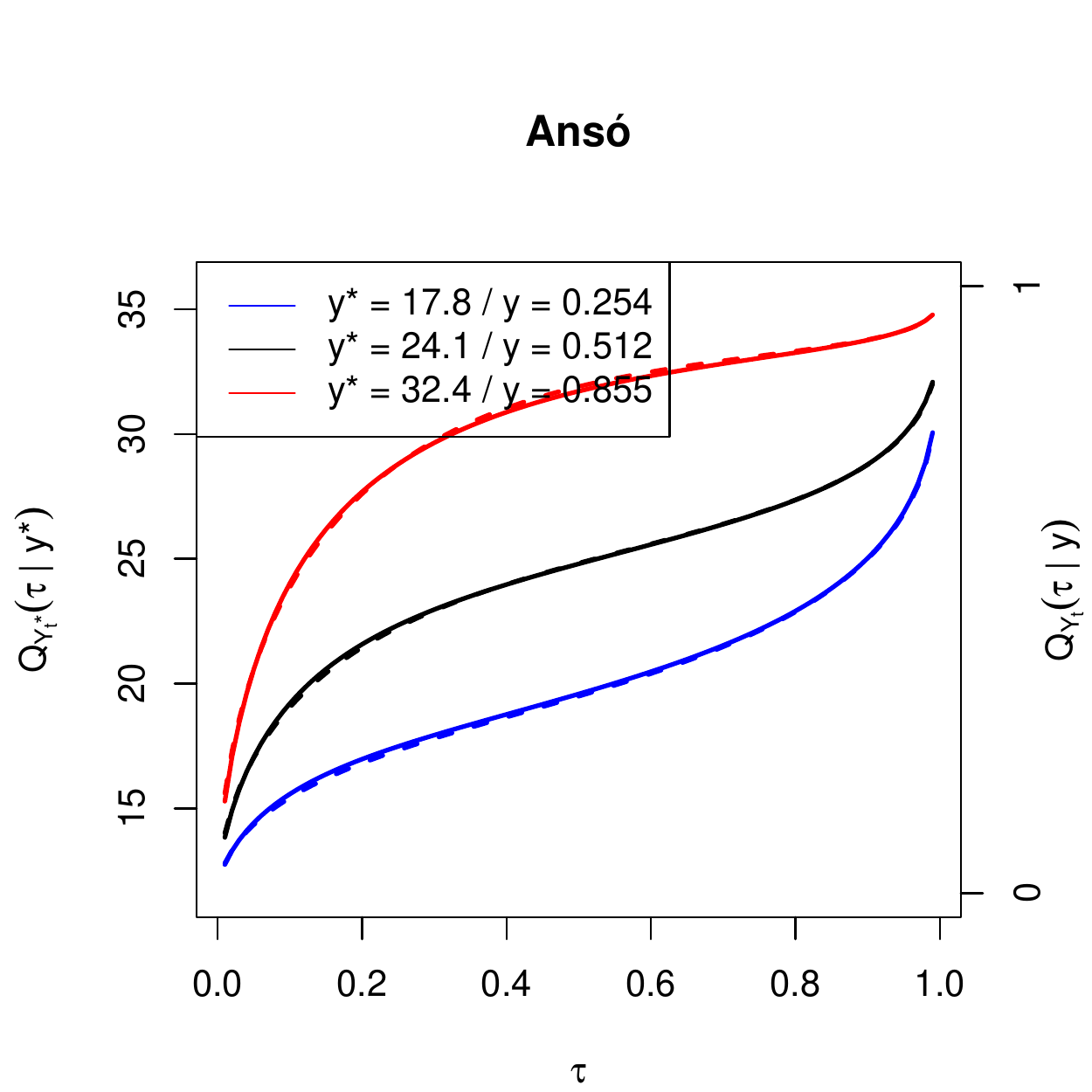}
\includegraphics[width=3cm]{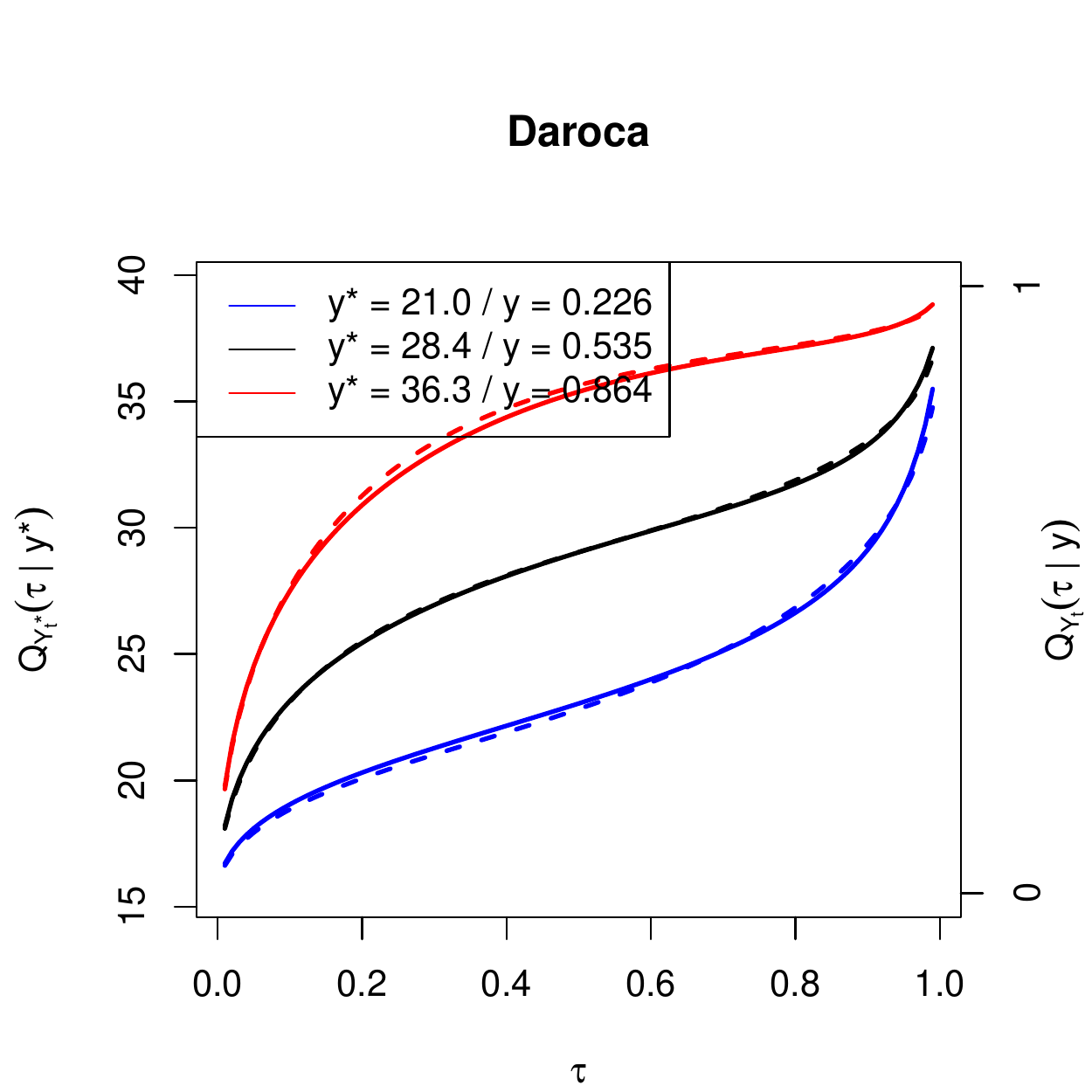} \\
\includegraphics[width=3cm]{quantileQAR1s13.pdf}
\includegraphics[width=3cm]{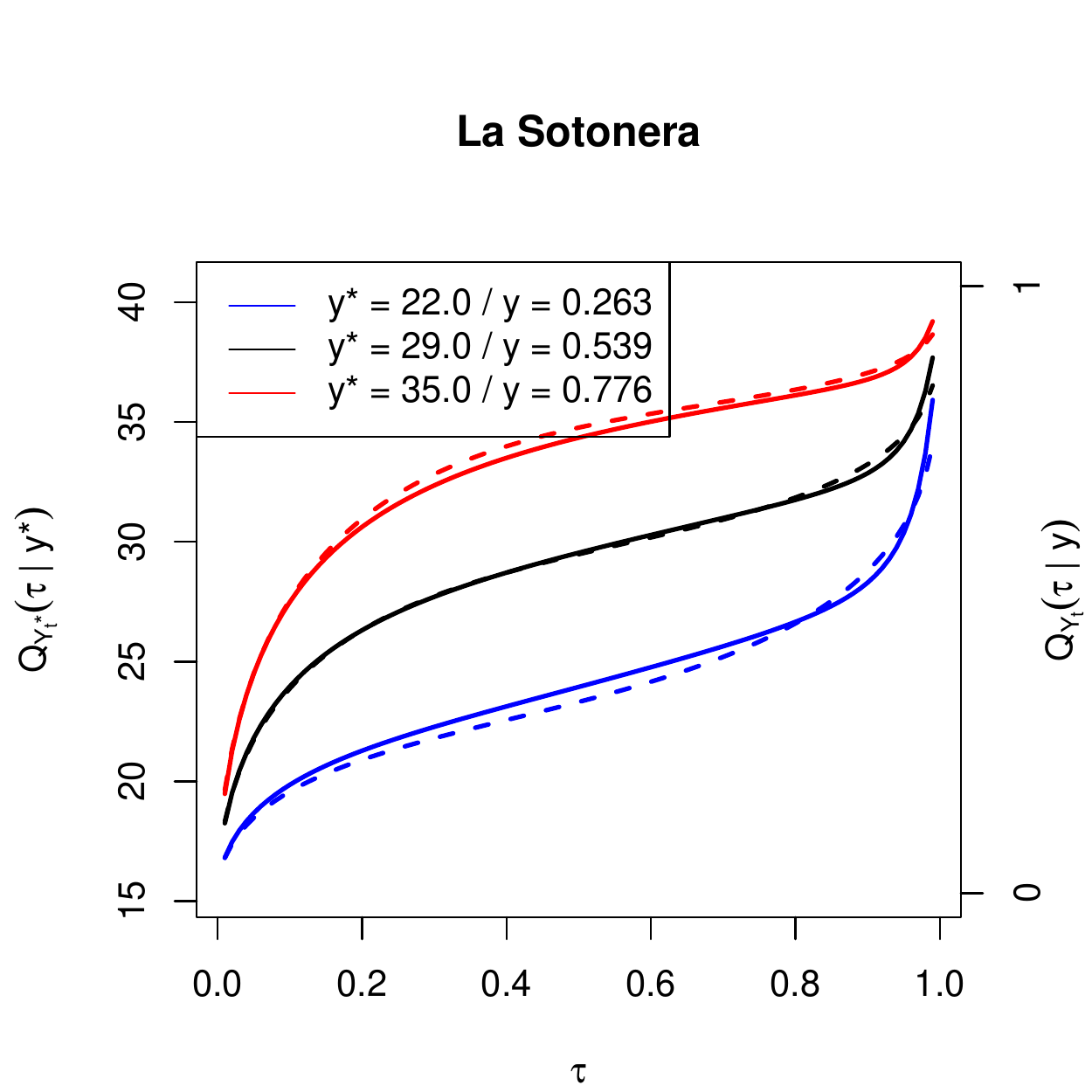}
\includegraphics[width=3cm]{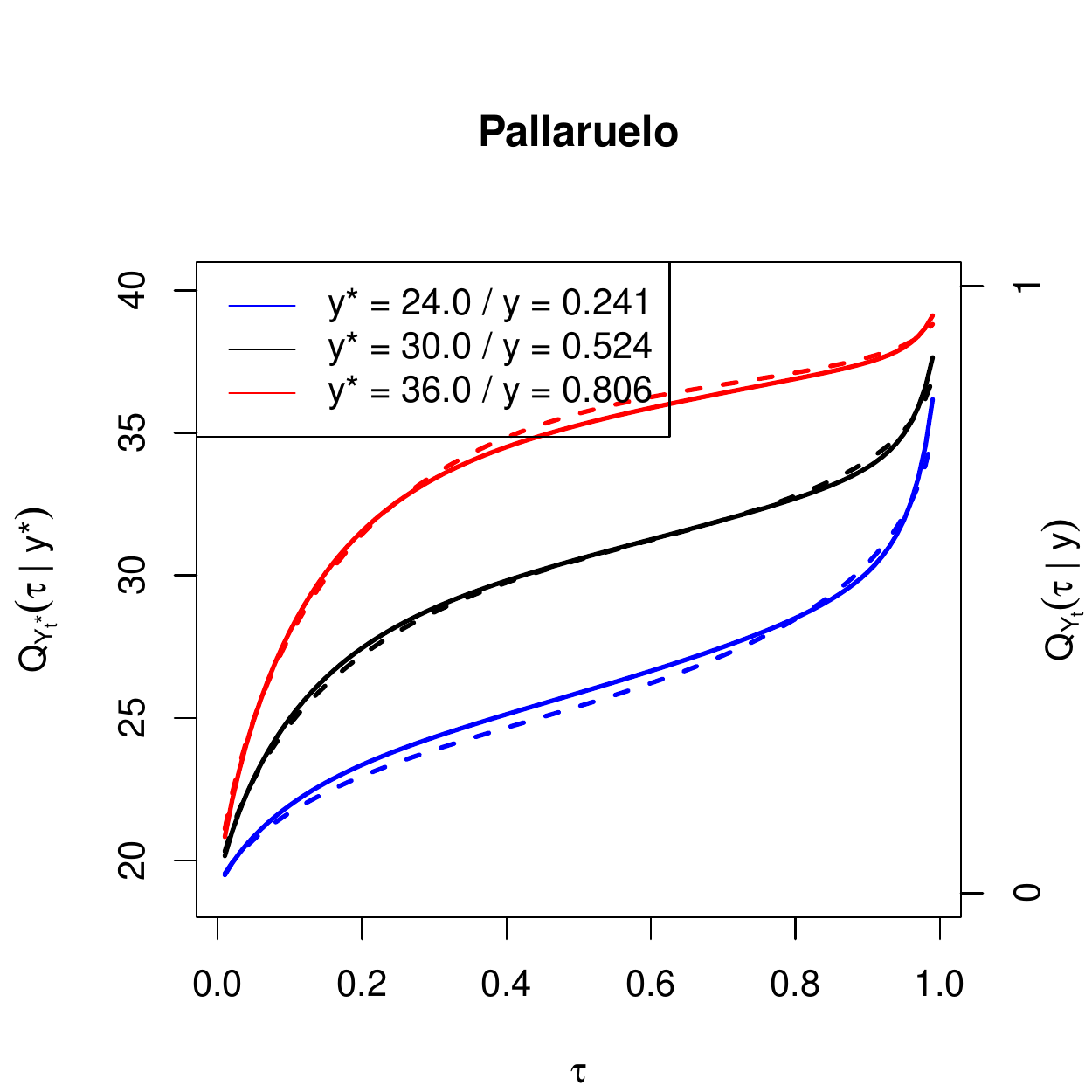}
\includegraphics[width=3cm]{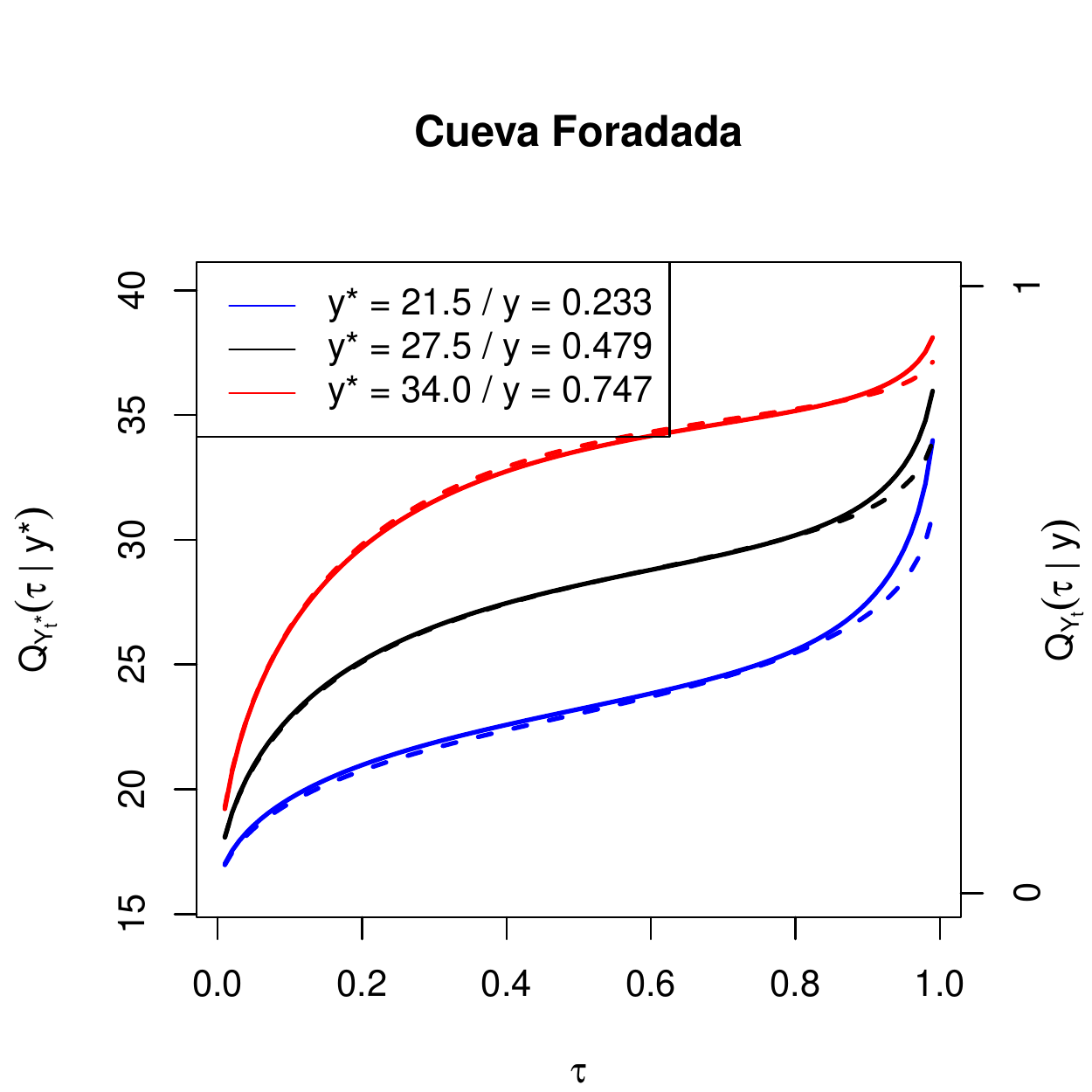} \\
\includegraphics[width=3cm]{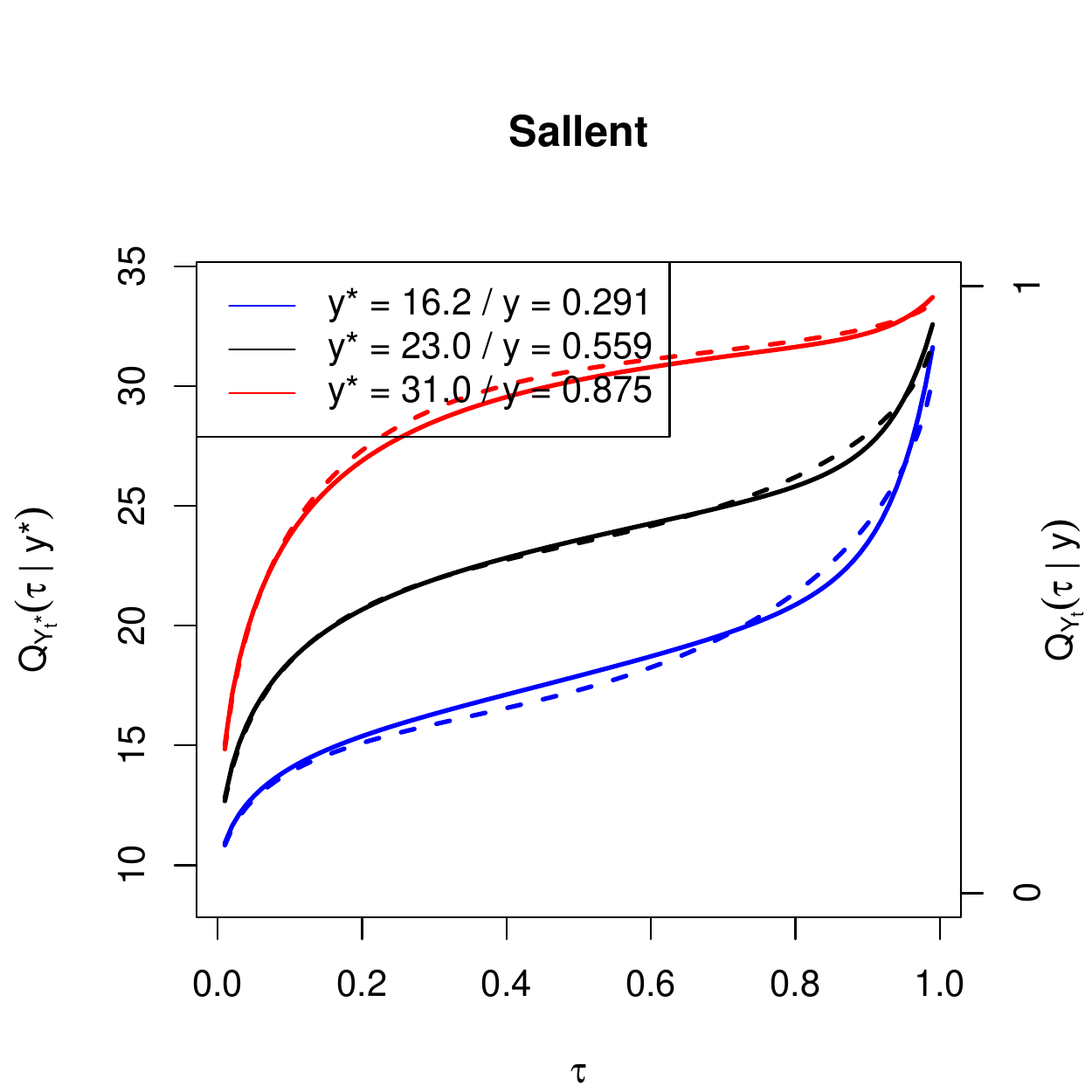}
\includegraphics[width=3cm]{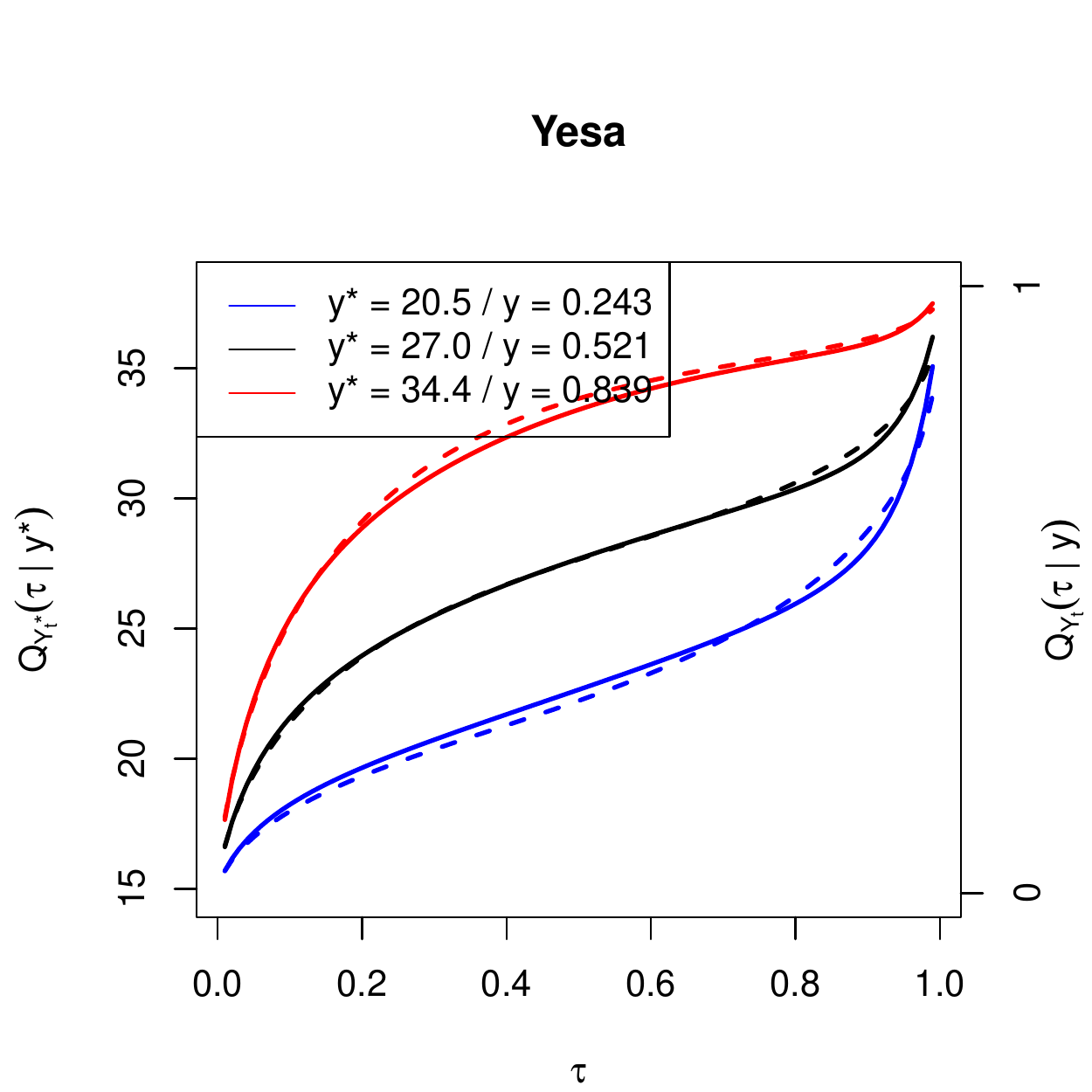}
\caption{Posterior mean of the quantile function $Q_{Y_{t}}(\tau \mid y)$ vs. $\tau$ for QAR1K1 (dashed) and QAR1K2 (solid).  Here, $y$ is the empirical marginal quantile for $\tau=0.1$ (blue), $0.5$ (black), $0.9$ (red). All locations, MJJAS, 2015. }
\label{fig:quantile}
\end{figure}

\begin{figure}[!ht]
\centering
\includegraphics[width=3cm]{densityQAR1s1.pdf}
\includegraphics[width=3cm]{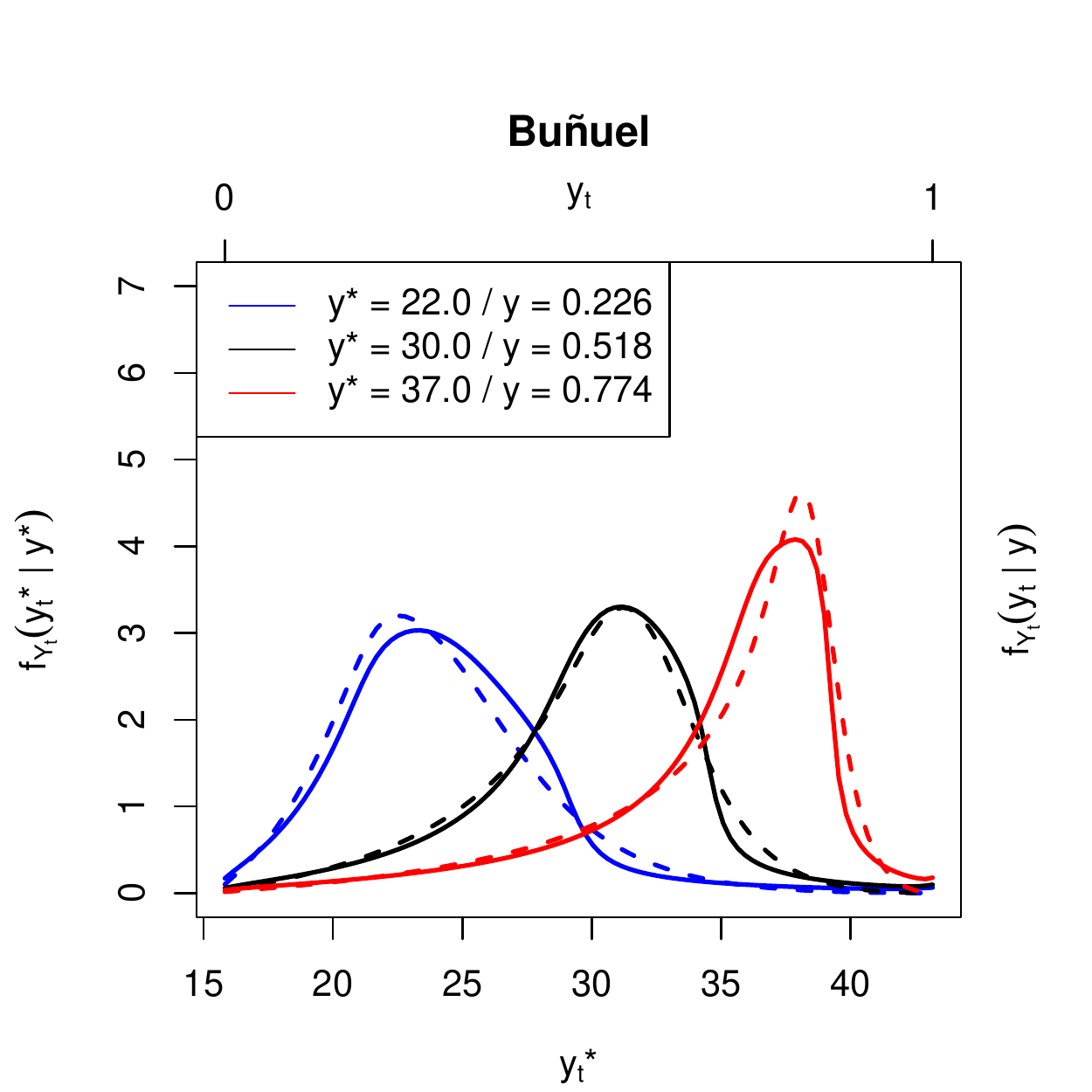}
\includegraphics[width=3cm]{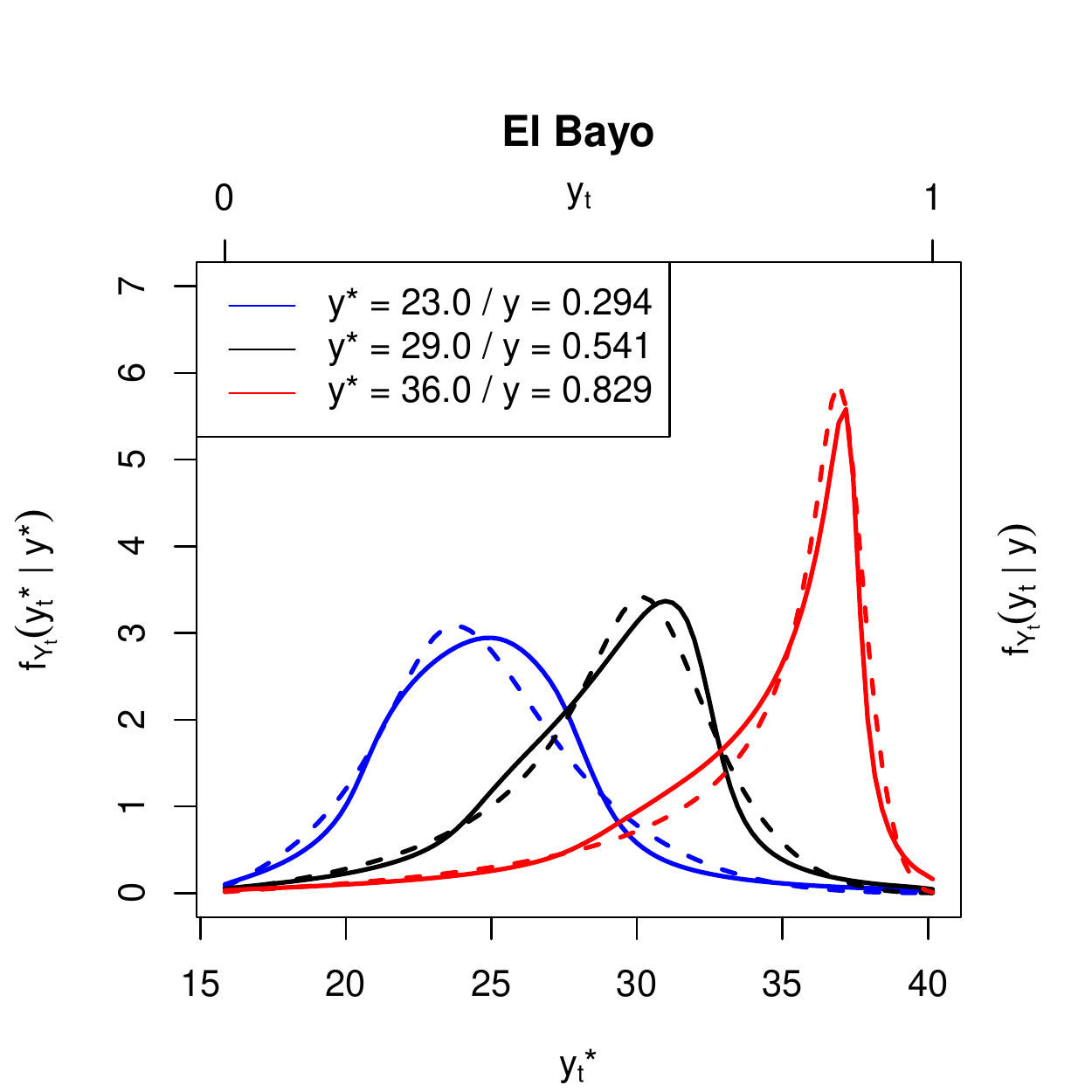}
\includegraphics[width=3cm]{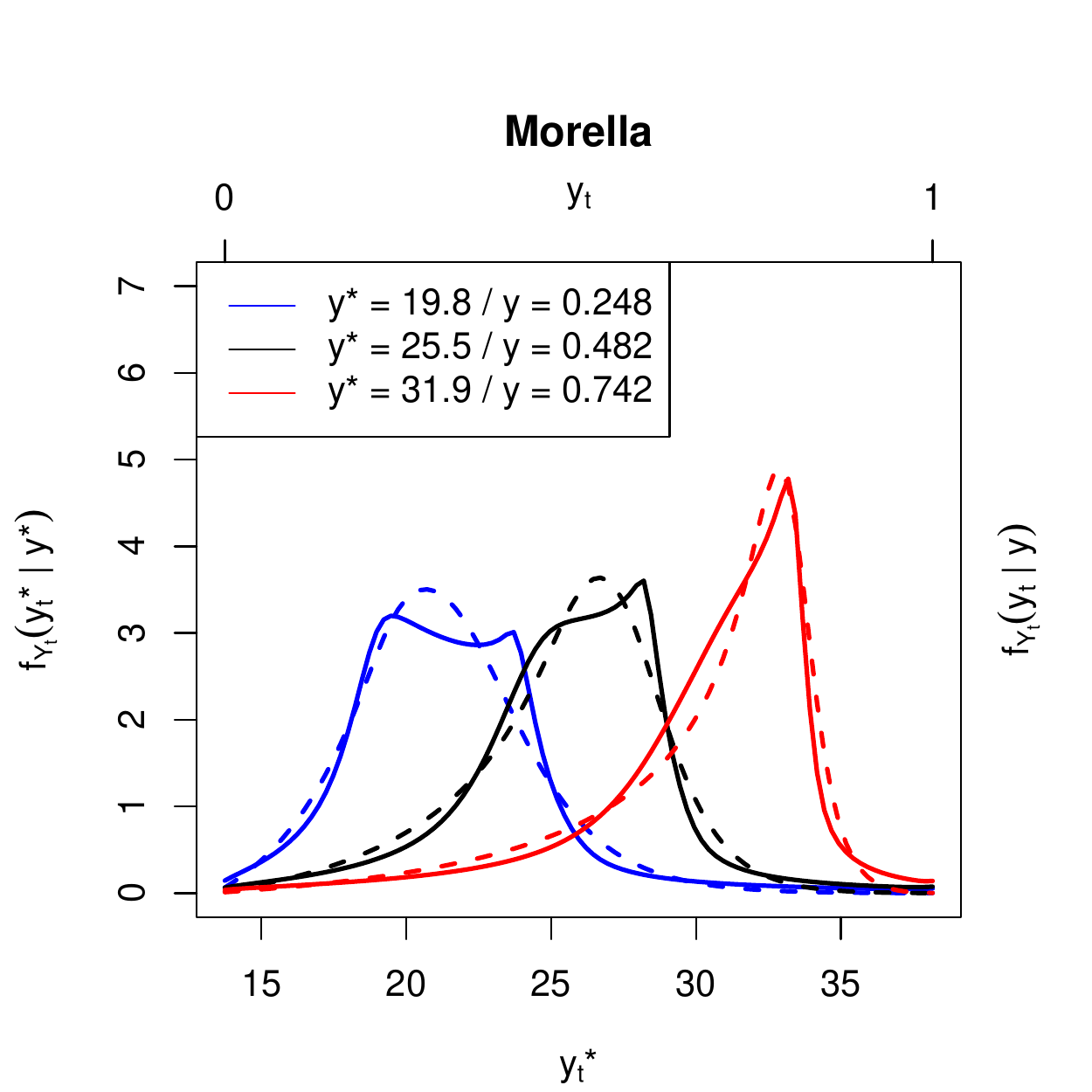} \\
\includegraphics[width=3cm]{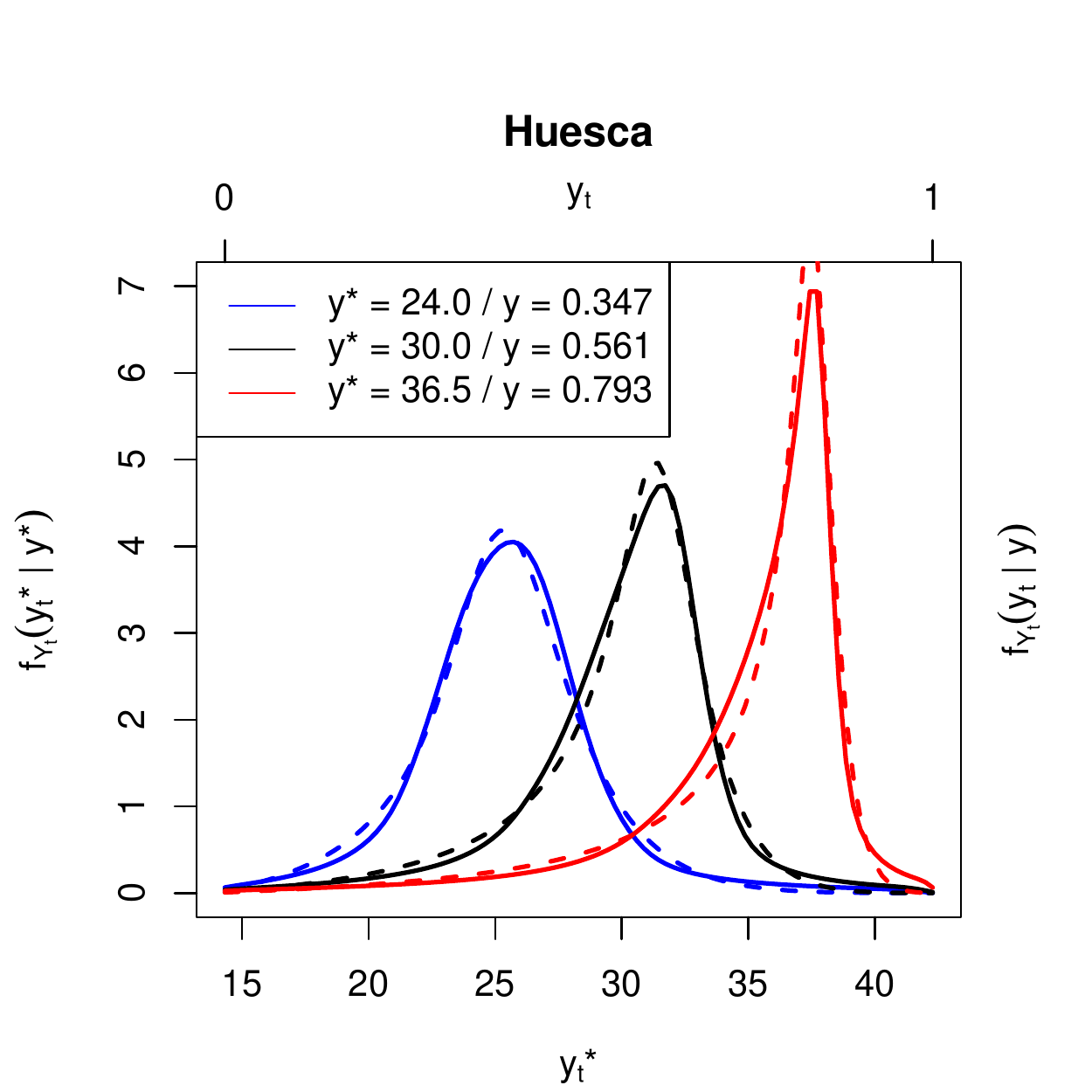}
\includegraphics[width=3cm]{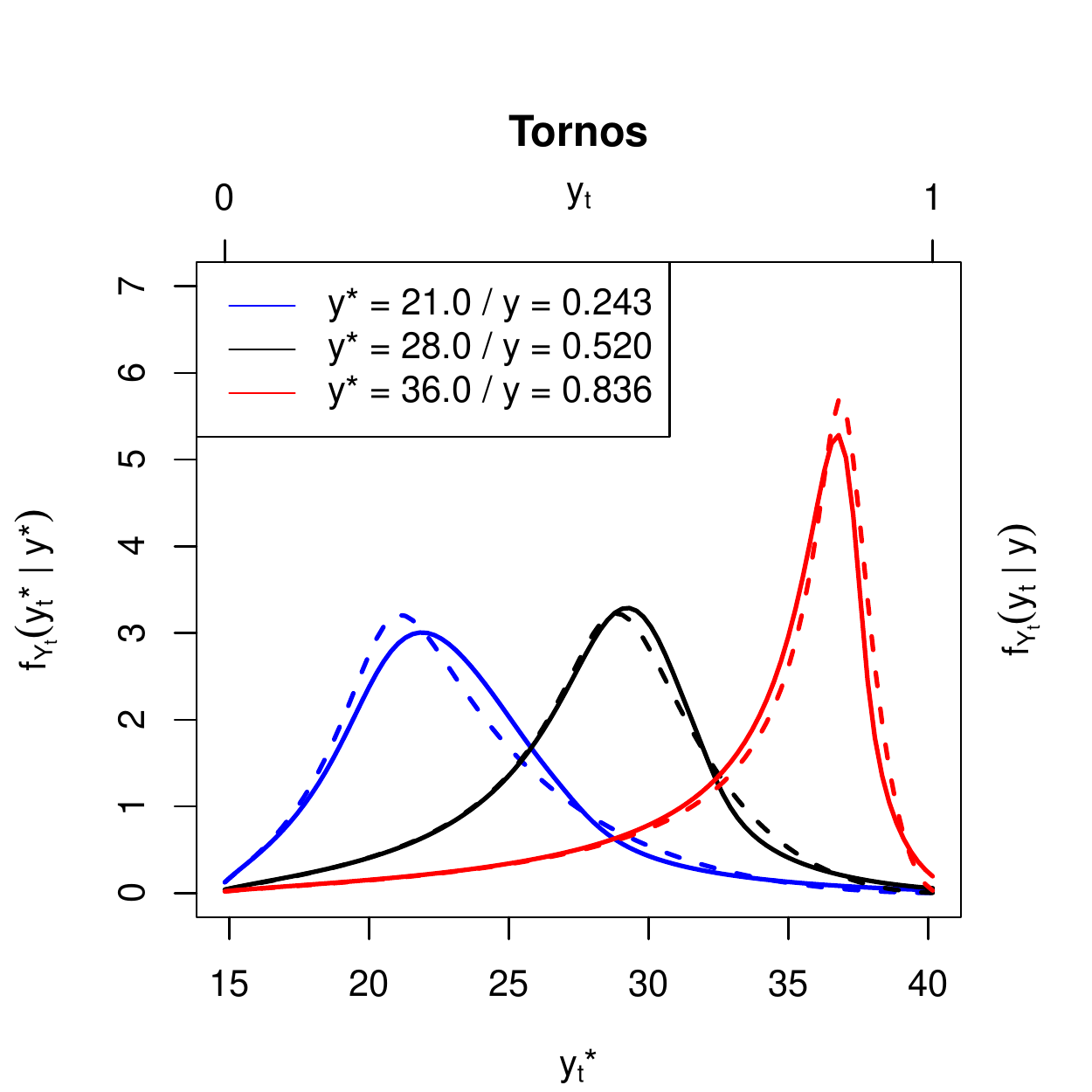}
\includegraphics[width=3cm]{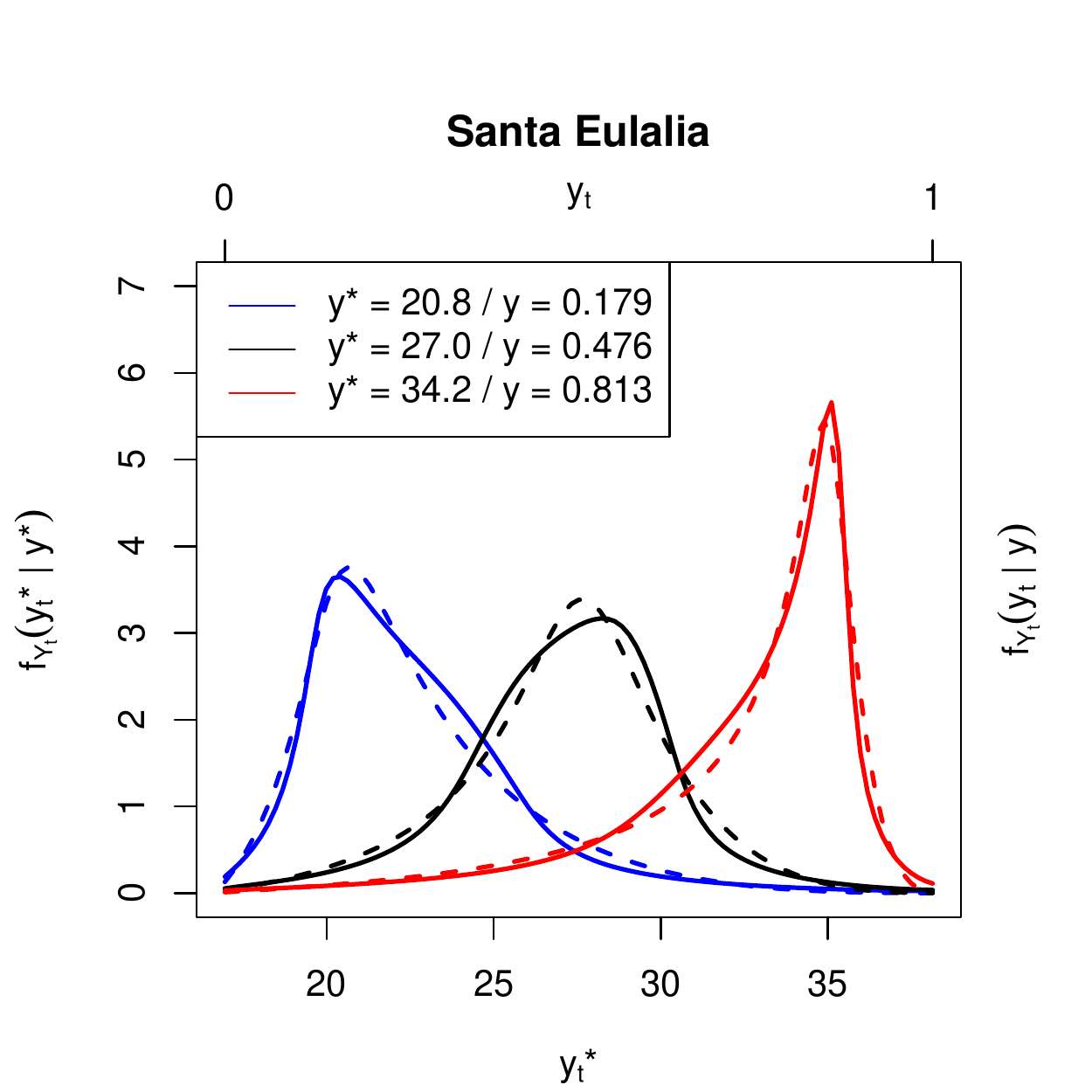}
\includegraphics[width=3cm]{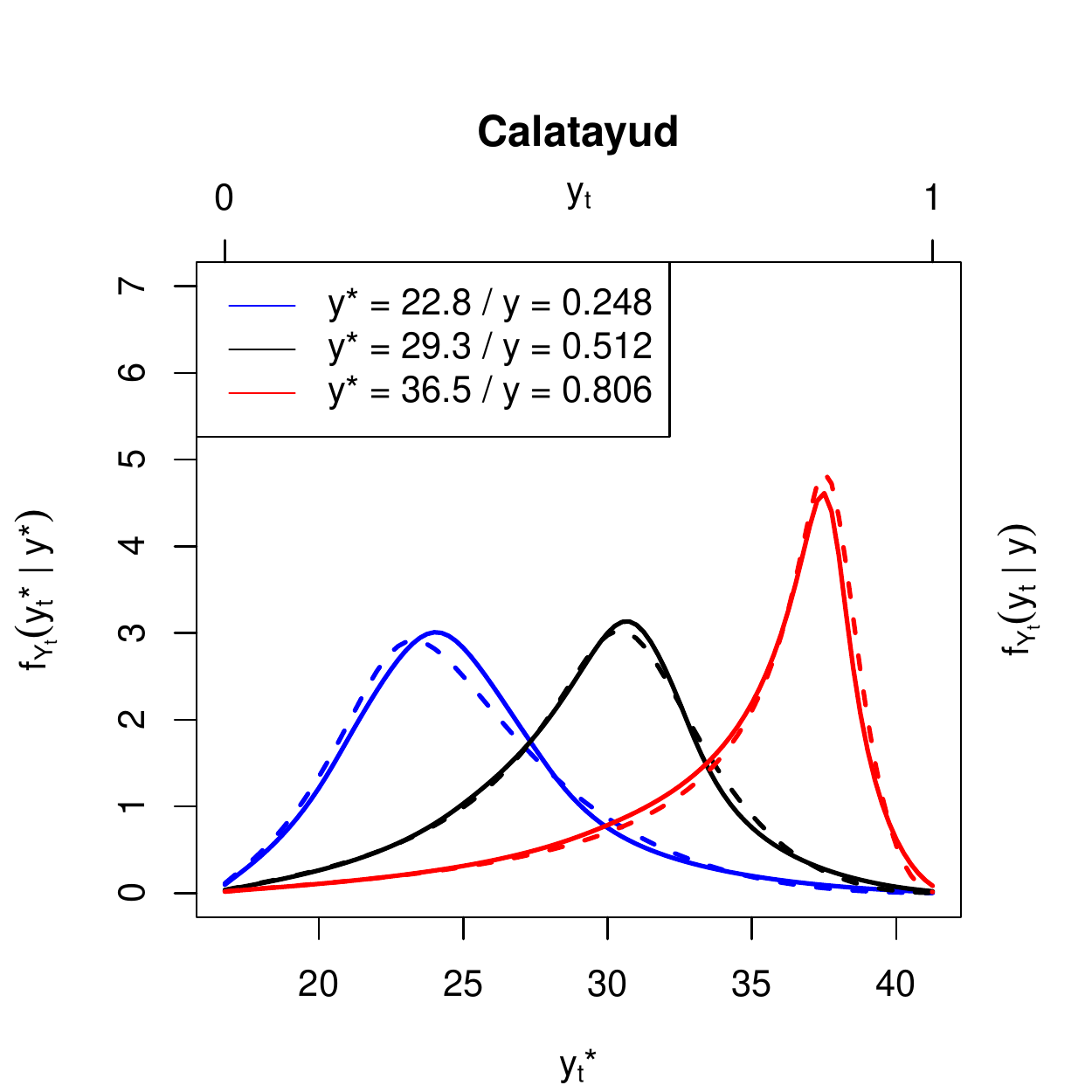} \\
\includegraphics[width=3cm]{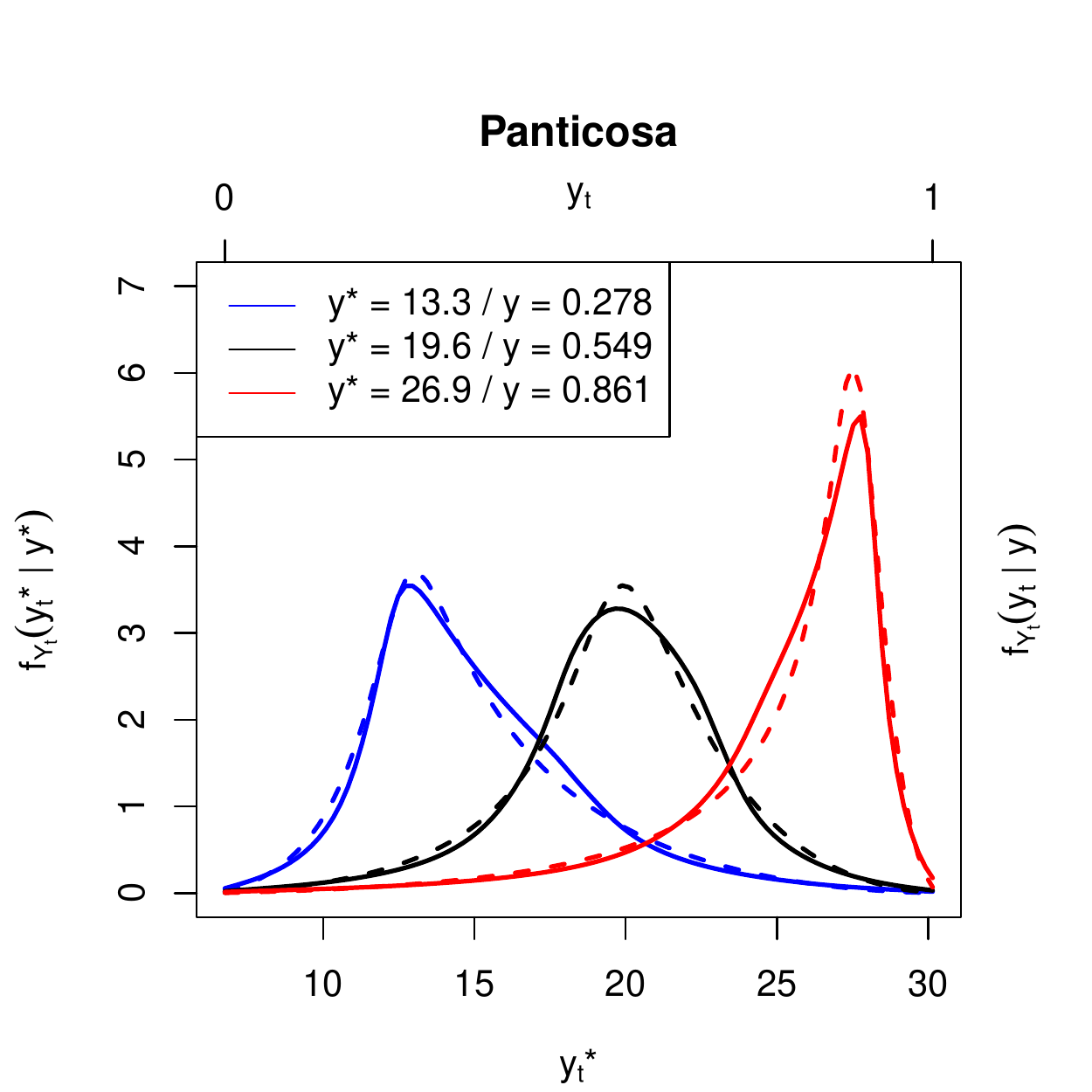}
\includegraphics[width=3cm]{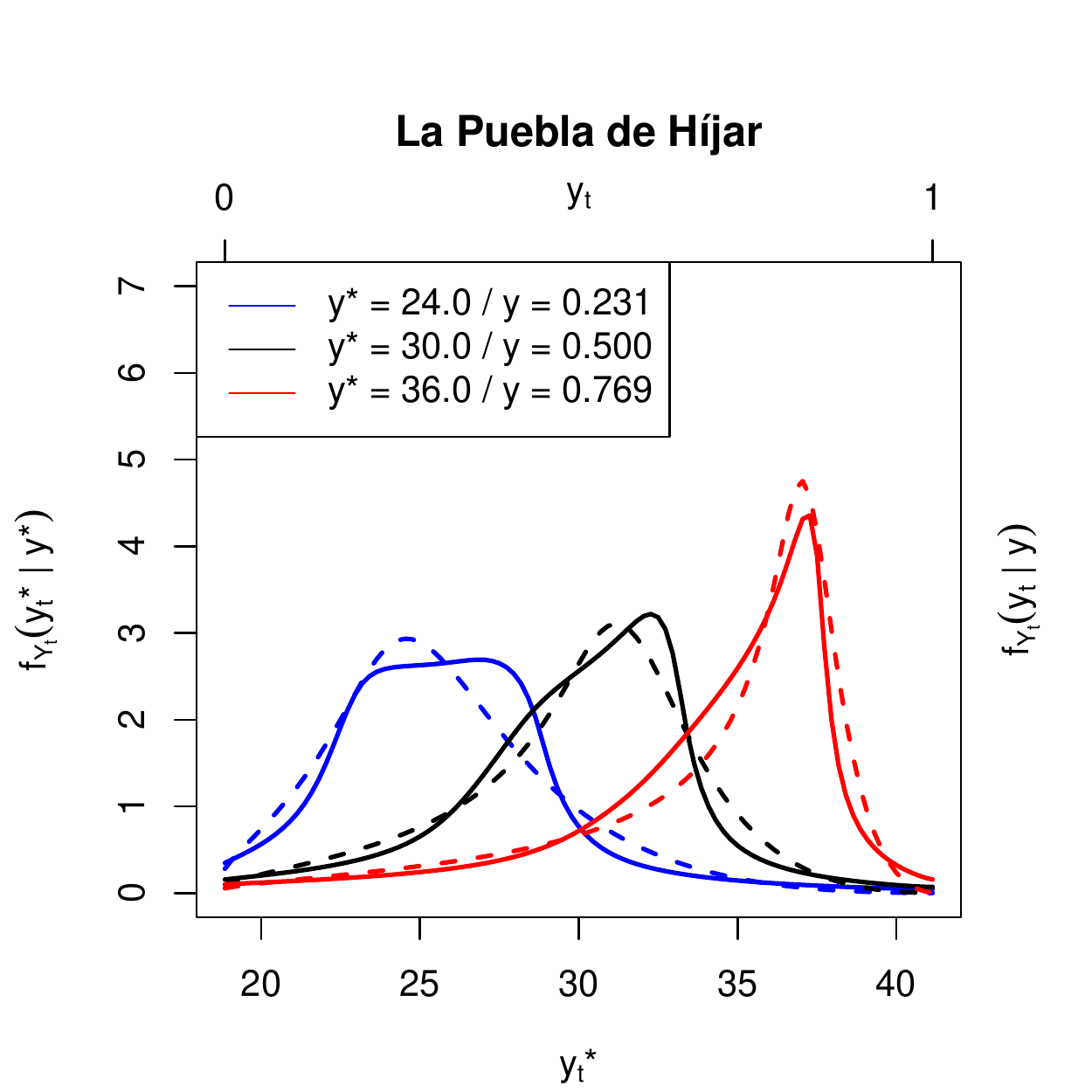}
\includegraphics[width=3cm]{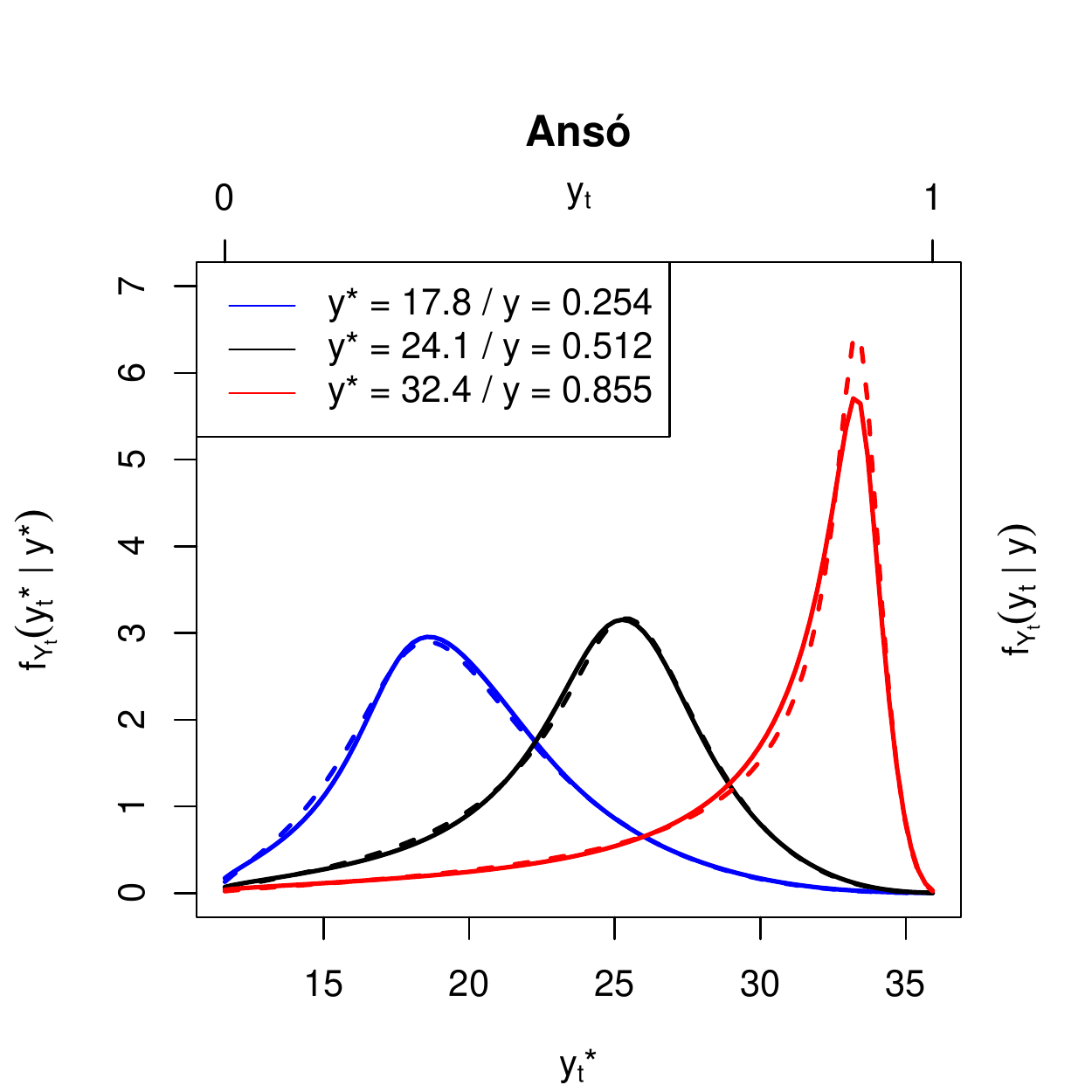}
\includegraphics[width=3cm]{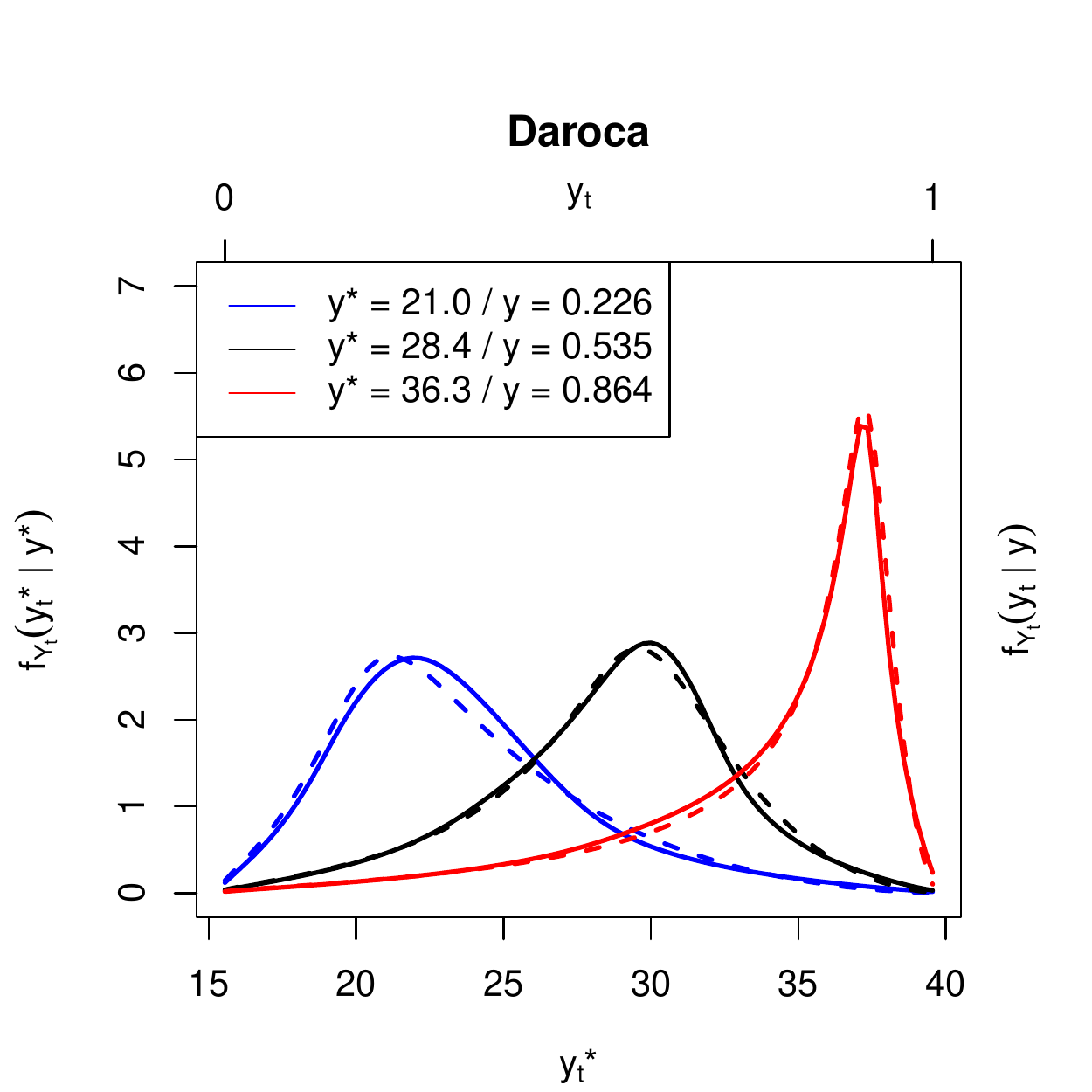} \\
\includegraphics[width=3cm]{densityQAR1s13.pdf}
\includegraphics[width=3cm]{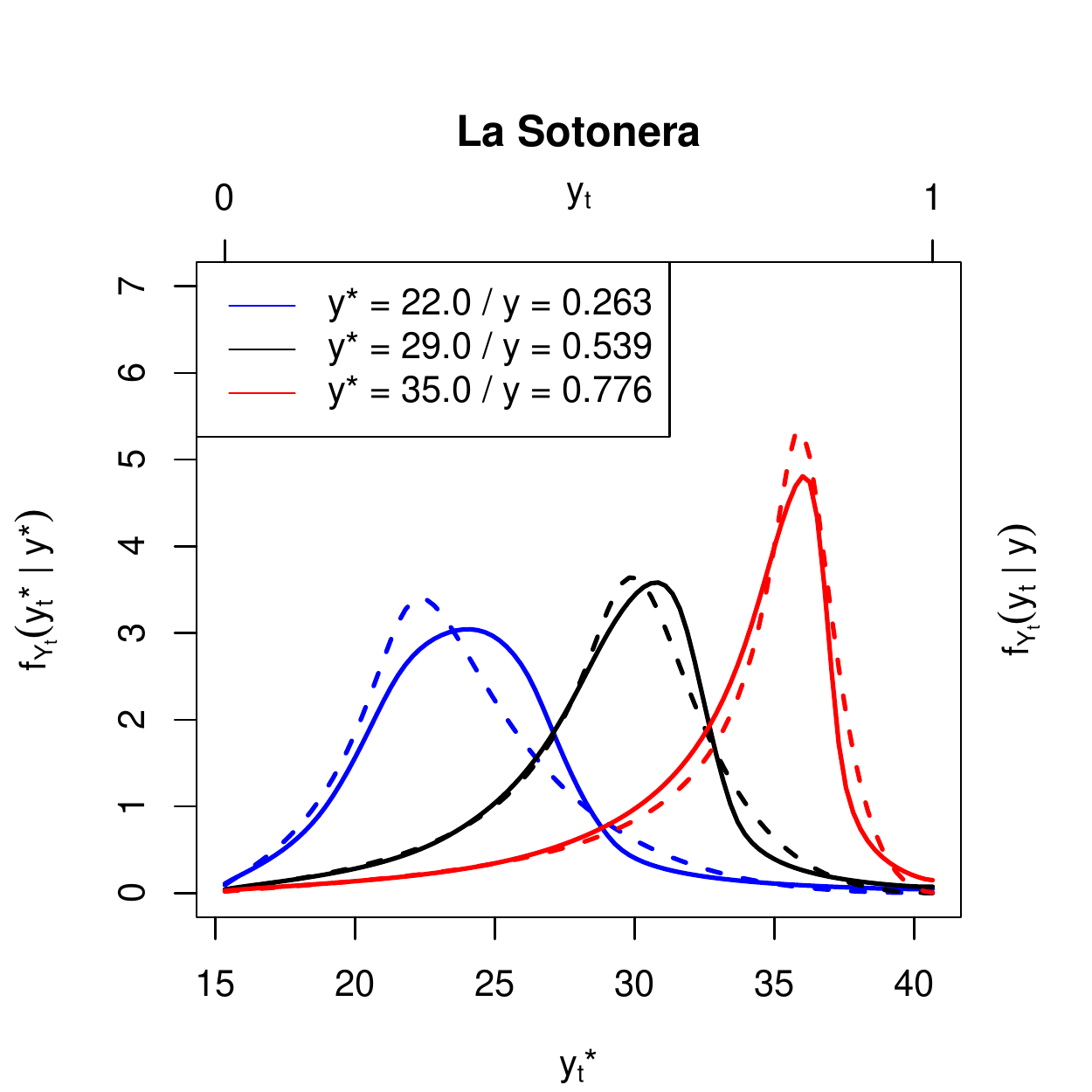}
\includegraphics[width=3cm]{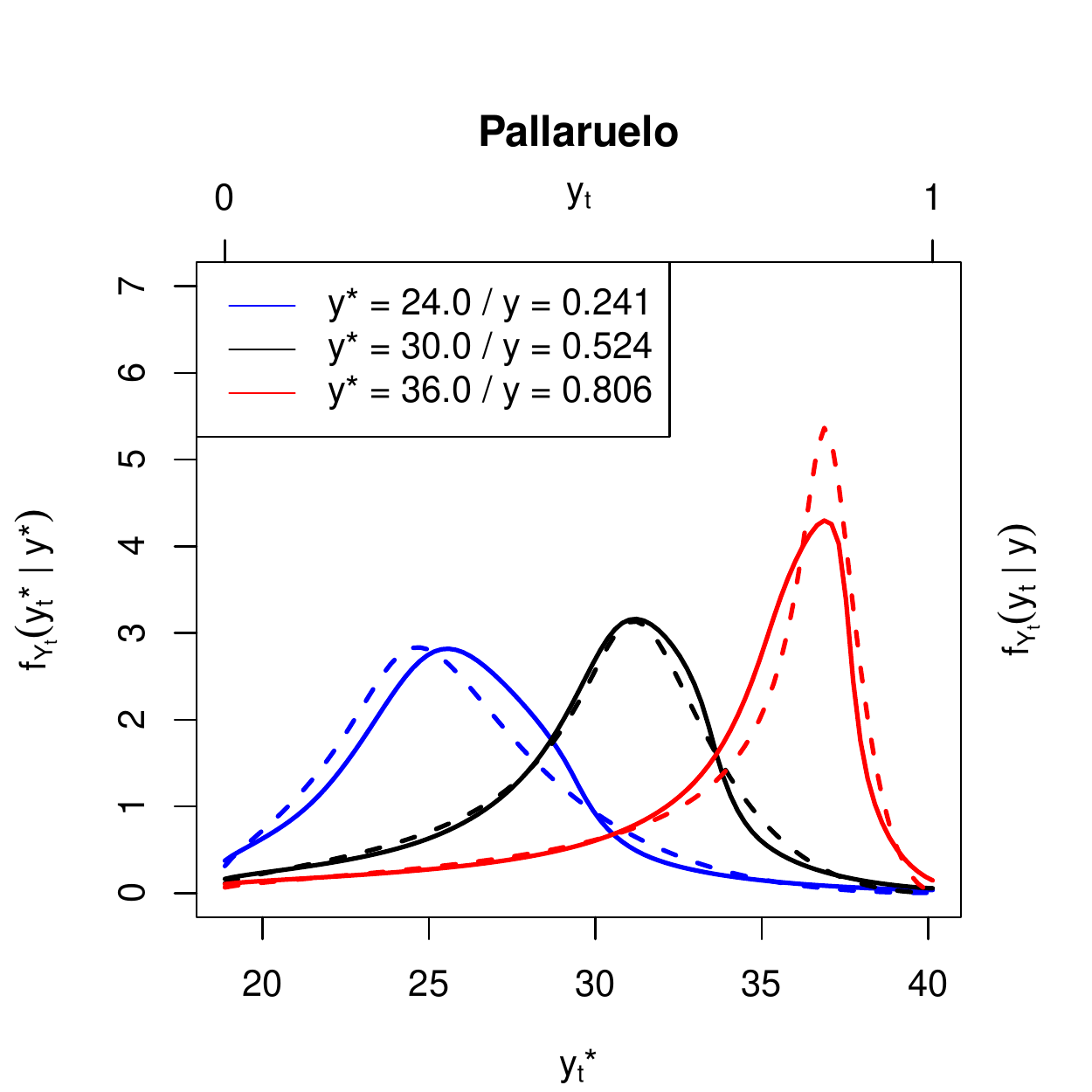}
\includegraphics[width=3cm]{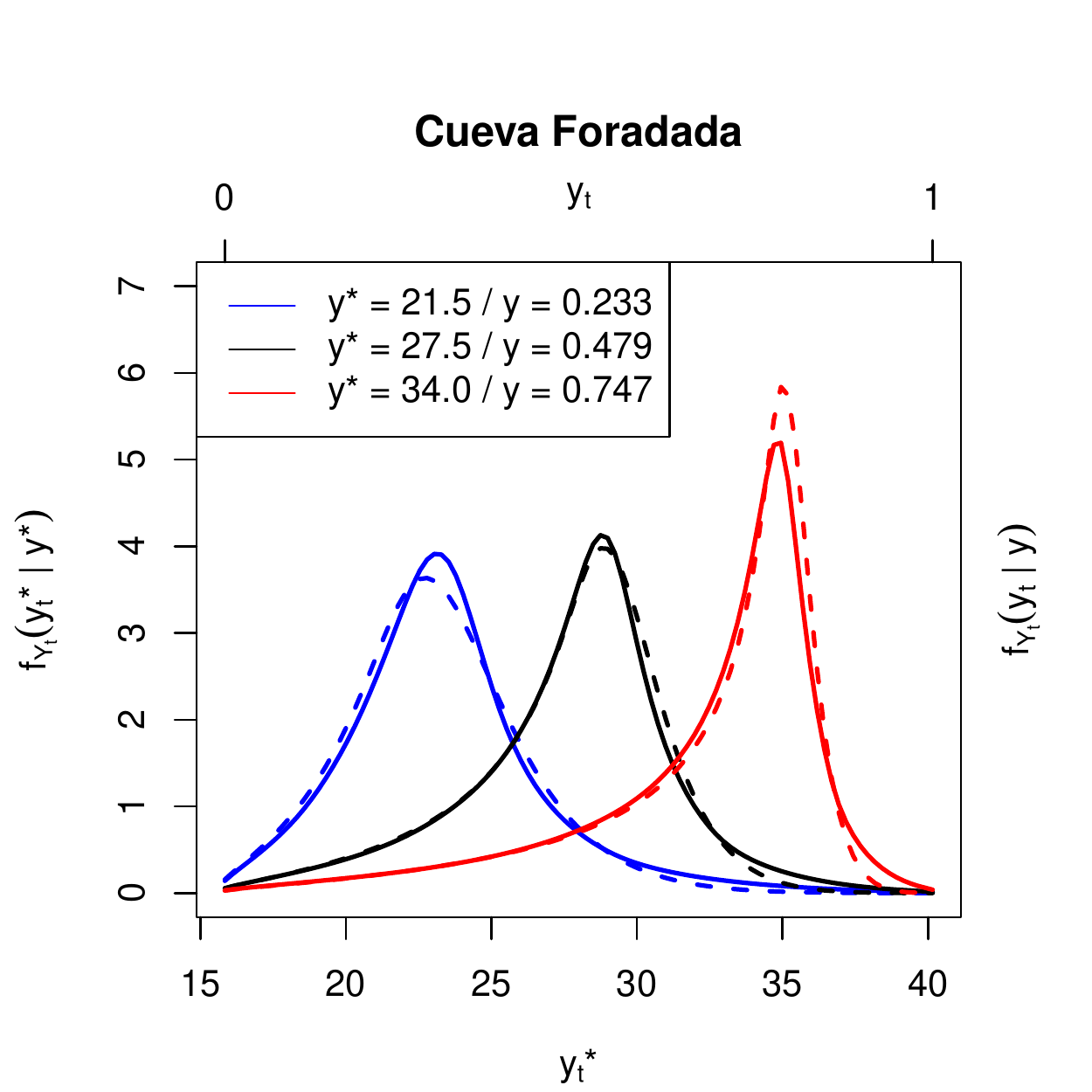} \\
\includegraphics[width=3cm]{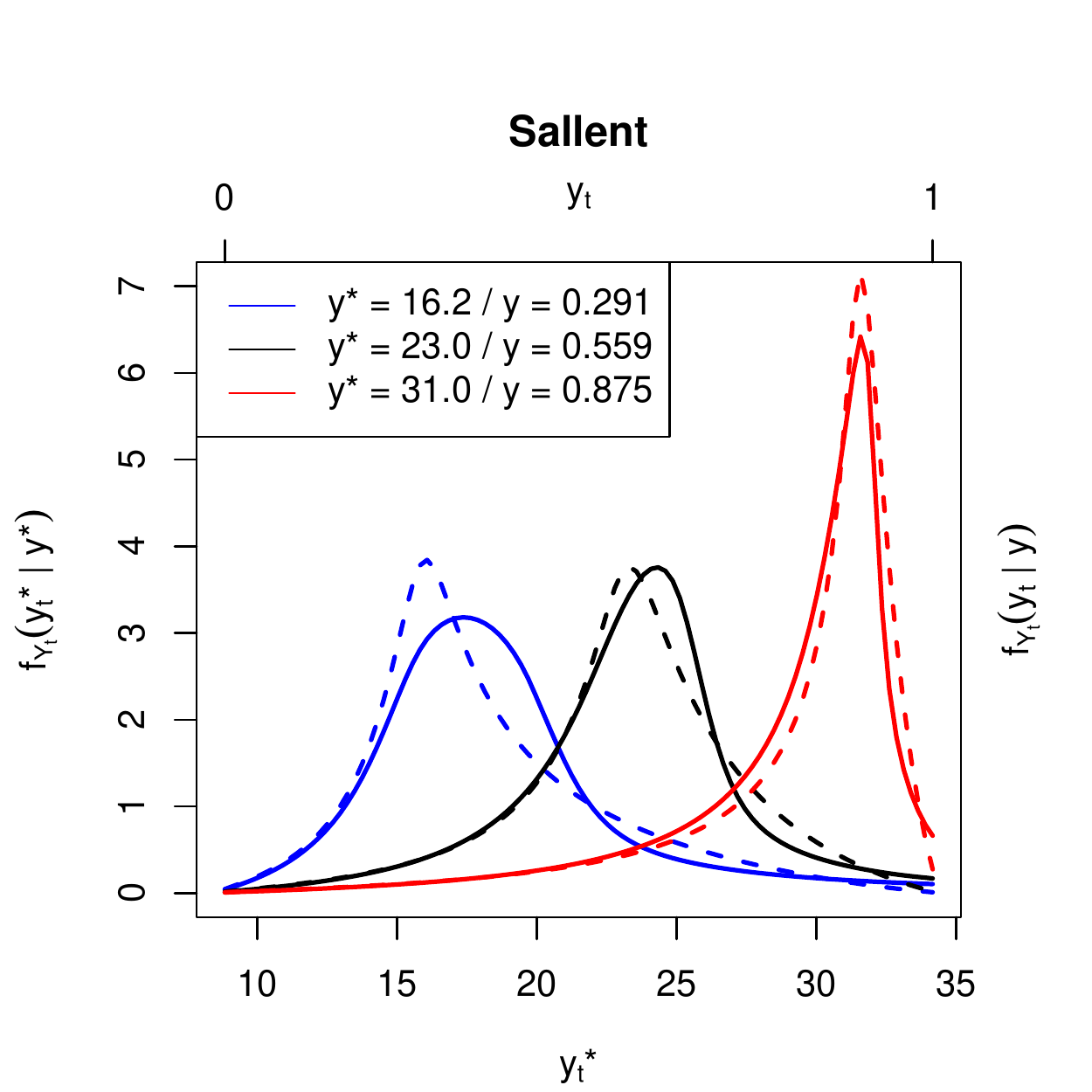}
\includegraphics[width=3cm]{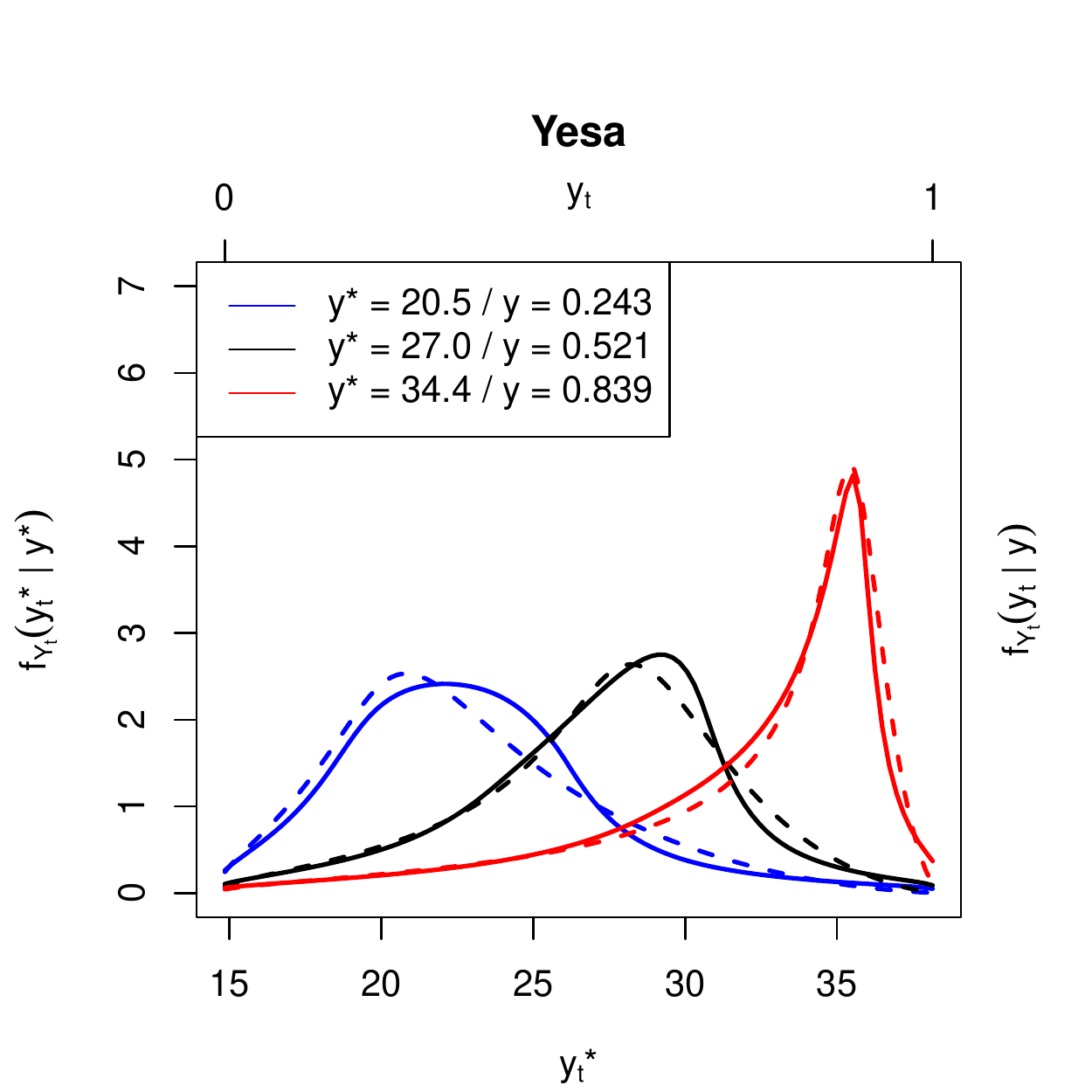}
\caption{Posterior mean of the density function $f_{Y_{t}}(y_t \mid y)$ for QAR1K1 (dashed) and QAR1K2 (solid).  Here, $y$ is the empirical marginal quantile for $\tau=0.1$ (blue), $0.5$ (black), $0.9$ (red). All locations, MJJAS, 2015. }
\label{fig:density}
\end{figure}

\clearpage

\subsection{The QAR(2) Case}

\begin{figure}[!ht]
\centering
\includegraphics[width=3cm]{thetaQAR2K1s1.pdf}
\includegraphics[width=3cm]{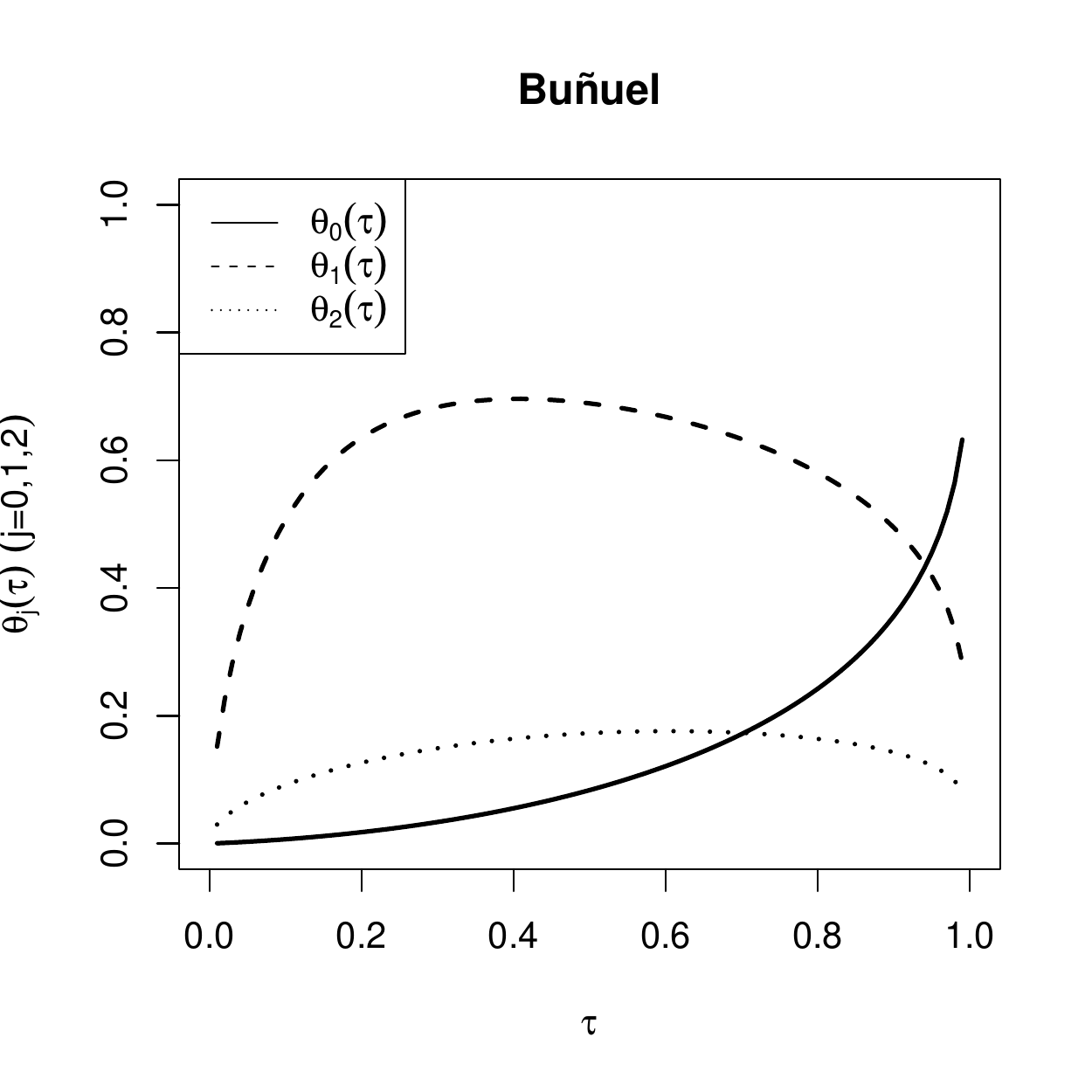}
\includegraphics[width=3cm]{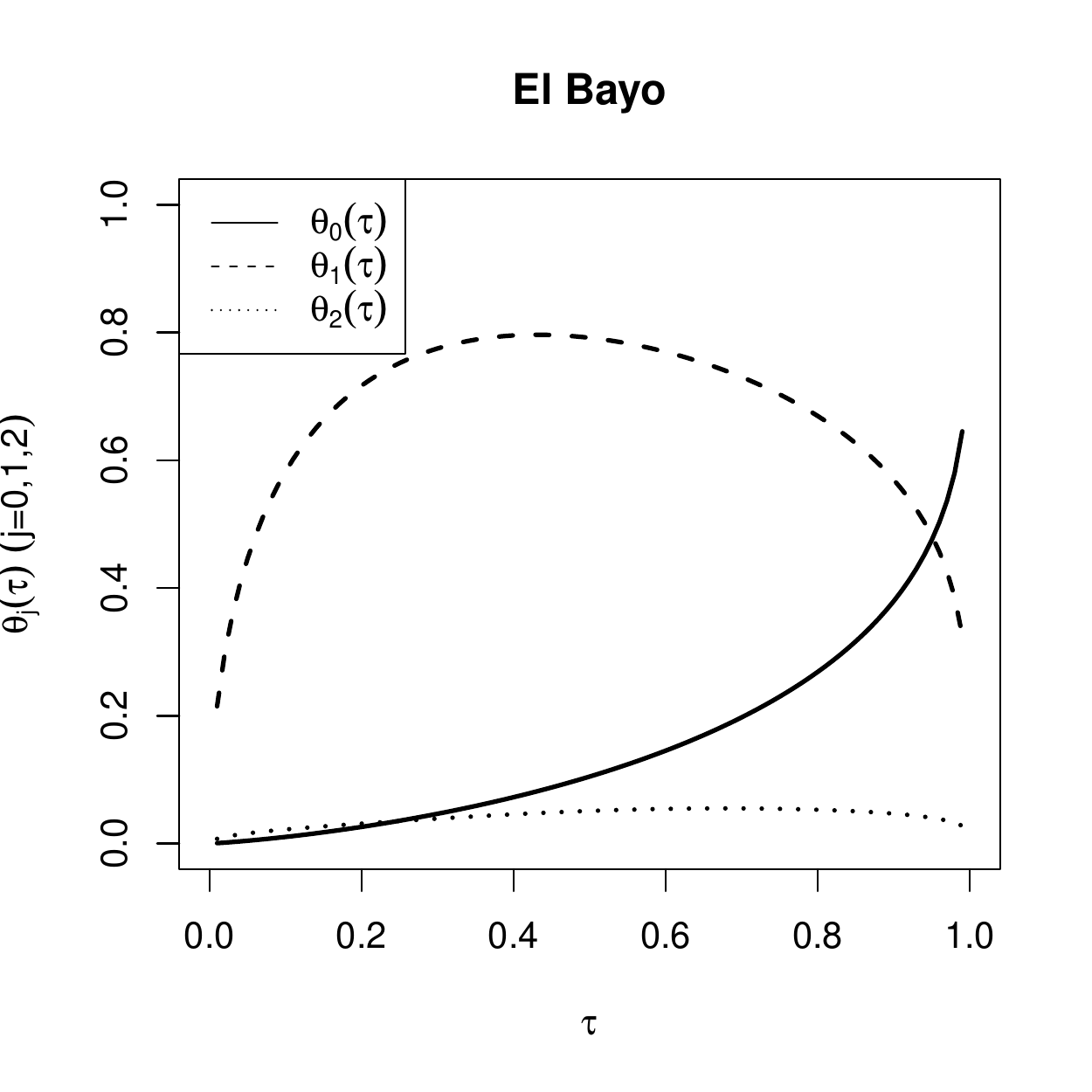}
\includegraphics[width=3cm]{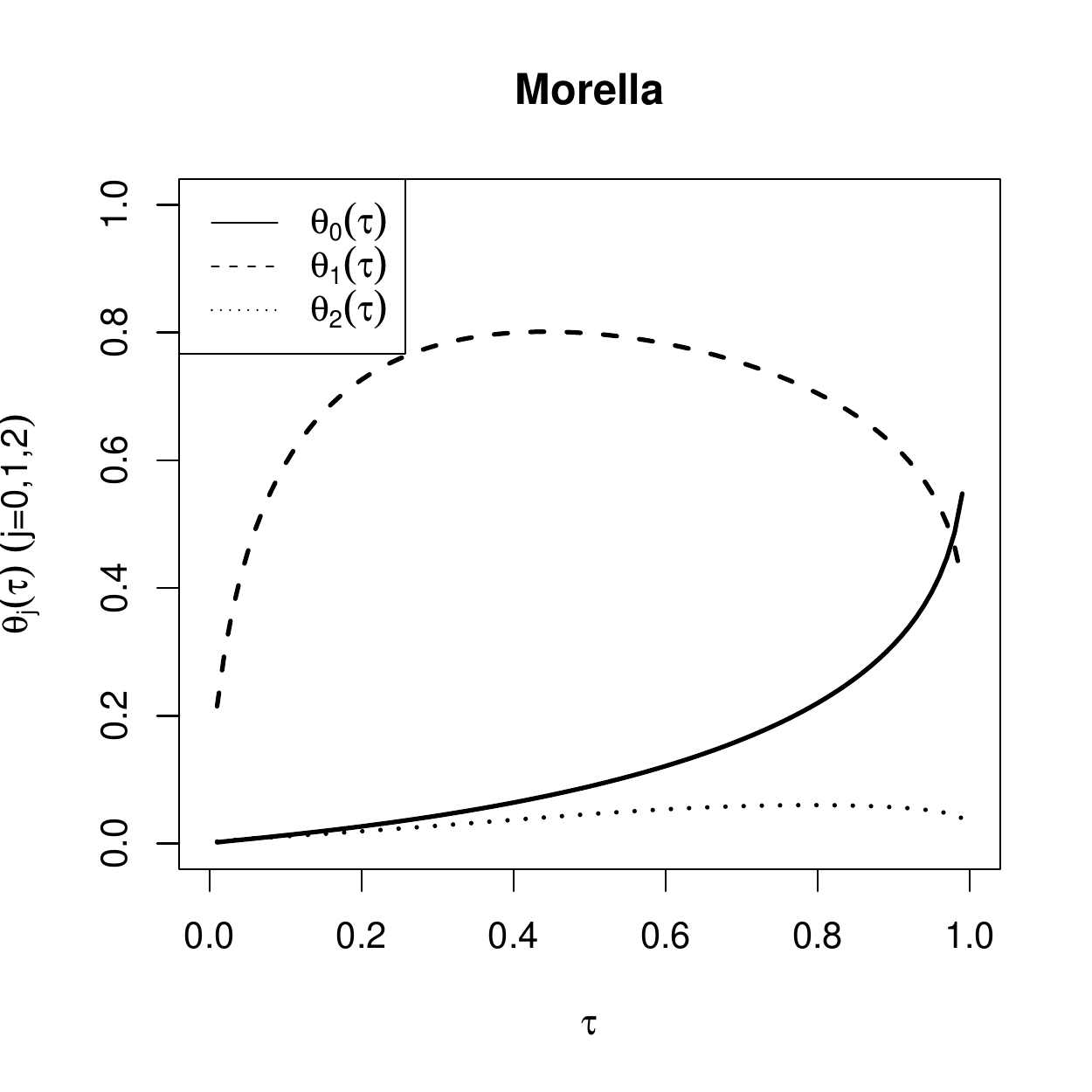} \\
\includegraphics[width=3cm]{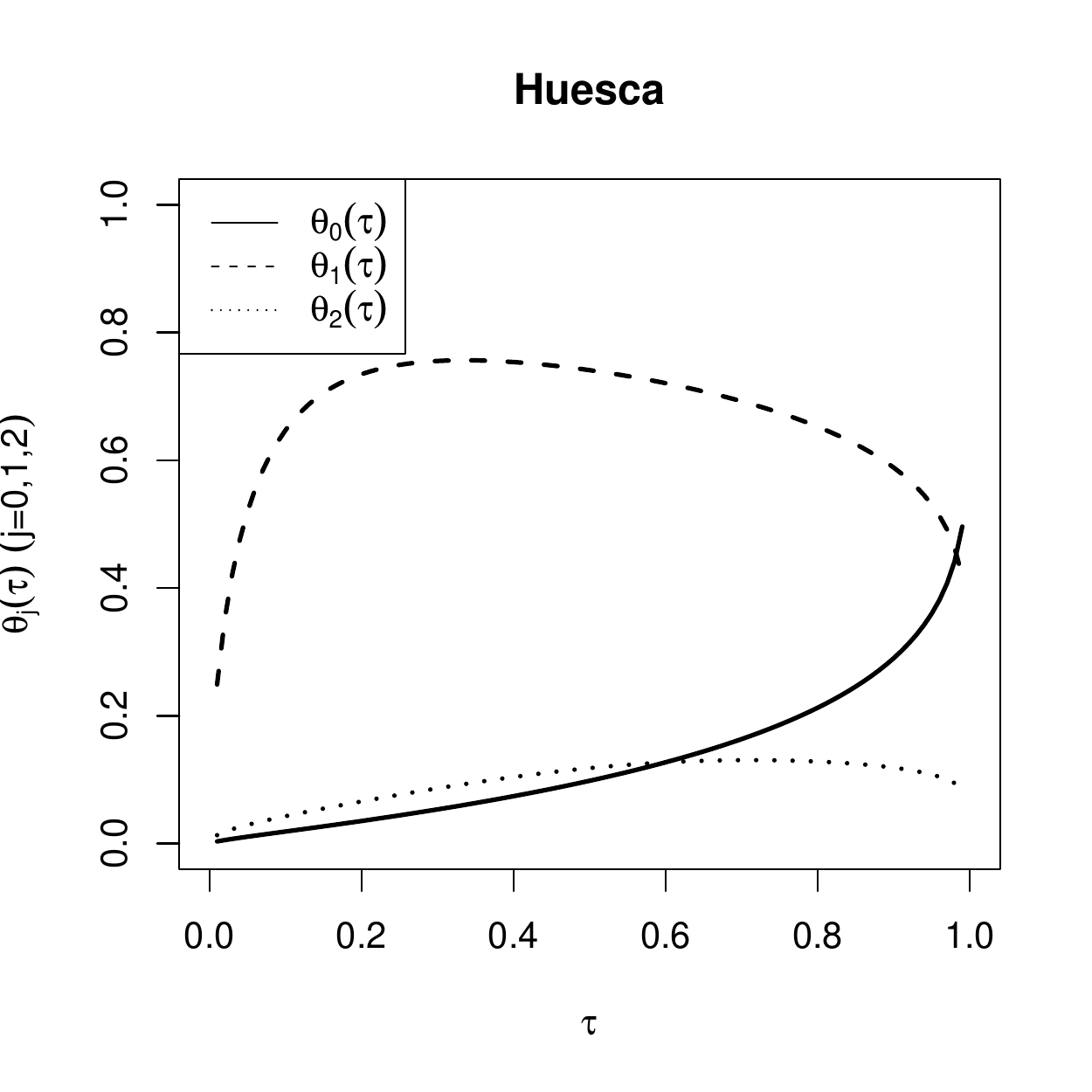}
\includegraphics[width=3cm]{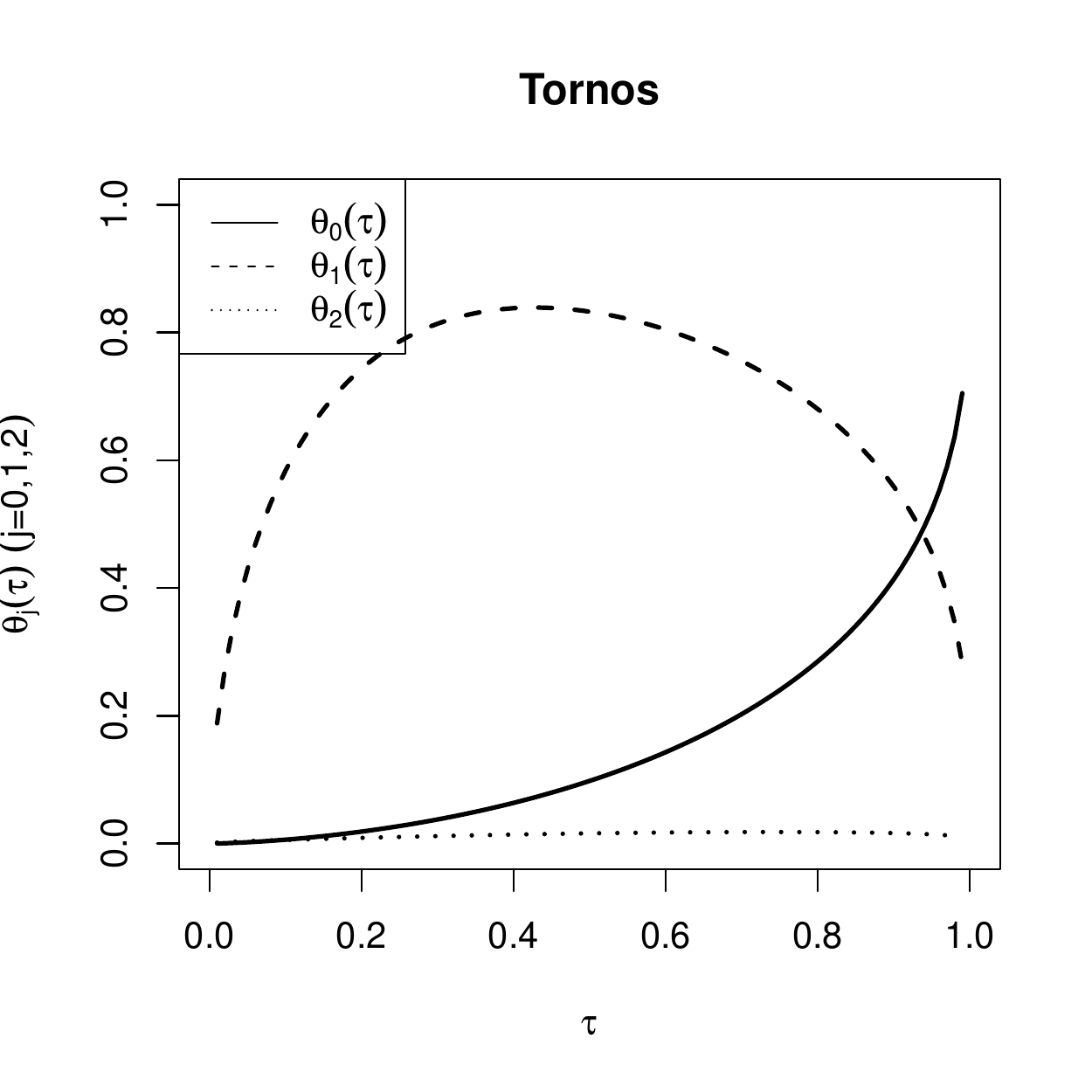}
\includegraphics[width=3cm]{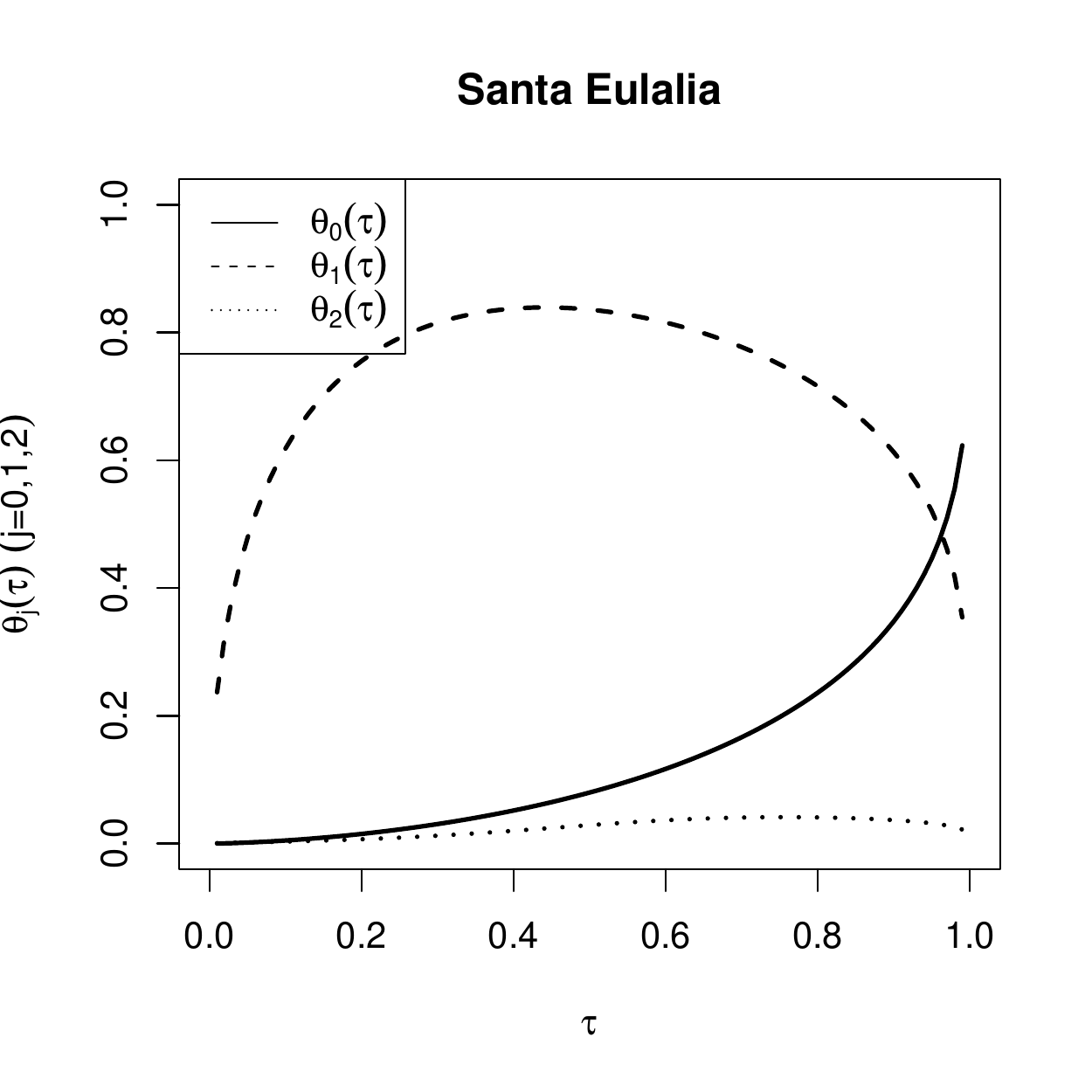}
\includegraphics[width=3cm]{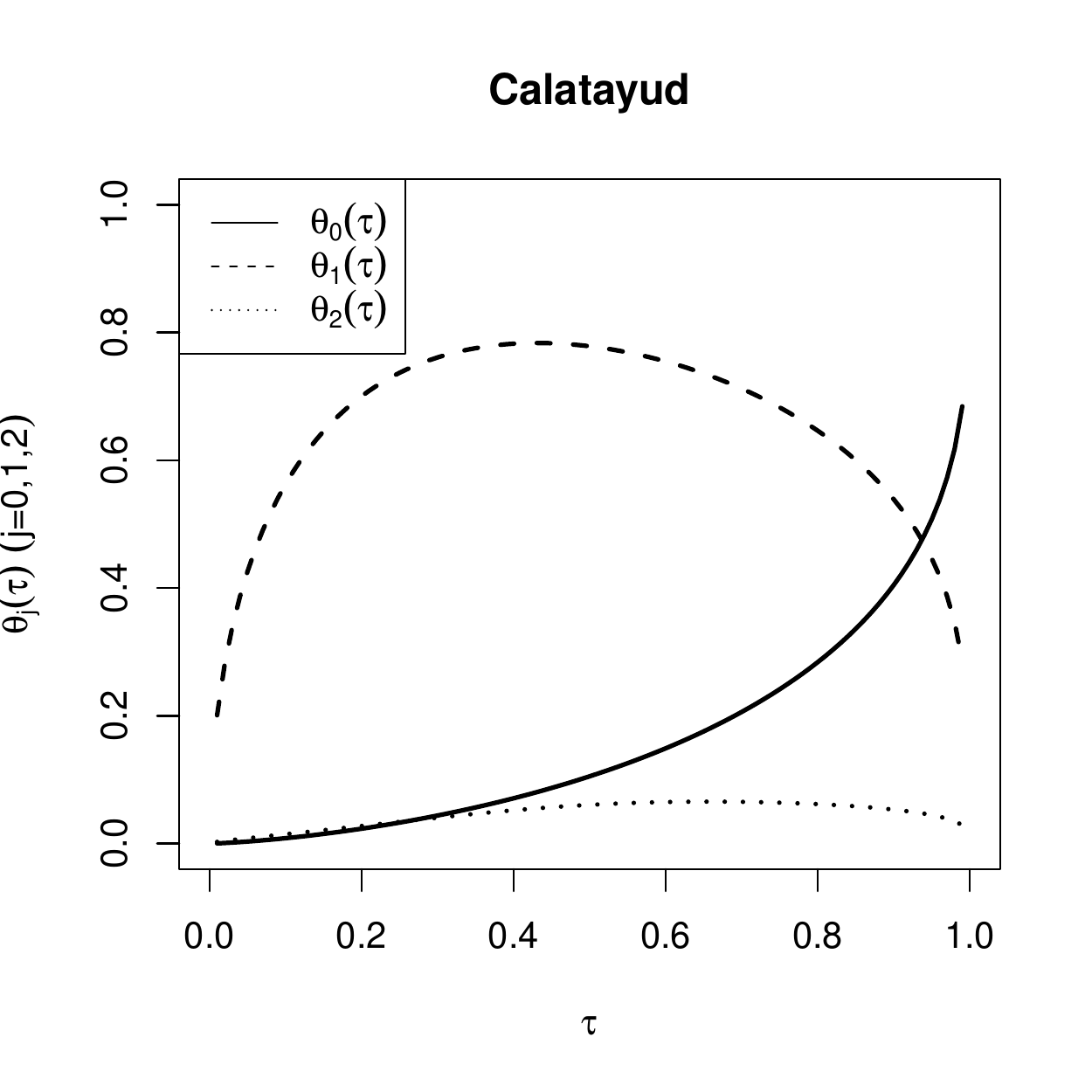} \\
\includegraphics[width=3cm]{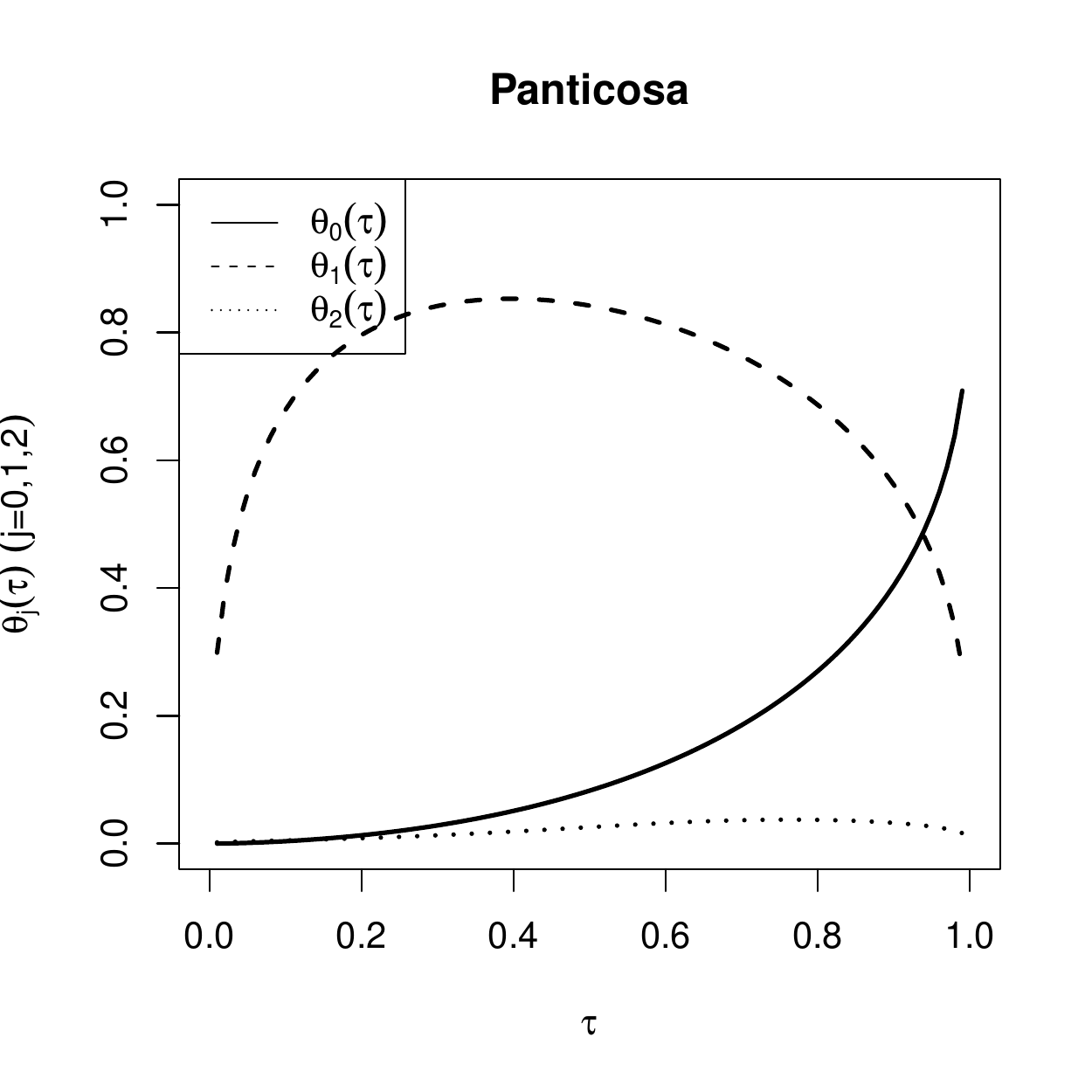}
\includegraphics[width=3cm]{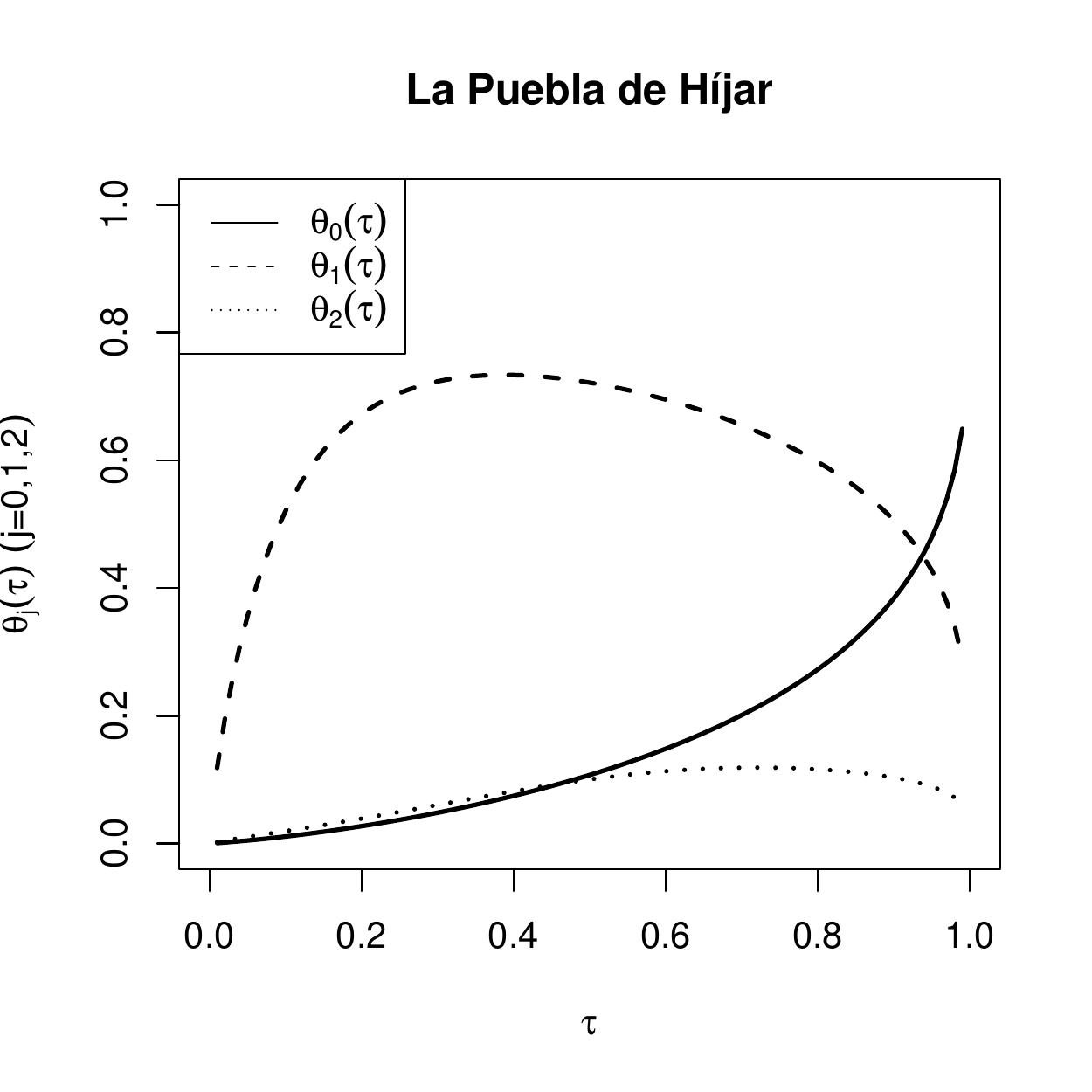}
\includegraphics[width=3cm]{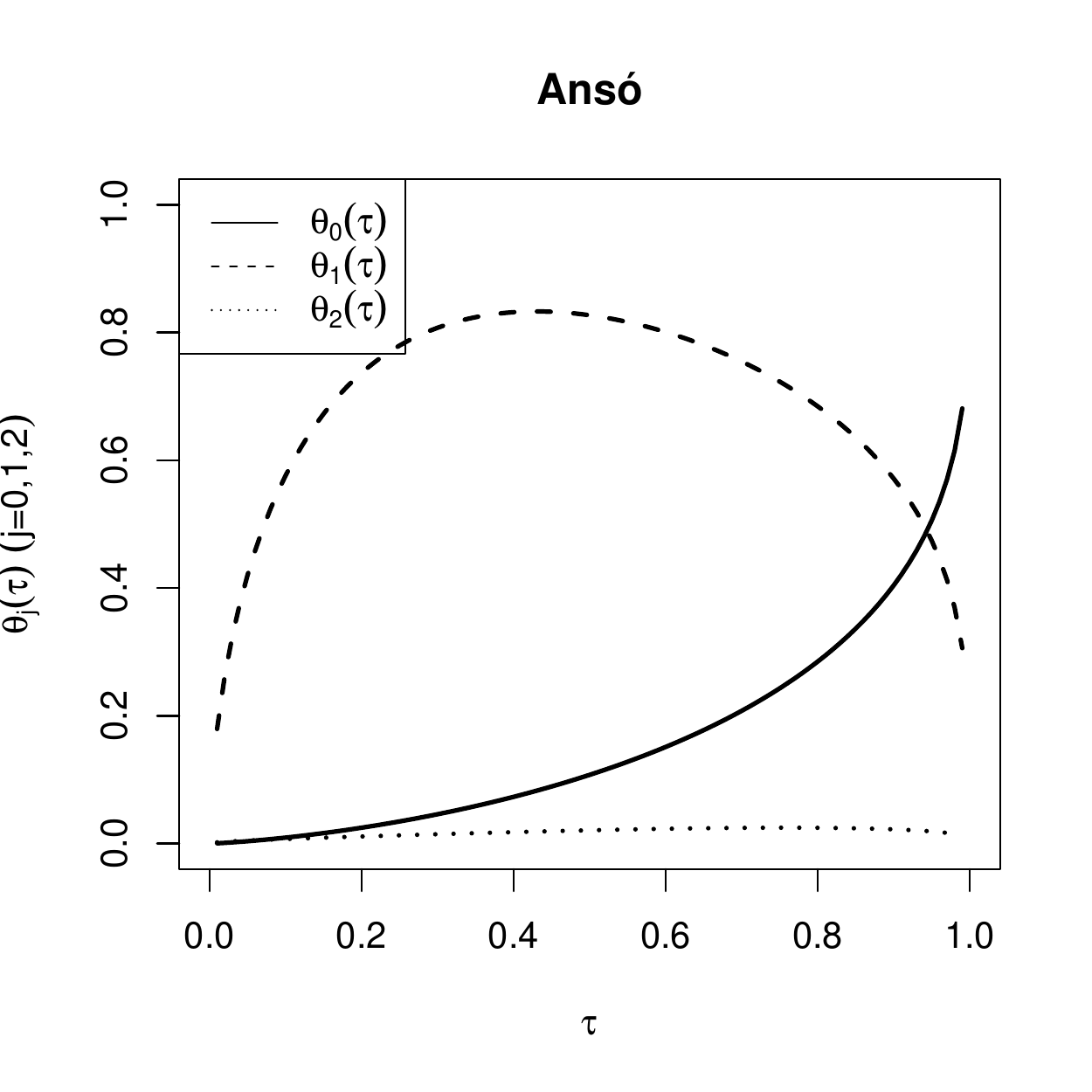}
\includegraphics[width=3cm]{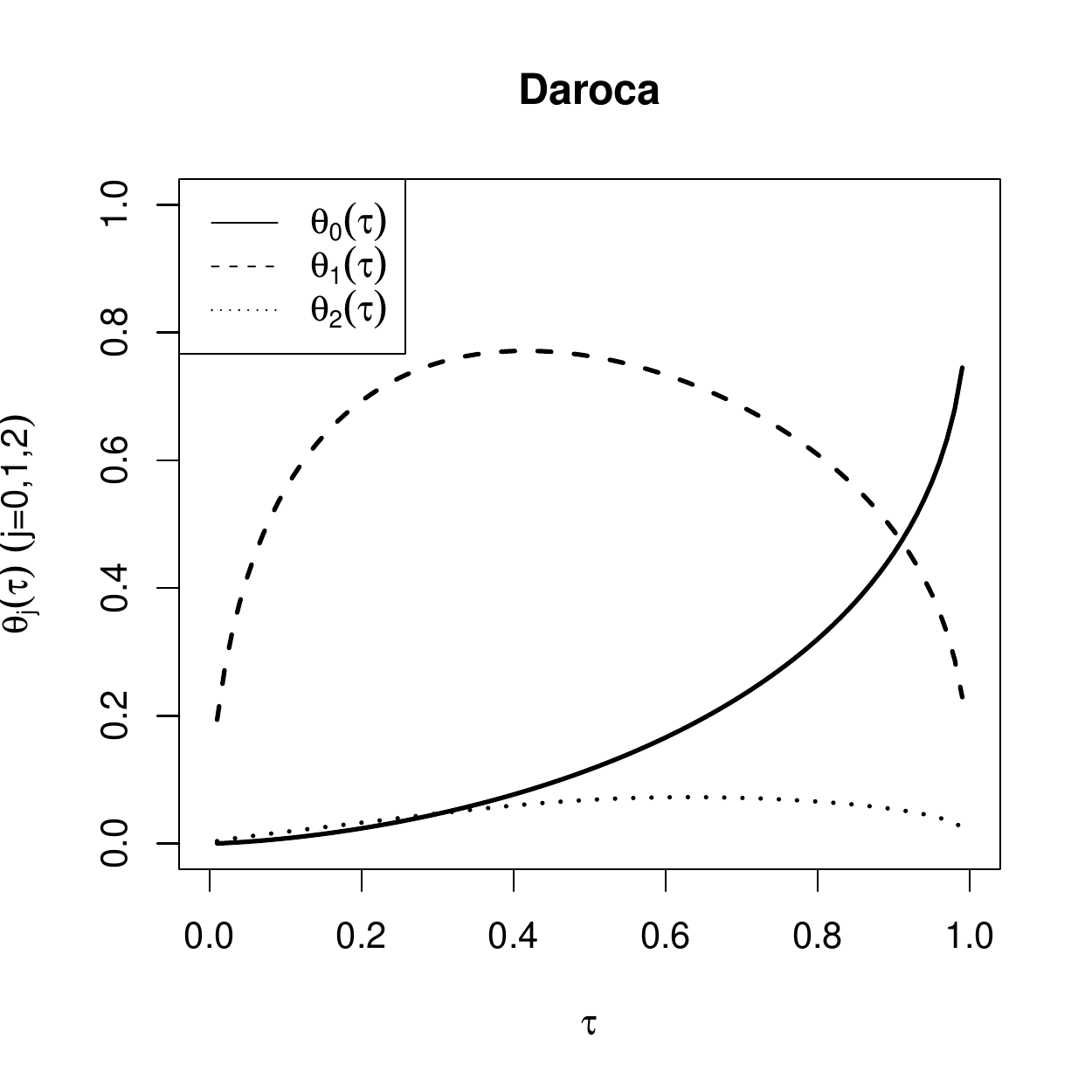} \\
\includegraphics[width=3cm]{thetaQAR2K1s13.pdf}
\includegraphics[width=3cm]{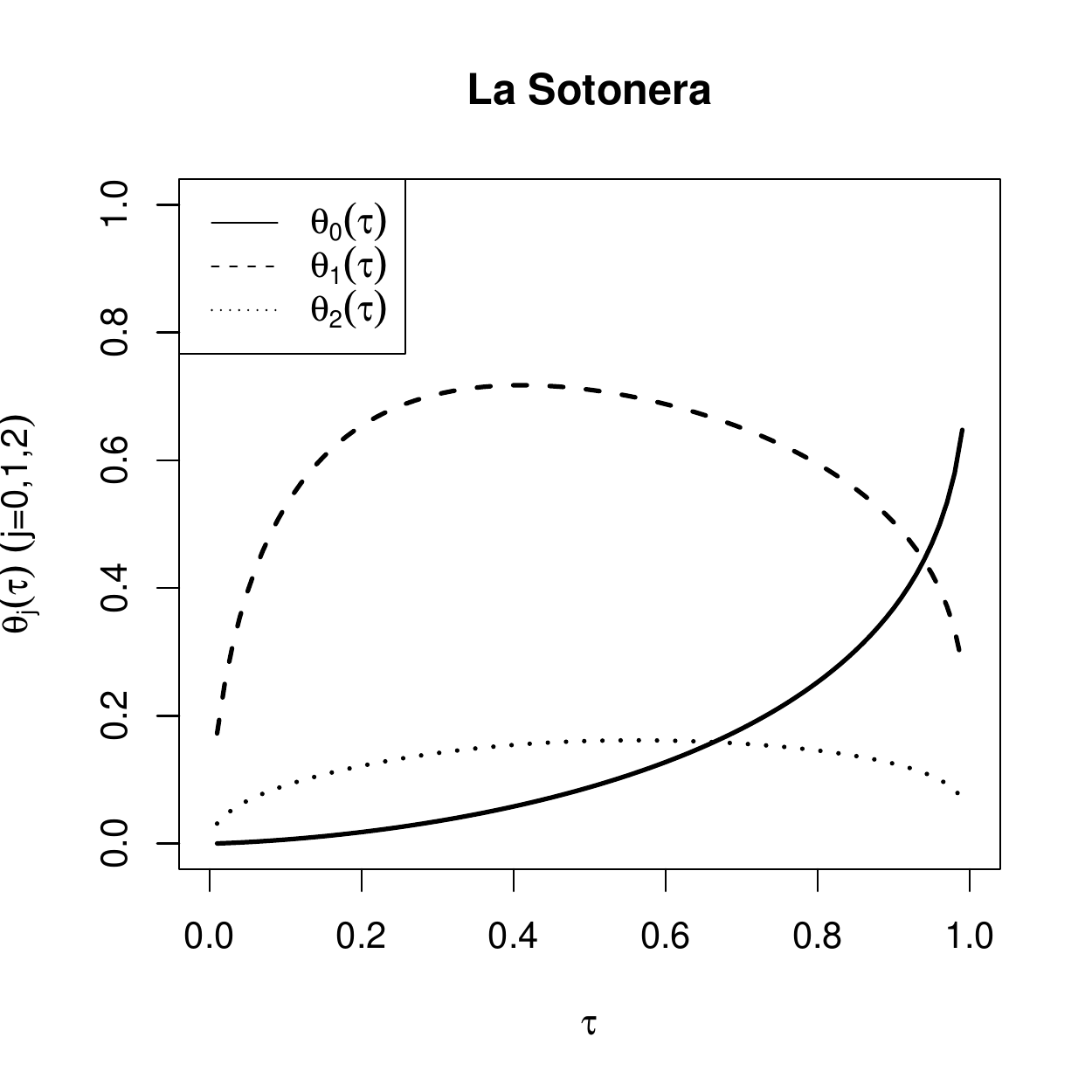}
\includegraphics[width=3cm]{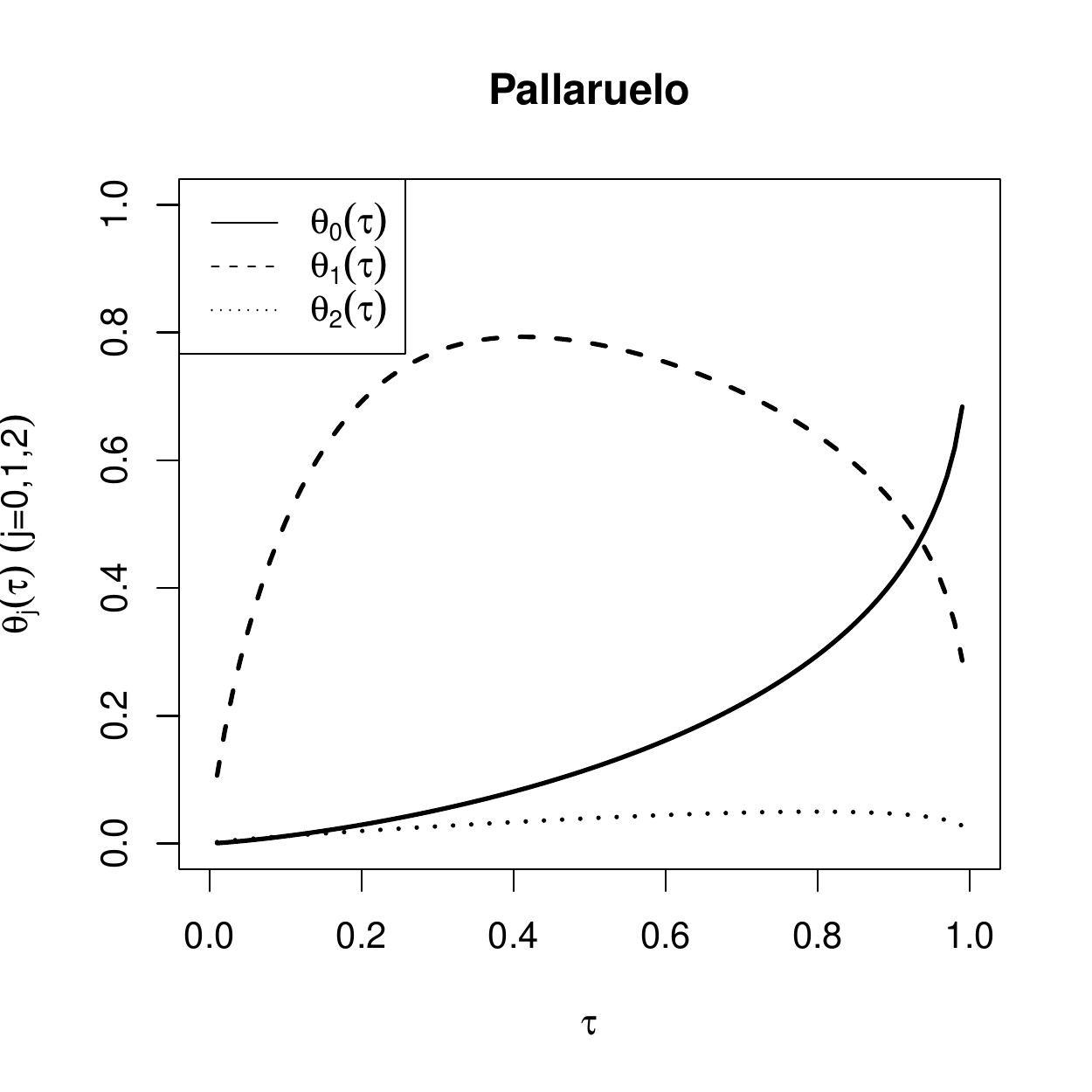}
\includegraphics[width=3cm]{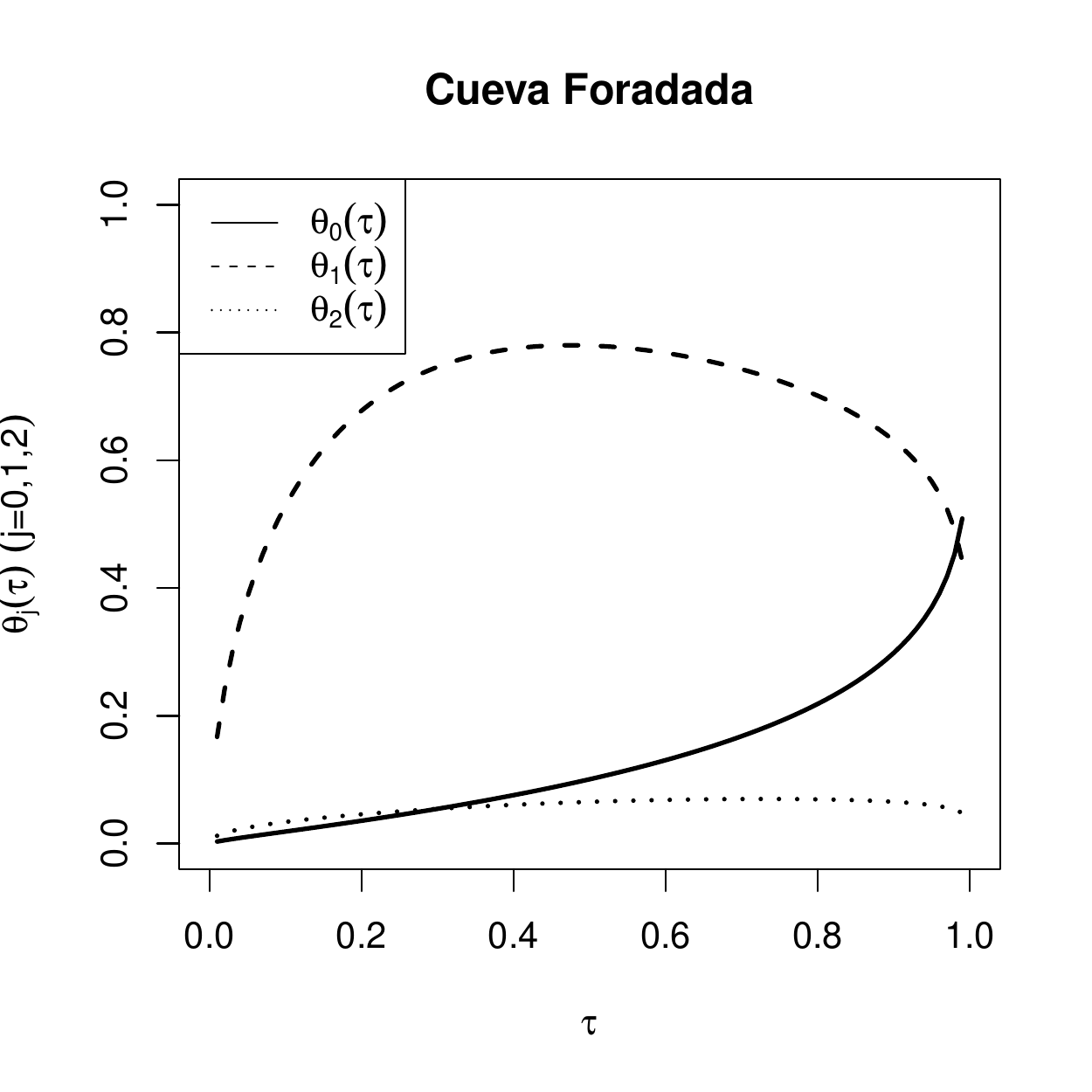} \\
\includegraphics[width=3cm]{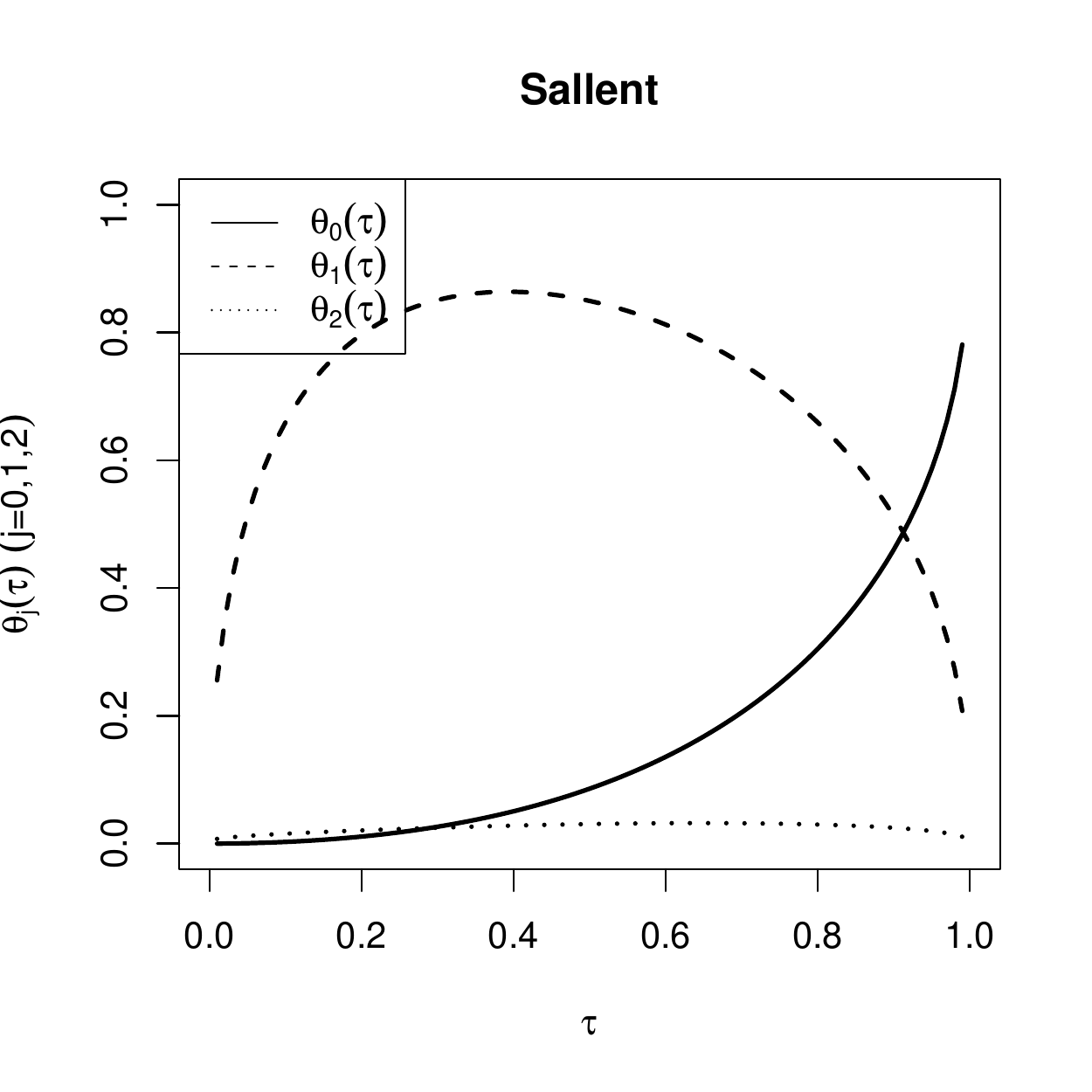}
\includegraphics[width=3cm]{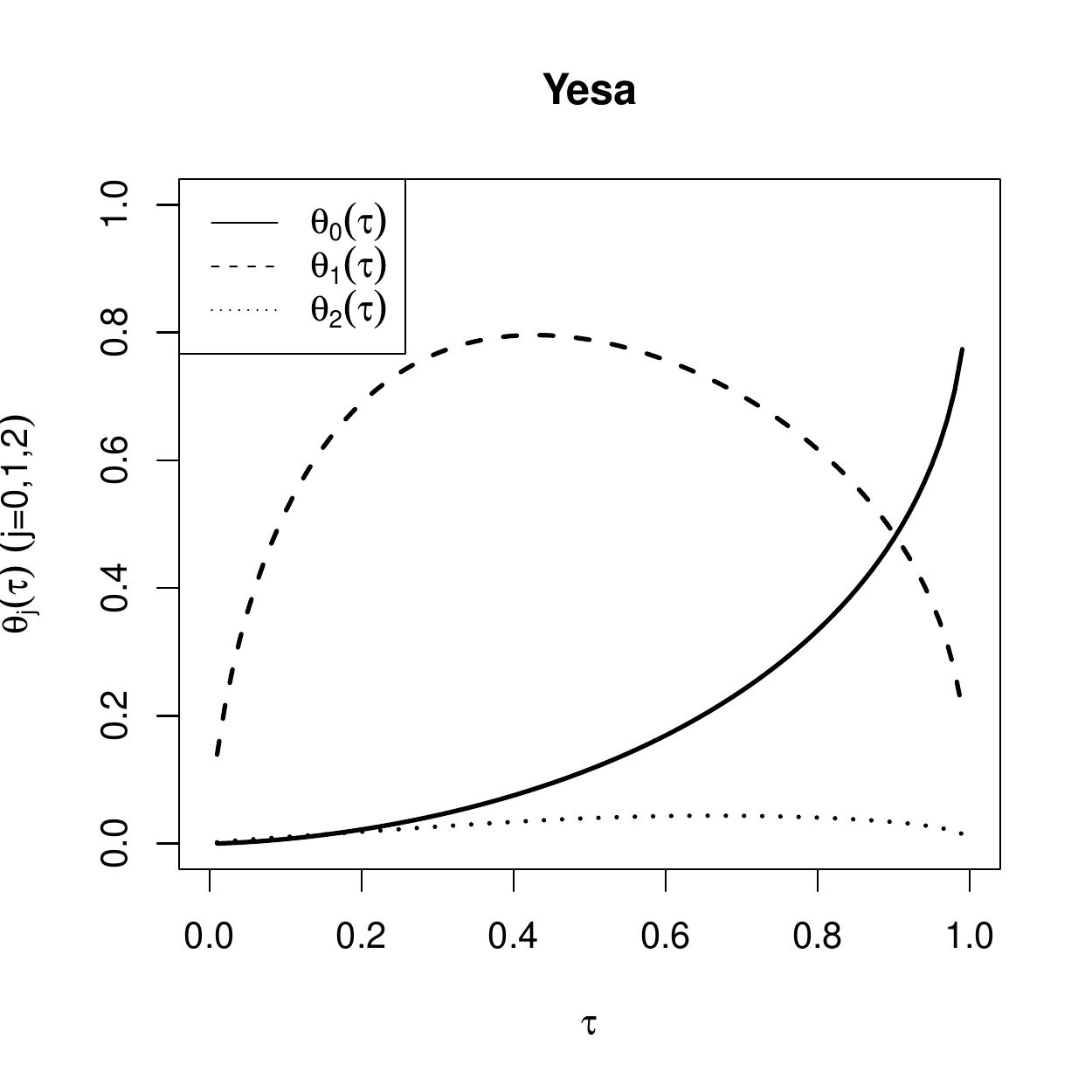}
\caption{Posterior mean of $\theta_0(\tau)$ (solid), $\theta_1(\tau)$ (dashed) and $\theta_2(\tau)$ (dotted) vs. $\tau$ for QAR2K1. All locations, MJJAS, 2015}
\end{figure}

\clearpage

\subsection{Multivariate QAR(1)}

\begin{figure}[!ht]
\centering
\includegraphics[width=5cm]{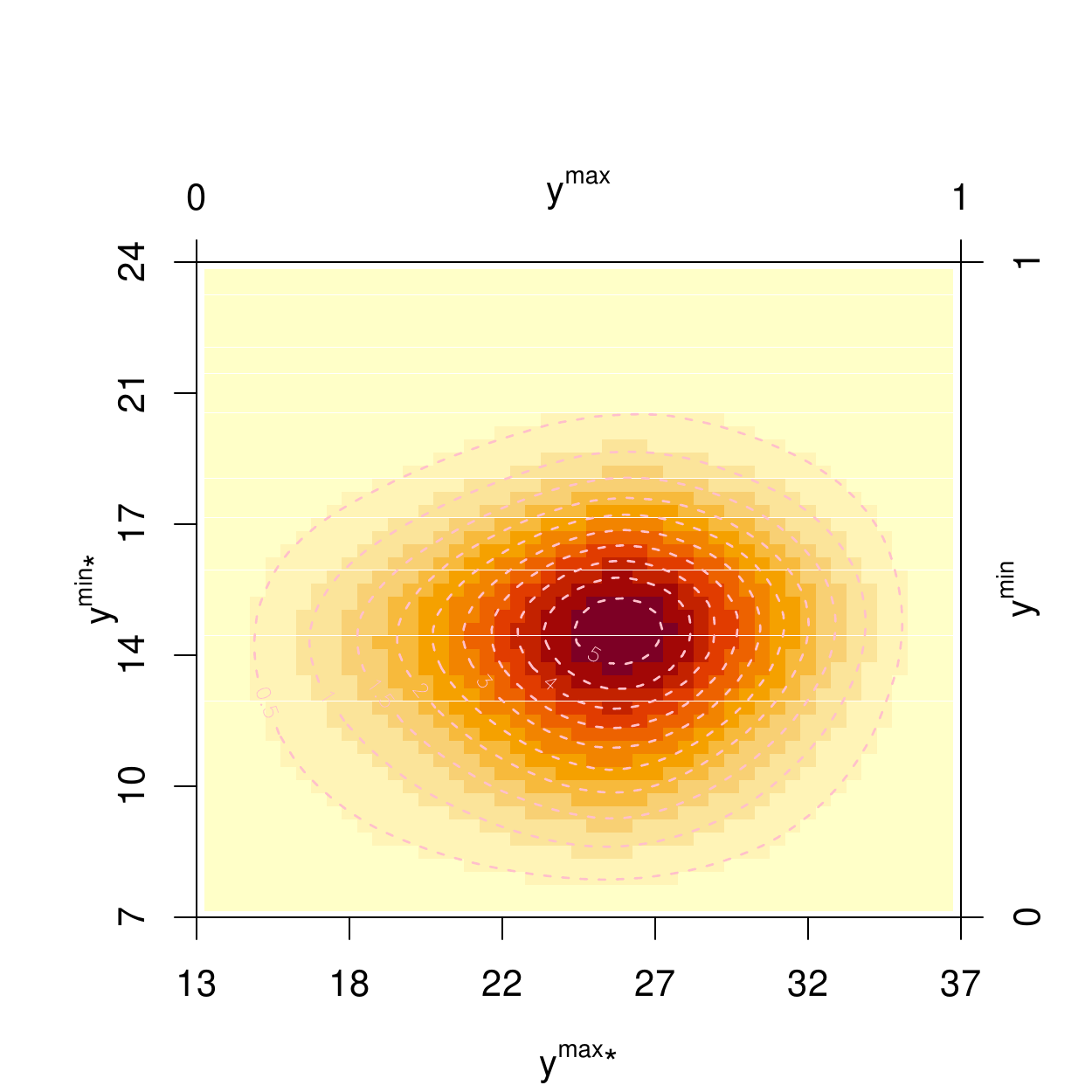}
\includegraphics[width=5cm]{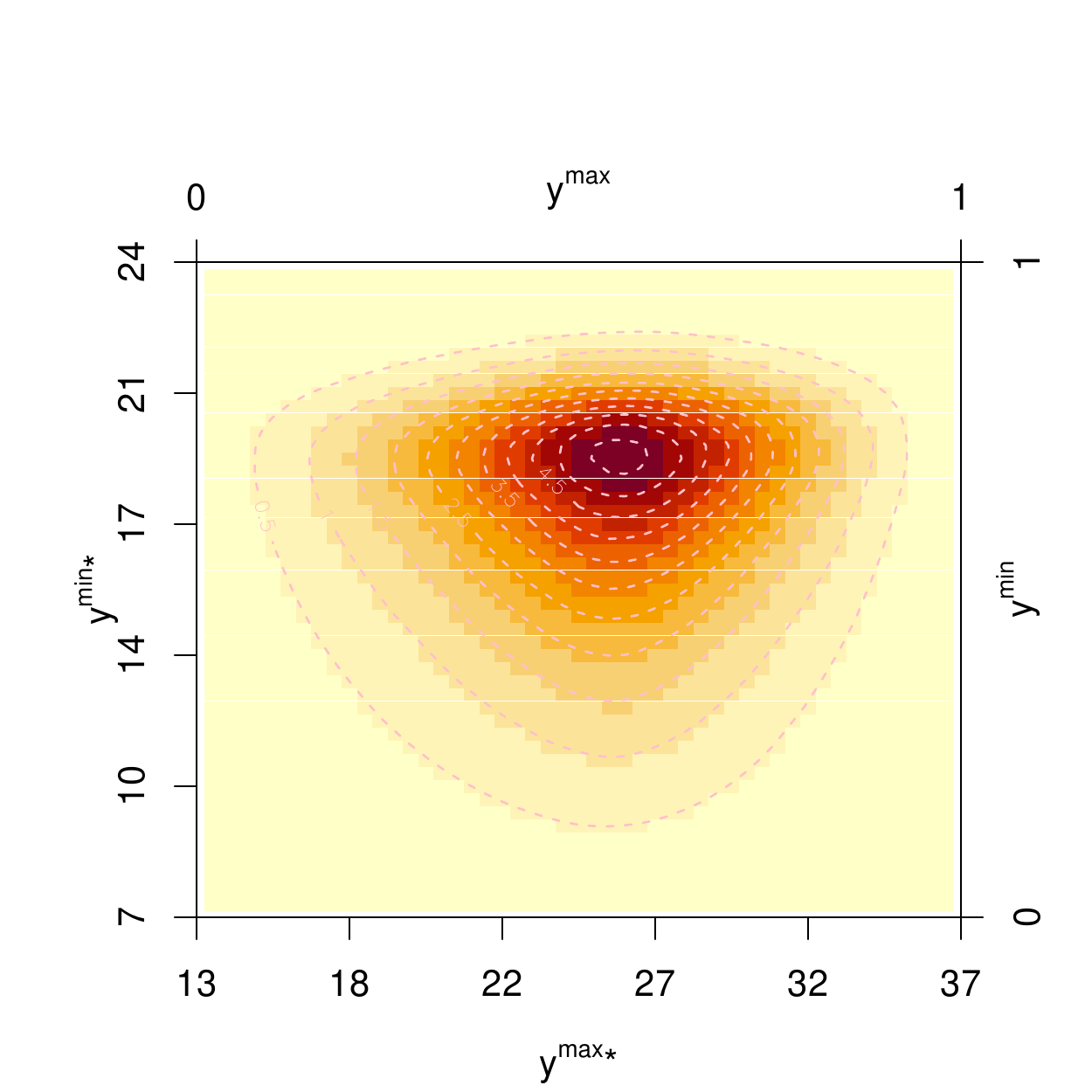} \\
\includegraphics[width=5cm]{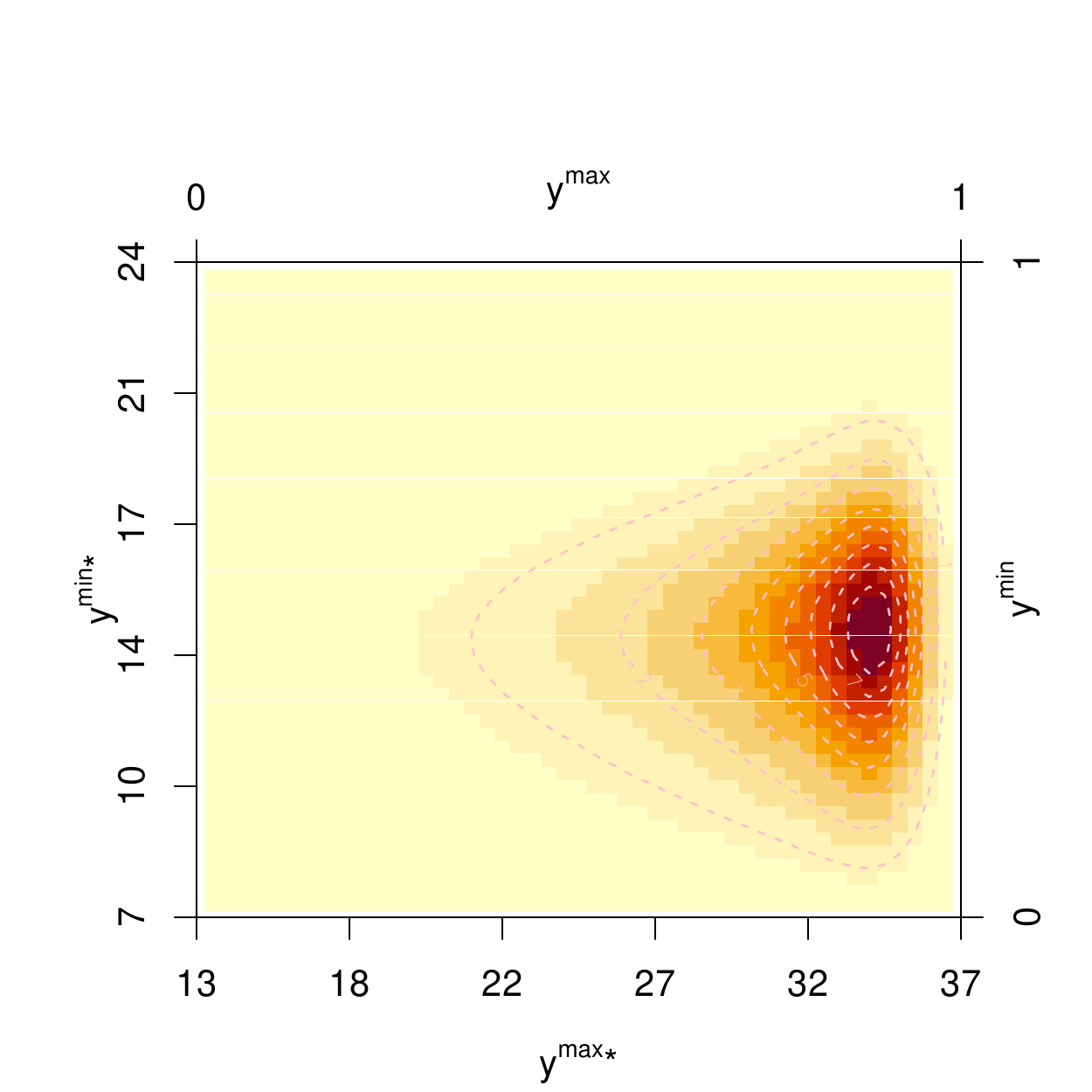}
\includegraphics[width=5cm]{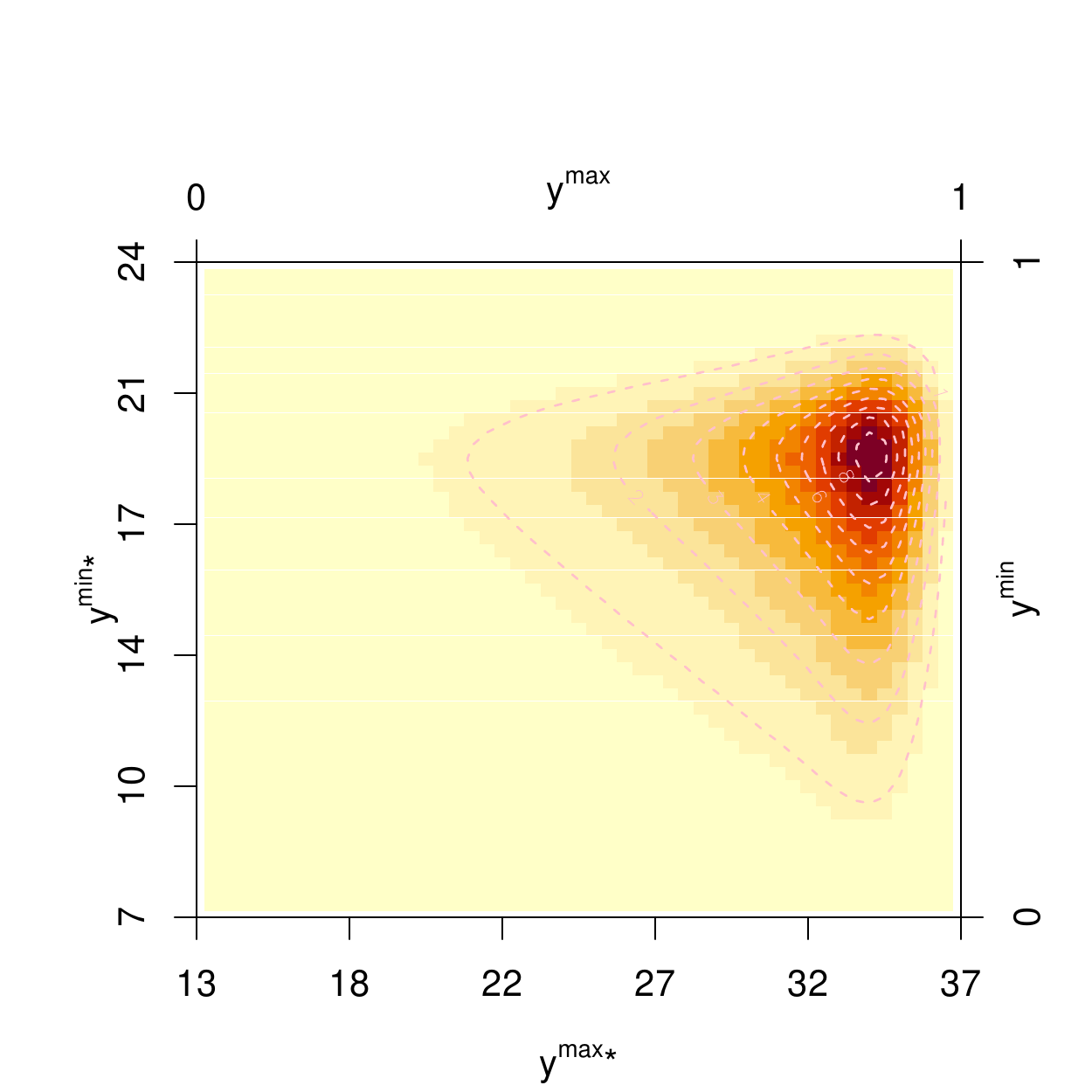}
\caption{Posterior mean of the density function of $(Y_{t}^{\text{max}}, Y_{t}^{\text{min}})$ conditioned on $(y^{\text{max}}, y^{\text{min}})$ for MQAR1K1. Here, $(y^{\text{max}}, y^{\text{min}})$ is equal to the respective empirical marginal quantiles for $\tau = 0.5$ (above), $0.9$ (below) for the maximum; and $\tau = 0.5$ (left), $0.9$ (right) for the minimum. Pamplona, MJJAS, 2015. }
\label{fig:fitted:density:Tx.Tn:Pamplona:MJJAS:2015}
\end{figure}

\begin{figure}[!ht]
\centering
\includegraphics[width=5cm]{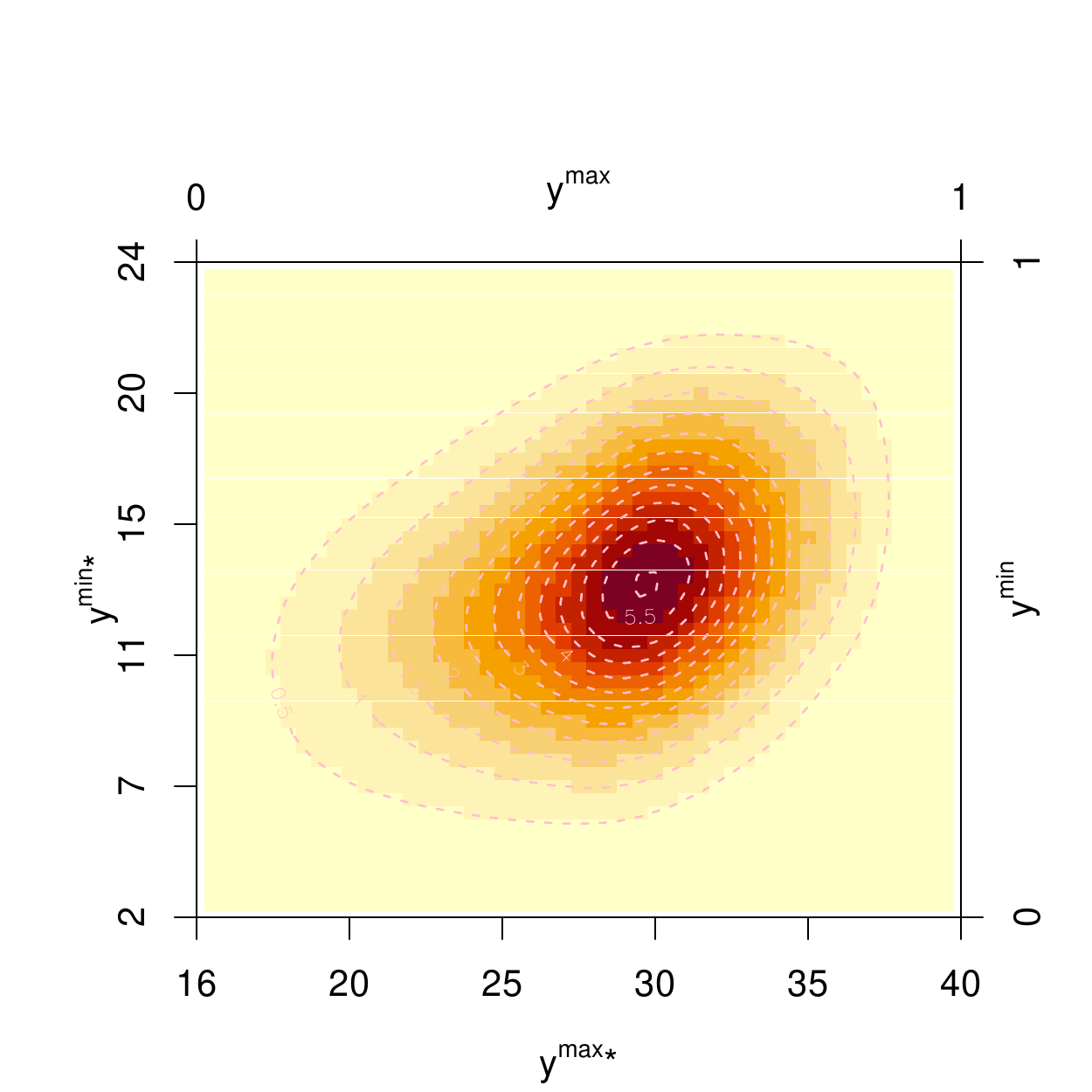}
\includegraphics[width=5cm]{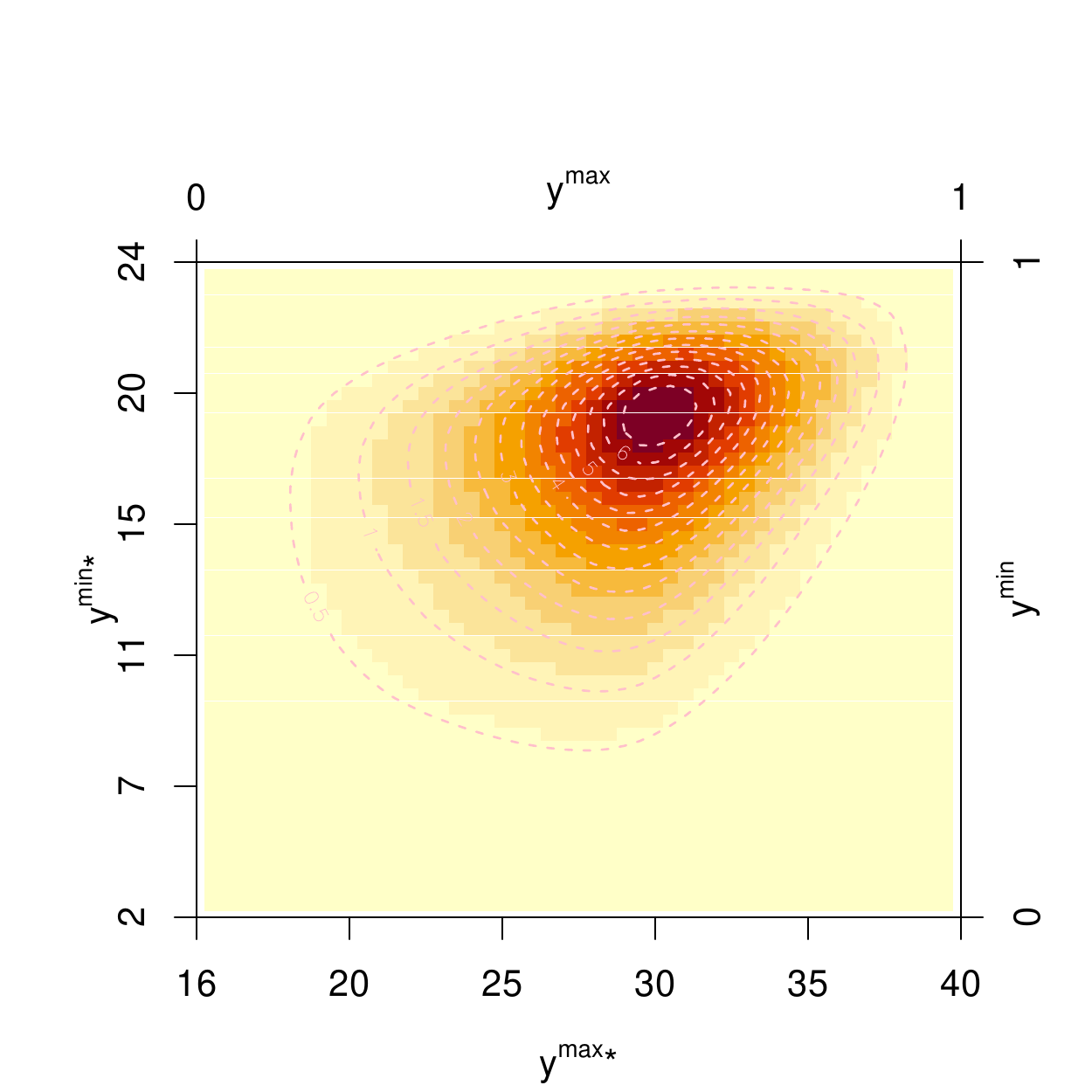} \\
\includegraphics[width=5cm]{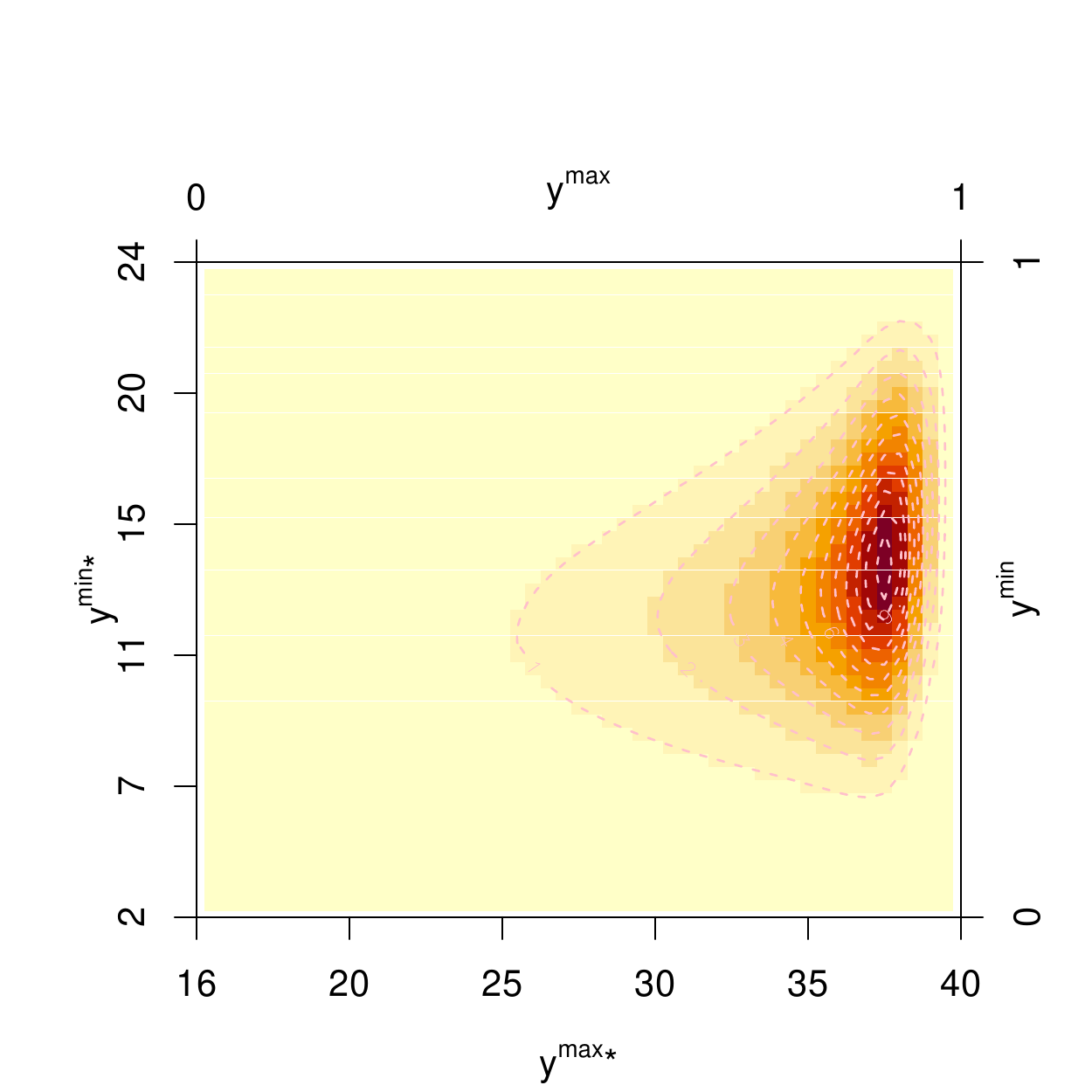}
\includegraphics[width=5cm]{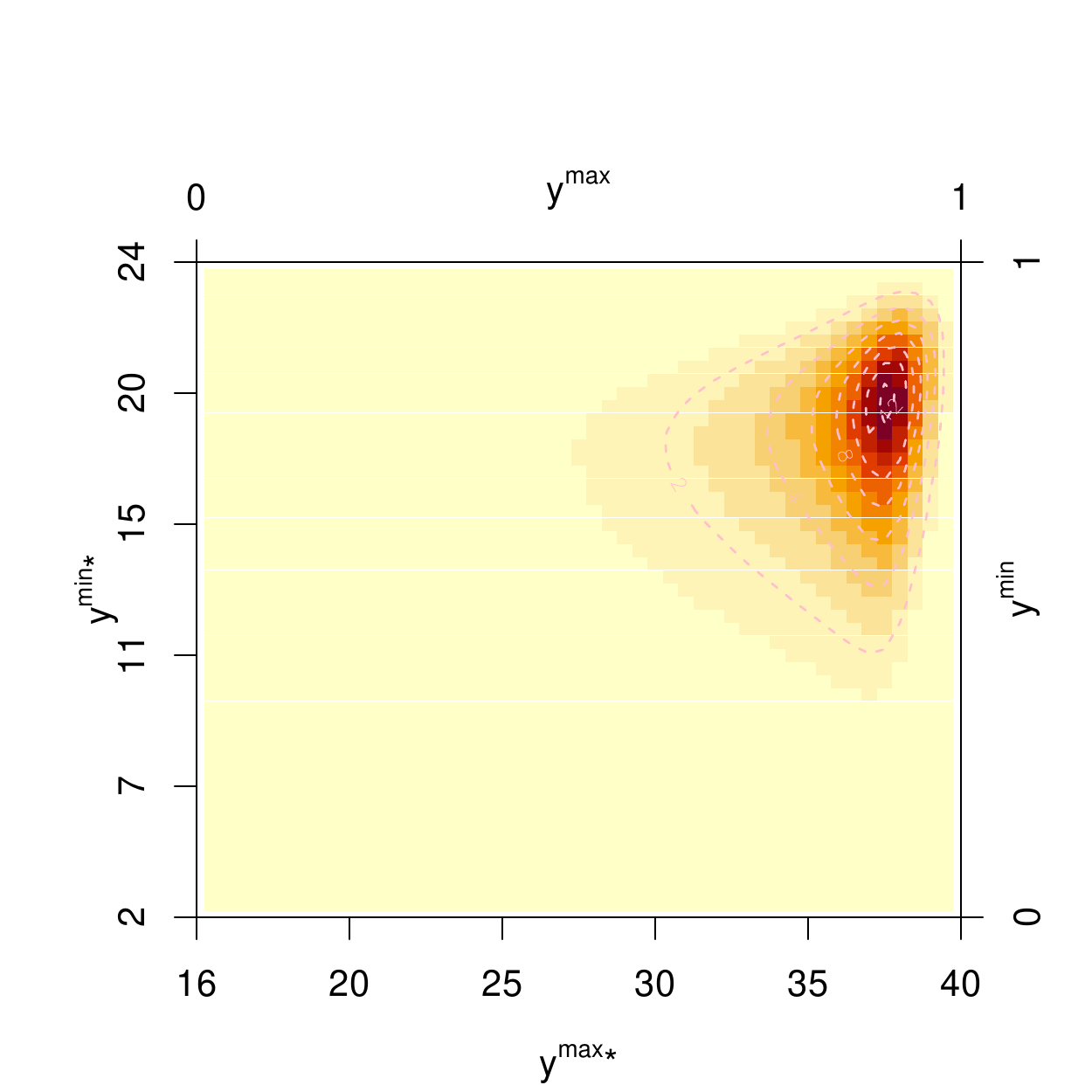}
\caption{Posterior mean of the density function of $(Y_{t}^{\text{max}}, Y_{t}^{\text{min}})$ conditioned on $(y^{\text{max}}, y^{\text{min}})$ for MQAR1K1. Here, $(y^{\text{max}}, y^{\text{min}})$ is equal to the respective empirical marginal quantiles for $\tau = 0.5$ (above), $0.9$ (below) for the maximum; and $\tau = 0.5$ (left), $0.9$ (right) for the minimum. Daroca, MJJAS, 2015. }
\label{fig:fitted:density:Tx.Tn:Daroca:MJJAS:2015}
\end{figure}

\clearpage

\subsection{Spatial QAR(1)}

\begin{figure}[!ht]
\centering
\includegraphics[width=5cm]{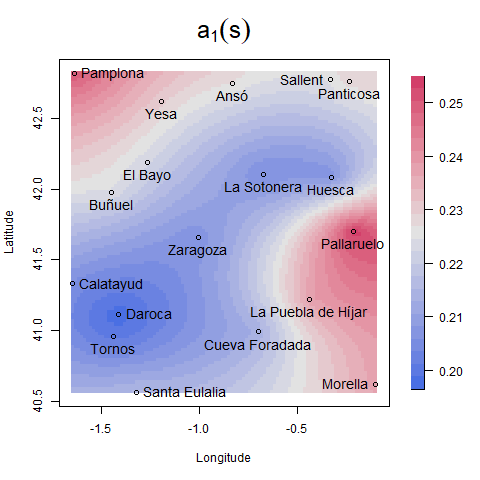}
\includegraphics[width=5cm]{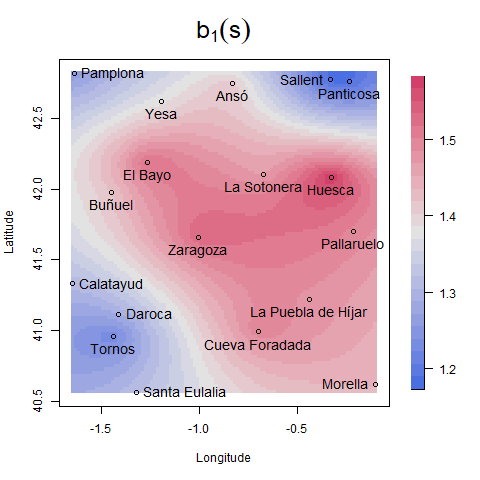} \\
\includegraphics[width=5cm]{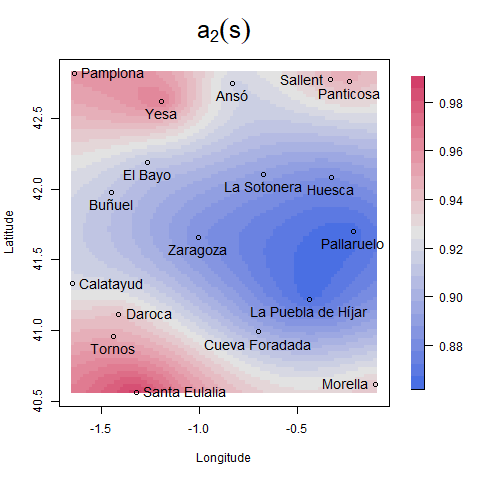}
\includegraphics[width=5cm]{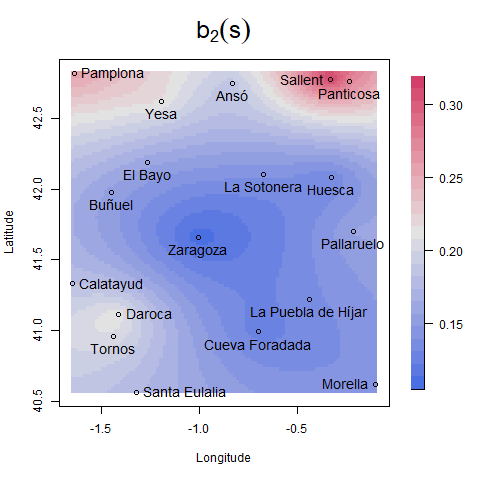}
\caption{Maps of the posterior mean of $a_1(\bs)$, $b_1(\bs)$, $a_2(\bs)$, and $b_2(\bs)$. } \label{fig:meanGP}
\end{figure}

\begin{figure}[!ht]
\centering
\includegraphics[width=5cm]{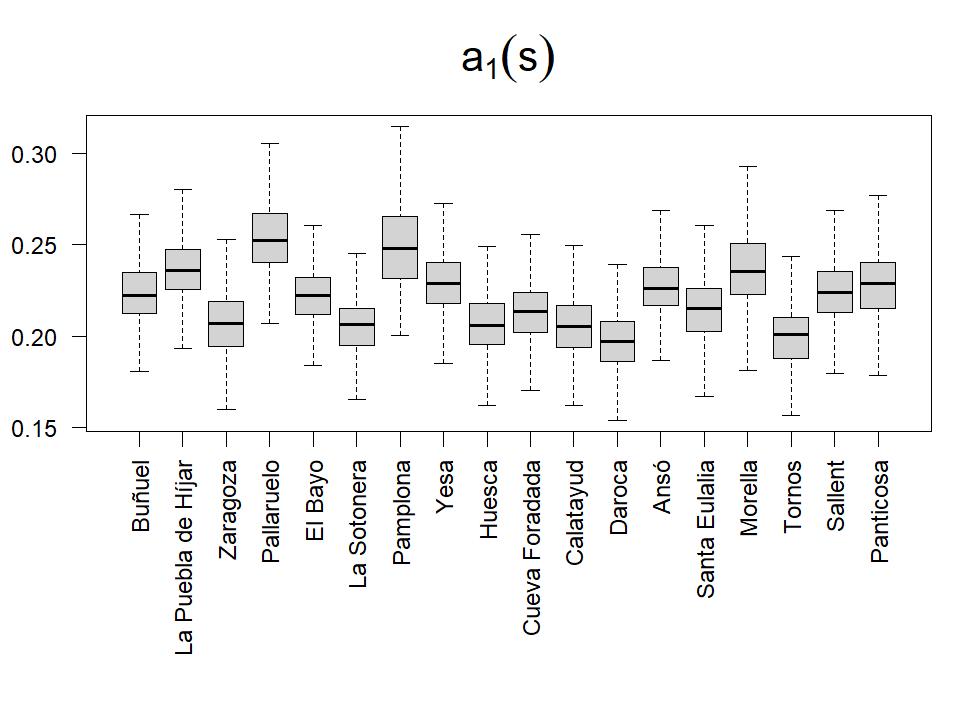}
\includegraphics[width=5cm]{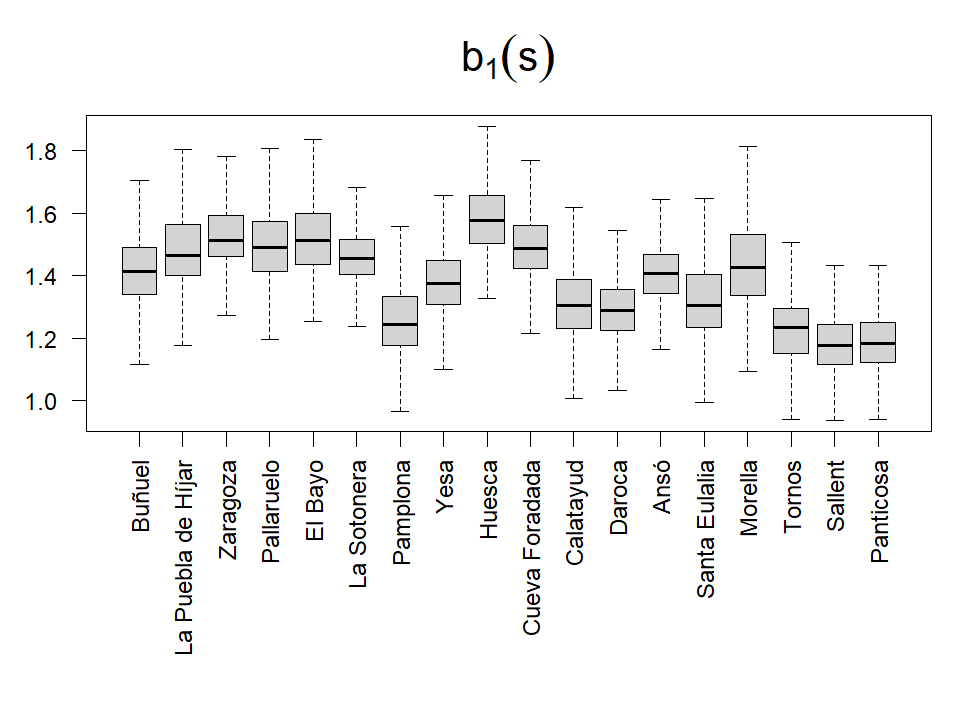} \\
\includegraphics[width=5cm]{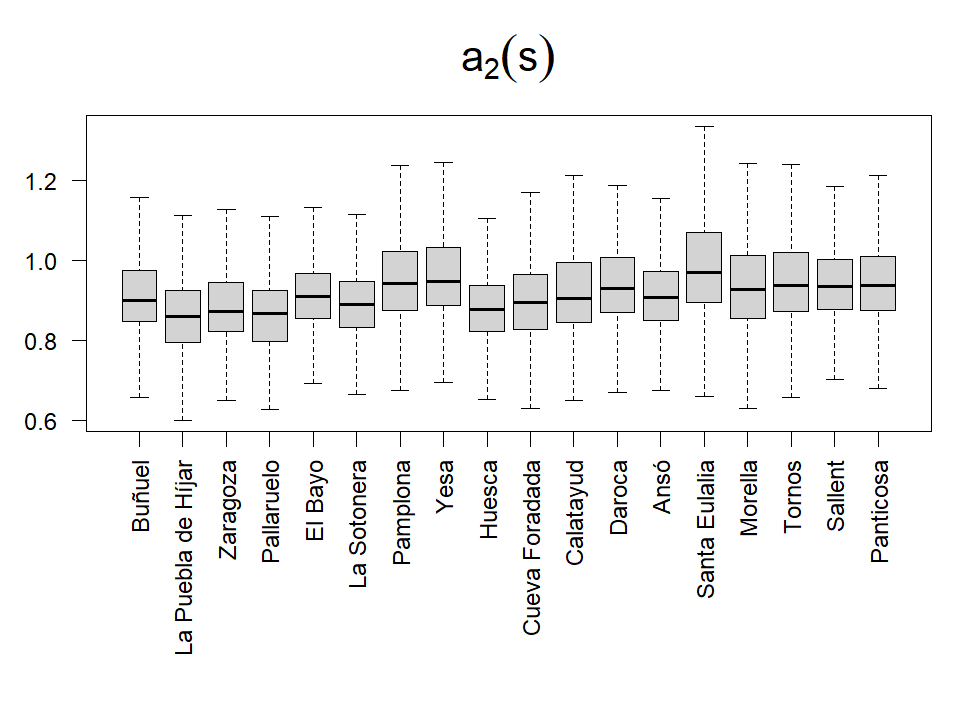}
\includegraphics[width=5cm]{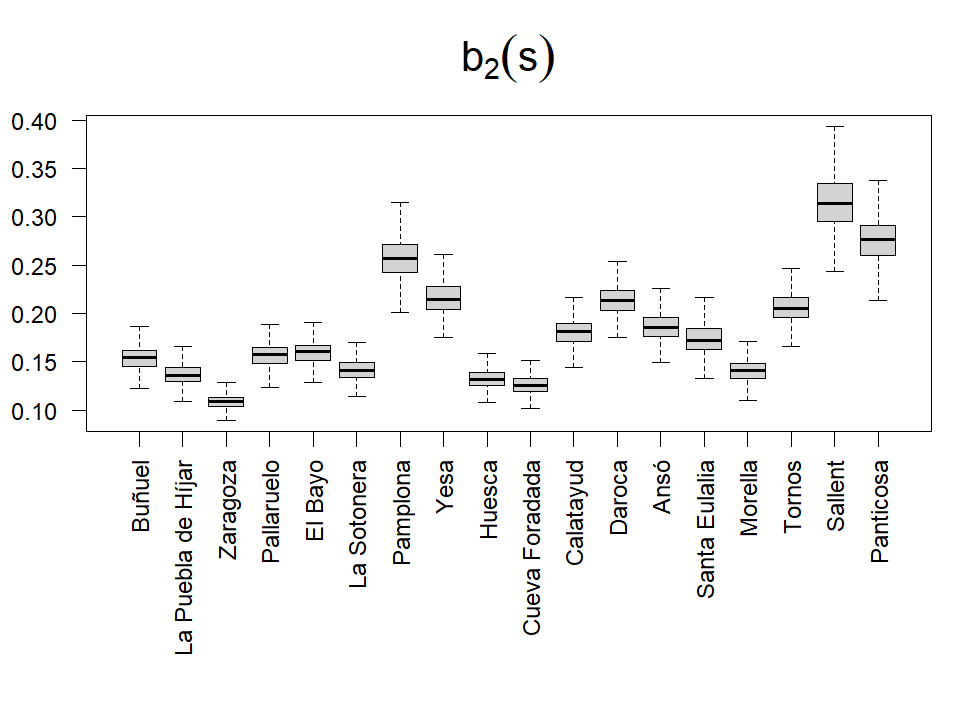}
\caption{Boxplots summarizing the posterior distribution of $a_1(\bs)$, $b_1(\bs)$, $a_2(\bs)$, and $b_2(\bs)$ at the observed locations sorted by elevation. } \label{fig:boxplotGP}
\end{figure}

\begin{figure}[!ht]
\centering
\includegraphics[width=5cm]{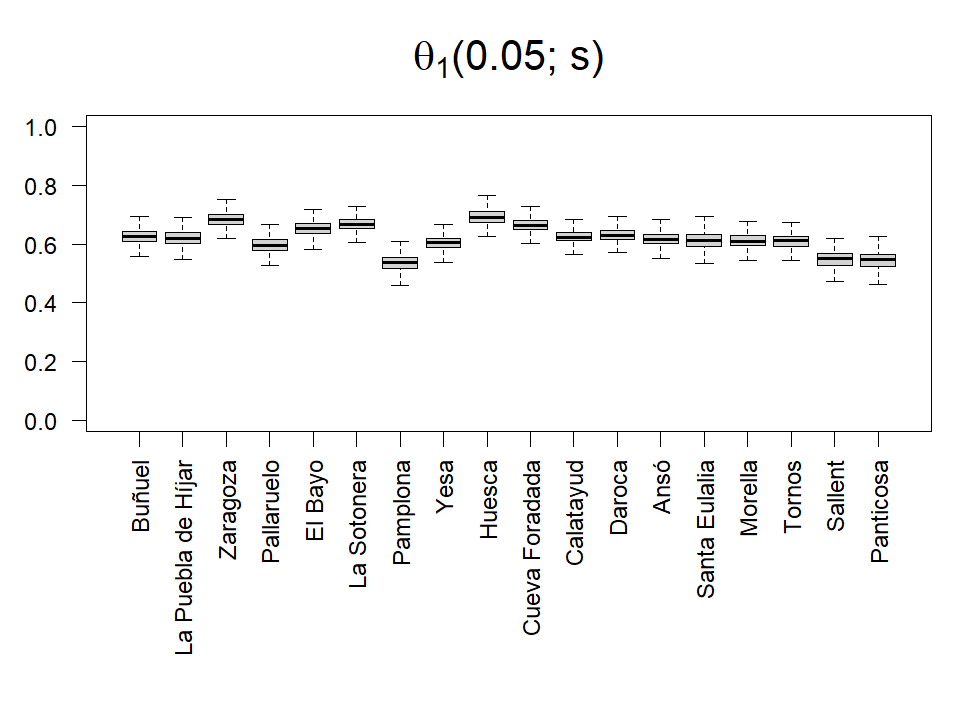}
\includegraphics[width=5cm]{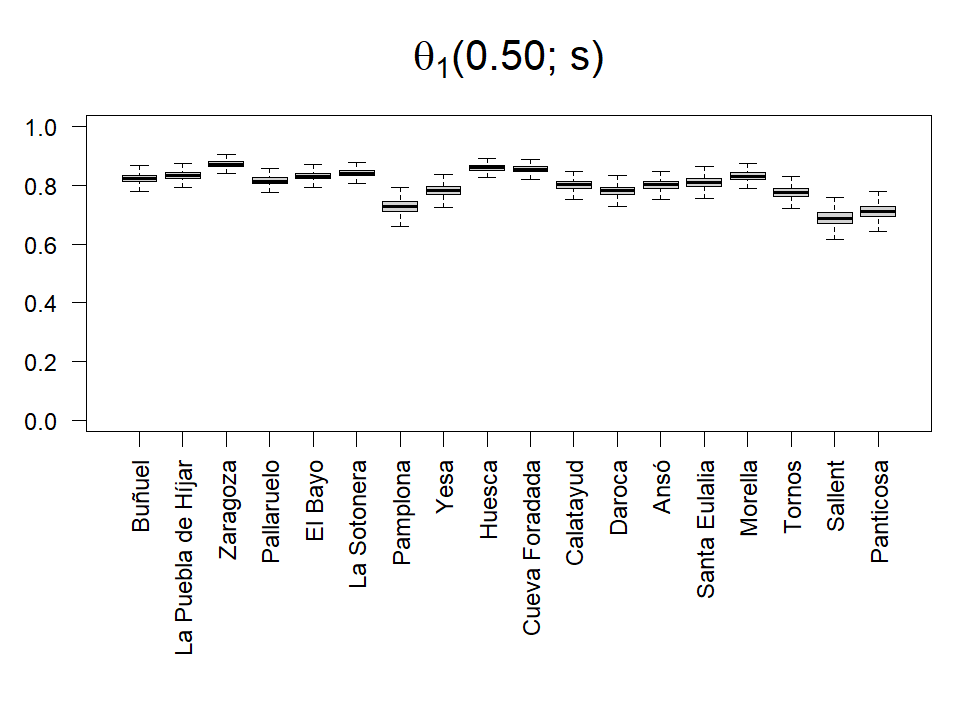}
\includegraphics[width=5cm]{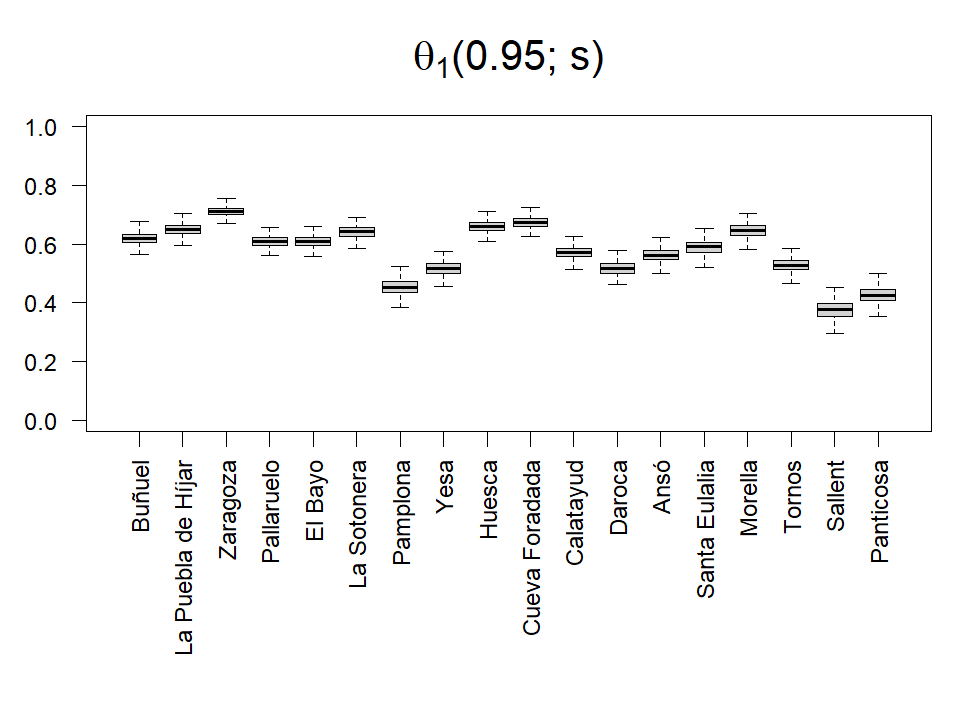}
\caption{Boxplots summarizing the posterior distribution of $\theta_1(\tau;\bs)$ for $\tau = 0.05,0.50,0.95$ at the observed locations sorted by elevation. } \label{fig:boxplottheta1}
\end{figure}

\begin{figure}[!ht]
\centering
\includegraphics[width=4cm]{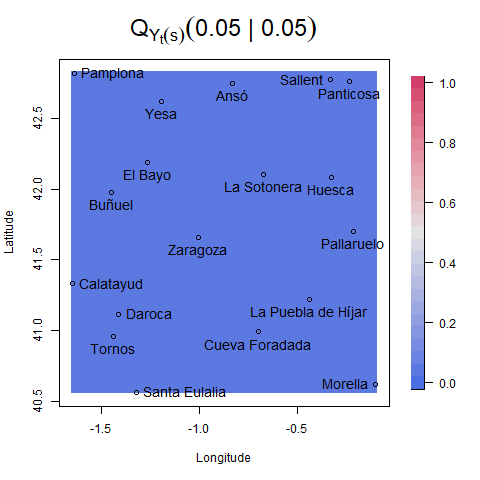}
\includegraphics[width=4cm]{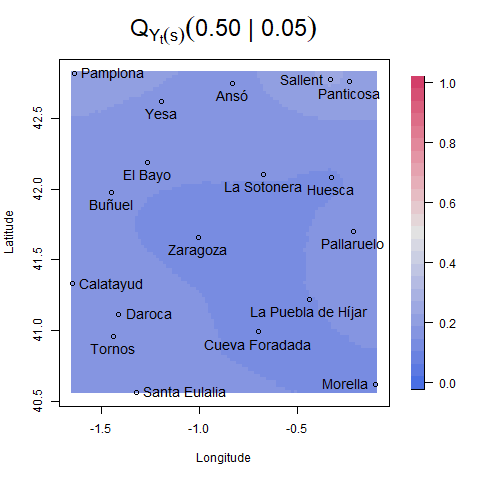}
\includegraphics[width=4cm]{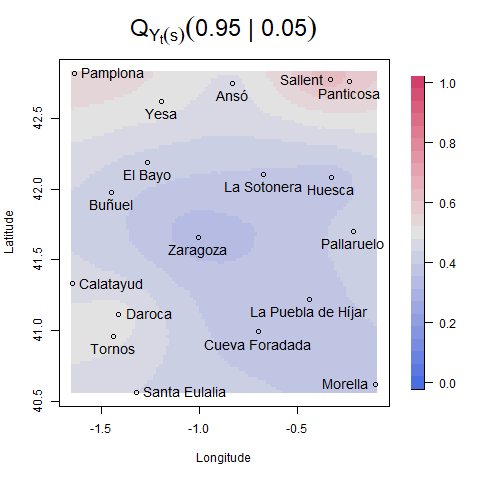} \\
\includegraphics[width=4cm]{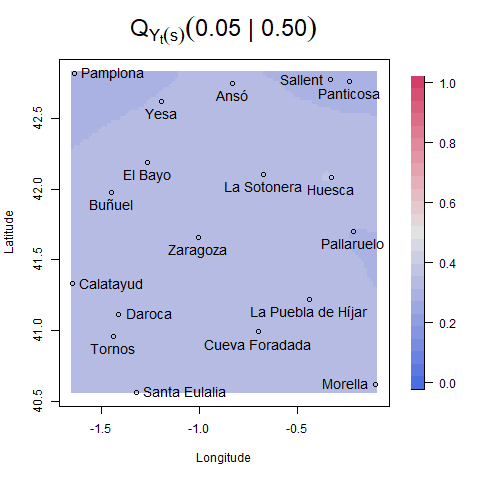}
\includegraphics[width=4cm]{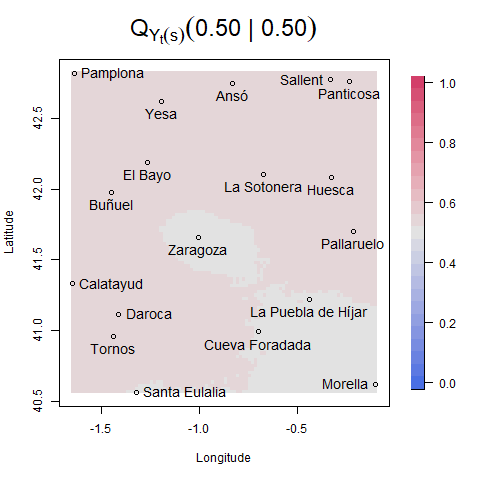}
\includegraphics[width=4cm]{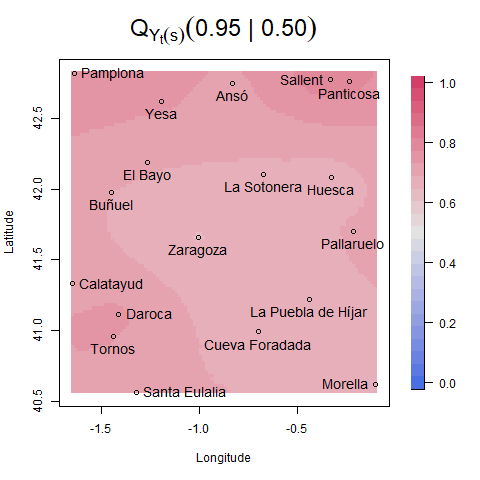} \\
\includegraphics[width=4cm]{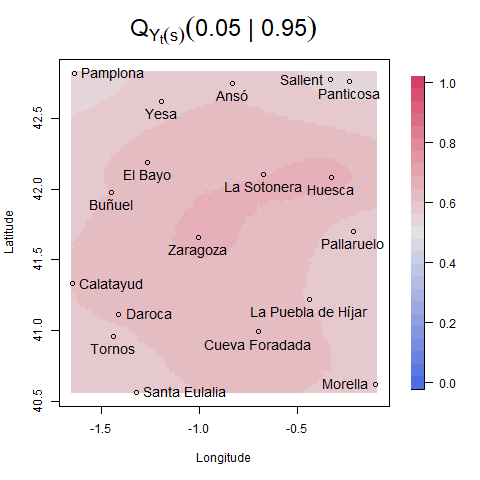}
\includegraphics[width=4cm]{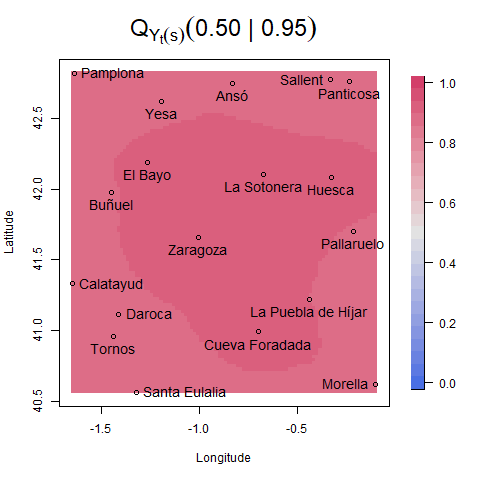}
\includegraphics[width=4cm]{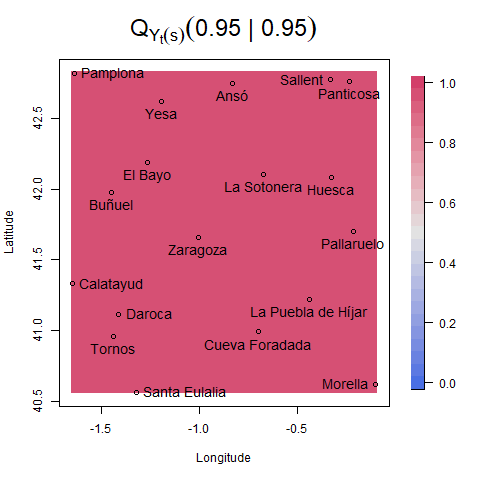}
\caption{Maps of the posterior mean of $Q_{Y_t(\bs)}(\tau \mid y)$ for $\tau = 0.05,0.50,0.95$ and $y = 0.05,0.50,0.95$. } \label{fig:meanQ}
\end{figure}

\clearpage

\bibliographystyle{sn-chicago}
\bibliography{arxiv_v1}